\let\oldvec\vec
\documentclass[runningheads]{llncs}
\usepackage{bcprules}\typicallabel{T-Hoge}
\usepackage{booktabs}
\usepackage{bcpproof}
\usepackage{graphicx,xcolor}

\usepackage{amsmath}
\usepackage{amssymb}
\usepackage{stmaryrd}
\usepackage{parcolumns}
\usepackage{cases}
\usepackage{listings}%
\usepackage{etoolbox}
\usepackage{microtype}
\usepackage{wrapfig}
\usepackage{paralist}
\usepackage[numbers,sectionbib,sort]{natbib}

\DeclareMathOperator{\SEQ}{;}

\DeclareMathOperator{\LET}{\mathbf{let}}
\DeclareMathOperator{\IN}{\mathbf{in}}
\DeclareMathOperator{\WRITE}{:=}

\DeclareMathOperator\IFZERO{\mathbf{ifz}}
\DeclareMathOperator\THEN{\mathbf{then}}
\DeclareMathOperator\ELSE{\mathbf{else}}
\DeclareMathOperator{\COL}{:}
\newcommand\ASSERT{\mathbf{assert}}

\newcommand\HOLE{[]}

\newcommand\wf{\mathit{WF}}

\let\name\consort

\newcommand{\DOM}{\mathit{dom}}
\newcommand{\CVar}{\textbf{CVar}}

\let\imp\lstinline

\newcommand\tuple[1]{\left\langle{#1}\right\rangle}

\newcommand\TINT{\mathbf{int}}
\DeclareMathOperator\TREF{\mathbf{ref}}

\newcommand\set[1]{\left\{{#1}\right\}}

\DeclareMathOperator{\produces}{\Rightarrow}
\DeclareMathOperator{\MKREF}{\mathbf{mkref}}

\newcommand\tenv{\Gamma}

\DeclareMathOperator{\ra}{\rightarrow}

\newcommand\ALIAS{\mathbf{alias}}

\newcommand\sem[1]{\left\llbracket{#1}\right\rrbracket}

\newif\ifdraftComments
\draftCommentstrue
\def\mkDraftFn#1#2{%
  \expandafter\def\csname #1\endcsname##1{\ifdraftComments\textcolor{#2}{[#1: ##1]}\marginpar[$\longrightarrow$]{$\longleftarrow$}\fi}%
}
\mkDraftFn{JT}{blue}
\mkDraftFn{AI}{red}
\mkDraftFn{KS}{purple}
\mkDraftFn{NK}{cyan}
\mkDraftFn{RS}{teal}
\usepackage{xspace}
\makeatletter
\def\needcite@with[#1]{\ifdraftComments\textcolor{blue}{[citation needed #1]}\else\empty\fi\xspace}
\def\needcite@bare{\ifdraftComments\textcolor{blue}{[citation needed]}\else\empty\fi\xspace}
\def\needcite{\@ifnextchar[{\needcite@with}{\needcite@bare}}
\makeatother

\lstdefinelanguage{Imp}{
  keywords=[0]{ifz,then,else,alias,assert,mkref,let,in,null,ifnull,return},
  morecomment=[l]{//},
  morecomment=[s]{/*}{*/}
}
\definecolor{comment-green}{rgb}{0,0.6,0}
\draftCommentsfalse
\newif\iffullversion

\fullversiontrue

\def\appref#1#2{\iffullversion #1\else#2\fi}

\ifdraftComments\else\linepenalty=1000\fi

\ifdraftComments
\makeatletter
\let\llncs@addcontentsline\addcontentsline
\patchcmd{\maketitle}{\addcontentsline}{\llncs@addcontentsline}{}{}
\patchcmd{\maketitle}{\addcontentsline}{\llncs@addcontentsline}{}{}
\patchcmd{\maketitle}{\addcontentsline}{\llncs@addcontentsline}{}{}
\makeatother
\fi

\iffullversion
\makeatletter
\let\llncs@addcontentsline\addcontentsline
\patchcmd{\maketitle}{\addcontentsline}{\llncs@addcontentsline}{}{}
\patchcmd{\maketitle}{\addcontentsline}{\llncs@addcontentsline}{}{}
\patchcmd{\maketitle}{\addcontentsline}{\llncs@addcontentsline}{}{}
\makeatother
\fi

\makeatletter
\RequirePackage[bookmarks,unicode,colorlinks=true]{hyperref}%
   \def\@citecolor{blue}%
   \def\@urlcolor{blue}%
   \def\@linkcolor{blue}%

\def\orcidID#1{\smash{\href{http://orcid.org/#1}{\protect\raisebox{-1.25pt}{\protect\includegraphics{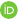}}}}}
\makeatother

\ifdraftComments
\hypersetup{colorlinks,
  linkcolor=red,
  citecolor=red}
\fi
\usepackage[noabbrev]{cleveref}

\newcommand{\ottnt}[1]{\mathit{#1}}
\newcommand{\ottmv}[1]{\mathit{#1}}
\newcommand{\ottkw}[1]{\mathbf{#1}}
\newcommand{\ottsym}[1]{#1}

\def\jhaliascount{6}

\def\jhmaxalias{3}

\begin{document}
\title{\name: Context- and Flow-Sensitive Ownership Refinement Types for Imperative Programs}
\titlerunning{\name: Context- and Flow-Sensitive Ownership Refinement Types}
\author{
  John Toman\inst{1} \and
  Ren Siqi\inst{1} \and
  Kohei Suenaga\inst{1}\orcidID{0000-0002-7466-8789} \and \\
  Atsushi Igarashi\inst{1}\orcidID{0000-0002-5143-9764} \and
  Naoki Kobayashi\inst{2}\orcidID{0000-0002-0537-0604}
}
\authorrunning{J. Toman et al.}
\institute{
  Kyoto University, Kyoto, Japan, \email{\{jtoman,shiki,ksuenaga,igarashi\}@fos.kuis.kyoto-u.ac.jp}
  \and
  The University of Tokyo, Tokyo, Japan, \email{koba@is.s.u-tokyo.ac.jp}
}
\maketitle              %
\begin{abstract}
  We present \name, a type system for safety verification in the
  presence of mutability and aliasing. Mutability requires
  \emph{strong updates} to model changing invariants
  during program execution, but aliasing between pointers
  makes it difficult
  to determine which invariants must be updated
  in response to mutation.
  Our type system addresses this difficulty with a
  novel combination of refinement types and
  fractional ownership types. Fractional ownership types
  provide flow-sensitive and precise aliasing information
  for reference variables. \name interprets this ownership information
  to soundly handle strong updates of potentially aliased references.
  We have proved \name sound and implemented a prototype,
  fully automated inference tool. We evaluated our tool
  and found it verifies non-trivial programs including data structure implementations.
\end{abstract}
\section{Introduction}
\label{sec:intro}
Driven by the increasing power of automated theorem provers and
recent high-profile software failures,
fully automated program verification has seen a surge of interest
in recent years \cite{leino2010dafny,cousot2005astree,ball2011decade,zave2012using,hawblitzel2015ironfleet,bhargavan2017everest}.
In particular, \emph{refinement types}
\cite{flanagan2006hybrid,xi1999dependent,freeman1991refinement,bengtson2011refinement},
which refine base types with logical predicates, have been shown to be a practical
approach for program verification that are amenable to
(sometimes full) automation \cite{rondon2008liquid,vazou2014refinement,vazou2013abstract,unno2009dependent}.
Despite promising advances \cite{rondon2010low,kahsai2017quantified,gordon2013rely},
the sound and precise application of refinement types (and program
verification in general) in settings with mutability and aliasing
(e.g., Java, Ruby, etc.) remains difficult.

\begin{figure}[t]
  \begin{minipage}[t]{0.45\textwidth}
  \begin{lstlisting}[numbers=left, numbersep=3pt, numberstyle=\tiny\color{black},numberblanklines=false]
mk(n) { mkref n }$\label{line:mk}$

let p = mk(3) in$\label{line:main-start}$
let q = mk(5) in
p := *p + 1;$\label{line:first-p-write}$
q := *q + 1;$\label{line:r-write}$
assert(*p = 4);$\label{line:rq-eq-assert}$
  \end{lstlisting}
  \caption{Example demonstrating the difficulty of effecting strong updates in the presence of aliasing. The function \imp{mk} is
    bound in the program from \crefrange{line:main-start}{line:rq-eq-assert}; its body is given within the braces.}
  \label{fig:strong-update-example}
\end{minipage}
\hfill
\begin{minipage}[t]{0.45\textwidth}
  \begin{lstlisting}[numbers=left, numbersep=3pt, numberstyle=\tiny\color{black},numberblanklines=false]
loop(a, b) {
  let aold = *a in
  b := *b + 1;$\label{line:b-write}$
  a := *a + 1;$\label{line:a-write}$
  assert(*a = aold + 1);
  if $\star$ then
    loop(b, mkref $\star$)$\label{line:alloc1}$
  else
    loop(b,a)$\label{line:call2}$
}
loop(mkref $\star$, mkref $\star$)$\label{line:start}$
\end{lstlisting}
  \caption{Example with non-trivial aliasing behavior.}
  \label{fig:hard-loop}
\end{minipage}
\end{figure}

One of the major challenges is how to precisely and soundly support
\emph{strong updates} for the invariants on memory cells.
In a setting with mutability, a single invariant may not necessarily
hold throughout the lifetime of a memory cell; while the program
mutates the memory the invariant may change or evolve.
To model these changes, a program verifier
must support different, incompatible invariants which hold at
different points during program execution. Further, precise
program verification requires supporting different invariants
on distinct pieces of memory.

One solution is to use refinement types on the static program names
(i.e., variables) which point to a memory location.
This approach can model evolving invariants
while tracking distinct invariants for each memory cell.
For example,
consider the (contrived) example in \Cref{fig:strong-update-example}.
This program is written in an ML-like language with mutable references;
references are updated with \imp{:=} and allocated with \imp{mkref}.
Variable \imp{p} can initially be given the type
$ \set{  \nu  \COL \TINT \mid  \nu \, \ottsym{=} \, \ottsym{3} } \TREF$, indicating it is a
reference to the integer 3.
Similarly, \imp{q} can be given the type $ \set{  \nu  \COL \TINT \mid  \nu \, \ottsym{=} \, \ottsym{5} }  \TREF$.
We can model the mutation of \imp{p}'s memory on \cref{line:first-p-write}
by strongly updating \imp{p}'s type to $ \set{  \nu  \COL \TINT \mid  \nu \, \ottsym{=} \, \ottsym{4} }  \TREF$.

Unfortunately, the precise application of this technique is confounded
by the existence of unrestricted aliasing.
In general, updating just the type of the mutated reference
is insufficient: due to aliasing, other variables may point to the mutated
memory and their refinements must be updated as well. However,
in the presence of conditional, \emph{may} aliasing, it is
impossible to strongly update the refinements on all possible
aliases; given the static uncertainty
about whether a variable points to the mutated memory, that variable's
refinement may only be \emph{weakly updated}. For example,
suppose we used a simple alias analysis that imprecisely (but soundly)
concluded all references allocated at the same program point \emph{might} alias.
Variables \imp{p} and \imp{q} share the allocation site on
\cref{line:mk}, so on \cref{line:first-p-write} we would have to weakly update
\imp{q}'s type to $ \set{  \nu  \COL \TINT \mid  \nu \, \ottsym{=} \, \ottsym{4}  \vee  \nu \, \ottsym{=} \, \ottsym{5} } $,
indicating it may hold either 4 \emph{or} 5.
Under this same imprecise aliasing assumption,
we would also have to weakly update
\imp{p}'s type on \cref{line:r-write},
preventing the verification of the example program.

Given the precision loss associated with weak updates,
it is critical that verification techniques built upon
refinement types use precise aliasing information and
avoid spuriously applied weak updates. Although it
is relatively simple to conclude that \imp{p} and \imp{q}
do not alias in \Cref{fig:strong-update-example}, consider
the example in \Cref{fig:hard-loop}. (In this example,
$\star$ represents non-deterministic values.) Verifying
this program requires proving \imp{a} and \imp{b} never
alias at the writes on \cref{line:b-write,line:a-write}.
In fact, \imp{a} and \imp{b} \emph{may} point to the same
memory location, but only in different invocations of \imp{loop};
this pattern may confound even sophisticated symbolic alias analyses.
Additionally, \imp{a} and \imp{b} share an allocation site on \cref{line:alloc1},
so an approach based on the simple alias analysis described above will also fail on
this example. This must-not alias proof obligation  \emph{can} be discharged with existing
techniques \cite{spath2016boomerang,spath2019context}, but requires an expensive, on-demand,
interprocedural, flow-sensitive alias analysis. 

This paper presents \name (CONtext Sensitive Ownership Refinement Types),
a type system for the automated verification
of program safety in imperative languages with mutability and aliasing.
\name is built upon the novel combination of refinement types and
fractional ownership types \cite{suenaga2009fractional,suenaga2012type}.
Fractional ownership types extend pointer types with a rational number in
the range $[0,1]$ called an \emph{ownership}.
These ownerships encapsulate the permission of the reference; only
references with ownership $1$ may be used for mutation. Fractional ownership
types also obey the following key invariant: any references with a mutable
alias must have ownership 0.
Thus, any reference with non-zero ownership \emph{cannot}
be an alias of a reference with ownership $1$. In other words,
ownerships encode precise aliasing information in the form of
\emph{must-not} aliasing relationships.\looseness=-1

To understand the benefit of this approach, let us return to
\Cref{fig:strong-update-example}. As \imp{mk} returns a freshly allocated
reference with no aliases, its type indicates it returns a
reference with ownership 1. Thus, our type system can initially
give \imp{p} and \imp{q} types $  \set{  \nu  \COL \TINT \mid  \nu \, \ottsym{=} \, \ottsym{3} }   \TREF^{ \ottsym{1} } $ and $  \set{  \nu  \COL \TINT \mid  \nu \, \ottsym{=} \, \ottsym{5} }   \TREF^{ \ottsym{1} } $ respectively. The ownership $\ottsym{1}$ on the reference type constructor $\TREF$ indicates both
pointers hold ``exclusive'' ownership of the pointed to reference cell;
from the invariant of fractional ownership types \imp{p} and \imp{q}
must \emph{not} alias. The types of both references
can be strongly updated \emph{without} requiring spurious weak updates. As a result,
at the assertion statement on \cref{line:rq-eq-assert}, \imp{p} has type
$  \set{  \nu  \COL \TINT \mid  \nu \, \ottsym{=} \, \ottsym{4} }   \TREF^{ \ottsym{1} } $ expressing the required invariant.\looseness=-1

Our type system can also verify the example in
\Cref{fig:hard-loop} \emph{without} expensive
side analyses. As \imp{a} and \imp{b} are both
mutated, they must both have ownership 1; i.e., they cannot alias.
This pre-condition is satisfied by all invocations of \imp{loop};
on \cref{line:alloc1}, \imp{b} has
ownership 1 (from the argument type), and the newly allocated
reference must also have ownership 1.
Similarly, both arguments on \cref{line:call2} have ownership $\ottsym{1}$
(from the assumed ownership on the argument types). %

Ownerships behave linearly;
they cannot be duplicated, only \emph{split}
when aliases are created. This linear behavior
preserves the critical ownership invariant. For example, if
we replace \cref{line:call2} in \Cref{fig:hard-loop} with
\imp{loop(b,b)}, the program becomes ill-typed; there
is no way to divide \imp{b}'s ownership of 1 to into \emph{two}
ownerships of 1.

Ownerships also obviate updating refinement information
of aliases at mutation. \name ensures that only the trivial refinement
$ \top $ is used in reference types with ownership $0$, i.e., mutably-aliased references.
When memory is mutated through a reference with ownership $1$,
\name simply updates the refinement of the mutated reference variable.
From the soundness of ownership types, all aliases have ownership 0
and must therefore only contain the $ \top $ refinement. Thus, the types of all aliases already
soundly describe \emph{all} possible contents.\footnote{This
  assumption holds only if updates do not change simple types, a condition our type-system enforces.}

\name is also \emph{context-sensitive}, and can
use different summaries of function behavior
at different points in the program.
For example, consider the variant
\begin{wrapfigure}{l}{0.29\textwidth}
  \begin{lstlisting}[numbers=left,numbersep=3pt,numberstyle=\tiny\color{black},numberblanklines=false]
get(p) { *p }

let p = mkref 3 in
let q = mkref 5 in
p := get(p) + 1;$\label{line:get1}$
q := get(q) + 1;$\label{line:get2}$
assert(*p = 4);
assert(*q = 6);
  \end{lstlisting}
  \caption{Example of context-sensitivity}
  \label{fig:running-example}
\end{wrapfigure}
of \Cref{fig:strong-update-example}
shown in \Cref{fig:running-example}.
The function \imp{get} returns the contents of its argument,
and is called on \cref{line:get1,line:get2}. To precisely verify
this program, on \cref{line:get1}
\imp{get} must be typed as a function that takes a reference to
3 and returns 3. Similarly, on \cref{line:get2}
\imp{get} must be typed as a function that
takes a reference to 5 and returns 5. Our type system can
give \imp{get} a function type that distinguishes between these two calling
contexts and selects the appropriate summary of \imp{get}'s behavior.

We have formalized \name as a type system for a small imperative calculus
and proved the system is sound: i.e., a well-typed program never encounters assertion
failures during execution. We have
implemented a prototype type inference tool targeting this
imperative language and found it can automatically verify several
non-trivial programs, including sorted lists and an array list data structure.

The rest of this paper is organized as follows. \Cref{sec:prelim} defines the
imperative language targeted by \name and its semantics. \Cref{sec:typesystem}
defines our type system and states our soundness theorem. \Cref{sec:infr} sketches
our implementation's inference algorithm and its current limitations.
\Cref{sec:eval} describes an evaluation of our prototype, \Cref{sec:rw} outlines
related work, and \Cref{sec:concl} concludes.\looseness=-1

\section{Target Language}
\label{sec:prelim}
This section describes a simple imperative language with mutable references and first-order, recursive functions.

\subsection{Syntax}
\label{sec:language}
We assume a set of \emph{variables}, ranged over by $x,y,z,\dots$,
a set of \emph{function names}, ranged over by $f$, and
a set of \emph{labels}, ranged over by $\ell_{{\mathrm{1}}}, \ell_{{\mathrm{2}}}, \dots$.
The grammar of the language is as follows.
\[
  \begin{array}{rcl}
  \ottnt{d} &::=& \mathit{f}  \mapsto  \ottsym{(}  \mathit{x_{{\mathrm{1}}}}  \ottsym{,} \, ... \, \ottsym{,}  \mathit{x_{\ottmv{n}}}  \ottsym{)}  \ottnt{e} \\
  \ottnt{e} &::= &
               \mathit{x} \mid
                \LET  \mathit{x}  =  \mathit{y}  \IN  \ottnt{e}  \mid
                \LET  \mathit{x}  =  n  \IN  \ottnt{e}  \mid
                \IFZERO  \mathit{x}  \THEN  \ottnt{e_{{\mathrm{1}}}}  \ELSE  \ottnt{e_{{\mathrm{2}}}}  \\
        &\mid&  \LET  \mathit{x}  =   \MKREF  \mathit{y}   \IN  \ottnt{e}  \mid
                \LET  \mathit{x}  =   *  \mathit{y}   \IN  \ottnt{e}  \mid
                \LET  \mathit{x}  =   \mathit{f} ^ \ell (  \mathit{y_{{\mathrm{1}}}} ,\ldots, \mathit{y_{\ottmv{n}}}  )   \IN  \ottnt{e}  \\
        &\mid&  \mathit{x}  \WRITE  \mathit{y}  \SEQ  \ottnt{e}  \mid
                \ALIAS( \mathit{x}  =  \mathit{y} ) \SEQ  \ottnt{e}  \mid
                \ALIAS( \mathit{x}  = *  \mathit{y} ) \SEQ  \ottnt{e}  \mid
                \ASSERT( \varphi ) \SEQ  \ottnt{e}  \mid
                \ottnt{e_{{\mathrm{1}}}}  \SEQ  \ottnt{e_{{\mathrm{2}}}}  \\
    \ottnt{P} &::=&  \tuple{ \ottsym{\{}  \ottnt{d_{{\mathrm{1}}}}  \ottsym{,} \, ... \, \ottsym{,}  \ottnt{d_{\ottmv{n}}}  \ottsym{\}} ,  \ottnt{e} }  %
  \end{array}
\]
$\varphi$ stands for a formula in propositional first-order
logic over variables, integers and contexts;
we discuss these formulas later in \Cref{sec:types}.

Variables are introduced by function parameters or let bindings.
Like ML, the variable bindings introduced by let expressions and parameters
are immutable.
Mutable variable declarations such as \texttt{int x = 1;} in C are achieved in our language with:\looseness=-1
\[
   \LET  \mathit{y}  =  \ottsym{1}  \IN  \ottsym{(}   \LET  \mathit{x}  =   \MKREF  \mathit{y}   \IN   \ldots    \ottsym{)} \ .
\]
As a convenience, we assume all variable names
introduced with let bindings and function parameters are distinct.

Unlike ML (and like C or Java) we do not allow general expressions on the right hand side of
let bindings. The simplest right hand forms are a variable $\mathit{y}$ or an integer literal $n$.
$ \MKREF  \mathit{y} $ creates a reference cell with value $\mathit{y}$, and
$ *  \mathit{y} $ accesses the contents of reference $\mathit{y}$. For simplicity,
we do not include an explicit null value;
an extension to support null is discussed in \Cref{sec:infr}.
Function calls must occur on the right hand side of a variable binding
and take the form $ \mathit{f} ^ \ell (  \mathit{x_{{\mathrm{1}}}} ,\ldots, \mathit{x_{\ottmv{n}}}  ) $, where $ \mathit{x_{{\mathrm{1}}}} ,\ldots, \mathit{x_{\ottmv{n}}} $ are distinct variables and $\ell$ is a
(unique) label. %
These labels are used to make our type system context-sensitive as
discussed in \Cref{sec:cs}.

The single base case for expressions is a single variable.
If the variable expression is executed in a tail position of a function,
then the value of that variable is the return value of the function,
otherwise the value is ignored.

The only intraprocedural control-flow operations in our language are if statements.
$\ottkw{ifz}$ checks whether the condition variable $\mathit{x}$ equals zero and chooses
the corresponding branch. Loops can be implemented with recursive functions and
we do not include them explicitly in our formalism.

Our grammar requires that side-effecting, result-free statements, \lstinline{assert}($\varphi$)
\lstinline{alias}($\mathit{x}$ = $\mathit{y}$), \lstinline{alias}($\mathit{x} = *\mathit{y}$) and
assignment $\mathit{x} := \mathit{y}$ are followed by a continuation expression.
We impose this requirement for technical reasons to ease our formal
presentation; this requirement does not reduce expressiveness
as dummy continuations can be inserted as needed.
The $ \ASSERT( \varphi ) \SEQ  \ottnt{e} $ form executes $\ottnt{e}$ if the predicate
$\varphi$ holds in the current state and aborts the program otherwise.
$ \ALIAS( \mathit{x}  =  \mathit{y} ) \SEQ  \ottnt{e} $ and $ \ALIAS( \mathit{x}  = *  \mathit{y} ) \SEQ  \ottnt{e} $ assert a must-aliasing relationship between $x$ and $y$ (resp. $\mathit{x}$ and $ *  \mathit{y} $)
and then execute $\ottnt{e}$.
\imp{alias} statements
are effectively \emph{annotations} that our type system exploits to gain
added precision.
$ \mathit{x}  \WRITE  \mathit{y}  \SEQ  \ottnt{e} $ updates the contents of the memory cell pointed to by $\mathit{x}$
with the value of $\mathit{y}$.
In addition to the above continuations, our language supports general sequencing with
$ \ottnt{e_{{\mathrm{1}}}}  \SEQ  \ottnt{e_{{\mathrm{2}}}} $.

A program is a pair $ \tuple{ \ottnt{D} ,  \ottnt{e} } $, where
$\ottnt{D}  \ottsym{=}  \ottsym{\{}  \ottnt{d_{{\mathrm{1}}}}  \ottsym{,} \, ... \, \ottsym{,}  \ottnt{d_{\ottmv{n}}}  \ottsym{\}}$ is a set of first-order, mutually recursive function definitions, and $\ottnt{e}$
is the program entry point. A function definition $\ottnt{d}$
maps the function name to a tuple of argument names $\mathit{x_{{\mathrm{1}}}}  \ottsym{,} \, ... \, \ottsym{,}  \mathit{x_{\ottmv{n}}}$
that are
bound within the function body $\ottnt{e}$.

\paragraph{Paper Syntax.}
In the remainder of the paper, we will write programs that
are technically illegal according to our grammar, but can be
easily ``de-sugared'' into an equivalent, valid program.
For example, we will write
\begin{lstlisting}
  let x = mkref 4 in assert(*x = 4)
\end{lstlisting}
as syntactic sugar for:
\begin{lstlisting}
  let f = 4 in let x = mkref f in
  let tmp = *x in assert(tmp = 4); let dummy = 0 in dummy
\end{lstlisting}

\bcprulessavespacetrue
\begin{figure}[t]
  \scriptsize
  \begin{center}
  \infrule[R-Var]{
  }{
      \tuple{ \ottnt{H} ,  \ottnt{R} ,  F  \ottsym{:}  \oldvec{F} ,  \mathit{x} }     \longrightarrow _{ \ottnt{D} }     \tuple{ \ottnt{H} ,  \ottnt{R} ,  \oldvec{F} ,  F  \ottsym{[}  \mathit{x}  \ottsym{]} }  
   }
   \infrule[R-Seq]{
   }{
       \tuple{ \ottnt{H} ,  \ottnt{R} ,  F  \ottsym{:}  \oldvec{F} ,  \ottnt{E}  \ottsym{[}   \mathit{x}  \SEQ  \ottnt{e}   \ottsym{]} }     \longrightarrow _{ \ottnt{D} }     \tuple{ \ottnt{H} ,  \ottnt{R} ,  \oldvec{F} ,  \ottnt{E}  \ottsym{[}  \ottnt{e}  \ottsym{]} }  
   }
  \infrule[R-Let]{
     \mathit{x'}  \not\in \DOM( \ottnt{R} ) 
  }{
     \begin{array}{l}  \tuple{ \ottnt{H} ,  \ottnt{R} ,  \oldvec{F} ,  \ottnt{E}  \ottsym{[}   \LET  \mathit{x}  =  \mathit{y}  \IN  \ottnt{e}   \ottsym{]} }   \\ \quad   \longrightarrow _{ \ottnt{D} }     \tuple{ \ottnt{H} ,  \ottnt{R}  \ottsym{\{}  \mathit{x'}  \mapsto  \ottnt{R}  \ottsym{(}  \mathit{y}  \ottsym{)}  \ottsym{\}} ,  \oldvec{F} ,  \ottnt{E}  \ottsym{[}    [  \mathit{x'}  /  \mathit{x}  ]    \ottnt{e}   \ottsym{]} }  \end{array} 
  }
  \infrule[R-LetInt]{
     \mathit{x'}  \not\in \DOM( \ottnt{R} ) 
  }{
     \begin{array}{l}  \tuple{ \ottnt{H} ,  \ottnt{R} ,  \oldvec{F} ,  \ottnt{E}  \ottsym{[}   \LET  \mathit{x}  =  n  \IN  \ottnt{e}   \ottsym{]} }   \\ \quad   \longrightarrow _{ \ottnt{D} }     \tuple{ \ottnt{H} ,  \ottnt{R}  \ottsym{\{}  \mathit{x'}  \mapsto  n  \ottsym{\}} ,  \oldvec{F} ,  \ottnt{E}  \ottsym{[}    [  \mathit{x'}  /  \mathit{x}  ]    \ottnt{e}   \ottsym{]} }  \end{array} 
  }
  \infrule[R-IfTrue]{
    \ottnt{R}  \ottsym{(}  \mathit{x}  \ottsym{)} \, \ottsym{=} \, \ottsym{0}
  }{
     \begin{array}{l}  \tuple{ \ottnt{H} ,  \ottnt{R} ,  \oldvec{F} ,  \ottnt{E}  \ottsym{[}   \IFZERO  \mathit{x}  \THEN  \ottnt{e_{{\mathrm{1}}}}  \ELSE  \ottnt{e_{{\mathrm{2}}}}   \ottsym{]} }   \\ \quad   \longrightarrow _{ \ottnt{D} }     \tuple{ \ottnt{H} ,  \ottnt{R} ,  \oldvec{F} ,  \ottnt{E}  \ottsym{[}  \ottnt{e_{{\mathrm{1}}}}  \ottsym{]} }  \end{array} 
  }
  \infrule[R-IfFalse]{
    \ottnt{R}  \ottsym{(}  \mathit{x}  \ottsym{)} \, \neq \, \ottsym{0}
  }{
     \begin{array}{l}  \tuple{ \ottnt{H} ,  \ottnt{R} ,  \oldvec{F} ,  \ottnt{E}  \ottsym{[}   \IFZERO  \mathit{x}  \THEN  \ottnt{e_{{\mathrm{1}}}}  \ELSE  \ottnt{e_{{\mathrm{2}}}}   \ottsym{]} }   \\ \quad   \longrightarrow _{ \ottnt{D} }     \tuple{ \ottnt{H} ,  \ottnt{R} ,  \oldvec{F} ,  \ottnt{E}  \ottsym{[}  \ottnt{e_{{\mathrm{2}}}}  \ottsym{]} }  \end{array} 
  }
  \infrule[R-MkRef]{
     \ottmv{a}  \not\in \DOM( \ottnt{H} )  \andalso  \mathit{x'}  \not\in \DOM( \ottnt{R} ) 
  }{
     \begin{array}{r}  \tuple{ \ottnt{H} ,  \ottnt{R} ,  \oldvec{F} ,  \ottnt{E}  \ottsym{[}   \LET  \mathit{x}  =   \MKREF  \mathit{y}   \IN  \ottnt{e}   \ottsym{]} }     \longrightarrow _{ \ottnt{D} }   \\   \tuple{ \ottnt{H}  \ottsym{\{}  \ottmv{a}  \mapsto  \ottnt{R}  \ottsym{(}  \mathit{y}  \ottsym{)}  \ottsym{\}} ,  \ottnt{R}  \ottsym{\{}  \mathit{x'}  \mapsto  \ottmv{a}  \ottsym{\}} ,  \oldvec{F} ,  \ottnt{E}  \ottsym{[}    [  \mathit{x'}  /  \mathit{x}  ]    \ottnt{e}   \ottsym{]} }  \end{array} 
  }
  \infrule[R-Deref]{
    \ottnt{R}  \ottsym{(}  \mathit{y}  \ottsym{)} \, \ottsym{=} \, \ottmv{a} \andalso \ottnt{H}  \ottsym{(}  \ottmv{a}  \ottsym{)} \, \ottsym{=} \, \ottnt{v} \andalso  \mathit{x'}  \not\in \DOM( \ottnt{R} ) 
  }{
     \begin{array}{r}  \tuple{ \ottnt{H} ,  \ottnt{R} ,  \oldvec{F} ,  \ottnt{E}  \ottsym{[}   \LET  \mathit{x}  =   *  \mathit{y}   \IN  \ottnt{e}   \ottsym{]} }     \longrightarrow _{ \ottnt{D} }   \\   \tuple{ \ottnt{H} ,  \ottnt{R}  \ottsym{\{}  \mathit{x'}  \mapsto  \ottnt{v}  \ottsym{\}} ,  \oldvec{F} ,  \ottnt{E}  \ottsym{[}    [  \mathit{x'}  /  \mathit{x}  ]    \ottnt{e}   \ottsym{]} }  \end{array} 
  }
\end{center}
\caption{Transition Rules (1).}
\label{fig:transitionRules1}
\end{figure}
\begin{figure}[t]
  \scriptsize
  \begin{center}
  \infrule[R-Call]{
     \mathit{f}  \mapsto  \ottsym{(}  \mathit{x_{{\mathrm{1}}}}  \ottsym{,} \, .. \, \ottsym{,}  \mathit{x_{\ottmv{n}}}  \ottsym{)}  \ottnt{e}  \in  \ottnt{D} 
  }{
     \begin{array}{l}  \tuple{ \ottnt{H} ,  \ottnt{R} ,  \oldvec{F} ,  \ottnt{E}  \ottsym{[}   \LET  \mathit{x}  =   \mathit{f} ^ \ell (  \mathit{y_{{\mathrm{1}}}} ,\ldots, \mathit{y_{\ottmv{n}}}  )   \IN  \ottnt{e'}   \ottsym{]} }   \\ \quad   \longrightarrow _{ \ottnt{D} }     \tuple{ \ottnt{H} ,  \ottnt{R} ,   \ottnt{E} [\LET  \mathit{x}  =   \HOLE^ \ell   \IN  \ottnt{e'}  ]   \ottsym{:}  \oldvec{F} ,     [  \mathit{y_{{\mathrm{1}}}}  /  \mathit{x_{{\mathrm{1}}}}  ]  \cdots  [  \mathit{y_{\ottmv{n}}}  /  \mathit{x_{\ottmv{n}}}  ]     \ottnt{e}  }  \end{array} 
  }
  \infrule[R-Assign]{
    \ottnt{R}  \ottsym{(}  \mathit{x}  \ottsym{)} \, \ottsym{=} \, \ottmv{a} \andalso  \ottmv{a}  \in \DOM( \ottnt{H} ) 
  }{
     \begin{array}{r}  \tuple{ \ottnt{H} ,  \ottnt{R} ,  \oldvec{F} ,  \ottnt{E}  \ottsym{[}   \mathit{x}  \WRITE  \mathit{y}  \SEQ  \ottnt{e}   \ottsym{]} }     \longrightarrow _{ \ottnt{D} }   \\   \tuple{ \ottnt{H}  \ottsym{\{}  \ottmv{a}  \hookleftarrow  \ottnt{R}  \ottsym{(}  \mathit{y}  \ottsym{)}  \ottsym{\}} ,  \ottnt{R} ,  \oldvec{F} ,  \ottnt{E}  \ottsym{[}  \ottnt{e}  \ottsym{]} }  \end{array} 
  }
  \infrule[R-Alias]{
    \ottnt{R}  \ottsym{(}  \mathit{x}  \ottsym{)} \, \ottsym{=} \, \ottnt{R}  \ottsym{(}  \mathit{y}  \ottsym{)}
  }{
     \begin{array}{l}  \tuple{ \ottnt{H} ,  \ottnt{R} ,  \oldvec{F} ,  \ottnt{E}  \ottsym{[}   \ALIAS( \mathit{x}  =  \mathit{y} ) \SEQ  \ottnt{e}   \ottsym{]} }   \\ \quad   \longrightarrow _{ \ottnt{D} }     \tuple{ \ottnt{H} ,  \ottnt{R} ,  \oldvec{F} ,  \ottnt{E}  \ottsym{[}  \ottnt{e}  \ottsym{]} }  \end{array} 
  }
  \infrule[R-AliasPtr]{
    \ottnt{R}  \ottsym{(}  \mathit{y}  \ottsym{)} \, \ottsym{=} \, \ottmv{a} \andalso \ottnt{H}  \ottsym{(}  \ottmv{a}  \ottsym{)} \, \ottsym{=} \, \ottnt{R}  \ottsym{(}  \mathit{x}  \ottsym{)}
  }{
      \tuple{ \ottnt{H} ,  \ottnt{R} ,  \oldvec{F} ,  \ottnt{E}  \ottsym{[}   \ALIAS( \mathit{x}  = *  \mathit{y} ) \SEQ  \ottnt{e}   \ottsym{]} }     \longrightarrow _{ \ottnt{D} }     \tuple{ \ottnt{H} ,  \ottnt{R} ,  \oldvec{F} ,  \ottnt{E}  \ottsym{[}  \ottnt{e}  \ottsym{]} }  
  }
  \infrule[R-AliasFail]{
    \ottnt{R}  \ottsym{(}  \mathit{x}  \ottsym{)} \, \neq \, \ottnt{R}  \ottsym{(}  \mathit{y}  \ottsym{)}
  }{
      \tuple{ \ottnt{H} ,  \ottnt{R} ,  \oldvec{F} ,  \ottnt{E}  \ottsym{[}   \ALIAS( \mathit{x}  =  \mathit{y} ) \SEQ  \ottnt{e}   \ottsym{]} }     \longrightarrow _{ \ottnt{D} }     \mathbf{AliasFail}  
  }
  \infrule[R-AliasPtrFail]{
    \ottnt{R}  \ottsym{(}  \mathit{x}  \ottsym{)} \, \neq \, \ottnt{H}  \ottsym{(}  \ottnt{R}  \ottsym{(}  \mathit{y}  \ottsym{)}  \ottsym{)}
  }{
      \tuple{ \ottnt{H} ,  \ottnt{R} ,  \oldvec{F} ,  \ottnt{E}  \ottsym{[}   \ALIAS( \mathit{x}  = *  \mathit{y} ) \SEQ  \ottnt{e}   \ottsym{]} }     \longrightarrow _{ \ottnt{D} }     \mathbf{AliasFail}  
  }
  \infrule[R-Assert]{
    \models  \ottsym{[}  \ottnt{R}  \ottsym{]} \, \varphi
  }{
     \begin{array}{l}  \tuple{ \ottnt{H} ,  \ottnt{R} ,  \oldvec{F} ,  \ottnt{E}  \ottsym{[}   \ASSERT( \varphi ) \SEQ  \ottnt{e}   \ottsym{]} }   \\ \quad   \longrightarrow _{ \ottnt{D} }     \tuple{ \ottnt{H} ,  \ottnt{R} ,  \oldvec{F} ,  \ottnt{E}  \ottsym{[}  \ottnt{e}  \ottsym{]} }  \end{array} 
  }
  \infrule[R-AssertFail]{
    \not\models  \ottsym{[}  \ottnt{R}  \ottsym{]} \, \varphi
  }{
      \tuple{ \ottnt{H} ,  \ottnt{R} ,  \oldvec{F} ,  \ottnt{E}  \ottsym{[}   \ASSERT( \varphi ) \SEQ  \ottnt{e}   \ottsym{]} }     \longrightarrow _{ \ottnt{D} }     \mathbf{AssertFail}  
  }
  \end{center}
\caption{Transition Rules (2).}
\label{fig:transitionRules2}
\end{figure}
\bcprulessavespacefalse

\subsection{Operational Semantics}
\label{sec:semantics}
We now introduce the operational semantics for our language.
We assume a countably infinite domain of heap addresses \textbf{Addr}:
we denote an arbitrary address with $\ottmv{a}$.
A runtime state is represented by a configuration $ \tuple{ \ottnt{H} ,  \ottnt{R} ,  \oldvec{F} ,  \ottnt{e} } $, which consists of a heap,
register file, stack, and currently reducing expression respectively.
The register file maps variables to runtime values $v$, which are
either integers $n$ or addresses $\ottmv{a}$. The heap maps a finite subset
of addresses to runtime values. The runtime stack represents pending
function calls as a sequence of return contexts, which we
describe below.
While the final configuration component is an expression,
the rewriting rules are defined in terms of $\ottnt{E}  \ottsym{[}  \ottnt{e}  \ottsym{]}$, which is
an evaluation context $\ottnt{E}$ and redex $\ottnt{e}$, as is standard.
The grammar for evaluation contexts is defined by: 
\(
    \ottnt{E}  ::=   \ottnt{E'} \SEQ \ottnt{e}  \mid  \HOLE .
  \)

Our operational semantics is given in \Cref{fig:transitionRules1,fig:transitionRules2}.
We write $ \DOM( \ottnt{H} ) $ to indicate the domain of
a function and $\ottnt{H}  \ottsym{\{}  \ottmv{a}  \mapsto  v  \ottsym{\}}$ where $ \ottmv{a}  \not\in   \DOM( \ottnt{H} )  $ to denote a map
which takes all values in $\DOM(H)$ to their values in $H$ and which
additionally takes $\ottmv{a}$ to $v$.
We will write $ \ottnt{H}  \ottsym{\{}  \ottmv{a}  \hookleftarrow  v  \ottsym{\}}$
where $ \ottmv{a}  \in   \DOM( \ottnt{H} )  $ to denote a map equivalent to $\ottnt{H}$ except that $\ottmv{a}$ takes value $v$.
We use similar notation for $ \DOM( \ottnt{R} ) $ and $\ottnt{R}  \ottsym{\{}  \mathit{x}  \mapsto  v  \ottsym{\}}$.
We also write $\emptyset$ for the empty register file and heap.
The step relation  $ \longrightarrow _{ \ottnt{D} } $ is parameterized
by a set of function definitions $\ottnt{D}$; a program $ \tuple{ \ottnt{D} ,  \ottnt{e} } $ is executed
by stepping the initial configuration $ \tuple{  \emptyset  ,   \emptyset  ,   \cdot  ,  \ottnt{e} } $ according
to $\longrightarrow_{\ottnt{D}}$.
The semantics is mostly standard; we highlight some important points below.

Return contexts $F$ take the form $\ottnt{E}  \ottsym{[}   \LET  \mathit{y}  =   \HOLE^ \ell   \IN  \ottnt{e}   \ottsym{]}$. A return context represents a pending
function call with label $\ell$, and indicates that $\mathit{y}$ should be bound to
the return value of the callee during the execution of $\ottnt{e}$ within the larger execution context $\ottnt{E}$.
The call stack $\oldvec{F}$ is a sequence of these contexts, with
the first such return context representing the most recent function call.
The stack grows at function calls as described by rule \rn{R-Call}.
For a call $\ottnt{E}  \ottsym{[}   \LET  \mathit{x}  =   \mathit{f} ^ \ell (  \mathit{y_{{\mathrm{1}}}} ,\ldots, \mathit{y_{\ottmv{n}}}  )   \IN  \ottnt{e}   \ottsym{]}$ where
$f$ is defined as $\ottsym{(}  \mathit{x_{{\mathrm{1}}}}  \ottsym{,} \, ... \, \ottsym{,}  \mathit{x_{\ottmv{n}}}  \ottsym{)}  \ottnt{e'}$, the return context $\ottnt{E}  \ottsym{[}   \LET  \mathit{y}  =   \HOLE^ \ell   \IN  \ottnt{e}   \ottsym{]}$ is
prepended onto the stack of the input configuration.
The substitution of formal arguments for parameters in $e'$, denoted by $   [  \mathit{y_{{\mathrm{1}}}}  /  \mathit{x_{{\mathrm{1}}}}  ]  \cdots  [  \mathit{y_{\ottmv{n}}}  /  \mathit{x_{\ottmv{n}}}  ]     \ottnt{e'} $,
becomes the currently reducing expression in the output configuration.
Function returns are handled by \rn{R-Var}.
Our semantics return values by name; when the currently executing function fully reduces to a single variable $x$,
$x$ is substituted into the return context on the top of the stack,
denoted by $\ottnt{E}  \ottsym{[}   \LET  \mathit{y}  =   \HOLE^ \ell   \IN  \ottnt{e}   \ottsym{]}[x]$.

In the rules \rn{R-Assert} we write $\models  \ottsym{[}  \ottnt{R}  \ottsym{]} \, \varphi$ to mean that the formula
yielded by substituting the concrete values in $\ottnt{R}$ for the variables in $\varphi$
is valid within some chosen logic (see \Cref{sec:types}); in \rn{R-AssertFail} we write $\not\models  \ottsym{[}  \ottnt{R}  \ottsym{]} \, \varphi$
when the formula is \emph{not} valid.
The substitution operation $\ottsym{[}  \ottnt{R}  \ottsym{]} \, \varphi$ is defined inductively as $\ottsym{[}   \emptyset   \ottsym{]} \, \varphi  \ottsym{=}  \varphi, \ottsym{[}  \ottnt{R}  \ottsym{\{}  \mathit{x}  \mapsto  n  \ottsym{\}}  \ottsym{]} \, \varphi  \ottsym{=}  \ottsym{[}  \ottnt{R}  \ottsym{]} \, \ottsym{[}  n  \ottsym{/}  \mathit{x}  \ottsym{]}  \varphi, \ottsym{[}  \ottnt{R}  \ottsym{\{}  \mathit{x}  \mapsto  \ottmv{a}  \ottsym{\}}  \ottsym{]} \, \varphi  \ottsym{=}  \ottsym{[}  \ottnt{R}  \ottsym{]} \, \varphi$.
In the case of an assertion failure, the semantics steps to a distinguished
configuration $ \mathbf{AssertFail} $. The goal of our type system is to show that no
execution of a well-typed program may reach this configuration.
The \imp{alias} form checks whether the two references actually alias;
i.e., if the must-alias assertion provided by the programmer is correct.
If not, our semantics steps to the distinguished $ \mathbf{AliasFail} $ configuration.
We assume that the operands $\mathit{x}$ and $\mathit{y}$ are distinct variables.
Our type system does \emph{not} guarantee that $ \mathbf{AliasFail} $ is unreachable;
aliasing assertions are effectively trusted annotations that are assumed to hold.

In order to avoid duplicate variable names in our register file due to recursive functions,
we refresh the bound variable $\mathit{x}$ in a let expression to $\mathit{x'}$.
Take expression $ \LET  \mathit{x}  =  \mathit{y}  \IN  \ottnt{e} $ as an example;
we substitute a fresh variable $\mathit{x'}$ for $\mathit{x}$ in $\ottnt{e}$, then bind $\mathit{x'}$ to the value
of variable $\mathit{y}$.
We assume this refreshing of variables preserves our assumption that all variable
bindings introduced with let and function parameters are unique, i.e. $\mathit{x'}$ does
not overlap with variable names that occur in the program.

\section{Typing}
\label{sec:typesystem}
We now introduce a
fractional ownership refinement type system that guarantees well-typed
programs do not encounter assertion failures.

\begin{figure}[t]
  \begin{minipage}{0.2\textwidth}
  \[
    \begin{array}{rrcl}
      \text{\scriptsize Types} & \tau %
                                         &::=&  \set{  \nu  \COL \TINT \mid  \varphi }  \mid  \tau  \TREF^{ r }  \\
      \text{\scriptsize Ownership} & r & \in & [0,1] \\
      \text{\scriptsize Refinements} & \varphi & ::= & \varphi_{{\mathrm{1}}}  \vee  \varphi_{{\mathrm{2}}} \mid  \neg  \varphi  \mid  \top  \\
                               & & \mid & \phi  \ottsym{(}  \widehat{v}_{{\mathrm{1}}}  \ottsym{,} \, .. \, \ottsym{,}  \widehat{v}_{\ottmv{n}}  \ottsym{)} \\
                               & & \mid & \widehat{v}_{{\mathrm{1}}} \, \ottsym{=} \, \widehat{v}_{{\mathrm{2}}} \\
                                 & & \mid & \mathcal{CP} \\
      \text{\scriptsize Ref. Values} & \widehat{v} & ::= & \mathit{x} \mid n \mid \nu \\
    \end{array}
  \]
\end{minipage}
\hfill
\begin{minipage}{0.59\textwidth}
  \[
    \begin{array}{rrcl}
      \text{\scriptsize Function Types} & \sigma & ::= &  \forall  \lambda .\tuple{ \mathit{x_{{\mathrm{1}}}} \COL \tau_{{\mathrm{1}}} ,\dots, \mathit{x_{\ottmv{n}}} \COL \tau_{\ottmv{n}} } \\ & & & \ra\tuple{ \mathit{x_{{\mathrm{1}}}} \COL \tau'_{{\mathrm{1}}} ,\dots, \mathit{x_{\ottmv{n}}} \COL \tau'_{\ottmv{n}}  \mid  \tau }  \\
      \text{\scriptsize Context Variables} & \lambda & \in & \CVar \\
      \text{\scriptsize Concrete Context} & \oldvec{\ell} & ::= & \ell  \ottsym{:}  \oldvec{\ell} \mid  \epsilon  \\
      \text{\scriptsize Pred. Context} & \mathcal{C} & ::= &  \ell  :  \mathcal{C}  \mid \lambda \mid  \epsilon  \\
      \text{\scriptsize Context Query} & \mathcal{CP} & ::= &   \oldvec{\ell}     \preceq    \mathcal{C}  \\
      \text{\scriptsize Typing Context} & \mathcal{L} & ::= & \lambda \mid \oldvec{\ell} \\
    \end{array}
  \]
  \end{minipage}
  \caption{Syntax of types, refinements, and contexts.}
  \label{fig:types}
\end{figure}

\subsection{Types and Contexts}
\label{sec:types}
The syntax of types is given in \Cref{fig:types}.
Our type system has two type constructors: references and integers. $ \tau  \TREF^{ r } $ is the
type of a (non-null) reference to a value of type $\tau$.
$r$ is an ownership which is a rational number in the
range $[0,1]$. An ownership of $\ottsym{0}$ indicates a
reference that cannot be written,
and for which there may exist a mutable alias.
By contrast, $\ottsym{1}$ indicates a pointer with exclusive
ownership that can be read and written. Reference types with ownership values between
these two extremes indicate a pointer that is readable but not writable, and for which no mutable aliases exist.
\name ensures that these
invariants hold while aliases are created and destroyed during execution.

Integers are refined with a predicate $\varphi$.  The language of predicates is built using the
standard logical connectives of first-order logic,
with (in)equality between variables and integers, and atomic predicate symbols
$\phi$ as the basic atoms. We include a special
``value'' variable $\nu$ representing the
value being refined by the predicate. For simplicity, we omit the
connectives $ \varphi_{{\mathrm{1}}}  \wedge  \varphi_{{\mathrm{2}}} $ and $\varphi_{{\mathrm{1}}}  \implies  \varphi_{{\mathrm{2}}}$; they can be written
as derived forms using the given connectives.
We do not fix a particular theory from which $\phi$ are drawn, provided a sound
(but not necessarily complete) decision procedure exists.
$\mathcal{CP}$ are context predicates, which are used for context sensitivity as
explained below.

\begin{example}
  $ \set{  \nu  \COL \TINT \mid  \nu \, \ottsym{>} \, \ottsym{0} } $ is the type of strictly positive integers.
  The type of immutable
  references to integers exactly equal to $3$ can be expressed by
  $  \set{  \nu  \COL \TINT \mid  \nu \, \ottsym{=} \, \ottsym{3} }   \TREF^{ \ottsym{0}  \ottsym{.}  \ottsym{5} } $.
\end{example}

As is standard, we denote a type environment with $\Gamma$, which is a
finite map from variable names to type $\tau$.
We write $\Gamma  \ottsym{[}  \mathit{x}  \ottsym{:}  \tau  \ottsym{]}$ to denote a type environment $\Gamma$ such
that $\Gamma  \ottsym{(}  \mathit{x}  \ottsym{)}  \ottsym{=}  \tau$ where $ \mathit{x}  \in   \DOM( \Gamma )  $, $\Gamma  \ottsym{,}  \mathit{x}  \ottsym{:}  \tau$ to indicate the
extension of $\Gamma$ with the type binding $ \mathit{x} \COL \tau $, and $\Gamma  \ottsym{[}  \mathit{x}  \hookleftarrow  \tau  \ottsym{]}$
to indicate the type environment $\Gamma$ with the binding of $\mathit{x}$
updated to $\tau$. We write the empty environment as $ \bullet $.
The treatment of type environments as mappings instead of sequences
in a dependent type system is somewhat non-standard.
The standard formulation based on ordered sequences of bindings and
its corresponding well-formedness condition did not easily admit
variables with mutually dependent refinements as introduced by
our function types (see below). We therefore use an unordered
environment and relax well-formedness to ignore variable binding order.

\paragraph{Function Types, Contexts, and Context Polymorphism.}
Our type system achieves context sensitivity
by allowing function types to depend on where a function
is called, i.e., the \emph{execution context} of the function invocation.
Our system represents a \emph{concrete} execution contexts with
strings of call site labels (or just ``call strings''),
defined by $\oldvec{\ell} ::=  \epsilon  \mid \ell  \ottsym{:}  \oldvec{\ell}$.
As is standard (e.g., \cite{sharir1978two,shivers1991control}), the string $\ell  \ottsym{:}  \oldvec{\ell}$ abstracts
an execution context where the most recent, active function call occurred
at call site $\ell$ which itself was executed in a context abstracted
by $\oldvec{\ell}$; $ \epsilon $ is the context under which program execution begins. \emph{Context variables},
drawn from a finite domain $\CVar$ and ranged over by $\lambda_{{\mathrm{1}}}, \lambda_{{\mathrm{2}}}, \ldots$,
represent arbitrary, unknown contexts.\looseness=-1

A function type takes the form
$ \forall  \lambda .\tuple{ \mathit{x_{{\mathrm{1}}}} \COL \tau_{{\mathrm{1}}} ,\dots, \mathit{x_{\ottmv{n}}} \COL \tau_{\ottmv{n}} }\ra\tuple{ \mathit{x_{{\mathrm{1}}}} \COL \tau'_{{\mathrm{1}}} ,\dots, \mathit{x_{\ottmv{n}}} \COL \tau'_{\ottmv{n}}  \mid  \tau } $.
The arguments of a function are an $n$-ary tuple of types $\tau_{\ottmv{i}}$.
To model side-effects on arguments, the function type includes the same number of \emph{output types}
$\tau'_{\ottmv{i}}$. In addition, function types have a direct return type $\tau$. 
The argument and output types are given names: refinements within the function type
may refer to these names.
Function types in our language are context polymorphic,
expressed by universal quantification ``$\forall \lambda.$''
over a context variable. Intuitively,
this context variable represents the many different execution contexts
under which a function may be called.

Argument and return types may depend on this context variable by
including \emph{context query predicates} in their refinements.
A context query predicate $\mathcal{CP}$ usually takes the
form $  \oldvec{\ell}     \preceq    \lambda $, and is true iff $\oldvec{\ell}$ is a
prefix of the concrete context represented by $\lambda$.
Intuitively, a refinement $  \oldvec{\ell}     \preceq    \lambda   \implies  \varphi$
states that $\varphi$ holds in any concrete execution context
with prefix $\oldvec{\ell}$, and provides no information in any other
context. In full generality, a context query predicate may
be of the form $  \oldvec{\ell}_{{\mathrm{1}}}     \preceq    \oldvec{\ell}_{{\mathrm{2}}} $ or $  \oldvec{\ell}     \preceq     \ell_{{\mathrm{1}}} \ldots \ell_{\ottmv{n}}   \ottsym{:}  \lambda $;
these forms may be immediately simplified to
$ \top $, $ \bot $ or $  \oldvec{\ell}'     \preceq    \lambda $.

\begin{example}
  \label{exmp:cs-type-example}
  The type $ \set{  \nu  \COL \TINT \mid   \ottsym{(}    \ell_{{\mathrm{1}}}     \preceq    \lambda   \implies  \nu \, \ottsym{=} \, \ottsym{3}  \ottsym{)}  \wedge  \ottsym{(}    \ell_{{\mathrm{2}}}     \preceq    \lambda   \implies  \nu \, \ottsym{=} \, \ottsym{5}  \ottsym{)}  } $ represents
  an integer that is 3 if the most recent active
  function call site is $\ell_{{\mathrm{1}}}$, 5 if the most recent call site is $\ell_{{\mathrm{2}}}$,
  and is otherwise unconstrained. This type may be used for the argument of
  \imp{f} in, e.g., \imp[mathescape]{f$^{\ell_{{\mathrm{1}}}}$(3) + f$^{\ell_{{\mathrm{2}}}}$(5)}.
\end{example}

As types in our type system may contain context variables, our typing judgment
(introduced below) includes a typing context $\mathcal{L}$, which is either a single
context variable $\lambda$ or a concrete context $\oldvec{\ell}$. This typing context
represents the assumptions about the execution context of the term being
typed. If the typing context is a context variable $\lambda$, then no assumptions
are made about the execution context of the term, although types
may depend upon $\lambda$ with context query predicates.
Accordingly, function bodies are typed under the context variable
universally quantified over in the corresponding function type; i.e.,
no assumptions are made about the exact execution context of the function body.
As in parametric polymorphism, consistent
substitution of a concrete context $\oldvec{\ell}$ for a context variable $\lambda$
in a typing derivation yields a valid type derivation under concrete context $\oldvec{\ell}$.

\begin{remark}
  \label{rem:cfa}
  The context-sensitivity scheme described here corresponds
  to the standard CFA approach \cite{shivers1991control} without \emph{a priori} call-string limiting.
  We chose this scheme because it can be easily encoded with equality over integer variables (see \Cref{sec:infr}),
  but in principle another context-sensitivity strategy could be used instead. The important
  feature of our type system is the inclusion of predicates over
  contexts, not the specific choice for these predicates.
\end{remark}

Function type environments are denoted with $\Theta$ and are finite
maps from function names ($\mathit{f}$) to function types ($\sigma$).

\paragraph{Well Formedness.}
We impose two well-formedness conditions on types:
\emph{ownership well-formedness} and \emph{refinement well-formedness}.
The ownership condition is purely syntactic:
$\tau$ is ownership well-formed if $\tau  \ottsym{=}   \tau'  \TREF^{ \ottsym{0} } $ implies
$\tau'  \ottsym{=}  \top_{\ottmv{n}}$ for some $\ottmv{n}$. $\top_{\ottmv{i}}$ is the ``maximal'' type
of a chain of $\ottmv{i}$ references, and is defined inductively as
$\top_{{\mathrm{0}}}  \ottsym{=}   \set{  \nu  \COL \TINT \mid   \top  } , \top_{\ottmv{i}}  \ottsym{=}   \top_{{\ottmv{i}-1}}  \TREF^{ \ottsym{0} } $.

The ownership well-formedness condition ensures that aliases introduced via
heap writes do not violate the invariant of ownership types \emph{and} that refinements
are consistent with updates performed through mutable aliases. Recall our ownership type invariant
ensures all aliases of a mutable reference have 0 ownership.
Any mutations through that mutable alias will therefore be consistent with the
``no information'' $ \top $ refinement required by this
well-formedness condition.

Refinement well-formedness, denoted $ \mathcal{L}   \mid   \Gamma   \vdash _{\wf}  \varphi $,
ensures that free program variables in refinement $\varphi$ are
bound in a type environment $\Gamma$ and have integer type. It
also requires that for a typing context
$\mathcal{L}  \ottsym{=}  \lambda$, only context query predicates over $\lambda$ are used (no such
predicates may be used if $\mathcal{L}  \ottsym{=}  \oldvec{\ell}$). Notice this condition
forbids refinements that refer to references. Although
ownership information can signal when refinements on a mutably-aliased
reference must be discarded, our current formulation
provides no such information for refinements that \emph{mention}
mutably-aliased references. We therefore conservatively reject
such refinements at the cost of some expressiveness in our type system.

We write $ \mathcal{L}   \mid   \Gamma   \vdash _{\wf}  \tau $ to indicate a well-formed type where
all refinements are well-formed with respect to $\mathcal{L}$ and $\Gamma$.
We write $ \mathcal{L}   \vdash _{\wf}  \Gamma $ for a type environment where all
types are well-formed. A function environment
is well-formed (written $ \vdash _{\wf}  \Theta $) if, for every $\sigma$ in $\Theta$,
the argument, result, and output types are well-formed with respect
to each other and the context variable quantified over in $\sigma$.
As the formal definition of refinement well-formedness is fairly standard, we
omit it for space reasons (the full definition may be found in \appref{\Cref{sec:aux-defn-lem}}{the full version \cite{toman2020consort}}).

\begin{figure}[t]
  \scriptsize
  \infrule[T-Var]{
  }{
     \Theta   \mid   \mathcal{L}   \mid   \Gamma  \ottsym{[}  \mathit{x}  \ottsym{:}  \tau_{{\mathrm{1}}}  \ottsym{+}  \tau_{{\mathrm{2}}}  \ottsym{]}   \vdash   \mathit{x}  :  \tau_{{\mathrm{1}}}   \produces   \Gamma  \ottsym{[}  \mathit{x}  \hookleftarrow  \tau_{{\mathrm{2}}}  \ottsym{]} 
  }
  \vspace*{2ex}
  \infrule[T-Let]{
     \Theta   \mid   \mathcal{L}   \mid   \Gamma  \ottsym{[}  \mathit{y}  \hookleftarrow   \tau_{{\mathrm{1}}}  \wedge_{ \mathit{y} }   \mathit{y}  =_{ \tau_{{\mathrm{1}}} }  \mathit{x}    \ottsym{]}  \ottsym{,}  \mathit{x}  \ottsym{:}  \ottsym{(}   \tau_{{\mathrm{2}}}  \wedge_{ \mathit{x} }   \mathit{x}  =_{ \tau_{{\mathrm{2}}} }  \mathit{y}    \ottsym{)}   \vdash   \ottnt{e}  :  \tau   \produces   \Gamma'  \andalso
     \mathit{x}  \not\in   \DOM( \Gamma' )  
  }{
     \Theta   \mid   \mathcal{L}   \mid   \Gamma  \ottsym{[}  \mathit{y}  \ottsym{:}  \tau_{{\mathrm{1}}}  \ottsym{+}  \tau_{{\mathrm{2}}}  \ottsym{]}   \vdash    \LET  \mathit{x}  =  \mathit{y}  \IN  \ottnt{e}   :  \tau   \produces   \Gamma' 
  }
  \vspace*{2ex}
  \infrule[T-LetInt]{
     \Theta   \mid   \mathcal{L}   \mid   \Gamma  \ottsym{,}  \mathit{x}  \ottsym{:}   \set{  \nu  \COL \TINT \mid  \nu \, \ottsym{=} \, n }    \vdash   \ottnt{e}  :  \tau   \produces   \Gamma'  \andalso
     \mathit{x}  \not\in   \DOM( \Gamma' )  
  }{
     \Theta   \mid   \mathcal{L}   \mid   \Gamma   \vdash    \LET  \mathit{x}  =  n  \IN  \ottnt{e}   :  \tau   \produces   \Gamma' 
  }
  \vspace*{2ex}
  \infrule[T-If]{
     \Theta   \mid   \mathcal{L}   \mid   \Gamma  \ottsym{[}  \mathit{x}  \hookleftarrow   \set{  \nu  \COL \TINT \mid   \varphi  \wedge  \nu \, \ottsym{=} \, \ottsym{0}  }   \ottsym{]}   \vdash   \ottnt{e_{{\mathrm{1}}}}  :  \tau   \produces   \Gamma'  \\
     \Theta   \mid   \mathcal{L}   \mid   \Gamma  \ottsym{[}  \mathit{x}  \hookleftarrow   \set{  \nu  \COL \TINT \mid   \varphi  \wedge  \nu \, \neq \, \ottsym{0}  }   \ottsym{]}   \vdash   \ottnt{e_{{\mathrm{2}}}}  :  \tau   \produces   \Gamma' 
  }{
     \Theta   \mid   \mathcal{L}   \mid   \Gamma  \ottsym{[}  \mathit{x}  \ottsym{:}   \set{  \nu  \COL \TINT \mid  \varphi }   \ottsym{]}   \vdash    \IFZERO  \mathit{x}  \THEN  \ottnt{e_{{\mathrm{1}}}}  \ELSE  \ottnt{e_{{\mathrm{2}}}}   :  \tau   \produces   \Gamma' 
  }
  \vspace*{2ex}
  \bcprulessavespacetrue
  \begin{center}
  \infrule[T-MkRef]{
     \Theta   \mid   \mathcal{L}   \mid   \Gamma  \ottsym{[}  \mathit{y}  \hookleftarrow  \tau_{{\mathrm{1}}}  \ottsym{]}  \ottsym{,}  \mathit{x}  \ottsym{:}   \ottsym{(}   \tau_{{\mathrm{2}}}  \wedge_{ \mathit{x} }   \mathit{x}  =_{ \tau_{{\mathrm{2}}} }  \mathit{y}    \ottsym{)}  \TREF^{ \ottsym{1} }    \vdash   \ottnt{e}  :  \tau   \produces   \Gamma'  \\
     \mathit{x}  \not\in   \DOM( \Gamma' )  
  }{
     \Theta   \mid   \mathcal{L}   \mid   \Gamma  \ottsym{[}  \mathit{y}  \ottsym{:}  \tau_{{\mathrm{1}}}  \ottsym{+}  \tau_{{\mathrm{2}}}  \ottsym{]}   \vdash    \LET  \mathit{x}  =   \MKREF  \mathit{y}   \IN  \ottnt{e}   :  \tau   \produces   \Gamma' 
  }
  \infrule[T-Seq]{
     \Theta   \mid   \mathcal{L}   \mid   \Gamma   \vdash   \ottnt{e_{{\mathrm{1}}}}  :  \tau'   \produces   \Gamma'  \\
     \Theta   \mid   \mathcal{L}   \mid   \Gamma'   \vdash   \ottnt{e_{{\mathrm{2}}}}  :  \tau''   \produces   \Gamma'' 
  }{
     \Theta   \mid   \mathcal{L}   \mid   \Gamma   \vdash    \ottnt{e_{{\mathrm{1}}}}  \SEQ  \ottnt{e_{{\mathrm{2}}}}   :  \tau''   \produces   \Gamma'' 
  }
  \infrule[T-Deref]{
    \tau' = \begin{cases}
       \tau_{{\mathrm{1}}}  \wedge_{ \mathit{y} }   \mathit{y}  =_{ \tau_{{\mathrm{1}}} }  \mathit{x}   &  r   \ottsym{>}   \ottsym{0}  \\
      \tau_{{\mathrm{1}}} & r  \ottsym{=}  \ottsym{0}
    \end{cases}  \\
     \Theta   \mid   \mathcal{L}   \mid   \Gamma  \ottsym{[}  \mathit{y}  \hookleftarrow   \tau'  \TREF^{ r }   \ottsym{]}  \ottsym{,}  \mathit{x}  \ottsym{:}  \tau_{{\mathrm{2}}}   \vdash   \ottnt{e}  :  \tau   \produces   \Gamma'  \\
     \mathit{x}  \not\in   \DOM( \Gamma' )  
  }{
     \Theta   \mid   \mathcal{L}   \mid   \Gamma  \ottsym{[}  \mathit{y}  \ottsym{:}   \ottsym{(}  \tau_{{\mathrm{1}}}  \ottsym{+}  \tau_{{\mathrm{2}}}  \ottsym{)}  \TREF^{ r }   \ottsym{]}   \vdash    \LET  \mathit{x}  =   *  \mathit{y}   \IN  \ottnt{e}   :  \tau   \produces   \Gamma' 
  }
  \infrule[T-Assert]{
    \Gamma  \models  \varphi \andalso
      \epsilon    \mid   \Gamma   \vdash _{\wf}  \varphi  \\
     \Theta   \mid   \mathcal{L}   \mid   \Gamma   \vdash   \ottnt{e}  :  \tau   \produces   \Gamma' 
  }{
     \Theta   \mid   \mathcal{L}   \mid   \Gamma   \vdash    \ASSERT( \varphi ) \SEQ  \ottnt{e}   :  \tau   \produces   \Gamma' 
  }
  \end{center}
  \bcprulessavespacefalse
  \caption{Expression typing rules.}
  \label{fig:intra-type}
\end{figure}

\subsection{Intraprocedural Type System}
\label{sec:intra-types}
We now introduce the type system for the intraprocedural fragment of
our language. Accordingly, this section focuses on the interplay of
mutability and refinement types. The typing rules are given in
\Cref{fig:intra-type,fig:pointer-typing}.  
A typing judgment takes the
form $ \Theta   \mid   \mathcal{L}   \mid   \Gamma   \vdash   \ottnt{e}  :  \tau   \produces   \Gamma' $, which indicates that $\ottnt{e}$
is well-typed under a function type
environment $\Theta$, typing context $\mathcal{L}$, and type environment $\Gamma$, 
and evaluates to a value of type $\tau$ and modifies the input environment
according to $\Gamma'$. Any valid typing derivation must
have $ \mathcal{L}   \vdash _{\wf}  \Gamma $, $ \mathcal{L}   \vdash _{\wf}  \Gamma' $, and $ \mathcal{L}   \mid   \Gamma'   \vdash _{\wf}  \tau $, i.e.,
the input and output type environments and result type must
be well-formed.

The typing rules in \Cref{fig:intra-type} handle the relatively
standard features in our language. The rule \rn{T-Seq} for sequential
composition is fairly straightforward except that the output type
environment for $\ottnt{e_{{\mathrm{1}}}}$ is the input type environment for $\ottnt{e_{{\mathrm{2}}}}$.
\rn{T-LetInt} is also straightforward; since $\mathit{x}$ is bound to a
constant, it is given type $ \set{  \nu  \COL \TINT \mid  \nu \, \ottsym{=} \, n } $ to indicate $\mathit{x}$ is
exactly $n$.  The output type environment $\Gamma'$ cannot mention
$\mathit{x}$ (expressed with $ \mathit{x}  \not\in   \DOM( \Gamma' )  $)
to prevent $\mathit{x}$ from escaping its scope.  This requirement
can be met by applying the subtyping rule (see below) to weaken
refinements to no longer mention $\mathit{x}$. As in other refinement type
systems \cite{rondon2008liquid}, this requirement is critical for
ensuring soundness.

Rule \rn{T-Let} is crucial to understanding our ownership type system.
The body of the let expression $\ottnt{e}$ is typechecked under a type environment
where the type of $\mathit{y}$ in $\Gamma$ is linearly split into two types:
$\tau_{{\mathrm{1}}}$ for $\mathit{y}$ and $\tau_{{\mathrm{2}}}$ for the newly created binding $\mathit{x}$.
This splitting is expressed using the $+$ operator. If $\mathit{y}$ is a reference
type, the split operation distributes some portion of $\mathit{y}$'s ownership information
to its new alias $\mathit{x}$. The split operation also distributes
refinement information between the two types.
For example, type $  \set{  \nu  \COL \TINT \mid  \nu \, \ottsym{>} \, \ottsym{0} }   \TREF^{ \ottsym{1} } $ can be split into (1)
$  \set{  \nu  \COL \TINT \mid  \nu \, \ottsym{>} \, \ottsym{0} }   \TREF^{ r } $ and $  \set{  \nu  \COL \TINT \mid  \nu \, \ottsym{>} \, \ottsym{0} }   \TREF^{ \ottsym{(}  \ottsym{1}  \ottsym{-}  r  \ottsym{)} } $ (for $r \in (0,1)$),
i.e., two \emph{immutable} references with non-trivial refinement information,
or (2) $  \set{  \nu  \COL \TINT \mid  \nu \, \ottsym{>} \, \ottsym{0} }   \TREF^{ \ottsym{1} } $ and $  \set{  \nu  \COL \TINT \mid   \top  }   \TREF^{ \ottsym{0} } $, where
one of the aliases is mutable and the other provides no refinement information.
How a type is split depends on the usage of $\mathit{x}$ and $\mathit{y}$ in $\ottnt{e}$.
Formally, we define the type addition operator as
the least commutative partial operation that satisfies the following
rules:\looseness=-1
\begin{align*}
   \set{  \nu  \COL \TINT \mid  \varphi_{{\mathrm{1}}} }   \ottsym{+}   \set{  \nu  \COL \TINT \mid  \varphi_{{\mathrm{2}}} }  & =  \set{  \nu  \COL \TINT \mid   \varphi_{{\mathrm{1}}}  \wedge  \varphi_{{\mathrm{2}}}  }  & (\rn{Tadd-Int})  \\
    \tau_{{\mathrm{1}}}  \TREF^{ r_{{\mathrm{1}}} }   \ottsym{+}  \tau_{{\mathrm{2}}}  \TREF^{ r_{{\mathrm{2}}} }  & =  \ottsym{(}  \tau_{{\mathrm{1}}}  \ottsym{+}  \tau_{{\mathrm{2}}}  \ottsym{)}  \TREF^{ r_{{\mathrm{1}}}  \ottsym{+}  r_{{\mathrm{2}}} }  & (\rn{Tadd-Ref})
\end{align*}
Viewed another way, type addition describes how to combine two types
for the same value such that the combination soundly incorporates all
information from the two original types.  Critically, the type
addition operation cannot create or destroy ownership and refinement
information, only combine or divide it between types.  Although not
explicit in the rules, by ownership well-formedness, if
the entirety of a reference's ownership is transferred to another type
during a split, all refinements in the remaining type must be
$ \top $.

The additional bits $\land_y  \mathit{y}  =_{ \tau_{{\mathrm{1}}} }  \mathit{x} $ and $\land_x  \mathit{x}  =_{ \tau_{{\mathrm{2}}} }  \mathit{y} $
express equality between $\mathit{x}$ and $\mathit{y}$ as refinements.  We use
the strengthening operation $ \tau  \wedge_{ \mathit{x} }  \varphi $ and typed equality
proposition $ \mathit{x}  =_{ \tau }  \mathit{y} $, defined respectively as:
\begin{align*}
    \set{  \nu  \COL \TINT \mid  \varphi }   \wedge_{ \mathit{y} }  \varphi'  & =  \set{  \nu  \COL \TINT \mid   \varphi  \wedge  \ottsym{[} \, \nu \, \ottsym{/}  \mathit{y}  \ottsym{]} \, \varphi'  }  &
  \ottsym{(}   \mathit{x}  =_{  \set{  \nu  \COL \TINT \mid  \varphi }  }  \mathit{y}   \ottsym{)} & =  \ottsym{(}  \mathit{x} \, \ottsym{=} \, \mathit{y}  \ottsym{)} \\ 
      \tau  \TREF^{ r }   \wedge_{ \mathit{y} }  \varphi'  & =  \tau  \TREF^{ r }  &
    \ottsym{(}   \mathit{x}  =_{  \tau  \TREF^{ r }  }  \mathit{y}   \ottsym{)} & =  \top 
\end{align*}
We do not track equality between references or between the contents of aliased
reference cells as doing so would violate our refinement
well-formedness condition.  These operations are also used in other rules
that can introduce equality.

Rule \rn{T-MkRef} is very similar to \rn{T-Let}, except that $\mathit{x}$
is given a reference type of ownership 1 pointing to $\tau_{{\mathrm{2}}}$, which
is obtained by splitting the type of $\mathit{y}$.  In \rn{T-Deref},
the content type of $\mathit{y}$ is split and distributed to
$\mathit{x}$.  The strengthening is \emph{conditionally} applied depending
on the ownership of the dereferenced pointer, that is, if $r = 0$,
$\tau'$ has to be a maximal type $\top_{\ottmv{i}}$.

Our type system also tracks path information; in the \rn{T-If}
rule, we update the refinement on the condition variable within the respective
branches to indicate whether the variable must be zero. By requiring
both branches to produce the same output type environment, we guarantee
that these conflicting refinements are rectified within the
type derivations of the two branches.

The type rule for assert statements has the precondition
$\Gamma  \models  \varphi$ which is defined to be $\models   \sem{ \Gamma }   \implies  \varphi$, i.e.,
the logical formula $ \sem{ \Gamma }   \implies  \varphi$ is valid in the chosen theory.
$ \sem{ \Gamma } $ lifts the refinements on the integer valued variables into a proposition
in the logic used for verification. This denotation
operation is defined as:
\[
\begin{array}{rlcrl}
   \sem{  \bullet  }  &=  \top  & \hspace{1em} &   \sem{  \set{  \nu  \COL \TINT \mid  \varphi }  }_{ \mathit{y} }  & = \ottsym{[}  \mathit{y}  \ottsym{/} \, \nu \, \ottsym{]} \, \varphi \\
   \sem{ \Gamma  \ottsym{,}  \mathit{x}  \ottsym{:}  \tau }  &=  \sem{ \Gamma }   \wedge   \sem{ \tau }_{ \mathit{x} }  &   &   \sem{  \tau'  \TREF^{ r }  }_{ \mathit{y} }  & =  \top  \\
\end{array}
\]
If the formula $ \sem{ \Gamma }   \implies  \varphi$
is valid, then in any context and under any valuation of program variables that satisfy
the refinements in $ \sem{ \Gamma } $,
the predicate $\varphi$ must be true and the
assertion must not fail. This intuition forms the foundation of our
soundness claim (\Cref{sec:soundness}).

\begin{figure}[t]
  \scriptsize
    \infrule[T-Assign]{
    (\text{The shapes of $\tau'$ and $\tau_{{\mathrm{2}}}$ are similar}) \\
     \Theta   \mid   \mathcal{L}   \mid   \Gamma  \ottsym{[}  \mathit{x}  \hookleftarrow  \tau_{{\mathrm{1}}}  \ottsym{]}  \ottsym{[}  \mathit{y}  \hookleftarrow   \ottsym{(}   \tau_{{\mathrm{2}}}  \wedge_{ \mathit{y} }   \mathit{y}  =_{ \tau_{{\mathrm{2}}} }  \mathit{x}    \ottsym{)}  \TREF^{ \ottsym{1} }   \ottsym{]}   \vdash   \ottnt{e}  :  \tau   \produces   \Gamma'  \\
  }{
     \Theta   \mid   \mathcal{L}   \mid   \Gamma  \ottsym{[}  \mathit{x}  \ottsym{:}  \tau_{{\mathrm{1}}}  \ottsym{+}  \tau_{{\mathrm{2}}}  \ottsym{]}  \ottsym{[}  \mathit{y}  \ottsym{:}   \tau'  \TREF^{ \ottsym{1} }   \ottsym{]}   \vdash    \mathit{y}  \WRITE  \mathit{x}  \SEQ  \ottnt{e}   :  \tau   \produces   \Gamma' 
  }
    \vspace*{2ex}
  \infrule[T-Alias]{
    \ottsym{(}    \tau_{{\mathrm{1}}}  \TREF^{ r_{{\mathrm{1}}} }   \ottsym{+}  \tau_{{\mathrm{2}}}  \TREF^{ r_{{\mathrm{2}}} }   \ottsym{)}  \approx  \ottsym{(}    \tau'_{{\mathrm{1}}}  \TREF^{ r'_{{\mathrm{1}}} }   \ottsym{+}  \tau'_{{\mathrm{2}}}  \TREF^{ r'_{{\mathrm{2}}} }   \ottsym{)} \\
     \Theta   \mid   \mathcal{L}   \mid   \Gamma  \ottsym{[}  \mathit{x}  \hookleftarrow   \tau'_{{\mathrm{1}}}  \TREF^{ r'_{{\mathrm{1}}} }   \ottsym{]}  \ottsym{[}  \mathit{y}  \hookleftarrow   \tau'_{{\mathrm{2}}}  \TREF^{ r'_{{\mathrm{2}}} }   \ottsym{]}   \vdash   \ottnt{e}  :  \tau   \produces   \Gamma' 
  }{
     \Theta   \mid   \mathcal{L}   \mid   \Gamma  \ottsym{[}  \mathit{x}  \ottsym{:}   \tau_{{\mathrm{1}}}  \TREF^{ r_{{\mathrm{1}}} }   \ottsym{]}  \ottsym{[}  \mathit{y}  \ottsym{:}   \tau_{{\mathrm{2}}}  \TREF^{ r_{{\mathrm{2}}} }   \ottsym{]}   \vdash    \ALIAS( \mathit{x}  =  \mathit{y} ) \SEQ  \ottnt{e}   :  \tau   \produces   \Gamma' 
  }
    \vspace*{2ex}
  \infrule[T-AliasPtr]{
    \ottsym{(}    \tau_{{\mathrm{1}}}  \TREF^{ r_{{\mathrm{1}}} }   \ottsym{+}  \tau_{{\mathrm{2}}}  \TREF^{ r_{{\mathrm{2}}} }   \ottsym{)}  \approx  \ottsym{(}    \tau'_{{\mathrm{1}}}  \TREF^{ r'_{{\mathrm{1}}} }   \ottsym{+}  \tau'_{{\mathrm{2}}}  \TREF^{ r'_{{\mathrm{2}}} }   \ottsym{)} \\
     \Theta   \mid   \mathcal{L}   \mid   \Gamma  \ottsym{[}  \mathit{x}  \hookleftarrow   \tau'_{{\mathrm{1}}}  \TREF^{ r'_{{\mathrm{1}}} }   \ottsym{]}  \ottsym{[}  \mathit{y}  \hookleftarrow   \ottsym{(}   \tau'_{{\mathrm{2}}}  \TREF^{ r'_{{\mathrm{2}}} }   \ottsym{)}  \TREF^{ r }   \ottsym{]}   \vdash   \ottnt{e}  :  \tau   \produces   \Gamma' 
  }{
     \Theta   \mid   \mathcal{L}   \mid   \Gamma  \ottsym{[}  \mathit{x}  \ottsym{:}   \tau_{{\mathrm{1}}}  \TREF^{ r_{{\mathrm{1}}} }   \ottsym{]}  \ottsym{[}  \mathit{y}  \ottsym{:}   \ottsym{(}   \tau_{{\mathrm{2}}}  \TREF^{ r_{{\mathrm{2}}} }   \ottsym{)}  \TREF^{ r }   \ottsym{]}   \vdash    \ALIAS( \mathit{x}  = *  \mathit{y} ) \SEQ  \ottnt{e}   :  \tau   \produces   \Gamma' 
  }
    \vspace*{2ex}
  \infrule[T-Sub]{
    \Gamma  \leq  \Gamma' \andalso
     \Theta   \mid   \mathcal{L}   \mid   \Gamma'   \vdash   \ottnt{e}  :  \tau   \produces   \Gamma''  \andalso
    \Gamma''  \ottsym{,}  \tau  \leq  \Gamma'''  \ottsym{,}  \tau'
  }{
     \Theta   \mid   \mathcal{L}   \mid   \Gamma   \vdash   \ottnt{e}  :  \tau'   \produces   \Gamma''' 
  }
  \vspace*{2ex}
  \begin{center}
    $\tau_{{\mathrm{1}}}  \approx  \tau_{{\mathrm{2}}}$ iff $ \bullet   \vdash  \tau_{{\mathrm{1}}}  \leq  \tau_{{\mathrm{2}}}$ and $ \bullet   \vdash  \tau_{{\mathrm{2}}}  \leq  \tau_{{\mathrm{1}}}$.
  \end{center}
  \caption{Pointer manipulation and subtyping}
  \label{fig:pointer-typing}
\end{figure}

\begin{figure}[t]
  \begin{multicols}{2}
  \scriptsize
  \infrule[S-Int]{
  \Gamma  \models  \varphi_{{\mathrm{1}}}  \implies  \varphi_{{\mathrm{2}}}
  }{
   \Gamma  \vdash   \set{  \nu  \COL \TINT \mid  \varphi_{{\mathrm{1}}} }   \leq   \set{  \nu  \COL \TINT \mid  \varphi_{{\mathrm{2}}} } 
 }
    \vspace*{2ex}
  \infrule[S-Ref]{
     r_{{\mathrm{1}}}   \ge   r_{{\mathrm{2}}} 
    \andalso
    \Gamma  \vdash  \tau_{{\mathrm{1}}}  \leq  \tau_{{\mathrm{2}}}
  }{
    \Gamma  \vdash   \tau_{{\mathrm{1}}}  \TREF^{ r_{{\mathrm{1}}} }   \leq   \tau_{{\mathrm{2}}}  \TREF^{ r_{{\mathrm{2}}} } 
  }
  \infrule[S-TyEnv]{
    \forall \,  \mathit{x}  \in \DOM( \Gamma' )   \ottsym{.}  \Gamma  \vdash  \Gamma  \ottsym{(}  \mathit{x}  \ottsym{)}  \leq  \Gamma'  \ottsym{(}  \mathit{x}  \ottsym{)}
  }{
    \Gamma  \leq  \Gamma'
  }
    \vspace*{2ex}
  \infrule[S-Res]{
    \Gamma  \ottsym{,}  \mathit{x}  \ottsym{:}  \tau  \leq  \Gamma'  \ottsym{,}  \mathit{x}  \ottsym{:}  \tau' \andalso  \mathit{x}  \not\in   \DOM( \Gamma )  
  }{
    \Gamma  \ottsym{,}  \tau  \leq  \Gamma  \ottsym{,}  \tau'
  }
  \end{multicols}
  \caption{Subtyping rules.}
  \label{fig:subtyping}
\end{figure}

\paragraph{Destructive Updates, Aliasing, and Subtyping.}
We now discuss the handling of assignment, aliasing annotations,
and subtyping as described in \Cref{fig:pointer-typing}.
Although apparently unrelated, all three concern
updating the refinements of (potentially) aliased reference cells.

Like the binding forms discussed above, \rn{T-Assign} splits the assigned
value's type into two types via the type addition operator,
and distributes these types between the right hand side of the
assignment and the mutated reference contents.
Refinement information in the fresh
contents \emph{may} be inconsistent with
any previous refinement information; only the shapes must be the same.
In a system with unrestricted aliasing, this
typing rule would be unsound as it would admit writes that
are inconsistent with refinements on aliases of the left hand side.
However, the assignment rule
requires that the updated reference has an ownership of $1$.
By the ownership type invariant, all aliases
with the updated reference have $0$ ownership, and
by ownership well-formedness may only contain the $ \top $ refinement.

\begin{example}
  We can type the program as follows:
\begin{lstlisting}
let x = mkref 5 in       // $\color{comment-green}{ \mathit{x} \COL   \set{  \nu  \COL \TINT \mid  \nu \, \ottsym{=} \, \ottsym{5} }   \TREF^{ \ottsym{1} }  }$
let y = x in             // $\color{comment-green}{\mathit{x}  \ottsym{:}  \top_{{\mathrm{1}}}  \ottsym{,}  \mathit{y}  \ottsym{:}    \set{  \nu  \COL \TINT \mid  \nu \, \ottsym{=} \, \ottsym{5} }   \TREF^{ \ottsym{1} } }$
  y := 4; assert(*y = 4) // $\color{comment-green}{\mathit{x}  \ottsym{:}  \top_{{\mathrm{1}}}  \ottsym{,}  \mathit{y}  \ottsym{:}    \set{  \nu  \COL \TINT \mid  \nu \, \ottsym{=} \, \ottsym{4} }   \TREF^{ \ottsym{1} } }$
\end{lstlisting}
  In this and later examples, we include type annotations within comments. We stress
  that these annotations are for expository purposes only; our tool can infer these types
  automatically with no manual annotations.
\end{example}

As described thus far, the type system is quite strict: if ownership
has been completely transferred from one reference to another, the
refinement information found in the original reference is effectively useless.
Additionally, once
a mutable pointer has been split through an assignment or let expression, there is no
way to recover mutability. The typing rule for
must alias assertions, \rn{T-Alias} and \rn{T-AliasPtr}, overcomes this restriction
by exploiting the must-aliasing information to
``shuffle'' or redistribute ownerships \emph{and refinements} between two
aliased pointers. The typing rule assigns
two fresh types $ \tau'_{{\mathrm{1}}}  \TREF^{ r'_{{\mathrm{1}}} } $ and $ \tau'_{{\mathrm{2}}}  \TREF^{ r'_{{\mathrm{2}}} } $ to the two operand pointers.
The choice of $\tau'_{{\mathrm{1}}}, r'_{{\mathrm{1}}}, \tau'_{{\mathrm{2}}}$, and $r'_{{\mathrm{2}}}$ is left open
provided that the sum of the new types, $\ottsym{(}   \tau'_{{\mathrm{1}}}  \TREF^{ r'_{{\mathrm{1}}} }   \ottsym{)}  \ottsym{+}  \ottsym{(}   \tau'_{{\mathrm{2}}}  \TREF^{ r'_{{\mathrm{2}}} }   \ottsym{)}$ is
equivalent (denoted $ \approx $) to the sum of the original types.
Formally, $ \approx $ is defined as in \Cref{fig:pointer-typing}; it implies
that any refinements in the two types must be logically equivalent and
that ownerships must also be equal. 
This redistribution is sound precisely because the two references are
assumed to alias; the total ownership for the single memory cell pointed
to by both references cannot be increased by this shuffling. Further,
any refinements that hold for the contents of one reference must necessarily
hold for contents of the other and vice versa.

\begin{example}[Shuffling ownerships and refinements]
  Let $ \varphi_{= n } $ be $\nu \, \ottsym{=} \, n$.
\begin{lstlisting}
let x = mkref 5 in     // $\color{comment-green}{ \mathit{x} \COL   \set{  \nu  \COL \TINT \mid   \varphi_{= \ottsym{5} }  }   \TREF^{ \ottsym{1} }  }$
let y = x in           // $\color{comment-green}{\mathit{x}  \ottsym{:}  \top_{{\mathrm{1}}}  \ottsym{,}  \mathit{y}  \ottsym{:}    \set{  \nu  \COL \TINT \mid   \varphi_{= \ottsym{5} }  }   \TREF^{ \ottsym{1} } }$
  y := 4; alias(x = y) // $\color{comment-green}{\mathit{x}  \ottsym{:}    \set{  \nu  \COL \TINT \mid   \varphi_{= \ottsym{4} }  }   \TREF^{ \ottsym{0}  \ottsym{.}  \ottsym{5} }   \ottsym{,}  \mathit{y}  \ottsym{:}    \set{  \nu  \COL \TINT \mid   \varphi_{= \ottsym{4} }  }   \TREF^{ \ottsym{0}  \ottsym{.}  \ottsym{5} } }$    
\end{lstlisting}
  The final type assignment for $\mathit{x}$ and $\mathit{y}$ is justified by
  \begin{align*}
    & \top_{{\mathrm{1}}}  \ottsym{+}   \set{  \nu  \COL \TINT \mid   \varphi_{= \ottsym{4} }  }   \TREF^{ \ottsym{1} }   \ottsym{=}    \set{  \nu  \COL \TINT \mid    \top   \wedge   \varphi_{= \ottsym{4} }   }   \TREF^{ \ottsym{1} }  \approx \\
    \,\,\,\,&  \set{  \nu  \COL \TINT \mid    \varphi_{= \ottsym{4} }   \wedge   \varphi_{= \ottsym{4} }   }   \TREF^{ \ottsym{1} }   \ottsym{=}     \set{  \nu  \COL \TINT \mid   \varphi_{= \ottsym{4} }  }   \TREF^{ \ottsym{0}  \ottsym{.}  \ottsym{5} }   \ottsym{+}   \set{  \nu  \COL \TINT \mid   \varphi_{= \ottsym{4} }  }   \TREF^{ \ottsym{0}  \ottsym{.}  \ottsym{5} }  .
  \end{align*}
\end{example}

The aliasing rules give fine-grained control over ownership information. This
flexibility allows mutation through two or more aliased references within the same scope.
Provided sufficient aliasing annotations, the type system may shuffle ownerships between
one or more live references, enabling and disabling mutability as required. Although the
reliance on these annotations appears to decrease the practicality of our type system,
we expect these aliasing annotations can be inserted by a conservative must-aliasing
analysis. Further, empirical experience from our prior work \cite{suenaga2009fractional} indicates
that only a small number of annotations are required for larger programs.

\begin{example}[Shuffling Mutability]
  \label{exmp:shuffle-example}
  Let $ \varphi_{= n } $ again be $\nu \, \ottsym{=} \, n$.
  The following program uses two live, aliased references to mutate the same memory location:
\begin{lstlisting}
let x = mkref 0 in
let y = x in            // $\color{comment-green}{\mathit{x}  \ottsym{:}    \set{  \nu  \COL \TINT \mid   \varphi_{= \ottsym{0} }  }   \TREF^{ \ottsym{1} }   \ottsym{,}  \mathit{y}  \ottsym{:}  \top_{{\mathrm{1}}}}$
  x := 1; alias(x = y); // $\color{comment-green}{\mathit{x}  \ottsym{:}  \top_{{\mathrm{1}}}  \ottsym{,}  \mathit{y}  \ottsym{:}    \set{  \nu  \COL \TINT \mid   \varphi_{= \ottsym{1} }  }   \TREF^{ \ottsym{1} } }$
  y := 2; alias(x = y); // $\color{comment-green}{\mathit{x}  \ottsym{:}    \set{  \nu  \COL \TINT \mid   \varphi_{= \ottsym{2} }  }   \TREF^{ \ottsym{0}  \ottsym{.}  \ottsym{5} }   \ottsym{,}  \mathit{y}  \ottsym{:}    \set{  \nu  \COL \TINT \mid   \varphi_{= \ottsym{2} }  }   \TREF^{ \ottsym{0}  \ottsym{.}  \ottsym{5} } }$
  assert(*x = 2)
\end{lstlisting}
  After the first aliasing
  statement the type system shuffles the (exclusive) mutability between $\mathit{x}$
  and $\mathit{y}$ to enable the write to $\mathit{y}$. After the second aliasing statement
  the ownership in $\mathit{y}$ is split with $\mathit{x}$; note that
  transferring all ownership from $\mathit{y}$ to $\mathit{x}$ would also yield a
  valid typing.
\end{example}

Finally, we describe the subtyping rule. The rules for subtyping types
and environments are shown in \Cref{fig:subtyping}. For integer types,
the rules require the refinement of a supertype is a logical consequence of
the subtype's refinement conjoined with the lifting of $\Gamma$.
The subtype rule for references is \emph{covariant} in the type
of reference contents. It is widely known that in a language with unrestricted aliasing
and mutable references such a rule is unsound: after a write into the coerced
pointer, reads from an alias may yield a value disallowed by the alias' type
\cite{pierce2002types}. However, as in
the assign case, ownership types prevent unsoundness; a write to the
coerced pointer requires the pointer to have ownership 1, which guarantees
any aliased pointers have the maximal type and provide no information about their contents
beyond simple types.\looseness=-1

\begin{figure}[t]
  \leavevmode
    \infrule[T-Call]{
    \Theta  \ottsym{(}  \mathit{f}  \ottsym{)}  \ottsym{=}   \forall  \lambda .\tuple{ \mathit{x_{{\mathrm{1}}}} \COL \tau_{{\mathrm{1}}} ,\dots, \mathit{x_{\ottmv{n}}} \COL \tau_{\ottmv{n}} }\ra\tuple{ \mathit{x_{{\mathrm{1}}}} \COL \tau'_{{\mathrm{1}}} ,\dots, \mathit{x_{\ottmv{n}}} \COL \tau'_{\ottmv{n}}  \mid  \tau }  \\
    \sigma_{\alpha}  \ottsym{=}  \ottsym{[}  \ell  \ottsym{:}  \mathcal{L}  \ottsym{/}  \lambda  \ottsym{]} \andalso \sigma_{x}  \ottsym{=}    [  \mathit{y_{{\mathrm{1}}}}  /  \mathit{x_{{\mathrm{1}}}}  ]  \cdots  [  \mathit{y_{\ottmv{n}}}  /  \mathit{x_{\ottmv{n}}}  ]   \\
     \Theta   \mid   \mathcal{L}   \mid   \Gamma  \ottsym{[}  \mathit{y_{\ottmv{i}}}  \hookleftarrow  \sigma_{\alpha} \, \sigma_{x} \, \tau'_{\ottmv{i}}  \ottsym{]}  \ottsym{,}  \mathit{x}  \ottsym{:}  \sigma_{\alpha} \, \sigma_{x} \, \tau   \vdash   \ottnt{e}  :  \tau'   \produces   \Gamma'  \andalso
     \mathit{x}  \not\in   \DOM( \Gamma' )  
  }{
     \Theta   \mid   \mathcal{L}   \mid   \Gamma  \ottsym{[}  \mathit{y_{\ottmv{i}}}  \ottsym{:}  \sigma_{\alpha} \, \sigma_{x} \, \tau_{\ottmv{i}}  \ottsym{]}   \vdash    \LET  \mathit{x}  =   \mathit{f} ^ \ell (  \mathit{y_{{\mathrm{1}}}} ,\ldots, \mathit{y_{\ottmv{n}}}  )   \IN  \ottnt{e}   :  \tau'   \produces   \Gamma' 
  }
  \infrule[T-FunDef]{
    \Theta  \ottsym{(}  \mathit{f}  \ottsym{)}  \ottsym{=}   \forall  \lambda .\tuple{ \mathit{x_{{\mathrm{1}}}} \COL \tau_{{\mathrm{1}}} ,\dots, \mathit{x_{\ottmv{n}}} \COL \tau_{\ottmv{n}} }\ra\tuple{ \mathit{x_{{\mathrm{1}}}} \COL \tau'_{{\mathrm{1}}} ,\dots, \mathit{x_{\ottmv{n}}} \COL \tau'_{\ottmv{n}}  \mid  \tau }  \\
     \Theta   \mid   \lambda   \mid    \mathit{x_{{\mathrm{1}}}} \COL \tau_{{\mathrm{1}}} ,\ldots, \mathit{x_{\ottmv{n}}} \COL \tau_{\ottmv{n}}    \vdash   \ottnt{e}  :  \tau   \produces    \mathit{x_{{\mathrm{1}}}} \COL \tau'_{{\mathrm{1}}} ,\ldots, \mathit{x_{\ottmv{n}}} \COL \tau'_{\ottmv{n}}  
  }{
    \Theta  \vdash  \mathit{f}  \mapsto  \ottsym{(}  \mathit{x_{{\mathrm{1}}}}  \ottsym{,} \, .. \, \ottsym{,}  \mathit{x_{\ottmv{n}}}  \ottsym{)}  \ottnt{e}
  }
  \bcprulessavespacetrue
  \infrule[T-Funs]{
    \forall  \mathit{f}  \mapsto  \ottsym{(}  \mathit{x_{{\mathrm{1}}}}  \ottsym{,} \, .. \, \ottsym{,}  \mathit{x_{\ottmv{n}}}  \ottsym{)}  \ottnt{e}  \in  \ottnt{D} .\Theta  \vdash  \mathit{f}  \mapsto  \ottsym{(}  \mathit{x_{{\mathrm{1}}}}  \ottsym{,} \, .. \, \ottsym{,}  \mathit{x_{\ottmv{n}}}  \ottsym{)}  \ottnt{e} \\  \DOM( \ottnt{D} )   \ottsym{=}   \DOM( \Theta ) 
  }{
    \Theta  \vdash  \ottnt{D}
  }
  \hfil
  \infrule[T-Prog]{
    \Theta  \vdash  \ottnt{D} \andalso  \vdash _{\wf}  \Theta  \\
     \Theta   \mid    \epsilon    \mid    \bullet    \vdash   \ottnt{e}  :  \tau   \produces   \Gamma 
  }{
    \vdash   \tuple{ \ottnt{D} ,  \ottnt{e} } 
  }
  \bcprulessavespacefalse
\caption{Program typing rules}
\label{fig:progTyping}
\end{figure}

\subsection{Interprocedural Fragment and Context-Sensitivity}
\label{sec:cs}
We now turn to a discussion of the interprocedural fragment of our
language, and how our type system propagates context information. The remaining
typing rules for our language are shown in \Cref{fig:progTyping}.
These rules concern the typing of function calls, function bodies, and
entire programs.

We first explain the \rn{T-Call} rule. The rule uses two substitution
maps. $\sigma_{x}$ translates between the parameter names used in the
function type and actual argument names at the call-site. $\sigma_{\alpha}$
instantiates all occurrences of $\lambda$ in the callee type with
$\ell  \ottsym{:}  \mathcal{L}$, where $\ell$ is the label of the call-site and $\mathcal{L}$
the typing context of the call.  The types of the arguments $\mathit{y_{\ottmv{i}}}$'s
are required to match the parameter types (post substitution). The
body of the let binding is then checked with the argument types
updated to reflect the changes in the function call (again, post
substitution). This update is well-defined because we require all
function arguments be distinct as described in \Cref{sec:language}.
Intuitively, the substitution $\sigma_{\alpha}$ represents incrementally
refining the behavior of the callee function with partial context
information.  If $\mathcal{L}$ is itself a context variable $\lambda'$, this
substitution effectively transforms any context prefix queries over
$\lambda$ in the argument/return/output types into a queries over
$\ell  \ottsym{:}  \lambda'$.  In other words, while the exact concrete execution context
of the callee is unknown, the context must at least begin with $\ell$
which can potentially rule out certain behaviors.

Rule \rn{T-FunDef} type checks a function definition
$\mathit{f}  \mapsto  \ottsym{(}  \mathit{x_{{\mathrm{1}}}}  \ottsym{,} \, .. \, \ottsym{,}  \mathit{x_{\ottmv{n}}}  \ottsym{)}  \ottnt{e}$ against the function type given in $\Theta$.
As a convenience we assume that the parameter
names in the function type match the formal parameters in the function
definition. The rule checks that under an initial environment given
by the argument types the function body produces a value of the
return type and transforms the arguments according to the output types.
As mentioned above, functions may be executed under many different contexts,
so type checking the function body is performed under the context
variable $\lambda$ that occurs in the function type.

Finally, the rule for typing programs (\rn{T-Prog}) checks that
all function definitions are well typed under a well-formed
function type environment, and that the entry point $\ottnt{e}$
is well typed in an empty type environment and the typing context
$ \epsilon $, i.e., the initial context.

\begin{example}[1-CFA]
  \label{exmp:1cfa}
  Recall the program in \Cref{fig:running-example} in \Cref{sec:intro}; assume
  the function calls are labeled as follows:
\begin{lstlisting}
p := get$^{\ell1}$(p) + 1;
// ...
q := get$^{\ell2}$(q) + 1;
\end{lstlisting}
  Taking $\tau_{\ottmv{p}}$ to be the type shown in \Cref{exmp:cs-type-example}:
  \[
     \set{  \nu  \COL \TINT \mid   \ottsym{(}    \ell_{{\mathrm{1}}}     \preceq    \lambda   \implies  \nu \, \ottsym{=} \, \ottsym{3}  \ottsym{)}  \wedge  \ottsym{(}    \ell_{{\mathrm{2}}}     \preceq    \lambda   \implies  \nu \, \ottsym{=} \, \ottsym{5}  \ottsym{)}  } 
  \]
  we can give \lstinline{get} the type $ \forall  \lambda .\tuple{ \mathit{z}  \ottsym{:}   \tau_{\ottmv{p}}  \TREF^{ \ottsym{1} }  }\ra\tuple{ \mathit{z}  \ottsym{:}   \tau_{\ottmv{p}}  \TREF^{ \ottsym{1} }  \mid \tau_{\ottmv{p}} } $.
\end{example}
\begin{example}[2-CFA]
  To see how context information propagates across multiple calls, consider
  the following change to the code considered in \Cref{exmp:1cfa}:
  \begin{lstlisting}
    get_real(z) { *z }
    get(z) { get_real$^{\ell_{{\mathrm{3}}}}$(z) }
  \end{lstlisting}
  The type of \imp{get} remains as in \Cref{exmp:1cfa}, and taking $\tau$
  to be \[
     \set{  \nu  \COL \TINT \mid   \ottsym{(}    \ell_{{\mathrm{3}}} \, \ell_{{\mathrm{1}}}     \preceq    \lambda'   \implies  \nu \, \ottsym{=} \, \ottsym{3}  \ottsym{)}  \wedge  \ottsym{(}    \ell_{{\mathrm{3}}} \, \ell_{{\mathrm{2}}}     \preceq    \lambda'   \implies  \nu \, \ottsym{=} \, \ottsym{5}  \ottsym{)}  } 
  \]
  the type of \imp{get_real} is: $ \forall  \lambda' .\tuple{ \mathit{z}  \ottsym{:}   \tau  \TREF^{ \ottsym{1} }  }\ra\tuple{ \mathit{z}  \ottsym{:}   \tau  \TREF^{ \ottsym{1} }  \mid \tau } $.

  We focus on the typing of the call to \imp{get_real} in \imp{get}; it
  is typed in context $\lambda$ and a type environment where
  \imp{p} is given type $\tau_{\ottmv{p}}$ from \Cref{exmp:1cfa}.

  Applying the substitution $\ottsym{[}  \ell_{{\mathrm{3}}}  \ottsym{:}  \lambda  \ottsym{/}  \lambda'  \ottsym{]}$ to the argument
  type of \imp{get_real} yields:
  \begin{align*}
    &   \set{  \nu  \COL \TINT \mid   \ottsym{(}    \ell_{{\mathrm{3}}} \, \ell_{{\mathrm{1}}}     \preceq     \ell_{{\mathrm{3}}}  :  \lambda    \implies  \nu \, \ottsym{=} \, \ottsym{3}  \ottsym{)}  \wedge  \ottsym{(}    \ell_{{\mathrm{3}}} \, \ell_{{\mathrm{2}}}     \preceq     \ell_{{\mathrm{3}}}  :  \lambda    \implies  \nu \, \ottsym{=} \, \ottsym{5}  \ottsym{)}  }   \TREF^{ \ottsym{1} }  \approx  \\
    & \,\,\,\,\,   \set{  \nu  \COL \TINT \mid   \ottsym{(}    \ell_{{\mathrm{1}}}     \preceq    \lambda   \implies  \nu \, \ottsym{=} \, \ottsym{3}  \ottsym{)}  \wedge  \ottsym{(}    \ell_{{\mathrm{2}}}     \preceq    \lambda   \implies  \nu \, \ottsym{=} \, \ottsym{5}  \ottsym{)}  }   \TREF^{ \ottsym{1} } 
  \end{align*}
  which is exactly the type of \imp{p}.
  A similar derivation applies to
  the return type of \imp{get_real} and thus \imp{get}.
\end{example}

\subsection{Soundness}
\label{sec:soundness}
We have proven that any program that type checks according to
the rules above will never experience an assertion failure.
We formalize this claim with the following soundness theorem.

\begin{theorem}[Soundness]
  \label{thm:soundness}
  If  $\vdash   \tuple{ \ottnt{D} ,  \ottnt{e} } $, then
  $  \tuple{  \emptyset  ,   \emptyset  ,   \cdot  ,  \ottnt{e} }      \not \longrightarrow^*  _{ \ottnt{D} }     \mathbf{AssertFail}  $.
  
  Further, any well-typed program either diverges, halts in the configuration $ \mathbf{AliasFail} $, or
  halts in a configuration $ \tuple{ \ottnt{H} ,  \ottnt{R} ,   \cdot  ,  \mathit{x} } $ for some $\ottnt{H}, \ottnt{R}$ and $\mathit{x}$, i.e.,
  evaluation does not get stuck.
\end{theorem}
\begin{proof}[Sketch]
  By standard progress and preservation lemmas; 
  the full proof has been omitted for space reasons and can be found
  in the \appref{accompanying appendix}{full version \cite{toman2020consort}}.\looseness=-1

\end{proof}

\section{Inference and Extensions}
\label{sec:infr}
We now briefly describe the inference algorithm implemented in
our tool \name. We sketch some implemented extensions needed to type more
interesting programs and close with a discussion of current
limitations of our prototype.

\subsection{Inference}
Our tool first runs a standard, simple type inference
algorithm to generate type templates for every function parameter
type, return type, and for every live variable at each program point.
For a variable $x$ of simple type $\tau_{S} ::=  \TINT  \mid  \tau_{S}  \TREF $
at program point $p$, \name generates a type template $ \sem{ \tau_{S} }_{ \mathit{x} , \ottsym{0} , p } $
as follows:
\begin{align*}
   \sem{  \TINT  }_{ \mathit{x} , n , p }  =  \set{  \nu  \COL \TINT \mid    \varphi _{ \mathit{x} , n , p }  ( \nu ;   \ottkw{FV} _{ p }  )  }  &&  \sem{  \tau_{S}  \TREF  }_{ \mathit{x} , n , p }  =   \sem{ \tau_{S} }_{ \mathit{x} , n  \ottsym{+}  \ottsym{1} , p }   \TREF^{  r _{ \mathit{x} , n , p }  } 
\end{align*}
$  \varphi _{ \mathit{x} , n , p }  ( \nu ;   \ottkw{FV} _{ p }  ) $ denotes a fresh relation symbol
applied to $\nu$ and the free variables of simple type $ \TINT $
at program point $p$ (denoted $ \ottkw{FV} _{ p } $). $ r _{ \mathit{x} , n , p } $ is a
fresh ownership variable. For each function $\mathit{f}$,
there are two synthetic program points, $ { \mathit{f} ^{b} } $ and $ { \mathit{f} ^{e} } $
for the beginning and end of the function respectively.
At both points, \name generates type template for each argument, where $ \ottkw{FV} _{  { \mathit{f} ^{b} }  } $
and $ \ottkw{FV} _{  { \mathit{f} ^{e} }  } $ are the names of integer typed parameters.
At $ { \mathit{f} ^{e} } $, \name also generates a type template for the return
value.
We write $ \Gamma ^ p $ to indicate the type environment at point $p$,
where every variable is mapped to its corresponding type template.
$ \sem{  \Gamma ^ p  } $ is thus equivalent to
$ \bigwedge_{  \mathit{x}  \in   \ottkw{FV} _{ p }   }    \varphi _{ \mathit{x} , \ottsym{0} , p }  ( \mathit{x} ;   \ottkw{FV} _{ p }  )  $.

When generating these type templates, our implementation
also generates ownership well-formedness constraints. Specifically,
for a type template of the form $  \set{  \nu  \COL \TINT \mid    \varphi _{ \mathit{x} , n  \ottsym{+}  \ottsym{1} , p }  ( \nu ;   \ottkw{FV} _{ p }  )  }   \TREF^{  r _{ \mathit{x} , n , p }  } $
\name emits the constraint: $ r _{ \mathit{x} , n , p }   \ottsym{=}  \ottsym{0}  \implies    \varphi _{ \mathit{x} , n  \ottsym{+}  \ottsym{1} , p }  ( \nu ;   \ottkw{FV} _{ p }  ) $
and for a type template $ \ottsym{(}   \tau  \TREF^{  r _{ \mathit{x} , n  \ottsym{+}  \ottsym{1} , p }  }   \ottsym{)}  \TREF^{  r _{ \mathit{x} , n , p }  } $
\name emits the constraint $ r _{ \mathit{x} , n , p }   \ottsym{=}  \ottsym{0}  \implies   r _{ \mathit{x} , n  \ottsym{+}  \ottsym{1} , p }   \ottsym{=}  \ottsym{0}$.

\name then walks the program, generating constraints between
relation symbols and ownership variables according to
the typing rules. These constraints take three forms, ownership
constraints, subtyping constraints, and assertion constraints.
Ownership constraints are simple linear (in)equalities over ownership variables
and constants, according to conditions imposed by the typing
rules. For example, if variable $\mathit{x}$ has the type template
$ \tau  \TREF^{  r _{ \mathit{x} , \ottsym{0} , p }  } $ for the expression $ \mathit{x}  \WRITE  \mathit{y}  \SEQ  \ottnt{e} $ at point
$p$, \name generates the constraint $ r _{ \mathit{x} , \ottsym{0} , p }   \ottsym{=}  \ottsym{1}$.

\name emits subtyping constraints between the relation symbols
at related program points according to the rules of the type system.
For example, for the term $ \LET  \mathit{x}  =  \mathit{y}  \IN  \ottnt{e} $
at program point $p$ (where $\ottnt{e}$ is at program point $p'$, and $\mathit{x}$ has simple type
$  \TINT   \TREF $) \name generates the following subtyping constraint:
\[
   \sem{  \Gamma ^ p  }   \wedge    \varphi _{ \mathit{y} , \ottsym{1} , p }  ( \nu ;   \ottkw{FV} _{ p }  )   \implies    \varphi _{ \mathit{y} , \ottsym{1} , p' }  ( \nu ;   \ottkw{FV} _{ p' }  )   \wedge    \varphi _{ \mathit{x} , \ottsym{1} , p' }  ( \nu ;   \ottkw{FV} _{ p' }  ) 
\]
in addition to the ownership constraint $ r _{ \mathit{y} , \ottsym{0} , p }   \ottsym{=}    r _{ \mathit{y} , \ottsym{0} , p' }   \ottsym{+}  r _{ \mathit{x} , \ottsym{0} , p' } $.

Finally, for each \lstinline[mathescape]{assert($\varphi$)} in the program, \name emits an assertion
constraint of the form: $ \sem{  \Gamma ^ p  }   \implies  \varphi$
which requires the refinements on integer typed variables
in scope are sufficient to prove $\varphi$.

\paragraph{Encoding Context Sensitivity.}
To make inference tractable, we require the user to fix \emph{a
  priori} the maximum length of prefix queries to a constant $k$ (this
choice is easily controlled with a command line parameter to our tool). We
supplement the arguments in \emph{every} predicate application with a set of
integer context variables $ \mathit{c_{{\mathrm{1}}}} ,\ldots, \mathit{c_{\ottmv{k}}} $; these variables do not
overlap with any program variables.\looseness=-1

\name uses these variables to infer context sensitive
refinements as follows. Consider a function call
$ \LET  \mathit{x}  =   \mathit{f} ^ \ell (  \mathit{y_{{\mathrm{1}}}} ,\ldots, \mathit{y_{\ottmv{n}}}  )   \IN  \ottnt{e} $ at point $p$
where $\ottnt{e}$ is at point $p'$.
\name generates the following constraint for a refinement
$  \varphi _{ \mathit{y_{\ottmv{i}}} , n , p }  ( \nu ,  \mathit{c_{{\mathrm{1}}}} ,\ldots, \mathit{c_{\ottmv{k}}}  ;  \ottkw{FV} _{ p }  ) $ which occurs in the type template of $\mathit{y_{\ottmv{i}}}$:
\begin{align*}
  &     \varphi _{ \mathit{y_{\ottmv{i}}} , n , p }  ( \nu ,  \mathit{c_{{\mathrm{0}}}} ,\ldots, \mathit{c_{\ottmv{k}}}  ;  \ottkw{FV} _{ p }  )   \implies  \sigma_{x} \, \varphi _{ \mathit{x_{\ottmv{i}}} , n ,  { \mathit{f} ^{b} }  }  ( \nu , \ell  \ottsym{,}   \mathit{c_{{\mathrm{0}}}} ,\ldots, \mathit{c_{{\ottmv{k}-1}}}  ;  \ottkw{FV} _{  { \mathit{f} ^{b} }  }  )  \\
  &     \sigma_{x} \, \varphi _{ \mathit{x_{\ottmv{i}}} , n ,  { \mathit{f} ^{e} }  }  ( \nu , \ell  \ottsym{,}   \mathit{c_{{\mathrm{0}}}} ,\ldots, \mathit{c_{{\ottmv{k}-1}}}  ;  \ottkw{FV} _{  { \mathit{f} ^{e} }  }  )   \implies  \varphi _{ \mathit{y_{\ottmv{i}}} , n , p' }  ( \nu ,  \mathit{c_{{\mathrm{0}}}} ,\ldots, \mathit{c_{\ottmv{k}}}  ;  \ottkw{FV} _{ p' }  )  \\
  & \sigma_{x}  \ottsym{=}    [  \mathit{y_{{\mathrm{1}}}}  /  \mathit{x_{{\mathrm{1}}}}  ]  \cdots  [  \mathit{y_{\ottmv{n}}}  /  \mathit{x_{\ottmv{n}}}  ]  
\end{align*}
Effectively, we have encoded $   \ell_{{\mathrm{1}}} \ldots \ell_{\ottmv{k}}      \preceq    \lambda $ as $\land_{0<i\leq k} c_i = \ell_{\ottmv{i}}$.
In the above, the shift from $ \mathit{c_{{\mathrm{0}}}} ,\ldots, \mathit{c_{\ottmv{k}}} $ to $\ell  \ottsym{,}   \mathit{c_{{\mathrm{0}}}} ,\ldots, \mathit{c_{{\ottmv{k}-1}}} $ plays the role of
$\sigma_{\alpha}$ in the \rn{T-Call} rule. The above constraint serves to
determine the value of $\mathit{c_{{\mathrm{0}}}}$ within the body of the function $\mathit{f}$.
If $\mathit{f}$ calls another function $\mathit{g}$, the above rule
propagates this value of $\mathit{c_{{\mathrm{0}}}}$ to $\mathit{c_{{\mathrm{1}}}}$ within $\mathit{g}$ and so on.
The solver may then instantiate relation symbols with predicates
that are conditional over the values of $\mathit{c_{\ottmv{i}}}$.

\paragraph{Solving Constraints.}
The results of the above process are two systems of constraints;
real arithmetic constraints over
ownership variables and constrained Horn clauses (CHC) over the refinement relations.
Under certain assumptions about the simple types in a program, the size of the ownership
and subtyping constraints will be polynomial to the size of the program.
These systems are not independent; the relation constraints may mention the value
of ownership variables due to the well-formedness constraints described above.
The ownership constraints are first solved with Z3 \cite{de2008z3}. These constraints
are non-linear but Z3 appears
particularly well-engineered to quickly find solutions for the instances generated
by \name.
We constrain Z3 to maximize the number of non-zero ownership
variables to ensure as few refinements as possible are constrained to be $ \top $ by ownership well-formedness.

The values of ownership variables inferred by Z3 are then substituted into the
constrained Horn clauses, and the resulting system is checked
for satisfiability with an off-the-shelf CHC solver.
Our implementation generates constraints in the industry standard
SMT-Lib2 format \cite{BarFT-SMTLIB}; any solver that accepts
this format can be used as a backend for \name. Our implementation
currently supports Spacer \cite{komuravelli2013automatic} (part of the Z3 solver \cite{de2008z3}),
HoICE \cite{champion2018hoice}, and Eldarica \cite{rummer2013disjunctive}
(adding a new backend requires only a handful of lines of glue code).
We found that different solvers are better tuned to different problems;
we also implemented \emph{parallel mode} which runs all supported solvers
in parallel, using the first available result.

\subsection{Extensions}

\paragraph{Primitive Operations.}
As defined in \Cref{sec:prelim}, our language can compare integers to zero and load
and store them from memory, but can perform no meaningful computation
over these numbers. To promote the flexibility of our type system and simplify
our soundness statement, we do not
fix a set of primitive operations and their static semantics.
Instead, we assume any set of primitive operations
used in a program are given sound function types in $\Theta$.
For example, under the assumption that $+$ has its usual semantics
and the underlying logic supports $+$, we can give $+$ the type
$ \forall  \lambda .\tuple{ \mathit{x}  \ottsym{:}  \top_{{\mathrm{0}}}  \ottsym{,}  \mathit{y}  \ottsym{:}  \top_{{\mathrm{0}}} }\ra\tuple{ \mathit{x}  \ottsym{:}  \top_{{\mathrm{0}}}  \ottsym{,}  \mathit{y}  \ottsym{:}  \top_{{\mathrm{0}}} \mid  \set{  \nu  \COL \TINT \mid  \nu \, \ottsym{=} \, \mathit{x}  \ottsym{+}  \mathit{y} }  } $.
Interactions with a nondeterministic environment
or unknown program inputs can then be modeled with a primitive
that returns integers refined with $ \top $.

\paragraph{Dependent Tuples.}
Our implementation supports types of the form: $( \mathit{x_{{\mathrm{1}}}} : \tau_{{\mathrm{1}}}, \ldots,$ $\mathit{x_{\ottmv{n}}}:\tau_{\ottmv{n}})$,
where $\mathit{x_{\ottmv{i}}}$ can appear within $\tau_{\ottmv{j}}$
($j\neq i$) if $\tau_{\ottmv{i}}$ is an integer type. For example,
$\ottsym{(}   \mathit{x} \COL  \set{  \nu  \COL \TINT \mid   \top  }    \ottsym{,}  \mathit{y}  \ottsym{:}   \set{  \nu  \COL \TINT \mid  \nu \, \ottsym{>} \, \mathit{x} }   \ottsym{)}$ is the type of tuples whose
second element is strictly greater than the first. We also extend
the language with tuple constructors as a new value form, and
let bindings with tuple patterns as the LHS.

The extension to type checking is relatively straightforward; the only
significant extensions are to the subtyping rules.
Specifically, the subtyping check for a tuple element
$ \mathit{x_{\ottmv{i}}} \COL \tau_{\ottmv{i}} $ is performed in a type environment elaborated with the
types and names of other tuple elements.
The extension to type inference is also
straightforward; the arguments for a
predicate symbol include any enclosing dependent tuple names
and the environment in subtyping constraints is likewise extended.

\paragraph{Recursive Types.}
Our language also supports some unbounded heap structures via recursive
reference types. To keep inference tractable, we forbid nested
recursive types, multiple occurrences of the recursive type variable,
and additionally fix the shape of refinements that occur within a
recursive type. For recursive refinements that fit the
above restriction, our approach for refinements
is broadly similar to that in \cite{kawaguchi2009type}, and we use
the ownership scheme of \cite{suenaga2009fractional} for handling
ownership. We first use simple type inference
to infer the shape of the recursive types, and automatically insert
fold/unfold annotations into the source program. As in
\cite{kawaguchi2009type},
the refinements within an unfolding of a recursive type may refer to
dependent tuple names bound by the enclosing type. These recursive
types can express, e.g., the invariants of a mutable,
sorted list. As in \cite{suenaga2009fractional},
recursive types are unfolded once before assigning ownership variables;
further unfoldings copy existing ownership variables.\looseness=-1

As in Java or C++, our language does not support sum types,
and any instantiation of a recursive type must use a null pointer.
Our implementation supports an \lstinline{ifnull} construct
in addition to a distinguished \lstinline{null} constant. Our implementation allows any refinement
to hold for the null constant, including $\bot$. Currently, our implementation
does \emph{not} detect null pointer dereferences, and all soundness
guarantees are made modulo freedom of null dereferences. As $ \sem{ \Gamma } $
omits refinements under reference types, null pointer refinements do not affect
the verification of programs without null pointer dereferences.

\paragraph{Arrays.}
Our implementation supports arrays of integers. Each array is given
an ownership describing the ownership of memory allocated for the entire array.
The array type contains two refinements: the first refines the length of the array itself,
and the second refines the entire array contents.
The content refinement may refer to a symbolic index variable
for precise, per-index refinements. At reads and writes to the
array, \name instantiates the refinement's symbolic index variable with
the concrete index used at the read/write.

As in \cite{suenaga2009fractional}, our restriction to arrays of integers
stems from the difficulty of ownership inference. Soundly handling
pointer arrays requires index-wise tracking of ownerships which
significantly complicates automated inference. We leave supporting
arrays of pointers to future work.

\subsection{Limitations}
Our current approach is not complete; there are safe programs that will be
rejected by our type system.
As mentioned in \Cref{sec:types}, our well-formedness
condition forbids refinements that refer to memory locations. As a result,
\name cannot in general express, e.g., that the contents of two references are equal.
Further, due to our reliance on automated theorem provers
we are restricted to logics with sound but potentially incomplete decision procedures.
\name also does not support conditional or context-sensitive ownerships, and therefore
cannot precisely handle conditional mutation or aliasing.

\section{Experiments}
\label{sec:eval}

We now present the results of preliminary experiments performed with the implementation
described in \Cref{sec:infr}. The goal of these experiments was to answer the following questions: \begin{inparaenum}[i)]
\item is the type system (and extensions of \Cref{sec:infr}) expressive enough to type and verify non-trivial programs? and
\item is type inference feasible?
\end{inparaenum}\looseness=-1

To answer these questions, we evaluated our prototype implementation
on two sets of benchmarks.\footnote{Our experiments and the \name source code are available
  at \url{https://www.fos.kuis.kyoto-u.ac.jp/projects/consort/}.} 
The first set is adapted from JayHorn \cite{kahsai2017quantified,kahsai2016jayhorn},
a verification tool for Java.
This test suite contains a combination of 82 safe and unsafe programs written in
Java. We chose this benchmark suite as, like \name, JayHorn
is concerned with the automated verification of
programs in a language with mutable, aliased memory cells.
Further, although some of their benchmark programs tested Java
specific features, most could be adapted
into our low-level language.
The tests we could adapt provide a comparison with existing state-of-the-art
verification techniques.
A detailed breakdown of the adapted benchmark suite can be found in \Cref{tab:breakdown}.

\begin{table}[t]
  \caption{Description of benchmark suite adapted from JayHorn. \textbf{Java} are programs
    that test Java-specific features. \textbf{Inc} are tests that cannot be handled by \name, e.g.,
    null checking, etc. \textbf{Bug} includes a ``safe'' program
    we discovered was actually incorrect.}
  \begin{center}
    \begin{tabular}{cccccc}\toprule
\textbf{Set} & \textbf{Orig.} & \textbf{Adapted} & \textbf{Java} & \textbf{Inc} & \textbf{Bug} \\ \midrule
Safe & 41 & 32 & 6 & 2 & 1
\\
Unsafe & 41 & 26 & 13 & 2 & 0
\end{tabular}

  \end{center}
  \label{tab:breakdown}
\end{table}

\begin{remark}
  The original JayHorn paper includes two additional benchmark sets, Mine Pump and CBMC.
  Both our tool and recent JayHorn versions time out on the Mine Pump benchmark. Further,
  the CBMC tests were either subsumed by our own test programs, tested Java specific
  features, or tested program synthesis functionality. We therefore omitted both of these
  benchmarks from our evaluation.
\end{remark}

The second benchmark set consists of data structure
implementations and microbenchmarks written directly in our low-level
imperative language. We developed this suite to
test the expressive power of our type system and inference.
The programs included in this suite are:
\begin{itemize}
\item \textbf{Array-List} Implementation of an unbounded list backed by an array.
\item \textbf{Sorted-List} Implementation of a mutable, sorted list maintained with an in-place insertion sort algorithm.
\item \textbf{Shuffle} Multiple live references are used to mutate the same location in program memory as in \Cref{exmp:shuffle-example}.
\item \textbf{Mut-List} Implementation of general linked lists with a clear operation.
\item \textbf{Array-Inv} A program which allocates a length $n$ array and writes the value $i$ at every index $i$.
\item \textbf{Intro2} The motivating program shown in \Cref{fig:hard-loop} in \Cref{sec:intro}.
\end{itemize}
We introduced unsafe mutations to these programs to
check our tool for unsoundness and translated these programs
into Java for further comparison with JayHorn.\looseness=-1

Our benchmarks and JayHorn's require a small number of trivially identified alias annotations.
The adapted JayHorn benchmarks contain a total of \jhaliascount{}
annotations; the most for any individual test was \jhmaxalias.
The number of annotations required for our benchmark suite are shown in column
\textbf{Ann.} of \Cref{tab:jh-results}.

We first ran \name on each program in our benchmark suite and ran version 0.7 of
JayHorn on the corresponding Java version. We recorded the final verification
result for both our tool and JayHorn. We also collected the
end-to-end runtime of \name for each test;
we do not give a performance comparison with
JayHorn given the many differences in target languages.
For the JayHorn suite, we first ran our tool on the adapted version of
each test program and ran JayHorn on the original Java version. We also did not
collect runtime information for this set of experiments because our goal
is a comparison of tool precision, not performance.
All tests were run on a machine with 16 GB RAM and 4 Intel i5 CPUs at 2GHz
and with a timeout of 60 seconds (the same timeout
was used in \cite{kahsai2017quantified}). We used \name's parallel backend (\Cref{sec:infr})
with Z3 version 4.8.4, HoICE version 1.8.1, and Eldarica version 2.0.1 and JayHorn's Eldarica backend.

\begin{table}[t]
  \caption{Comparison of \name to JayHorn on the benchmark set of \cite{kahsai2017quantified} (top) and our custom benchmark suite (bottom). \emph{T/O}
  indicates a time out.}
  \begin{center}
    \begin{tabular}{ccccccc}\toprule
& & \multicolumn{2}{l}{\textbf{\name}} & \multicolumn{3}{c}{\textbf{JayHorn}} \\
\cmidrule(lr){3-4} \cmidrule(lr){5-7}
\textbf{Set} & \textbf{N. Tests} & \emph{Correct} & \emph{T/O} & \emph{Correct} & \emph{T/O} & \emph{Imp.} \\ \midrule
\textbf{Safe} & 32 & 29 & 3 &  24 &  5 & 3\\
\textbf{Unsafe} & 26 & 26 & 0 &  19 &  0 & 7
\end{tabular}

    \begin{tabular}{lcccc|lcccc}\toprule
\textbf{Name} & \textbf{Safe?} & \textbf{Time(s)} & \textbf{Ann} & \textbf{JH} & \textbf{Name} & \textbf{Safe?} & \textbf{Time(s)} & \textbf{Ann} & \textbf{JH} \\ \midrule
\textbf{Array-Inv} & \checkmark & 10.07 & 0 & \text{T/O} &
\textbf{Array-Inv-BUG} & \text{\sffamily X} & 5.29 & 0 & \text{T/O} \\
\textbf{Array-List} & \checkmark & 16.76 & 0 & \text{T/O} &
\textbf{Array-List-BUG} & \text{\sffamily X} & 1.13 & 0 & \text{T/O} \\
\textbf{Intro2} & \checkmark & 0.08 & 0 & \text{T/O} &
\textbf{Intro2-BUG} & \text{\sffamily X} & 0.02 & 0 & \text{T/O} \\
\textbf{Mut-List} & \checkmark & 1.45 & 3 & \text{T/O} &
\textbf{Mut-List-BUG} & \text{\sffamily X} & 0.41 & 3 & \text{T/O} \\
\textbf{Shuffle} & \checkmark & 0.13 & 3 & \checkmark &
\textbf{Shuffle-BUG} & \text{\sffamily X} & 0.07 & 3 & \text{\sffamily X} \\
\textbf{Sorted-List} & \checkmark & 1.90 & 3 & \text{T/O} &
\textbf{Sorted-List-BUG} & \text{\sffamily X} & 1.10 & 3 & \text{T/O} \\
\end{tabular}

  \end{center}
  \label{tab:jh-results}
\end{table}

\subsection{Results}
The results of our experiments are shown in \Cref{tab:jh-results}. On the JayHorn benchmark
suite \name performs competitively with JayHorn, correctly identifying 29 of the 32 safe programs
as such. For all 3 tests on which \name timed out after 60 seconds, JayHorn also timed out
(column \emph{T/O}).
For the unsafe programs, \name correctly identified all programs as unsafe within 60 seconds;
JayHorn answered \textsc{Unknown} for 7 tests (column \emph{Imp.}).

On our own benchmark set, \name correctly verifies all safe versions
of the programs within 60 seconds. For the unsafe variants,
\name was able to quickly and definitively determine these programs unsafe.
JayHorn times out on all tests except for \textbf{Shuffle} and \textbf{ShuffleBUG} (column \textbf{JH}).
We investigated the cause of time outs and discovered that after verification failed
with an unbounded heap model, JayHorn attempts verification on increasingly
larger bounded heaps. In every case, JayHorn exceeded the 60 second timeout
before reaching a preconfigured limit on the heap bound. This result suggests JayHorn struggles in the presence
of per-object invariants and unbounded allocations; 
the only two tests JayHorn successfully analyzed contain just a single object allocation.\looseness=-1

We do not believe
this struggle is indicative of a shortcoming in JayHorn's implementation,
but stems from the fundamental limitations of JayHorn's memory representation.
Like many verification tools (see \Cref{sec:rw}), JayHorn uses a single, unchanging
invariant to for every object allocated at the same syntactic location;
effectively, all objects allocated at the same location
are assumed to alias with one another.
This representation cannot, in general, handle programs with different
invariants for distinct objects that evolve over time.
We hypothesize other tools that adopt a similar approach will exhibit
the same difficulty.

\section{Related Work}
\label{sec:rw}

The difficulty in handling programs with mutable references and
aliasing has been well-studied. Like JayHorn, many approaches model the heap explicitly
at verification time, approximating concrete heap locations with allocation site
labels \cite{kahsai2017quantified,kahsai2016jayhorn,chugh2012dependent,rondon2010low,fink2008effective};
each \emph{abstract location} is also associated with a refinement.
As abstract locations summarize many concrete locations, this approach
does not in general admit strong updates and flow-sensitivity; in particular, the refinement
associated with an abstract location is fixed for the lifetime of the program.
The techniques cited above include various workarounds for this limitation. For
example, \cite{rondon2010low,chugh2012dependent} temporarily allows
breaking these invariants through a distinguished program name
as long as the abstract location is not accessed through another name.
The programmer must therefore eventually bring the invariant back in sync with the summary location.
As a result, these systems ultimately cannot precisely handle programs that require
evolving invariants on mutable memory.\looseness=-1

A similar approach was taken in
CQual \cite{foster2002flow} by Aiken et al. \cite{aiken2003checking}.
They used an explicit \emph{restrict} binding for pointers.
Strong updates are permitted through pointers bound
with \emph{restrict}, but the program is forbidden from
using any pointers which share an allocation site while the restrict binding is live.

A related technique used in the field of object-oriented verification is
to declare object invariants at the class level and allow these
invariants on object fields to be broken during a limited period of time
\cite{barnett2011specification,flanagan2002extended}. In particular, the work on
Spec\# \cite{barnett2011specification} uses an ownership system which tracks
whether object $a$ owns object $b$; like \name's ownership system,
these ownerships contain the effects of mutation. However, Spec\#'s ownership
is quite strict and does not admit references to $b$ outside of the owning object $a$.

Viper \cite{heule2013verification,mueller2016viper}
(and its related projects \cite{heule2013abstract,leino2010deadlock}) uses
access annotations (expressed as permission predicates) to explicitly transfer access/mutation
permissions for references between static program names. Like \name,
permissions may be fractionally transferred, allowing temporary shared, immutable
access to a mutable memory cell. However,
while \name automatically infers many ownership transfers, Viper requires extensive
annotations for each transfer.

F*, a dependently typed dialect of ML, includes an update/select
theory of heaps and requires explicit annotations summarizing the heap
effects of a method
\cite{protzenko2017verified,swamy2016dependent,swamy2013verifying}.
This approach enables modular reasoning and precise specification of
pre- and post-conditions with respect to the heap, but precludes full automation.

The work on rely--guarantee reference types by Gordon et
al. \cite{gordon2013rely,gordon2017verifying} uses refinement
types in a language mutable references and aliasing.
Their approach extends reference types with rely/guarantee predicates;
the rely predicate describes possible mutations via aliases,
and the guarantee predicate describes the admissible mutations through the current reference.
If two references may alias, then the guarantee predicate of one
reference implies the rely predicate of the other and vice versa. 
This invariant is maintained with a splitting operation that
is similar to our $+$ operator.
Further, their type system allows strong updates to reference refinements
provided the new refinements are preserved by the rely predicate.
Thus, rely--guarantee refinement support multiple
mutable, aliased references with non-trivial refinement information. Unfortunately this
expressiveness comes at the cost of automated inference and verification; an embedding
of this system into Liquid Haskell \cite{vazou2014refinement} described in
\cite{gordon2017verifying} was forced to sacrifice strong updates.

Work by Degen et al. \cite{degen2007tracking} introduced linear \emph{state annotations} to Java.
To effect strong updates in the presence of aliasing, like \name,
their system requires annotated memory locations are mutated only through
a distinguished reference. Further, all aliases of this mutable reference give
no information about the state of the object much like our $0$ ownership pointers.
However, their system cannot handle multiple, immutable aliases with non-trivial
annotation information; \emph{only} the mutable reference may
have non-trivial annotation information.

The fractional ownerships in \name and their counterparts in
\cite{suenaga2009fractional,suenaga2012type} have a clear relation to
linear type systems. Many authors have explored the use of linear type
systems to reason in contexts with aliased mutable references
\cite{fahndrich2002adoption,deline2001enforcing,smith2000alias}, and
in particular with the goal of supporting strong updates
\cite{ahmed20073}. A closely related approach is RustHorn by
Matsushita et al. \cite{matsushita2020rusthorn}. Much like
\name, RustHorn uses CHC and linear aliasing information for
the sound and---unlike \name---complete verification of
programs with aliasing and mutability. However, their approach depends on Rust's
strict \emph{borrowing discipline}, and cannot handle programs
where multiple aliased references are used 
in the same lexical region. In contrast, \name supports
fine-grained, per-statement changes in mutability and even further
control with \imp{alias} annotations, which allows it to verify larger
classes of programs.

The ownerships of \name also have
a connection to separation logic \cite{reynolds2002separation};
the separating conjunction isolates write effects to local subheaps,
while \name's ownership system isolates effects to local updates of pointer types.
Other researchers have used separation logic to precisely
support strong updates of abstract state.
For example, in work by Kloos et al. \cite{kloos2015asynchronous}
resources are associated with static, abstract names; each
resource (represented by its static name) may be owned (and thus,
mutated) by exactly one thread. Unlike \name, their ownership
system forbids even temporary immutable,
shared ownership, or transferring ownerships at arbitrary program points.
An approach proposed by Bakst and Jhala \cite{bakst2016predicate}
uses a similar technique,
combining separation logic with refinement types. Their approach gives
allocated memory cells abstract names, and associates these names with
refinements in an abstract heap. Like the approach of Kloos et al.
and \name's ownership 1 pointers, they
ensure these abstract locations are distinct
in all concrete heaps, enabling sound, strong updates.

The idea of using a rational number to express permissions to access a
reference dates back to the type system of \emph{fractional
permissions} by Boyland~\cite{DBLP:conf/sas/Boyland03}.  His work used
fractional permissions to verify race freedom of a concurrent program
without a may-alias analysis.  Later,
Terauchi~\cite{terauchi2008checking} proposed a type-inference
algorithm that reduces typing constraints to a set of linear
inequalities over rational numbers.  Boyland's idea also inspired a
variant of separation logic for a concurrent programming
language~\cite{DBLP:conf/popl/BornatCOP05} to express sharing of read
permissions among several threads.  Our previous
work~\cite{suenaga2009fractional,suenaga2012type}, inspired by that
in \cite{terauchi2008checking,DBLP:conf/popl/BornatCOP05}, proposed methods for type-based
verification of resource-leak freedom, in which a rational number
expresses an \emph{obligation} to deallocate certain resource, not
just a permission.\looseness=-1

The issue of context-sensitivity (sometimes called \emph{polyvariance})
is well-studied in the field of abstract interpretation
(e.g., \cite{kashyap2014jsai,hardekopf2014widening,shivers1991control,smaragdakis2011pick,milanova2005parameterized}, see \cite{gilray2013survey}
for a recent survey).
Polyvariance has also been used in type systems to assign
different behaviors to the same function depending on its call site
\cite{banerjee1997modular,wells2002calculus,amtoft2000faithful}.
In the area of refinement type systems, Zhu and Jagannathan developed
a context-sensitive dependent type system for a functional language \cite{zhu2013compositional}
that indexed function types by unique labels attached to call-sites. Our
context-sensitivity approach was inspired by this work.
In fact, we could have formalized context-polymorphism within the framework of full dependent
types, but chose the current presentation for simplicity.\looseness=-1

\section{Conclusion}
\label{sec:concl}

We presented \name, a novel type system for safety verification of imperative programs with mutability
and aliasing. \name is built upon the novel combination of fractional ownership types and refinement
types. Ownership types flow-sensitively and precisely track the existence of mutable aliases. \name
admits sound strong updates by discarding refinement information on mutably-aliased references as
indicated by ownership types.
Our type system is amenable
to automatic type inference; we have implemented a prototype of this inference tool and found it can verify several
non-trivial programs and outperforms a state-of-the-art program verifier.
As an area of future work, we plan to investigate using fractional ownership types to soundly allow refinements that mention memory
locations.\looseness=-1

\paragraph{\footnotesize Acknowledgments}
{\footnotesize The authors would like to the reviewers for their thoughtful
feedback and suggestions, and Yosuke Fukuda and Alex Potanin for their feedback on early drafts.
This work was supported in part by JSPS KAKENHI, grant numbers JP15H05706 and JP19H04084,
and in part by the JST ERATO MMSD Project.\looseness=-1}

\bibliographystyle{splncs04}
\bibliography{main}

\vfill

{\small\medskip\noindent{\bf Open Access} This chapter is licensed under the terms of the Creative Commons\break Attribution 4.0 International License (\url{http://creativecommons.org/licenses/by/4.0/}), which permits use, sharing, adaptation, distribution and reproduction in any medium or format, as long as you give appropriate credit to the original author(s) and the source, provide a link to the Creative Commons license and indicate if changes were made.}

{\small \spaceskip .28em plus .1em minus .1em The images or other third party material in this chapter are included in the chapter's Creative Commons license, unless indicated otherwise in a credit line to the material.~If material is not included in the chapter's Creative Commons license and your intended\break use is not permitted by statutory regulation or exceeds the permitted use, you will need to obtain permission directly from the copyright holder.}

\medskip\noindent\includegraphics{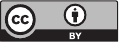}
\newpage
\iffullversion\appendix

\begin{figure}
  \leavevmode
  \infrule[TE-Seq]{
    \Theta  \mid  \HOLE  \ottsym{:}  \tau  \produces  \Gamma  \mid  \mathcal{L}  \vdash_{\mathit{ectx} }  \ottnt{E}  \ottsym{:}  \tau'  \produces  \Gamma' \andalso
     \Theta   \mid   \mathcal{L}   \mid   \Gamma'   \vdash   \ottnt{e}  :  \tau''   \produces   \Gamma'' 
  }{
    \Theta  \mid  \HOLE  \ottsym{:}  \tau  \produces  \Gamma  \mid  \mathcal{L}  \vdash_{\mathit{ectx} }   \ottnt{E} \SEQ \ottnt{e}   \ottsym{:}  \tau''  \produces  \Gamma''
  }
  \infrule[TE-Hole]{}{
    \Theta  \mid  \HOLE  \ottsym{:}  \tau  \produces  \Gamma  \mid  \mathcal{L}  \vdash_{\mathit{ectx} }  \HOLE  \ottsym{:}  \tau  \produces  \Gamma
  }
  \infrule[TE-Stack]{
    \Theta  \mid  \HOLE  \ottsym{:}  \tau'  \produces  \Gamma'  \mid  \mathcal{L}  \vdash_{\mathit{ectx} }  \ottnt{E}  \ottsym{:}  \tau''  \produces  \Gamma'' \\
     \Theta   \mid   \mathcal{L}   \mid   \Gamma  \ottsym{,}  \mathit{x}  \ottsym{:}  \tau   \vdash   \ottnt{e}  :  \tau'   \produces   \Gamma'  \\
     \mathit{x}  \not\in   \DOM( \Gamma' )  
  }{
    \Theta  \mid  \HOLE  \ottsym{:}  \tau  \produces  \Gamma  \mid  \mathcal{L}  \vdash_{\mathit{ectx} }   \ottnt{E} [\LET  \mathit{x}  =   \HOLE^ \ell   \IN  \ottnt{e}  ]   \ottsym{:}  \tau''  \produces  \Gamma''
  }

  \begin{align*}
    \ottsym{(}   \ottnt{E} \SEQ \ottnt{e}   \ottsym{)}  \ottsym{[}  \ottnt{e'}  \ottsym{]} & =  \ottnt{E}  \ottsym{[}  \ottnt{e'}  \ottsym{]}  \SEQ  \ottnt{e}  \\
    \HOLE  \ottsym{[}  \ottnt{e'}  \ottsym{]} & = e' \\
     \ottnt{E} [\LET  \mathit{y}  =   \HOLE^ \ell   \IN  \ottnt{e}  ]   \ottsym{[}  \mathit{x}  \ottsym{]} & = \ottnt{E}  \ottsym{[}   \LET  \mathit{y}  =  \mathit{x}  \IN  \ottnt{e}   \ottsym{]}
  \end{align*}
\caption{Context typing and substitution}
\label{fig:context-typing}
\end{figure}

\begin{figure}
  \leavevmode
  \infrule{
    \oldvec{\ell}  \ottsym{=}  \ottkw{Trace} \, \ottsym{(}  \oldvec{F}  \ottsym{)} \andalso n  \ottsym{=}   | \oldvec{\ell} |   \ottsym{=}   | \oldvec{F} |  \andalso \Theta  \vdash  \ottnt{D} \andalso \forall j \in \set{1..n}.\oldvec{\ell}_{\ottmv{j}} = tail^{n-j+1}(\oldvec{\ell}) \\
    \ottkw{Cons} \, \ottsym{(}  \ottnt{H}  \ottsym{,}  \ottnt{R}  \ottsym{,}  \Gamma  \ottsym{)} \andalso
    \forall i\in\set{1..n}.\Theta  \mid  \HOLE  \ottsym{:}  \tau_{\ottmv{i}}  \produces  \Gamma_{\ottmv{i}}  \mid  \oldvec{\ell}_{\ottmv{i}}  \vdash_{\mathit{ectx} }  F_{\ottmv{i}}  \ottsym{:}  \tau_{{\ottmv{i}-1}}  \produces  \Gamma_{{\ottmv{i}-1}} \\
    \oldvec{F} = F_{\ottmv{n}} : \cdots : F_{{\mathrm{1}}} : \cdot \andalso
     \Theta   \mid   \oldvec{\ell}   \mid   \Gamma   \vdash   \ottnt{e}  :  \tau_{\ottmv{n}}   \produces   \Gamma_{\ottmv{n}} 
  }{
     \vdash_{\mathit{conf} }^D   \tuple{ \ottnt{H} ,  \ottnt{R} ,  \oldvec{F} ,  \ottnt{e} }  
  }
  \infax[]{
     \vdash_{\mathit{conf} }^D   \mathbf{AliasFail}  
  }
  \begin{align*}
    \ottkw{Trace} \, \ottsym{(}   \cdot   \ottsym{)} &=  \epsilon  \\
    \ottkw{Trace} \, \ottsym{(}   \ottnt{E} [\LET  \mathit{x}  =   \HOLE^ \ell   \IN  \ottnt{e}  ]   \ottsym{:}  \oldvec{F}  \ottsym{)} &= \ell  \ottsym{:}  \ottkw{Trace} \, \ottsym{(}  \oldvec{F}  \ottsym{)} \\
   \end{align*}
   \begin{align*}
    \ottkw{Cons} \, \ottsym{(}  \ottnt{H}  \ottsym{,}  \ottnt{R}  \ottsym{,}  \Gamma  \ottsym{)} & \stackrel{\textrm{\tiny def}}{\iff} \ottkw{SAT} \, \ottsym{(}  \ottnt{H}  \ottsym{,}  \ottnt{R}  \ottsym{,}  \Gamma  \ottsym{)} \wedge \forall \,  \ottmv{a}  \in \DOM( \ottnt{H} )   \ottsym{.}  \ottkw{Own} \, \ottsym{(}  \ottnt{H}  \ottsym{,}  \ottnt{R}  \ottsym{,}  \Gamma  \ottsym{)}  \ottsym{(}  \ottmv{a}  \ottsym{)} \leq 1 \\
    \ottkw{SAT} \, \ottsym{(}  \ottnt{H}  \ottsym{,}  \ottnt{R}  \ottsym{,}  \Gamma  \ottsym{)} & \stackrel{\textrm{\tiny def}}{\iff} \forall \,  \mathit{x}  \in \DOM( \Gamma )   \ottsym{.}   \mathit{x}  \in \DOM( \ottnt{R} )   \wedge   \ottkw{SATv} ( \ottnt{H} , \ottnt{R} , \ottnt{R}  \ottsym{(}  \mathit{x}  \ottsym{)} , \Gamma  \ottsym{(}  \mathit{x}  \ottsym{)} )  \\
     \ottkw{SATv} ( \ottnt{H} , \ottnt{R} , v , \tau )  & \stackrel{\textrm{\tiny def}}{\iff} \begin{cases}
       v  \in  \mathbb{Z}   \wedge  \ottsym{[}  \ottnt{R}  \ottsym{]} \, \ottsym{[}  v  \ottsym{/}  \nu  \ottsym{]}  \varphi & \tau  \ottsym{=}   \set{  \nu  \COL \TINT \mid  \varphi }  \\
       \ottmv{a}  \in \DOM( \ottnt{H} )   \wedge   \ottkw{SATv} ( \ottnt{H} , \ottnt{R} , \ottnt{H}  \ottsym{(}  \ottmv{a}  \ottsym{)} , \tau' )  & \tau  \ottsym{=}   \tau'  \TREF^{ r }   \wedge  v \, \ottsym{=} \, \ottmv{a}
    \end{cases} \\
    \ottsym{[}   \emptyset   \ottsym{]} \, \varphi & = \varphi \\
    \ottsym{[}  \ottnt{R}  \ottsym{\{}  \mathit{y}  \mapsto  n  \ottsym{\}}  \ottsym{]} \, \varphi & =  \ottsym{[}  \ottnt{R}  \ottsym{]} \, \ottsym{[}  n  \ottsym{/}  \mathit{y}  \ottsym{]}  \varphi \\
    \ottsym{[}  \ottnt{R}  \ottsym{\{}  \mathit{y}  \mapsto  \ottmv{a}  \ottsym{\}}  \ottsym{]} \, \varphi & =  \ottsym{[}  \ottnt{R}  \ottsym{]} \, \varphi \\
    \ottkw{Own} \, \ottsym{(}  \ottnt{H}  \ottsym{,}  \ottnt{R}  \ottsym{,}  \Gamma  \ottsym{)} & =    \Sigma _{ \mathit{x} \in  \DOM( \Gamma )  }\, \ottkw{own} \, \ottsym{(}  \ottnt{H}  \ottsym{,}  \ottnt{R}  \ottsym{(}  \mathit{x}  \ottsym{)}  \ottsym{,}  \Gamma  \ottsym{(}  \mathit{x}  \ottsym{)}  \ottsym{)}  \\
    \ottkw{own} \, \ottsym{(}  \ottnt{H}  \ottsym{,}  \ottnt{v}  \ottsym{,}  \tau  \ottsym{)} & = \begin{cases}
      \ottsym{\{}  \ottmv{a}  \mapsto  r  \ottsym{\}}  \ottsym{+}  \ottkw{own} \, \ottsym{(}  \ottnt{H}  \ottsym{,}  \ottnt{H}  \ottsym{(}  \ottmv{a}  \ottsym{)}  \ottsym{,}  \tau'  \ottsym{)} & \ottnt{v} \, \ottsym{=} \, \ottmv{a}  \wedge   \ottmv{a}  \in \DOM( \ottnt{H} )   \wedge  \tau  \ottsym{=}   \tau'  \TREF^{ r }  \\
       \emptyset  & o.w.
    \end{cases} \\
  \end{align*}
\caption{Machine state typing}
\label{fig:state-typing}
\end{figure}

\section{Proof of Type Soundness (\Cref{thm:soundness})}

We first define a typing relation for machine configurations $ \tuple{ \ottnt{H} ,  \ottnt{R} ,  \oldvec{F} ,  \ottnt{e} } $ as shown in
\Cref{fig:state-typing}. The critical component of this typing relation is the consistency relation $\ottkw{Cons}$.
Intuitively, $\ottkw{Cons}$ expresses that the current heap and registers are consistent with the ownership
and refinement information implied by $\Gamma$. We say triple $( \ottnt{H},\ottnt{R},\Gamma )$ is \emph{consistent},
and write $\ottkw{Cons} \, \ottsym{(}  \ottnt{H}  \ottsym{,}  \ottnt{R}  \ottsym{,}  \Gamma  \ottsym{)}$. %
In the definitions for $\ottkw{own}$ we write $\ottsym{\{}  \ottmv{a}  \mapsto  r  \ottsym{\}}$ to denote a function $\textbf{Addr} \rightarrow \mathbb{Q}_{\geq 0}$
which returns $r$ for $\ottmv{a}$, and $0$ otherwise. We write $ \emptyset $ to denote a constant
function $\textbf{Addr} \rightarrow \mathbb{Q}_{\geq 0}$ which always returns $0$. We define the addition between
two functions $\ottnt{O_{{\mathrm{1}}}}, \ottnt{O_{{\mathrm{2}}}}: \textbf{Addr} \rightarrow \mathbb{Q}_{\geq 0}$ as: $\ottsym{(}  \ottnt{O_{{\mathrm{1}}}}  \ottsym{+}  \ottnt{O_{{\mathrm{2}}}}  \ottsym{)}  \ottsym{(}  \ottmv{a}  \ottsym{)}  \ottsym{=}  \ottnt{O_{{\mathrm{1}}}}  \ottsym{(}  \ottmv{a}  \ottsym{)}  \ottsym{+}  \ottnt{O_{{\mathrm{2}}}}  \ottsym{(}  \ottmv{a}  \ottsym{)}$.
Individual ownership annotations remain rationals in $[0,1]$, but an intermediate
pointwise sum need not itself lie in $[0,1]$; the final inequality in $\ottkw{Cons}$
is what bounds the total live ownership.
Finally, if a summation $ \Sigma $ has no summands, we take its result to be $ \emptyset $.

The proof of \Cref{thm:soundness} requires the following four key lemmas. These lemmas are stated
with respect to some well-typed program $ \tuple{ \ottnt{D} ,  \ottnt{e} } $, i.e. $\vdash   \tuple{ \ottnt{D} ,  \ottnt{e} } $.
\begin{lemma}
  \label{lem:initial}
  $ \vdash_{\mathit{conf} }^D   \tuple{  \emptyset  ,   \emptyset  ,   \cdot  ,  \ottnt{e} }  $
\end{lemma}
\begin{proof}
  Trivial, taking $\Gamma  \ottsym{=}   \bullet $ and by inversion on $\vdash   \tuple{ \ottnt{D} ,  \ottnt{e} } $.
\end{proof}

\begin{lemma}
  \label{lem:assertfail}
  $ \vdash_{\mathit{conf} }^D  \mathbf{C} $ implies $ \mathbf{C}  \neq   \mathbf{AssertFail}  $
\end{lemma}
\begin{proof}
  Simple proof by contradiction, as the $ \mathbf{AssertFail} $ is not well-typed.
\end{proof}
  
\begin{lemma}
  \label{lem:preservation}
  If $ \vdash_{\mathit{conf} }^D   \tuple{ \ottnt{H} ,  \ottnt{R} ,  \oldvec{F} ,  \ottnt{e} }  $ and $  \tuple{ \ottnt{H} ,  \ottnt{R} ,  \oldvec{F} ,  \ottnt{e} }     \longrightarrow _{ \ottnt{D} }    \mathbf{C} $, then $ \vdash_{\mathit{conf} }^D  \mathbf{C} $
\end{lemma}
\begin{lemma}
  \label{lem:progress}
  If $ \vdash_{\mathit{conf} }^D  \mathbf{C} $, then, one of the following conditions hold:
  \begin{enumerate}
  \item $\exists \mathbf{C}', \mathbf{C}    \longrightarrow _{ \ottnt{D} }    \mathbf{C}' $, or
  \item $\mathbf{C}  \ottsym{=}   \mathbf{AliasFail} $, or
  \item $\mathbf{C}  \ottsym{=}   \tuple{ \ottnt{H} ,  \ottnt{R} ,   \cdot  ,  \mathit{x} } $
  \end{enumerate}
\end{lemma}
\Cref{lem:preservation,lem:progress} are the heart of proof effort, we give their proofs in \Cref{sec:preservation-proof,sec:progress-proof}
respectively.

We can now prove \Cref{thm:soundness}:
\begin{proof}[\Cref{thm:soundness}: Soundness]
  From \Cref{lem:initial,lem:preservation} and an inductive argument, any configuration
  reachable from the initial state must be well-typed. Then, by \Cref{lem:assertfail} every
  configuration reachable from the initial state cannot be $ \mathbf{AssertFail} $, i.e., a well-typed
  program never experiences an assertion failure. This completes the first part of the proof.

  To prove the second portion of the theorem, it suffices to show that any configuration reachable
  from the initial state can step or is a final configuration. Again from
  \Cref{lem:initial,lem:preservation} and a simple inductive argument, we must have that
  for any state   $\mathbf{C}$ such that $  \tuple{  \emptyset  ,   \emptyset  ,   \cdot  ,  \ottnt{e} }     \longrightarrow^* _{ \ottnt{D} }    \mathbf{C} $ $ \vdash_{\mathit{conf} }^D  \mathbf{C} $.
  Then by \Cref{lem:progress} we have the configuration may step or is one of the final configurations.
\end{proof}

The remainder of this appendix proves \Cref{lem:preservation,lem:progress}. We introduce some auxiliary definitions and
lemmas in \Cref{sec:aux-defn-lem}, give the proof of \Cref{lem:preservation} in \Cref{sec:preservation-proof}, and prove
\Cref{lem:progress} in \Cref{sec:progress-proof}.

\section{Auxiliary Lemmas and Definitions}
\label{sec:aux-defn-lem}

\begin{figure}[t]
  \bcprulessavespacetrue
  \scriptsize
  \leavevmode
  \begin{center}
  \infrule[WF-Env]{
    \forall \,  \mathit{x}  \in \DOM( \Gamma )   \ottsym{.}   \mathcal{L}   \mid   \Gamma   \vdash _{\wf}  \Gamma  \ottsym{(}  \mathit{x}  \ottsym{)} 
  }{
     \mathcal{L}   \vdash _{\wf}  \Gamma 
  }
  \infrule[WF-Int]{
     \mathcal{L}   \mid   \Gamma   \vdash _{\wf}  \varphi 
  }{
     \mathcal{L}   \mid   \Gamma   \vdash _{\wf}   \set{  \nu  \COL \TINT \mid  \varphi }   
  }
  \infrule[WF-Ref]{
     \mathcal{L}   \mid   \Gamma   \vdash _{\wf}  \tau 
  }{
     \mathcal{L}   \mid   \Gamma   \vdash _{\wf}   \tau  \TREF^{ r }  
  }
  \infrule[WF-Phi]{
    \forall \,  \mathit{x}  \in   \ottkw{FPV} \, \ottsym{(}  \varphi  \ottsym{)}  \setminus   \set{ \nu }     \ottsym{.}  \Gamma  \ottsym{(}  \mathit{x}  \ottsym{)}  \ottsym{=}   \set{  \nu  \COL \TINT \mid \_ }  \\
    \mathbf{FCV} \, \ottsym{(}  \varphi  \ottsym{)} \subseteq \ottkw{CV} \, \ottsym{(}  \mathcal{L}  \ottsym{)}
  }{
     \mathcal{L}   \mid   \Gamma   \vdash _{\wf}  \varphi 
  }
  \infrule[WF-Result]{
     \mathcal{L}   \mid   \Gamma   \vdash _{\wf}  \tau  \andalso
     \mathcal{L}   \vdash _{\wf}  \Gamma 
  }{
     \mathcal{L}   \vdash _{\wf}  \tau   \produces   \Gamma 
  }
  \infrule[WF-FunType]{
     \lambda   \vdash _{\wf}   \mathit{x_{{\mathrm{1}}}} \COL \tau_{{\mathrm{1}}} ,\ldots, \mathit{x_{\ottmv{n}}} \COL \tau_{\ottmv{n}}   \\  \lambda   \vdash _{\wf}  \tau   \produces    \mathit{x_{{\mathrm{1}}}} \COL \tau'_{{\mathrm{1}}} ,\ldots, \mathit{x_{\ottmv{n}}} \COL \tau'_{\ottmv{n}}  
  }{
     \vdash _{\wf}   \forall  \lambda .\tuple{ \mathit{x_{{\mathrm{1}}}} \COL \tau_{{\mathrm{1}}} ,\dots, \mathit{x_{\ottmv{n}}} \COL \tau_{\ottmv{n}} }\ra\tuple{ \mathit{x_{{\mathrm{1}}}} \COL \tau'_{{\mathrm{1}}} ,\dots, \mathit{x_{\ottmv{n}}} \COL \tau'_{\ottmv{n}}  \mid  \tau }  
  }
  \infrule[WF-FunEnv]{
    \forall \,  \mathit{f}  \in \DOM( \Theta )   \ottsym{.}   \vdash _{\wf}  \Theta  \ottsym{(}  \mathit{f}  \ottsym{)} 
  }{
     \vdash _{\wf}  \Theta 
  }
  \end{center}
  \bcprulessavespacefalse
  \[
    \begin{array}{lrcl}
      \text{Simple Types} & \tau_{S} & ::= &  \TINT  \mid  \tau_{S}  \TREF  \\
      \text{Erasure} &  \llparenthesis   \set{  \nu  \COL \TINT \mid  \varphi }  \rrparenthesis  & = &  \TINT  \\
                      &  \llparenthesis   \tau  \TREF^{ r }  \rrparenthesis  & = &   \llparenthesis  \tau \rrparenthesis   \TREF 
    \end{array}
  \]
  \[
    \begin{array}{rrl}
      \text{Free Ctxt Vars} & \mathbf{FCV} \, \ottsym{(}  \varphi_{{\mathrm{1}}}  \vee  \varphi_{{\mathrm{2}}}  \ottsym{)} & =  \mathbf{FCV} \, \ottsym{(}  \varphi_{{\mathrm{1}}}  \ottsym{)}  \cup  \mathbf{FCV} \, \ottsym{(}  \varphi_{{\mathrm{2}}}  \ottsym{)}  \\
                            & \mathbf{FCV} \, \ottsym{(}   \neg  \varphi   \ottsym{)} & = \mathbf{FCV} \, \ottsym{(}  \varphi  \ottsym{)} \\
                            & \mathbf{FCV} \, \ottsym{(}  \widehat{v}_{{\mathrm{1}}} \, \ottsym{=} \, \widehat{v}_{{\mathrm{2}}}  \ottsym{)} & = \mathbf{FCV} \, \ottsym{(}  \phi  \ottsym{(}  \widehat{v}_{{\mathrm{1}}}  \ottsym{,} \, .. \, \ottsym{,}  \widehat{v}_{\ottmv{n}}  \ottsym{)}  \ottsym{)}  \ottsym{=}   \emptyset  \\
                            & \mathbf{FCV} \, \ottsym{(}    \oldvec{\ell}     \preceq    \mathcal{C}   \ottsym{)} & = \mathbf{FCV} \, \ottsym{(}  \mathcal{C}  \ottsym{)} \\
                            & \mathbf{FCV} \, \ottsym{(}   \ell  :  \mathcal{C}   \ottsym{)} & = \mathbf{FCV} \, \ottsym{(}  \mathcal{C}  \ottsym{)} \\
                            & \mathbf{FCV} \, \ottsym{(}  \mathcal{L}  \ottsym{)} & = \ottkw{CV} \, \ottsym{(}  \mathcal{L}  \ottsym{)} \\
      \text{Ctxt Vars} & \ottkw{CV} \, \ottsym{(}  \oldvec{\ell}  \ottsym{)} & =  \emptyset  \\
                            & \ottkw{CV} \, \ottsym{(}  \lambda  \ottsym{)} & =  \set{ \lambda } 
    \end{array}
  \]
\caption{Well-formedness of types and environments. }
  \label{fig:type-wf}
\end{figure}

The well-formedness rules omitted from the main paper are found in \Cref{fig:type-wf}. We write
$ \mathcal{L}   \vdash _{\wf}  \tau   \produces   \Gamma $ as shorthand for $ \mathcal{L}   \vdash _{\wf}  \Gamma $ and $ \mathcal{L}   \mid   \Gamma   \vdash _{\wf}  \tau $.

We first prove that the subtyping relations are transitive.

\begin{lemma} %
  \label{lem:subtype-transitive}
  \leavevmode
  \begin{enumerate}
  \item \label{part:sub-env-impl} If $\Gamma  \leq  \Gamma'$ then $\models   \sem{ \Gamma }   \implies   \sem{ \Gamma' } $.
  \item \label{part:single-env-trans} If $\Gamma  \vdash  \tau_{{\mathrm{1}}}  \leq  \tau_{{\mathrm{2}}}$ and $\Gamma  \vdash  \tau_{{\mathrm{2}}}  \leq  \tau_{{\mathrm{3}}}$, then $\Gamma  \vdash  \tau_{{\mathrm{1}}}  \leq  \tau_{{\mathrm{3}}}$
  \item \label{part:env-sup-subtype} If $\Gamma  \leq  \Gamma'$ and $\Gamma'  \vdash  \tau_{{\mathrm{1}}}  \leq  \tau_{{\mathrm{2}}}$, then $\Gamma  \vdash  \tau_{{\mathrm{1}}}  \leq  \tau_{{\mathrm{2}}}$
  \item \label{part:env-sub-trans} If $\Gamma  \leq  \Gamma'$,  $\Gamma  \vdash  \tau_{{\mathrm{1}}}  \leq  \tau_{{\mathrm{2}}}$, and $\Gamma'  \vdash  \tau_{{\mathrm{2}}}  \leq  \tau_{{\mathrm{3}}}$, then $\Gamma  \vdash  \tau_{{\mathrm{1}}}  \leq  \tau_{{\mathrm{3}}}$.
  \item If $\Gamma  \leq  \Gamma'$ and $\Gamma'  \leq  \Gamma''$, then $\Gamma  \leq  \Gamma''$.
  \end{enumerate}
\end{lemma}
\begin{proof}
  \leavevmode
  \begin{enumerate}
  \item It suffices to show that $\models   \sem{ \Gamma }   \implies  \ottsym{[}  \mathit{x}  \ottsym{/} \, \nu \, \ottsym{]} \, \varphi'$ for any $ \mathit{x}  \in   \DOM( \Gamma' )  $ where $\Gamma'  \ottsym{(}  \mathit{x}  \ottsym{)}  \ottsym{=}   \set{  \nu  \COL \TINT \mid  \varphi' } $.
    From $\Gamma  \leq  \Gamma'$ we have $\models   \sem{ \Gamma }   \wedge  \varphi  \implies  \varphi'$ where $\Gamma  \ottsym{(}  \mathit{x}  \ottsym{)}  \ottsym{=}   \set{  \nu  \COL \TINT \mid  \varphi } $.
    We must then have $\models   \sem{ \Gamma }   \wedge  \ottsym{[}  \mathit{x}  \ottsym{/} \, \nu \, \ottsym{]} \, \varphi  \implies  \ottsym{[}  \mathit{x}  \ottsym{/} \, \nu \, \ottsym{]} \, \varphi'$. From the definition of $ \sem{ \Gamma } $ we have $ \sem{ \Gamma }   \wedge  \ottsym{[}  \mathit{x}  \ottsym{/} \, \nu \, \ottsym{]} \, \varphi  \iff   \sem{ \Gamma } $, giving the desired result.
  \item By induction on $\Gamma  \vdash  \tau_{{\mathrm{1}}}  \leq  \tau_{{\mathrm{2}}}$. We only consider the base case where $\tau_{{\mathrm{1}}}  \ottsym{=}   \set{  \nu  \COL \TINT \mid  \varphi_{{\mathrm{1}}} } $ and $\tau_{{\mathrm{2}}}  \ottsym{=}   \set{  \nu  \COL \TINT \mid  \varphi_{{\mathrm{2}}} } $,
    the case for reference types follows from the induction hypothesis.
    By further inversion on $\Gamma  \vdash  \tau_{{\mathrm{2}}}  \leq  \tau_{{\mathrm{3}}}$ we therefore have:
    
    \begin{bcpcasearray}
      \tau_{{\mathrm{3}}}  \ottsym{=}   \set{  \nu  \COL \TINT \mid  \varphi_{{\mathrm{3}}} }  & \\
      \models   \sem{ \Gamma }   \wedge  \varphi_{{\mathrm{1}}}  \implies  \varphi_{{\mathrm{2}}} & \models   \sem{ \Gamma }   \wedge  \varphi_{{\mathrm{2}}}  \implies  \varphi_{{\mathrm{3}}}
    \end{bcpcasearray}
    
    From which it is immediate that we must have $\models   \sem{ \Gamma }   \wedge  \varphi_{{\mathrm{1}}}  \implies  \varphi_{{\mathrm{3}}}$, whereby \rn{S-Int} gives $\Gamma  \vdash  \tau_{{\mathrm{1}}}  \leq  \tau_{{\mathrm{3}}}$.
  \item By induction on $\Gamma'  \vdash  \tau_{{\mathrm{1}}}  \leq  \tau_{{\mathrm{2}}}$. The case for reference types is immediate from the inductive hypothesis, we focus
    on the base case where $\tau_{{\mathrm{1}}}  \ottsym{=}   \set{  \nu  \COL \TINT \mid  \varphi_{{\mathrm{1}}} } $ and $\tau_{{\mathrm{2}}}  \ottsym{=}   \set{  \nu  \COL \TINT \mid  \varphi_{{\mathrm{2}}} } $, and where $\models   \sem{ \Gamma' }   \wedge  \varphi_{{\mathrm{1}}}  \implies  \varphi_{{\mathrm{2}}}$.
    From $\Gamma  \leq  \Gamma'$ and \Cref{part:sub-env-impl} above, we have $\models   \sem{ \Gamma }   \implies   \sem{ \Gamma' } $ from which we can derive $ \sem{ \Gamma }   \wedge  \varphi_{{\mathrm{1}}}  \implies  \varphi_{{\mathrm{2}}}$, i.e.,
    $\Gamma  \vdash  \tau_{{\mathrm{1}}}  \leq  \tau_{{\mathrm{2}}}$.
  \item Immediate from \Cref{part:single-env-trans,part:env-sup-subtype}.
  \item Immediate corollary of \Cref{part:env-sub-trans}.
  \end{enumerate}
\end{proof}

\begin{definition}
  A value $v$ reaches an integer with $n$ dereferences in heap $\ottnt{H}$ when it is in the relation  $ \ottnt{H} \vdash   v  \Downarrow  n $
  defined as the smallest relation closed under the following rules:
  \begin{enumerate}
  \item If $ v  \in  \mathbb{Z} $ then $ \ottnt{H} \vdash   v  \Downarrow  \ottsym{0} $
  \item If $ \ottnt{H} \vdash   v  \Downarrow  n $ and $\ottnt{H}  \ottsym{(}  \ottmv{a}  \ottsym{)} \, \ottsym{=} \, v$ then $ \ottnt{H} \vdash   \ottmv{a}  \Downarrow  n  \ottsym{+}  \ottsym{1} $
  \end{enumerate}
  We will write $ \ottnt{H} \vdash   v  \Downarrow   | \tau |  $ to indicate a value $v$ is shape consistent with $\tau$ in heap $\ottnt{H}$, where
  $ | \tau | $ is the number of reference constructors in the type $\tau$.
\end{definition}

It is straightforward to show that $\ottmv{n}$ is uniquely determined by $\ottnt{H}$ and $\ottnt{v}$.

We also prove a standard inversion lemma to
handle the fact our typing rules are not syntax directed.

\begin{lemma}[Inversion] %
  \label{lem:inversion}
  If 
  $ \Theta   \mid   \mathcal{L}   \mid   \Gamma   \vdash   \ottnt{e_{{\mathrm{0}}}}  :  \tau   \produces   \Gamma' $,
  then there exists some $\Gamma_{\ottmv{p}}$, $\tau_{\ottmv{p}}$, and $\Gamma'_{\ottmv{p}}$
  such that $\Gamma  \leq  \Gamma_{\ottmv{p}}$, $ \mathcal{L}   \vdash _{\wf}  \Gamma_{\ottmv{p}} $, $\Gamma'_{\ottmv{p}}  \ottsym{,}  \tau_{\ottmv{p}}  \leq  \Gamma'  \ottsym{,}  \tau$, and:
  \begin{enumerate}
  \item If $\ottnt{e_{{\mathrm{0}}}}  \ottsym{=}  \mathit{x}$ then $\Gamma_{\ottmv{p}}  \ottsym{(}  \mathit{x}  \ottsym{)}  \ottsym{=}  \tau_{\ottmv{p}}  \ottsym{+}  \tau'$,  $\Gamma'_{\ottmv{p}}  \ottsym{=}  \Gamma_{\ottmv{p}}  \ottsym{[}  \mathit{x}  \hookleftarrow  \tau'  \ottsym{]}$.
  \item If $\ottnt{e_{{\mathrm{0}}}}  \ottsym{=}   \LET  \mathit{x}  =  \mathit{y}  \IN  \ottnt{e} $, then
    $ \Theta   \mid   \mathcal{L}   \mid   \Gamma_{\ottmv{p}}  \ottsym{[}  \mathit{y}  \hookleftarrow   \tau_{{\mathrm{1}}}  \wedge_{ \mathit{y} }   \mathit{y}  =_{ \tau_{{\mathrm{1}}} }  \mathit{x}    \ottsym{]}  \ottsym{,}  \mathit{x}  \ottsym{:}  \ottsym{(}   \tau_{{\mathrm{2}}}  \wedge_{ \mathit{x} }   \mathit{x}  =_{ \tau_{{\mathrm{2}}} }  \mathit{y}    \ottsym{)}   \vdash   \ottnt{e}  :  \tau_{\ottmv{p}}   \produces   \Gamma'_{\ottmv{p}} $
    and $ \mathit{x}  \not\in   \DOM( \Gamma'_{\ottmv{p}} )  $ where $\Gamma_{\ottmv{p}}  \ottsym{(}  \mathit{y}  \ottsym{)}  \ottsym{=}  \tau_{{\mathrm{1}}}  \ottsym{+}  \tau_{{\mathrm{2}}}$.
  \item If $\ottnt{e_{{\mathrm{0}}}}  \ottsym{=}   \LET  \mathit{x}  =  n  \IN  \ottnt{e} $ then $ \Theta   \mid   \mathcal{L}   \mid   \Gamma_{\ottmv{p}}  \ottsym{,}  \mathit{x}  \ottsym{:}   \set{  \nu  \COL \TINT \mid  \nu \, \ottsym{=} \, n }    \vdash   \ottnt{e}  :  \tau_{\ottmv{p}}   \produces   \Gamma'_{\ottmv{p}} $ and $ \mathit{x}  \not\in   \DOM( \Gamma'_{\ottmv{p}} )  $.
  \item If $\ottnt{e_{{\mathrm{0}}}}  \ottsym{=}   \IFZERO  \mathit{x}  \THEN  \ottnt{e_{{\mathrm{1}}}}  \ELSE  \ottnt{e_{{\mathrm{2}}}} $ then:
    \begin{itemize}
    \item $\Gamma_{\ottmv{p}}  \ottsym{(}  \mathit{x}  \ottsym{)}  \ottsym{=}   \set{  \nu  \COL \TINT \mid  \varphi } $
    \item $ \Theta   \mid   \mathcal{L}   \mid   \Gamma_{\ottmv{p}}  \ottsym{[}  \mathit{x}  \hookleftarrow   \set{  \nu  \COL \TINT \mid   \varphi  \wedge  \nu \, \ottsym{=} \, \ottsym{0}  }   \ottsym{]}   \vdash   \ottnt{e_{{\mathrm{1}}}}  :  \tau_{\ottmv{p}}   \produces   \Gamma'_{\ottmv{p}} $
    \item $ \Theta   \mid   \mathcal{L}   \mid   \Gamma_{\ottmv{p}}  \ottsym{[}  \mathit{x}  \hookleftarrow   \set{  \nu  \COL \TINT \mid   \varphi  \wedge  \nu \, \neq \, \ottsym{0}  }   \ottsym{]}   \vdash   \ottnt{e_{{\mathrm{2}}}}  :  \tau_{\ottmv{p}}   \produces   \Gamma'_{\ottmv{p}} $
    \end{itemize}
  \item If $\ottnt{e_{{\mathrm{0}}}}  \ottsym{=}   \LET  \mathit{x}  =   \MKREF  \mathit{y}   \IN  \ottnt{e} $, then $\Gamma_{\ottmv{p}}  \ottsym{(}  \mathit{y}  \ottsym{)}  \ottsym{=}  \tau_{{\mathrm{1}}}  \ottsym{+}  \tau_{{\mathrm{2}}}$, $ \Theta   \mid   \mathcal{L}   \mid   \Gamma_{\ottmv{p}}  \ottsym{[}  \mathit{y}  \hookleftarrow  \tau_{{\mathrm{1}}}  \ottsym{]}  \ottsym{,}  \mathit{x}  \ottsym{:}   \ottsym{(}   \tau_{{\mathrm{2}}}  \wedge_{ \mathit{x} }   \mathit{x}  =_{ \tau_{{\mathrm{2}}} }  \mathit{y}    \ottsym{)}  \TREF^{ \ottsym{1} }    \vdash   \ottnt{e}  :  \tau_{\ottmv{p}}   \produces   \Gamma'_{\ottmv{p}} $, and $ \mathit{x}  \not\in   \DOM( \Gamma'_{\ottmv{p}} )  $
  \item If $\ottnt{e_{{\mathrm{0}}}}  \ottsym{=}   \LET  \mathit{x}  =   *  \mathit{y}   \IN  \ottnt{e} $, then:
    \begin{itemize}
    \item $\Gamma_{\ottmv{p}}  \ottsym{(}  \mathit{y}  \ottsym{)}  \ottsym{=}   \tau_{{\mathrm{1}}}  \ottsym{+}  \tau_{{\mathrm{2}}}  \TREF^{ r } $
    \item $ \Theta   \mid   \mathcal{L}   \mid   \Gamma_{\ottmv{p}}  \ottsym{[}  \mathit{y}  \hookleftarrow   \tau''  \TREF^{ r }   \ottsym{]}  \ottsym{,}  \mathit{x}  \ottsym{:}  \tau_{{\mathrm{2}}}   \vdash   \ottnt{e}  :  \tau_{\ottmv{p}}   \produces   \Gamma'_{\ottmv{p}} $
    \item $ \mathit{x}  \not\in \DOM( \Gamma'_{\ottmv{p}} ) $
    \item \[
        \tau'' = \begin{cases}
          \ottsym{(}   \tau_{{\mathrm{1}}}  \wedge_{ \mathit{y} }   \mathit{y}  =_{ \tau_{{\mathrm{1}}} }  \mathit{x}    \ottsym{)} &  r   \ottsym{>}   \ottsym{0}  \\
          \tau_{{\mathrm{1}}} & r  \ottsym{=}  \ottsym{0}
        \end{cases}
      \]
    \end{itemize}
  \item If $\ottnt{e_{{\mathrm{0}}}}  \ottsym{=}   \LET  \mathit{x}  =   \mathit{f} ^ \ell (  \mathit{y_{{\mathrm{1}}}} ,\ldots, \mathit{y_{\ottmv{n}}}  )   \IN  \ottnt{e} $ then:
    \begin{itemize}
    \item $\Gamma_{\ottmv{p}}  \ottsym{(}  \mathit{y_{\ottmv{i}}}  \ottsym{)}  \ottsym{=}  \sigma_{\alpha} \, \sigma_{x} \, \tau_{\ottmv{i}}$ for each $i \in \set{1,\ldots,n}$
    \item $ \Theta   \mid   \mathcal{L}   \mid   \Gamma_{\ottmv{p}}  \ottsym{[}  \mathit{y_{\ottmv{i}}}  \hookleftarrow  \sigma_{\alpha} \, \sigma_{x} \, \tau'_{\ottmv{i}}  \ottsym{]}  \ottsym{,}  \mathit{x}  \ottsym{:}  \sigma_{\alpha} \, \sigma_{x} \, \tau   \vdash   \ottnt{e}  :  \tau_{\ottmv{p}}   \produces   \Gamma'_{\ottmv{p}} $
    \item $\Theta  \ottsym{(}  \mathit{f}  \ottsym{)}  \ottsym{=}   \forall  \lambda .\tuple{ \mathit{x_{{\mathrm{1}}}} \COL \tau_{{\mathrm{1}}} ,\dots, \mathit{x_{\ottmv{n}}} \COL \tau_{\ottmv{n}} }\ra\tuple{ \mathit{x_{{\mathrm{1}}}} \COL \tau'_{{\mathrm{1}}} ,\dots, \mathit{x_{\ottmv{n}}} \COL \tau'_{\ottmv{n}}  \mid  \tau } $
    \item $\sigma_{\alpha}  \ottsym{=}  \ottsym{[}  \ell  \ottsym{:}  \mathcal{L}  \ottsym{/}  \lambda  \ottsym{]}$
    \item $\sigma_{x}  \ottsym{=}    [  \mathit{y_{{\mathrm{1}}}}  /  \mathit{x_{{\mathrm{1}}}}  ]  \cdots  [  \mathit{y_{\ottmv{n}}}  /  \mathit{x_{\ottmv{n}}}  ]  $
    \item $ \mathit{x}  \not\in   \DOM( \Gamma'_{\ottmv{p}} )  $
    \end{itemize}
  \item If $\ottnt{e_{{\mathrm{0}}}}  \ottsym{=}   \mathit{y}  \WRITE  \mathit{x}  \SEQ  \ottnt{e} $ then:
    \begin{itemize}
    \item $\mathit{x} \, \neq \, \mathit{y}$
    \item $\Gamma_{\ottmv{p}}  \ottsym{(}  \mathit{x}  \ottsym{)}  \ottsym{=}  \tau_{{\mathrm{1}}}  \ottsym{+}  \tau_{{\mathrm{2}}}$
    \item $\Gamma_{\ottmv{p}}  \ottsym{(}  \mathit{y}  \ottsym{)}  \ottsym{=}   \tau'  \TREF^{ \ottsym{1} } $
    \item $ \Theta   \mid   \mathcal{L}   \mid   \Gamma_{\ottmv{p}}  \ottsym{[}  \mathit{x}  \hookleftarrow  \tau_{{\mathrm{1}}}  \ottsym{]}  \ottsym{[}  \mathit{y}  \hookleftarrow   \ottsym{(}   \tau_{{\mathrm{2}}}  \wedge_{ \mathit{y} }   \mathit{y}  =_{ \tau_{{\mathrm{2}}} }  \mathit{x}    \ottsym{)}  \TREF^{ \ottsym{1} }   \ottsym{]}   \vdash   \ottnt{e}  :  \tau_{\ottmv{p}}   \produces   \Gamma'_{\ottmv{p}} $
    \item The shapes of $\tau'$ and $\tau_{{\mathrm{2}}}$ are similar, i.e, $ \llparenthesis  \tau' \rrparenthesis   \ottsym{=}   \llparenthesis  \tau_{{\mathrm{2}}} \rrparenthesis $.
    \end{itemize}
  \item If $\ottnt{e_{{\mathrm{0}}}}  \ottsym{=}   \ALIAS( \mathit{x}  =  \mathit{y} ) \SEQ  \ottnt{e} $ then there exist some $\tau_{{\mathrm{1}}}, \tau_{{\mathrm{2}}}, \tau'_{{\mathrm{1}}}, \tau'_{{\mathrm{2}}}, r_{{\mathrm{1}}}, r_{{\mathrm{2}}}, r'_{{\mathrm{1}}}, r'_{{\mathrm{2}}}$ such that:
    \begin{itemize}
    \item $\mathit{x} \, \neq \, \mathit{y}$
    \item $  \tau_{{\mathrm{1}}}  \TREF^{ r_{{\mathrm{1}}} }   \ottsym{+}  \tau_{{\mathrm{2}}}  \TREF^{ r_{{\mathrm{2}}} }   \approx    \tau'_{{\mathrm{1}}}  \TREF^{ r'_{{\mathrm{1}}} }   \ottsym{+}  \tau'_{{\mathrm{2}}}  \TREF^{ r'_{{\mathrm{2}}} } $
    \item $\Gamma_{\ottmv{p}}  \ottsym{(}  \mathit{x}  \ottsym{)}  \ottsym{=}   \tau_{{\mathrm{1}}}  \TREF^{ r_{{\mathrm{1}}} } $ and $\Gamma_{\ottmv{p}}  \ottsym{(}  \mathit{y}  \ottsym{)}  \ottsym{=}   \tau_{{\mathrm{2}}}  \TREF^{ r_{{\mathrm{2}}} } $
    \item $ \Theta   \mid   \mathcal{L}   \mid   \Gamma_{\ottmv{p}}  \ottsym{[}  \mathit{x}  \hookleftarrow   \tau'_{{\mathrm{1}}}  \TREF^{ r'_{{\mathrm{1}}} }   \ottsym{]}  \ottsym{[}  \mathit{y}  \hookleftarrow   \tau'_{{\mathrm{2}}}  \TREF^{ r'_{{\mathrm{2}}} }   \ottsym{]}   \vdash   \ottnt{e}  :  \tau_{\ottmv{p}}   \produces   \Gamma'_{\ottmv{p}} $
    \end{itemize}
  \item If $\ottnt{e_{{\mathrm{0}}}}  \ottsym{=}   \ALIAS( \mathit{x}  = *  \mathit{y} ) \SEQ  \ottnt{e} $ then there exist some $\tau_{{\mathrm{1}}}, \tau_{{\mathrm{2}}}, \tau'_{{\mathrm{1}}}, \tau'_{{\mathrm{2}}}, r_{{\mathrm{1}}}, r_{{\mathrm{2}}}, r'_{{\mathrm{1}}}, r'_{{\mathrm{2}}}, r$, such that:
    \begin{itemize}
    \item $  \tau_{{\mathrm{1}}}  \TREF^{ r_{{\mathrm{1}}} }   \ottsym{+}  \tau_{{\mathrm{2}}}  \TREF^{ r_{{\mathrm{2}}} }   \approx    \tau'_{{\mathrm{1}}}  \TREF^{ r'_{{\mathrm{1}}} }   \ottsym{+}  \tau'_{{\mathrm{2}}}  \TREF^{ r'_{{\mathrm{2}}} } $
    \item $\Gamma_{\ottmv{p}}  \ottsym{(}  \mathit{x}  \ottsym{)}  \ottsym{=}   \tau_{{\mathrm{1}}}  \TREF^{ r_{{\mathrm{1}}} } $ and $\Gamma_{\ottmv{p}}  \ottsym{(}  \mathit{y}  \ottsym{)}  \ottsym{=}   \ottsym{(}   \tau_{{\mathrm{2}}}  \TREF^{ r_{{\mathrm{2}}} }   \ottsym{)}  \TREF^{ r } $
    \item $ \Theta   \mid   \mathcal{L}   \mid   \Gamma_{\ottmv{p}}  \ottsym{[}  \mathit{x}  \hookleftarrow   \tau'_{{\mathrm{1}}}  \TREF^{ r'_{{\mathrm{1}}} }   \ottsym{]}  \ottsym{[}  \mathit{y}  \hookleftarrow   \ottsym{(}   \tau'_{{\mathrm{2}}}  \TREF^{ r'_{{\mathrm{2}}} }   \ottsym{)}  \TREF^{ r }   \ottsym{]}   \vdash   \ottnt{e}  :  \tau_{\ottmv{p}}   \produces   \Gamma'_{\ottmv{p}} $
    \end{itemize}
  \item If $ \ottnt{e_{{\mathrm{0}}}}  \ottsym{=}  \ottnt{e_{{\mathrm{1}}}}  \SEQ  \ottnt{e_{{\mathrm{2}}}} $ then $ \Theta   \mid   \mathcal{L}   \mid   \Gamma_{\ottmv{p}}   \vdash   \ottnt{e_{{\mathrm{1}}}}  :  \tau_{{\mathrm{1}}}   \produces   \Gamma_{{\mathrm{1}}} $ and $ \Theta   \mid   \mathcal{L}   \mid   \Gamma_{{\mathrm{1}}}   \vdash   \ottnt{e_{{\mathrm{2}}}}  :  \tau_{\ottmv{p}}   \produces   \Gamma'_{\ottmv{p}} $
  \item If $ \ottnt{e_{{\mathrm{0}}}}  \ottsym{=}  \mathit{x}  \SEQ  \ottnt{e'} $ then $ \Theta   \mid   \mathcal{L}   \mid   \Gamma_{\ottmv{p}}  \ottsym{[}  \mathit{x}  \ottsym{:}  \tau'  \ottsym{+}  \tau_{{\mathrm{0}}}  \ottsym{]}   \vdash   \mathit{x}  :  \tau_{{\mathrm{1}}}   \produces   \Gamma_{\ottmv{p}}  \ottsym{[}  \mathit{x}  \hookleftarrow  \tau_{{\mathrm{0}}}  \ottsym{]} $ and $ \Theta   \mid   \mathcal{L}   \mid   \Gamma_{\ottmv{p}}  \ottsym{[}  \mathit{x}  \hookleftarrow  \tau_{{\mathrm{0}}}  \ottsym{]}   \vdash   \ottnt{e'}  :  \tau_{\ottmv{p}}   \produces   \Gamma'_{\ottmv{p}} $
  \item If $\ottnt{e_{{\mathrm{0}}}}  \ottsym{=}   \ASSERT( \varphi ) \SEQ  \ottnt{e} $ then $\Gamma_{\ottmv{p}}  \models  \varphi$ and $ \Theta   \mid   \mathcal{L}   \mid   \Gamma_{\ottmv{p}}   \vdash   \ottnt{e}  :  \tau_{\ottmv{p}}   \produces   \Gamma'_{\ottmv{p}} $
  \end{enumerate}
\end{lemma}
\begin{proof}
  By straightforward induction on the typing relation and
  the transitivity of the subtyping relation \Cref{lem:subtype-transitive}.

  The only case of note is the case for $ \ottnt{e_{{\mathrm{0}}}}  \ottsym{=}  \mathit{x}  \SEQ  \ottnt{e_{{\mathrm{2}}}} $. If the subderivation for $\mathit{x}$
  has applications of \rn{T-Sub} then the subtypings on the output environment can be pushed
  into application subtyping on input environments when typing $\ottnt{e'}$. Similarly,
  any input subtypings on the input environment of the derivation of $\mathit{x}$ can be pushed
  into \rn{T-Sub} rules such that $\Gamma  \leq  \Gamma_{\ottmv{p}}  \ottsym{[}  \mathit{x}  \ottsym{:}  \tau'  \ottsym{+}  \tau_{{\mathrm{0}}}  \ottsym{]}$.
\end{proof}

\Cref{lem:stack_var,lem:ectxt-sub-well-typed} prove some standard properties of execution contexts:
any decomposition of a well-typed expression into a execution context and redex can be well-typed,
and substituting a well-typed expression matching a context's hole type yields a well-typed expression

\begin{lemma} %
  \label{lem:stack_var}
  For any $\ottnt{E}$ and $\ottnt{e'}$ such that $\ottnt{E}  \ottsym{[}  \ottnt{e'}  \ottsym{]}  \ottsym{=}  \ottnt{e}$ where
  $ \Theta   \mid   \mathcal{L}   \mid   \Gamma   \vdash   \ottnt{e}  :  \tau   \produces   \Gamma' $ there exists some $\tau_{{\mathrm{0}}}$, $\Gamma_{{\mathrm{0}}}$ such that
  $\Theta  \mid  \HOLE  \ottsym{:}  \tau_{{\mathrm{0}}}  \produces  \Gamma_{{\mathrm{0}}}  \mid  \mathcal{L}  \vdash_{\mathit{ectx} }  \ottnt{E}  \ottsym{:}  \tau  \produces  \Gamma'$ and
  $ \Theta   \mid   \mathcal{L}   \mid   \Gamma   \vdash   \ottnt{e'}  :  \tau_{{\mathrm{0}}}   \produces   \Gamma_{{\mathrm{0}}} $.
\end{lemma}
\begin{proof}
  By induction on the structure of $\ottnt{E}$.
  \begin{rncase}{$\ottnt{E}  \ottsym{=}  \HOLE$}
    Trivial, by taking $\tau_{{\mathrm{0}}}  \ottsym{=}  \tau$ and $\Gamma_{{\mathrm{0}}}  \ottsym{=}  \Gamma'$.
  \end{rncase}
  \begin{rncase}{$ \ottnt{E}  \ottsym{=}  \ottnt{E'} \SEQ \ottnt{e''} $}
    Then $ \ottnt{E}  \ottsym{[}  \ottnt{e'}  \ottsym{]}  \ottsym{=}  \ottnt{E'}  \ottsym{[}  \ottnt{e'}  \ottsym{]}  \SEQ  \ottnt{e''}   \ottsym{=}  \ottnt{e}$. By
    \Cref{lem:inversion} we have
	\[
      \begin{bcpcasearray}
         \Theta   \mid   \mathcal{L}   \mid   \Gamma_{\ottmv{p}}   \vdash   \ottnt{E'}  \ottsym{[}  \ottnt{e'}  \ottsym{]}  :  \tau_{{\mathrm{1}}}   \produces   \Gamma_{{\mathrm{1}}}  &  \Theta   \mid   \mathcal{L}   \mid   \Gamma_{{\mathrm{1}}}   \vdash   \ottnt{e''}  :  \tau_{\ottmv{p}}   \produces   \Gamma'_{\ottmv{p}}  \\
        \Gamma  \leq  \Gamma_{\ottmv{p}} & \Gamma'_{\ottmv{p}}  \ottsym{,}  \tau_{\ottmv{p}}  \leq  \Gamma'  \ottsym{,}  \tau
      \end{bcpcasearray}
    \]
    for some $\Gamma_{\ottmv{p}}$, $\Gamma'_{\ottmv{p}}$, and $\tau_{\ottmv{p}}$.
    
    By the induction hypothesis
    we then have $ \Theta   \mid   \mathcal{L}   \mid   \Gamma_{\ottmv{p}}   \vdash   \ottnt{e'}  :  \tau_{{\mathrm{0}}}   \produces   \Gamma_{{\mathrm{0}}} $ and
    $\Theta  \mid  \HOLE  \ottsym{:}  \tau_{{\mathrm{0}}}  \produces  \Gamma_{{\mathrm{0}}}  \mid  \mathcal{L}  \vdash_{\mathit{ectx} }  \ottnt{E'}  \ottsym{:}  \tau_{{\mathrm{1}}}  \produces  \Gamma_{{\mathrm{1}}}$. for some $\tau_{{\mathrm{0}}}$ and $\Gamma_{{\mathrm{0}}}$.
    
    Next, as $\Gamma'_{\ottmv{p}}  \ottsym{,}  \tau_{\ottmv{p}}  \leq  \Gamma'  \ottsym{,}  \tau$ by an application of \rn{T-Sub},
    we have $ \Theta   \mid   \mathcal{L}   \mid   \Gamma_{{\mathrm{1}}}   \vdash   \ottnt{e''}  :  \tau   \produces   \Gamma' $. By \rn{TE-Seq}, we therefore
    have: $\Theta  \mid  \HOLE  \ottsym{:}  \tau_{{\mathrm{0}}}  \produces  \Gamma_{{\mathrm{0}}}  \mid  \mathcal{L}  \vdash_{\mathit{ectx} }   \ottnt{E'} \SEQ \ottnt{e''}   \ottsym{:}  \tau  \produces  \Gamma'$.

    Finally, from $\Gamma  \leq  \Gamma_{\ottmv{p}}$ and $ \Theta   \mid   \mathcal{L}   \mid   \Gamma_{\ottmv{p}}   \vdash   \ottnt{e'}  :  \tau_{{\mathrm{0}}}   \produces   \Gamma_{{\mathrm{0}}} $,
    and application of \rn{T-Sub}, we have $ \Theta   \mid   \mathcal{L}   \mid   \Gamma   \vdash   \ottnt{e'}  :  \tau_{{\mathrm{0}}}   \produces   \Gamma_{{\mathrm{0}}} $.
  \end{rncase}
\end{proof}

\begin{lemma} %
  \label{lem:ectxt-sub-well-typed}
  If $\Theta  \mid  \HOLE  \ottsym{:}  \tau  \produces  \Gamma'  \mid  \mathcal{L}  \vdash_{\mathit{ectx} }  \ottnt{E}  \ottsym{:}  \tau''  \produces  \Gamma''$ and $ \Theta   \mid   \mathcal{L}   \mid   \Gamma   \vdash   \ottnt{e}  :  \tau   \produces   \Gamma' $ for some $\Gamma$,
  then  $ \Theta   \mid   \mathcal{L}   \mid   \Gamma   \vdash   \ottnt{E}  \ottsym{[}  \ottnt{e}  \ottsym{]}  :  \tau''   \produces   \Gamma'' $.
\end{lemma}
\begin{proof}
  By induction on the typing derivation of $\ottnt{E}$.
  \begin{rneqncase}{TE-Seq}{
      \ottnt{E} =  \ottnt{E'} \SEQ \ottnt{e'}  \\  \ottnt{E}  \ottsym{[}  \ottnt{e}  \ottsym{]}  \ottsym{=}  \ottnt{E'}  \ottsym{[}  \ottnt{e}  \ottsym{]}  \SEQ  \ottnt{e'}  \\
      \Theta  \mid  \HOLE  \ottsym{:}  \tau  \produces  \Gamma'  \mid  \mathcal{L}  \vdash_{\mathit{ectx} }  \ottnt{E'}  \ottsym{:}  \tau_{{\mathrm{0}}}  \produces  \Gamma_{{\mathrm{0}}} \\  \Theta   \mid   \mathcal{L}   \mid   \Gamma_{{\mathrm{0}}}   \vdash   \ottnt{e'}  :  \tau''   \produces   \Gamma''  \\
    }
    By the induction hypothesis we have $ \Theta   \mid   \mathcal{L}   \mid   \Gamma   \vdash   \ottnt{E'}  \ottsym{[}  \ottnt{e}  \ottsym{]}  :  \tau_{{\mathrm{0}}}   \produces   \Gamma_{{\mathrm{0}}} $. We then
    have our result via an application of \rn{T-Seq}.
  \end{rneqncase}
  \begin{rncase}{TE-Hole}
    Trivial, as $\tau  \ottsym{=}  \tau''$ and $\Gamma'  \ottsym{=}  \Gamma''$ and $\ottnt{E}  \ottsym{[}  \ottnt{e}  \ottsym{]}  \ottsym{=}  \ottnt{e}$.
  \end{rncase}
\end{proof}

\def\subref#1#2{\Cref{#1} (part \labelcref{#2})}

\begin{lemma}[Context Variable Substitution]\label{lem:ctxt-substitution}
  \leavevmode
  \begin{enumerate}
  \item \label{itm:ctxt-sub-distribute} If $\tau_{{\mathrm{3}}}  \ottsym{=}  \tau_{{\mathrm{1}}}  \ottsym{+}  \tau_{{\mathrm{2}}}$ then $\ottsym{[}  \mathcal{L}  \ottsym{/}  \lambda  \ottsym{]} \, \tau_{{\mathrm{3}}}  \ottsym{=}  \ottsym{[}  \mathcal{L}  \ottsym{/}  \lambda  \ottsym{]} \, \tau_{{\mathrm{1}}}  \ottsym{+}  \ottsym{[}  \mathcal{L}  \ottsym{/}  \lambda  \ottsym{]} \, \tau_{{\mathrm{2}}}$
  \item \label{itm:ctxt-sub-wf}For any $\oldvec{\ell}$:
    \begin{enumerate}
    \item If $ \lambda   \vdash _{\wf}  \Gamma $ then $ \oldvec{\ell}   \vdash _{\wf}  \ottsym{[}  \oldvec{\ell}  \ottsym{/}  \lambda  \ottsym{]}  \Gamma $
    \item If $ \lambda   \mid   \Gamma   \vdash _{\wf}  \tau $ then $ \oldvec{\ell}   \mid   \ottsym{[}  \oldvec{\ell}  \ottsym{/}  \lambda  \ottsym{]}  \Gamma   \vdash _{\wf}  \ottsym{[}  \oldvec{\ell}  \ottsym{/}  \lambda  \ottsym{]} \, \tau $
    \item If $ \lambda   \vdash _{\wf}  \tau   \produces   \Gamma $ then $ \oldvec{\ell}   \vdash _{\wf}  \ottsym{[}  \oldvec{\ell}  \ottsym{/}  \lambda  \ottsym{]} \, \tau   \produces   \ottsym{[}  \oldvec{\ell}  \ottsym{/}  \lambda  \ottsym{]}  \Gamma $
    \end{enumerate}
  \item \label{itm:ctxt-sub-subtype} For any $\Gamma$, $\tau_{{\mathrm{1}}}$, $\tau_{{\mathrm{2}}}$, $\lambda$ and $\oldvec{\ell}$, If $\Gamma  \vdash  \tau_{{\mathrm{1}}}  \leq  \tau_{{\mathrm{2}}}$, then $\ottsym{[}  \oldvec{\ell}  \ottsym{/}  \lambda  \ottsym{]}  \Gamma  \vdash  \ottsym{[}  \oldvec{\ell}  \ottsym{/}  \lambda  \ottsym{]} \, \tau_{{\mathrm{1}}}  \leq  \ottsym{[}  \oldvec{\ell}  \ottsym{/}  \lambda  \ottsym{]} \, \tau_{{\mathrm{2}}}$
  \item \label{itm:ctxt-subst-assert} If $\Gamma  \models  \varphi$ where $ \lambda  \not\in  \mathbf{FCV} \, \ottsym{(}  \varphi  \ottsym{)} $ then $\ottsym{[}  \oldvec{\ell}  \ottsym{/}  \lambda  \ottsym{]}  \Gamma  \models  \varphi$
  \item \label{itm:ctxt-subst-well-typed} If $ \Theta   \mid   \lambda   \mid   \Gamma   \vdash   \ottnt{e}  :  \tau   \produces   \Gamma' $ then $ \Theta   \mid   \oldvec{\ell}   \mid   \ottsym{[}  \oldvec{\ell}  \ottsym{/}  \lambda  \ottsym{]}  \Gamma   \vdash   \ottnt{e}  :  \ottsym{[}  \oldvec{\ell}  \ottsym{/}  \lambda  \ottsym{]} \, \tau   \produces   \ottsym{[}  \oldvec{\ell}  \ottsym{/}  \lambda  \ottsym{]}  \Gamma' $
  \end{enumerate}
\end{lemma}
\begin{proof}\leavevmode
  \begin{enumerate}
  \item By straightforward induction on the definition of $\tau_{{\mathrm{1}}}  \ottsym{+}  \tau_{{\mathrm{2}}}  \ottsym{=}  \tau_{{\mathrm{3}}}$.
  \item Observe that any substitution of context variables cannot change simple
    types within $\Gamma$ and thus all types and refinements remain well-formed with respect
    to integer variables in $\Gamma$.
    It thus suffices to show that $\mathbf{FCV} \, \ottsym{(}  \ottsym{[}  \oldvec{\ell}  \ottsym{/}  \lambda  \ottsym{]} \, \varphi  \ottsym{)} \subseteq \ottkw{CV} \, \ottsym{(}  \oldvec{\ell}  \ottsym{)} = \emptyset$
    for any refinement $\varphi$ appearing in $\tau$ or a type in $\Gamma$.
    By the assumed well-formedness of $\tau$ with respect to context
    variable $\lambda$ (resp. $\Gamma$), after substitution all free
    context variables in $\tau$ (resp. the types in $\Gamma$) will be
    replaced with $\oldvec{\ell}$. Thus, post-substitution no free context
    variables appear in the refinement of $\ottsym{[}  \oldvec{\ell}  \ottsym{/}  \lambda  \ottsym{]} \, \tau$ (resp. refinements of
    types in $\ottsym{[}  \oldvec{\ell}  \ottsym{/}  \lambda  \ottsym{]}  \Gamma$), trivially satisfying our requirements.
  \item If $\lambda$ does not appear free in $\tau_{{\mathrm{1}}}$, $\tau_{{\mathrm{2}}}$ or $\Gamma$, then the result trivially holds. Let us then assume
    $\lambda$ appears free. We prove the result by induction on the subtyping derivation.

    \begin{rneqncase}{S-Ref}{
        \tau_{{\mathrm{1}}}  \ottsym{=}   \tau'_{{\mathrm{1}}}  \TREF^{ r_{{\mathrm{1}}} }  & \tau_{{\mathrm{2}}}  \ottsym{=}   \tau'_{{\mathrm{2}}}  \TREF^{ r_{{\mathrm{2}}} }  \\
        \ottsym{[}  \oldvec{\ell}  \ottsym{/}  \lambda  \ottsym{]} \, \tau_{{\mathrm{1}}}  \ottsym{=}   \ottsym{(}  \ottsym{[}  \oldvec{\ell}  \ottsym{/}  \lambda  \ottsym{]} \, \tau'_{{\mathrm{1}}}  \ottsym{)}  \TREF^{ r_{{\mathrm{1}}} }  & \ottsym{[}  \oldvec{\ell}  \ottsym{/}  \lambda  \ottsym{]} \, \tau_{{\mathrm{2}}}  \ottsym{=}   \ottsym{(}  \ottsym{[}  \oldvec{\ell}  \ottsym{/}  \lambda  \ottsym{]} \, \tau'_{{\mathrm{2}}}  \ottsym{)}  \TREF^{ r_{{\mathrm{2}}} }  \\
        \Gamma  \vdash  \tau'_{{\mathrm{1}}}  \leq  \tau'_{{\mathrm{2}}} & r_{{\mathrm{2}}}  \le  r_{{\mathrm{1}}}
      }
      We must show that $\ottsym{[}  \oldvec{\ell}  \ottsym{/}  \lambda  \ottsym{]}  \Gamma  \vdash  \ottsym{[}  \oldvec{\ell}  \ottsym{/}  \lambda  \ottsym{]} \, \tau'_{{\mathrm{1}}}  \leq  \ottsym{[}  \oldvec{\ell}  \ottsym{/}  \lambda  \ottsym{]} \, \tau'_{{\mathrm{2}}}$
      which holds immediately from the induction hypothesis.
    \end{rneqncase}

    \begin{rneqncase}{S-Int}{
        \tau_{{\mathrm{1}}}  \ottsym{=}   \set{  \nu  \COL \TINT \mid  \varphi_{{\mathrm{1}}} }  & \tau_{{\mathrm{2}}}  \ottsym{=}   \set{  \nu  \COL \TINT \mid  \varphi_{{\mathrm{2}}} }  \\
        \ottsym{[}  \oldvec{\ell}  \ottsym{/}  \lambda  \ottsym{]} \, \tau_{{\mathrm{1}}}  \ottsym{=}   \set{  \nu  \COL \TINT \mid  \ottsym{[}  \oldvec{\ell}  \ottsym{/}  \lambda  \ottsym{]} \, \varphi_{{\mathrm{1}}} }  & \ottsym{[}  \oldvec{\ell}  \ottsym{/}  \lambda  \ottsym{]} \, \tau_{{\mathrm{2}}}  \ottsym{=}   \set{  \nu  \COL \TINT \mid  \ottsym{[}  \oldvec{\ell}  \ottsym{/}  \lambda  \ottsym{]} \, \varphi_{{\mathrm{2}}} }  \\
        \Gamma  \models  \varphi_{{\mathrm{1}}}  \implies  \varphi_{{\mathrm{2}}}
      }
      We must show that $\ottsym{[}  \oldvec{\ell}  \ottsym{/}  \lambda  \ottsym{]}  \Gamma  \models  \ottsym{[}  \oldvec{\ell}  \ottsym{/}  \lambda  \ottsym{]} \, \varphi_{{\mathrm{1}}}  \implies  \ottsym{[}  \oldvec{\ell}  \ottsym{/}  \lambda  \ottsym{]} \, \varphi_{{\mathrm{2}}}$, i.e. $\models   \sem{ \ottsym{[}  \oldvec{\ell}  \ottsym{/}  \lambda  \ottsym{]}  \Gamma }   \wedge  \ottsym{[}  \oldvec{\ell}  \ottsym{/}  \lambda  \ottsym{]} \, \varphi_{{\mathrm{1}}}  \implies  \ottsym{[}  \oldvec{\ell}  \ottsym{/}  \lambda  \ottsym{]} \, \varphi_{{\mathrm{2}}}$.
      From our assumption that $\Gamma  \models  \varphi_{{\mathrm{1}}}  \implies  \varphi_{{\mathrm{2}}}$ we have that $\models   \sem{ \Gamma }   \wedge  \varphi_{{\mathrm{1}}}  \implies  \varphi_{{\mathrm{2}}}$ is valid,
      whereby the formula $ \sem{ \Gamma }   \wedge  \varphi_{{\mathrm{1}}}  \implies  \varphi_{{\mathrm{2}}}$ is true for any possible concrete valuation of the free context
      variable $\lambda$. As $\ottsym{[}  \oldvec{\ell}  \ottsym{/}  \lambda  \ottsym{]} \,  \sem{ \Gamma } $ is equivalent to $ \sem{ \ottsym{[}  \oldvec{\ell}  \ottsym{/}  \lambda  \ottsym{]}  \Gamma } $
      we have the formula $ \sem{ \ottsym{[}  \oldvec{\ell}  \ottsym{/}  \lambda  \ottsym{]}  \Gamma }   \wedge  \ottsym{[}  \oldvec{\ell}  \ottsym{/}  \lambda  \ottsym{]} \, \varphi_{{\mathrm{1}}}  \implies  \ottsym{[}  \oldvec{\ell}  \ottsym{/}  \lambda  \ottsym{]} \, \varphi_{{\mathrm{2}}}$ must also be valid.
    \end{rneqncase}
  \item If $\lambda$ does not appear free in $ \sem{ \Gamma } $, then the result trivially holds. Otherwise
    $\models   \sem{ \Gamma }   \implies  \varphi$ holds for any concrete valuation of the free context variable $\lambda$.
    Then the formula $\models   \sem{ \ottsym{[}  \oldvec{\ell}  \ottsym{/}  \lambda  \ottsym{]}  \Gamma }   \implies  \varphi$ must be valid from the equivalence of $\ottsym{[}  \oldvec{\ell}  \ottsym{/}  \lambda  \ottsym{]} \,  \sem{ \Gamma } $ and $ \sem{ \ottsym{[}  \oldvec{\ell}  \ottsym{/}  \lambda  \ottsym{]}  \Gamma } $.
  \item By induction on the typing derivation $ \Theta   \mid   \lambda   \mid   \Gamma   \vdash   \ottnt{e}  :  \tau   \produces   \Gamma' $.
    In every case, that $ \oldvec{\ell}   \vdash _{\wf}  \ottsym{[}  \oldvec{\ell}  \ottsym{/}  \lambda  \ottsym{]} \, \tau   \produces   \ottsym{[}  \oldvec{\ell}  \ottsym{/}  \lambda  \ottsym{]}  \Gamma' $ and $ \oldvec{\ell}   \vdash _{\wf}  \ottsym{[}  \oldvec{\ell}  \ottsym{/}  \lambda  \ottsym{]}  \Gamma $
    holds from \Cref{itm:ctxt-sub-wf}.
    
    \begin{rneqncase}{T-Var}{
        \ottnt{e}  \ottsym{=}  \mathit{x} & \tau  \ottsym{=}  \tau_{{\mathrm{2}}} \\
        \Gamma  \ottsym{=}  \Gamma_{{\mathrm{0}}}  \ottsym{[}  \mathit{x}  \ottsym{:}  \tau_{{\mathrm{1}}}  \ottsym{+}  \tau_{{\mathrm{2}}}  \ottsym{]} & \Gamma'  \ottsym{=}  \Gamma_{{\mathrm{0}}}  \ottsym{[}  \mathit{x}  \hookleftarrow  \tau_{{\mathrm{2}}}  \ottsym{]}
      }
      By application of \Cref{itm:ctxt-sub-distribute}.
    \end{rneqncase}
    \begin{rneqncase}{T-LetInt}{
        \ottnt{e}  \ottsym{=}   \LET  \mathit{x}  =  n  \IN  \ottnt{e'}  &  \Theta   \mid   \lambda   \mid   \Gamma  \ottsym{,}  \mathit{x}  \ottsym{:}   \set{  \nu  \COL \TINT \mid  \nu \, \ottsym{=} \, n }    \vdash   \ottnt{e'}  :  \tau   \produces   \Gamma'  \\
         \mathit{x}  \not\in   \DOM( \Gamma' )   &
      }
      The induction hypothesis gives
      \[
         \Theta   \mid   \oldvec{\ell}   \mid   \ottsym{[}  \oldvec{\ell}  \ottsym{/}  \lambda  \ottsym{]}  \Gamma  \ottsym{,}  \mathit{x}  \ottsym{:}   \set{  \nu  \COL \TINT \mid  \nu \, \ottsym{=} \, n }    \vdash   \ottnt{e}  :  \ottsym{[}  \oldvec{\ell}  \ottsym{/}  \lambda  \ottsym{]} \, \tau   \produces   \ottsym{[}  \oldvec{\ell}  \ottsym{/}  \lambda  \ottsym{]}  \Gamma' 
      \]
      We conclude $ \Theta   \mid   \oldvec{\ell}   \mid   \ottsym{[}  \oldvec{\ell}  \ottsym{/}  \lambda  \ottsym{]}  \Gamma   \vdash    \LET  \mathit{x}  =  n  \IN  \ottnt{e}   :  \ottsym{[}  \oldvec{\ell}  \ottsym{/}  \lambda  \ottsym{]} \, \tau   \produces   \ottsym{[}  \oldvec{\ell}  \ottsym{/}  \lambda  \ottsym{]}  \Gamma' $
      as required.
    \end{rneqncase}

    \begin{rneqncase}{T-Let}{
        & \ottnt{e}  \ottsym{=}   \LET  \mathit{x}  =  \mathit{y}  \IN  \ottnt{e'}  &  \mathit{x}  \not\in   \DOM( \Gamma' )   \\
        &  \Theta   \mid   \lambda   \mid   \Gamma_{{\mathrm{1}}}   \vdash   \ottnt{e'}  :  \tau   \produces   \Gamma'  & \Gamma_{{\mathrm{1}}}  \ottsym{=}  \Gamma  \ottsym{[}  \mathit{y}  \hookleftarrow  \ottsym{(}   \tau_{{\mathrm{1}}}  \wedge_{ \mathit{y} }   \mathit{y}  =_{ \tau_{{\mathrm{1}}} }  \mathit{x}    \ottsym{)}  \ottsym{]}  \ottsym{,}  \mathit{x}  \ottsym{:}  \ottsym{(}   \tau_{{\mathrm{2}}}  \wedge_{ \mathit{x} }   \mathit{x}  =_{ \tau_{{\mathrm{2}}} }  \mathit{y}    \ottsym{)} \\
        & \Gamma \quad \ottsym{(}  \mathit{y}  \ottsym{)}  \ottsym{=}   \mathit{y} \COL \tau_{{\mathrm{1}}}  \ottsym{+}  \tau_{{\mathrm{2}}} 
      }
      By \Cref{itm:ctxt-sub-distribute},
      $\ottsym{(}  \ottsym{[}  \oldvec{\ell}  \ottsym{/}  \lambda  \ottsym{]}  \Gamma  \ottsym{)}  \ottsym{(}  \mathit{y}  \ottsym{)}  \ottsym{=}  \ottsym{[}  \oldvec{\ell}  \ottsym{/}  \lambda  \ottsym{]} \, \ottsym{(}  \tau_{{\mathrm{1}}}  \ottsym{+}  \tau_{{\mathrm{2}}}  \ottsym{)}  \ottsym{=}  \ottsym{(}  \ottsym{[}  \oldvec{\ell}  \ottsym{/}  \lambda  \ottsym{]} \, \tau_{{\mathrm{1}}}  \ottsym{+}  \ottsym{[}  \oldvec{\ell}  \ottsym{/}  \lambda  \ottsym{]} \, \tau_{{\mathrm{2}}}  \ottsym{)}$. We must then
      show that $ \Theta   \mid   \oldvec{\ell}   \mid   \Gamma'_{{\mathrm{1}}}   \vdash   \ottnt{e'}  :  \ottsym{[}  \oldvec{\ell}  \ottsym{/}  \lambda  \ottsym{]} \, \tau   \produces   \ottsym{[}  \oldvec{\ell}  \ottsym{/}  \lambda  \ottsym{]}  \Gamma' $
      where
      \[
        \Gamma'_{{\mathrm{1}}}  \ottsym{=}  \ottsym{(}  \ottsym{[}  \oldvec{\ell}  \ottsym{/}  \lambda  \ottsym{]}  \Gamma  \ottsym{)}  \ottsym{[}  \mathit{y}  \hookleftarrow   \ottsym{[}  \oldvec{\ell}  \ottsym{/}  \lambda  \ottsym{]} \, \tau_{{\mathrm{1}}}  \wedge_{ \mathit{y} }  \mathit{y} \, \ottsym{=} \, \mathit{x}   \ottsym{]}  \ottsym{,}  \mathit{x}  \ottsym{:}  \ottsym{(}   \ottsym{[}  \oldvec{\ell}  \ottsym{/}  \lambda  \ottsym{]} \, \tau_{{\mathrm{2}}}  \wedge_{ \mathit{x} }  \mathit{x} \, \ottsym{=} \, \mathit{y}   \ottsym{)}
      \]
      As $\Gamma'_{{\mathrm{1}}}  \ottsym{=}  \ottsym{[}  \oldvec{\ell}  \ottsym{/}  \lambda  \ottsym{]}  \Gamma_{{\mathrm{1}}}$ the induction hypothesis gives the required typing judgment.
    \end{rneqncase}

    \begin{rncase}{T-If,T-Seq}
      By trivial application of the inductive hypothesis.
    \end{rncase}

    \begin{rncase}{T-MkRef,T-Deref}
      By reasoning similar to \rn{T-Let}.
    \end{rncase}

    \begin{rneqncase}{T-Call}{
        \ottnt{e}  \ottsym{=}   \LET  \mathit{x}  =   \mathit{f} ^ \ell (  \mathit{y_{{\mathrm{1}}}} ,\ldots, \mathit{y_{\ottmv{n}}}  )   \IN  \ottnt{e'}  \\
        \sigma_{x}  \ottsym{=}    [  \mathit{y_{{\mathrm{1}}}}  /  \mathit{x_{{\mathrm{1}}}}  ]  \cdots  [  \mathit{y_{\ottmv{n}}}  /  \mathit{x_{\ottmv{n}}}  ]   \\
        \sigma_{\alpha}  \ottsym{=}  \ottsym{[}  \ell  \ottsym{:}  \lambda  \ottsym{/}  \lambda'  \ottsym{]} \\
         \Theta   \mid   \lambda   \mid   \Gamma_{{\mathrm{1}}}   \vdash   \ottnt{e'}  :  \tau   \produces   \Gamma'  \\
         \mathit{y}  \not\in   \DOM( \Gamma' )  \\
        \Theta  \ottsym{(}  \mathit{f}  \ottsym{)}  \ottsym{=}   \forall  \lambda' .\tuple{ \mathit{x_{{\mathrm{1}}}} \COL \tau_{{\mathrm{1}}} ,\dots, \mathit{x_{\ottmv{n}}} \COL \tau_{\ottmv{n}} }\ra\tuple{ \mathit{x_{{\mathrm{1}}}} \COL \tau'_{{\mathrm{1}}} ,\dots, \mathit{x_{\ottmv{n}}} \COL \tau'_{\ottmv{n}}  \mid  \tau' }   \\
        \Gamma_{{\mathrm{1}}}  \ottsym{=}  \Gamma  \ottsym{[}  \mathit{y_{\ottmv{i}}}  \hookleftarrow  \sigma_{\alpha} \, \sigma_{x} \, \tau'_{\ottmv{i}}  \ottsym{]}  \ottsym{,}  \mathit{x}  \ottsym{:}  \sigma_{\alpha} \, \sigma_{x} \, \tau'
      }
      We must first show that for $\sigma_{\alpha}'  \ottsym{=}  \ottsym{[}  \ell  \ottsym{:}  \oldvec{\ell}  \ottsym{/}  \lambda'  \ottsym{]}$:
      \[
         \Theta   \mid   \oldvec{\ell}   \mid   \Gamma_{{\mathrm{3}}}   \vdash   \ottnt{e'}  :  \ottsym{[}  \oldvec{\ell}  \ottsym{/}  \lambda  \ottsym{]} \, \tau   \produces   \ottsym{[}  \oldvec{\ell}  \ottsym{/}  \lambda  \ottsym{]}  \Gamma' 
      \]
      where $\Gamma_{{\mathrm{3}}}  \ottsym{=}  \ottsym{(}  \ottsym{[}  \oldvec{\ell}  \ottsym{/}  \lambda  \ottsym{]}  \Gamma  \ottsym{)}  \ottsym{[}  \mathit{y_{\ottmv{i}}}  \hookleftarrow  \sigma_{\alpha}' \, \sigma_{x} \, \tau'_{\ottmv{i}}  \ottsym{]}  \ottsym{,}  \mathit{x}  \ottsym{:}  \sigma_{\alpha}' \, \sigma_{x} \, \tau'$.

      We first observe that $\Gamma_{{\mathrm{3}}}  \ottsym{=}  \ottsym{[}  \oldvec{\ell}  \ottsym{/}  \lambda  \ottsym{]}  \Gamma_{{\mathrm{1}}}$ (this follows from the
      equivalence of $\ottsym{[}  \oldvec{\ell}  \ottsym{/}  \lambda  \ottsym{]} \, \ottsym{[}  \ell  \ottsym{:}  \lambda  \ottsym{/}  \lambda'  \ottsym{]}$ and
      $\ottsym{[}  \ell  \ottsym{:}  \oldvec{\ell}  \ottsym{/}  \lambda'  \ottsym{]}$) whereby the induction hypothesis
      gives the required typing derivation.

      We must also show that
      $\forall i \in \set{1..n}.\ottsym{(}  \ottsym{[}  \ell  \ottsym{:}  \oldvec{\ell}  \ottsym{/}  \lambda  \ottsym{]}  \Gamma  \ottsym{)}  \ottsym{(}  \mathit{y_{\ottmv{i}}}  \ottsym{)}  \ottsym{=}  \sigma_{\alpha}' \, \sigma_{x} \, \tau_{\ottmv{i}}$.
      From the assumed well-typing of the term under $\lambda$ we have
      that $\forall i \in\set{1..n}.\Gamma  \ottsym{(}  \mathit{y_{\ottmv{i}}}  \ottsym{)}  \ottsym{=}  \sigma_{\alpha} \, \sigma_{x} \, \tau_{\ottmv{i}}$. Recall
      that $\sigma_{\alpha}'$ is equivalent to $\ottsym{[}  \oldvec{\ell}  \ottsym{/}  \lambda  \ottsym{]} \, \sigma_{\alpha}$, whereby we have
      $\ottsym{[}  \oldvec{\ell}  \ottsym{/}  \lambda  \ottsym{]}  \Gamma  \ottsym{(}  \mathit{y_{\ottmv{i}}}  \ottsym{)}  \ottsym{=}  \ottsym{[}  \oldvec{\ell}  \ottsym{/}  \lambda  \ottsym{]} \, \sigma_{\alpha} \, \sigma_{x} \, \tau_{\ottmv{i}}  \ottsym{=}  \sigma_{\alpha}' \, \sigma_{x} \, \tau_{\ottmv{i}}$ for any $\ottmv{i}$ as
      equality is preserved by consistent substitution.
    \end{rneqncase}

    \begin{rncase}{T-Assign}
      By the inductive hypothesis and application of \Cref{itm:ctxt-sub-distribute}.
    \end{rncase}

    \begin{rneqncase}{T-Alias}{
         \Theta   \mid   \lambda   \mid   \Gamma  \ottsym{[}  \mathit{x}  \ottsym{:}   \tau_{{\mathrm{1}}}  \TREF^{ r_{{\mathrm{1}}} }   \ottsym{]}  \ottsym{[}  \mathit{y}  \ottsym{:}   \tau_{{\mathrm{2}}}  \TREF^{ r_{{\mathrm{2}}} }   \ottsym{]}   \vdash    \ALIAS( \mathit{x}  =  \mathit{y} ) \SEQ  \ottnt{e}   :  \tau   \produces   \Gamma  \\
          \tau_{{\mathrm{1}}}  \TREF^{ r_{{\mathrm{1}}} }   \ottsym{+}  \tau_{{\mathrm{2}}}  \TREF^{ r_{{\mathrm{2}}} }   \approx    \tau'_{{\mathrm{1}}}  \TREF^{ r'_{{\mathrm{1}}} }   \ottsym{+}  \tau'_{{\mathrm{2}}}  \TREF^{ r'_{{\mathrm{2}}} }  \\
         \Theta   \mid   \lambda   \mid   \Gamma  \ottsym{[}  \mathit{x}  \hookleftarrow   \tau'_{{\mathrm{1}}}  \TREF^{ r'_{{\mathrm{1}}} }   \ottsym{]}  \ottsym{[}  \mathit{y}  \hookleftarrow   \tau'_{{\mathrm{2}}}  \TREF^{ r'_{{\mathrm{2}}} }   \ottsym{]}   \vdash   \ottnt{e}  :  \tau   \produces   \Gamma 
      }
      From \Cref{itm:ctxt-sub-distribute} we have that $\ottsym{[}  \oldvec{\ell}  \ottsym{/}  \lambda  \ottsym{]} \, \ottsym{(}   \tau_{{\mathrm{1}}}  \TREF^{ r_{{\mathrm{1}}} }   \ottsym{)}  \ottsym{+}  \ottsym{[}  \oldvec{\ell}  \ottsym{/}  \lambda  \ottsym{]} \, \ottsym{(}   \tau_{{\mathrm{2}}}  \TREF^{ r_{{\mathrm{2}}} }   \ottsym{)}  \ottsym{=}  \ottsym{[}  \oldvec{\ell}  \ottsym{/}  \lambda  \ottsym{]} \, \ottsym{(}    \tau_{{\mathrm{1}}}  \TREF^{ r_{{\mathrm{1}}} }   \ottsym{+}  \tau_{{\mathrm{2}}}  \TREF^{ r_{{\mathrm{2}}} }   \ottsym{)}$
      and similarly for $  \tau'_{{\mathrm{1}}}  \TREF^{ r'_{{\mathrm{1}}} }   \ottsym{+}  \tau'_{{\mathrm{2}}}  \TREF^{ r'_{{\mathrm{2}}} } $. It therefore remains to show that: \[
        \ottsym{[}  \oldvec{\ell}  \ottsym{/}  \lambda  \ottsym{]} \, \ottsym{(}    \tau_{{\mathrm{1}}}  \TREF^{ r_{{\mathrm{1}}} }   \ottsym{+}  \tau_{{\mathrm{2}}}  \TREF^{ r_{{\mathrm{2}}} }   \ottsym{)}  \approx  \ottsym{[}  \oldvec{\ell}  \ottsym{/}  \lambda  \ottsym{]} \, \ottsym{(}    \tau'_{{\mathrm{1}}}  \TREF^{ r'_{{\mathrm{1}}} }   \ottsym{+}  \tau'_{{\mathrm{2}}}  \TREF^{ r'_{{\mathrm{2}}} }   \ottsym{)}
      \]
      For which it suffices to show that $ \bullet   \vdash  \ottsym{[}  \oldvec{\ell}  \ottsym{/}  \lambda  \ottsym{]} \, \ottsym{(}    \tau_{{\mathrm{1}}}  \TREF^{ r_{{\mathrm{1}}} }   \ottsym{+}  \tau_{{\mathrm{2}}}  \TREF^{ r_{{\mathrm{2}}} }   \ottsym{)}  \leq  \ottsym{[}  \oldvec{\ell}  \ottsym{/}  \lambda  \ottsym{]} \, \ottsym{(}    \tau'_{{\mathrm{1}}}  \TREF^{ r'_{{\mathrm{1}}} }   \ottsym{+}  \tau'_{{\mathrm{2}}}  \TREF^{ r'_{{\mathrm{2}}} }   \ottsym{)}$
      and $ \bullet   \vdash  \ottsym{[}  \oldvec{\ell}  \ottsym{/}  \lambda  \ottsym{]} \, \ottsym{(}    \tau'_{{\mathrm{1}}}  \TREF^{ r'_{{\mathrm{1}}} }   \ottsym{+}  \tau'_{{\mathrm{2}}}  \TREF^{ r'_{{\mathrm{2}}} }   \ottsym{)}  \leq  \ottsym{[}  \oldvec{\ell}  \ottsym{/}  \lambda  \ottsym{]} \, \ottsym{(}    \tau_{{\mathrm{1}}}  \TREF^{ r_{{\mathrm{1}}} }   \ottsym{+}  \tau_{{\mathrm{2}}}  \TREF^{ r_{{\mathrm{2}}} }   \ottsym{)}$. From
      $  \tau_{{\mathrm{1}}}  \TREF^{ r_{{\mathrm{1}}} }   \ottsym{+}  \tau_{{\mathrm{2}}}  \TREF^{ r_{{\mathrm{2}}} }   \approx    \tau'_{{\mathrm{1}}}  \TREF^{ r'_{{\mathrm{1}}} }   \ottsym{+}  \tau'_{{\mathrm{2}}}  \TREF^{ r'_{{\mathrm{2}}} } $ these both follow from \Cref{itm:ctxt-sub-subtype},
      whereby the result follows from the inductive hypothesis.
    \end{rneqncase}

    \begin{rncase}{T-AliasPtr}
      By similar reasoning to the \rn{T-Alias} case.
    \end{rncase}

    \begin{rneqncase}{T-Sub}{
         \Theta   \mid   \lambda   \mid   \Gamma_{{\mathrm{1}}}   \vdash   \ottnt{e}  :  \tau_{{\mathrm{1}}}   \produces   \Gamma_{{\mathrm{2}}}  & \Gamma  \leq  \Gamma_{{\mathrm{1}}} \\
        \Gamma_{{\mathrm{2}}}  \ottsym{,}  \tau_{{\mathrm{1}}}  \leq  \Gamma'  \ottsym{,}  \tau \\
      }
      By the induction hypothesis we have that: $ \Theta   \mid   \oldvec{\ell}   \mid   \ottsym{[}  \oldvec{\ell}  \ottsym{/}  \lambda  \ottsym{]}  \Gamma_{{\mathrm{1}}}   \vdash   \ottnt{e}  :  \ottsym{[}  \oldvec{\ell}  \ottsym{/}  \lambda  \ottsym{]} \, \tau_{{\mathrm{1}}}   \produces   \ottsym{[}  \oldvec{\ell}  \ottsym{/}  \lambda  \ottsym{]}  \Gamma_{{\mathrm{2}}} $.
      If we show that $\ottsym{[}  \oldvec{\ell}  \ottsym{/}  \lambda  \ottsym{]}  \Gamma  \leq  \ottsym{[}  \oldvec{\ell}  \ottsym{/}  \lambda  \ottsym{]}  \Gamma_{{\mathrm{1}}}$ and $\ottsym{[}  \oldvec{\ell}  \ottsym{/}  \lambda  \ottsym{]}  \Gamma_{{\mathrm{2}}}  \ottsym{,}  \ottsym{[}  \oldvec{\ell}  \ottsym{/}  \lambda  \ottsym{]} \, \tau_{{\mathrm{1}}}  \leq  \ottsym{[}  \oldvec{\ell}  \ottsym{/}  \lambda  \ottsym{]}  \Gamma'  \ottsym{,}  \ottsym{[}  \oldvec{\ell}  \ottsym{/}  \lambda  \ottsym{]} \, \tau$
      we will have the required result. To show the first requirement, for any $ \mathit{x}  \in \DOM( \Gamma ) $ we have that
      $\ottsym{[}  \oldvec{\ell}  \ottsym{/}  \lambda  \ottsym{]}  \Gamma  \vdash  \ottsym{[}  \oldvec{\ell}  \ottsym{/}  \lambda  \ottsym{]}  \Gamma  \ottsym{(}  \mathit{x}  \ottsym{)}  \leq  \ottsym{[}  \oldvec{\ell}  \ottsym{/}  \lambda  \ottsym{]}  \Gamma_{{\mathrm{1}}}  \ottsym{(}  \mathit{x}  \ottsym{)}$ from \Cref{itm:ctxt-sub-subtype} so we have $\ottsym{[}  \oldvec{\ell}  \ottsym{/}  \lambda  \ottsym{]}  \Gamma  \leq  \ottsym{[}  \oldvec{\ell}  \ottsym{/}  \lambda  \ottsym{]}  \Gamma_{{\mathrm{1}}}$.
      To show the latter requirement, we observe that $\ottsym{[}  \oldvec{\ell}  \ottsym{/}  \lambda  \ottsym{]}  \Gamma_{{\mathrm{2}}}  \ottsym{,}  \ottsym{[}  \oldvec{\ell}  \ottsym{/}  \lambda  \ottsym{]} \, \tau_{{\mathrm{1}}}  \leq  \ottsym{[}  \oldvec{\ell}  \ottsym{/}  \lambda  \ottsym{]}  \Gamma'  \ottsym{,}  \ottsym{[}  \oldvec{\ell}  \ottsym{/}  \lambda  \ottsym{]} \, \tau$ is equivalent to showing
      $\ottsym{[}  \oldvec{\ell}  \ottsym{/}  \lambda  \ottsym{]}  \ottsym{(}  \Gamma_{{\mathrm{2}}}  \ottsym{,}  \mathit{x}  \ottsym{:}  \tau_{{\mathrm{1}}}  \ottsym{)}  \leq  \ottsym{[}  \oldvec{\ell}  \ottsym{/}  \lambda  \ottsym{]}  \ottsym{(}  \Gamma'  \ottsym{,}  \mathit{x}  \ottsym{:}  \tau  \ottsym{)}$ for some $ \mathit{x}  \not\in   \DOM( \Gamma_{{\mathrm{2}}} )  $,
      whereby we have the required subtyping relationship from the application of \Cref{itm:ctxt-sub-subtype}.
    \end{rneqncase}
    \begin{rneqncase}{T-Assert}{
         \Theta   \mid   \lambda   \mid   \Gamma   \vdash    \ASSERT( \varphi ) \SEQ  \ottnt{e}   :  \tau   \produces   \Gamma'  & \Gamma  \models  \varphi \\
         \Theta   \mid   \lambda   \mid   \Gamma   \vdash   \ottnt{e}  :  \tau   \produces   \Gamma'  &   \epsilon    \mid   \Gamma   \vdash _{\wf}  \varphi 
      }
      By induction hypothesis, the result holds if we can show $\ottsym{[}  \oldvec{\ell}  \ottsym{/}  \lambda  \ottsym{]}  \Gamma  \models  \varphi$ which
      follows from \Cref{itm:ctxt-subst-assert} (that $ \lambda  \not\in  \mathbf{FCV} \, \ottsym{(}  \varphi  \ottsym{)} $ follows
      from the well-formedness of $\varphi$ with respect to $ \epsilon $).
    \end{rneqncase}
  \end{enumerate}
\end{proof}

\begin{lemma}[Substitution] %
  \label{lem:substitution}
  If $ \Theta   \mid   \mathcal{L}   \mid   \Gamma   \vdash   \ottnt{e}  :  \tau   \produces   \Gamma' $ and $ \mathit{x'}  \not\in \DOM( \Gamma ) $, then
  $ \Theta   \mid   \mathcal{L}   \mid    [  \mathit{x'}  /  \mathit{x}  ]  \, \Gamma   \vdash     [  \mathit{x'}  /  \mathit{x}  ]    \ottnt{e}   :   [  \mathit{x'}  /  \mathit{x}  ]  \, \tau   \produces    [  \mathit{x'}  /  \mathit{x}  ]  \, \Gamma' $.
\end{lemma}
\begin{proof}
  By straightforward induction of typing rules.
\end{proof}

\begin{lemma}%
  \label{lem:var-trace-transport}
  If $ \Theta   \mid   \oldvec{\ell}   \mid   \Gamma   \vdash   \mathit{x}  :  \tau   \produces   \Gamma' $, then $ \Theta   \mid   \oldvec{\ell}'   \mid   \Gamma   \vdash   \mathit{x}  :  \tau   \produces   \Gamma' $.
\end{lemma}
\begin{proof}
  Induction on the variable typing derivation.
\end{proof}

We now prove that if every variable satisfies its refinement in a type environment $\Gamma$,
we must have $\models  \ottsym{[}  \ottnt{R}  \ottsym{]} \,  \sem{ \Gamma } $.

\begin{lemma}
  \label{lem:sat-implies-gamma}
  If $\ottkw{SAT} \, \ottsym{(}  \ottnt{H}  \ottsym{,}  \ottnt{R}  \ottsym{,}  \Gamma  \ottsym{)}$ then $\models  \ottsym{[}  \ottnt{R}  \ottsym{]} \,  \sem{ \Gamma } $.
\end{lemma}
\begin{proof}
  To show $\models  \ottsym{[}  \ottnt{R}  \ottsym{]} \,  \sem{ \Gamma } $, it suffices to show that for any $ \mathit{x}  \in   \DOM( \Gamma )  $ where $\Gamma  \ottsym{(}  \mathit{x}  \ottsym{)}  \ottsym{=}   \set{  \nu  \COL \TINT \mid  \varphi } $
  $\models  \ottsym{[}  \ottnt{R}  \ottsym{]} \, \ottsym{[}  \mathit{x}  \ottsym{/} \, \nu \, \ottsym{]} \, \varphi$ holds. From $\ottkw{SAT} \, \ottsym{(}  \ottnt{H}  \ottsym{,}  \ottnt{R}  \ottsym{,}  \Gamma  \ottsym{)}$, we must have $ \ottkw{SATv} ( \ottnt{H} , \ottnt{R} , \ottnt{R}  \ottsym{(}  \mathit{x}  \ottsym{)} , \Gamma  \ottsym{(}  \mathit{x}  \ottsym{)} ) $, whereby
  we have $ \ottnt{R}  \ottsym{(}  \mathit{x}  \ottsym{)}  \in  \mathbb{Z} $ and $\ottsym{[}  \ottnt{R}  \ottsym{]} \, \ottsym{[}  \ottnt{R}  \ottsym{(}  \mathit{x}  \ottsym{)}  \ottsym{/}  \nu  \ottsym{]}  \varphi$. As $\ottsym{[}  \ottnt{R}  \ottsym{]} \, \ottsym{[}  \mathit{x}  \ottsym{/} \, \nu \, \ottsym{]} \, \varphi$ is equivalent to $\ottsym{[}  \ottnt{R}  \ottsym{]} \, \ottsym{[}  \ottnt{R}  \ottsym{(}  \mathit{x}  \ottsym{)}  \ottsym{/}  \nu  \ottsym{]}  \varphi$,
  and we have the desired result.
\end{proof}

We prove that subtyping preserves the consistency relation in the following sense.

\begin{lemma} %
  \label{lem:subtyp-preserves-cons}
  If $\Gamma  \leq  \Gamma'$ and $\ottkw{Cons} \, \ottsym{(}  \ottnt{H}  \ottsym{,}  \ottnt{R}  \ottsym{,}  \Gamma  \ottsym{)}$ then:
  \begin{enumerate}
  \item For any $ \mathit{x}  \in   \DOM( \Gamma' )  $, $\forall \,  \ottmv{a}  \in \DOM( \ottnt{H} )   \ottsym{.}  \ottkw{own} \, \ottsym{(}  \ottnt{H}  \ottsym{,}  \ottnt{R}  \ottsym{(}  \mathit{x}  \ottsym{)}  \ottsym{,}  \Gamma'  \ottsym{(}  \mathit{x}  \ottsym{)}  \ottsym{)}  \ottsym{(}  \ottmv{a}  \ottsym{)}  \le  \ottkw{own} \, \ottsym{(}  \ottnt{H}  \ottsym{,}  \ottnt{R}  \ottsym{(}  \mathit{x}  \ottsym{)}  \ottsym{,}  \Gamma  \ottsym{(}  \mathit{x}  \ottsym{)}  \ottsym{)}  \ottsym{(}  \ottmv{a}  \ottsym{)}$
  \item $\forall \,  \ottmv{a}  \in \DOM( \ottnt{H} )   \ottsym{.}  \ottkw{Own} \, \ottsym{(}  \ottnt{H}  \ottsym{,}  \ottnt{R}  \ottsym{,}  \Gamma'  \ottsym{)}  \ottsym{(}  \ottmv{a}  \ottsym{)}  \le  \ottsym{1}$
  \item If $\Gamma  \vdash  \tau  \leq  \tau'$ and $ \ottkw{SATv} ( \ottnt{H} , \ottnt{R} , \ottnt{v} , \tau ) $ then $ \ottkw{SATv} ( \ottnt{H} , \ottnt{R} , \ottnt{v} , \tau' ) $
  \item $\ottkw{SAT} \, \ottsym{(}  \ottnt{H}  \ottsym{,}  \ottnt{R}  \ottsym{,}  \Gamma'  \ottsym{)}$
  \item $\ottkw{Cons} \, \ottsym{(}  \ottnt{H}  \ottsym{,}  \ottnt{R}  \ottsym{,}  \Gamma'  \ottsym{)}$
  \end{enumerate}
\end{lemma}
\begin{proof}
  \leavevmode
  \begin{enumerate}
  \item By induction on $\Gamma  \vdash  \Gamma  \ottsym{(}  \mathit{x}  \ottsym{)}  \leq  \Gamma'  \ottsym{(}  \mathit{x}  \ottsym{)}$.
  \item Direct consequence of 1 and that $\forall \,  \ottmv{a}  \in \DOM( \ottnt{H} )   \ottsym{.}  \ottkw{Own} \, \ottsym{(}  \ottnt{H}  \ottsym{,}  \ottnt{R}  \ottsym{,}  \Gamma  \ottsym{)}  \ottsym{(}  \ottmv{a}  \ottsym{)}  \le  \ottsym{1}$ from $\ottkw{Cons} \, \ottsym{(}  \ottnt{H}  \ottsym{,}  \ottnt{R}  \ottsym{,}  \Gamma  \ottsym{)}$.
  \item From $\ottkw{Cons} \, \ottsym{(}  \ottnt{H}  \ottsym{,}  \ottnt{R}  \ottsym{,}  \Gamma  \ottsym{)}$ we have $\ottkw{SAT} \, \ottsym{(}  \ottnt{H}  \ottsym{,}  \ottnt{R}  \ottsym{,}  \Gamma  \ottsym{)}$ which by \Cref{lem:sat-implies-gamma} we have $\models  \ottsym{[}  \ottnt{R}  \ottsym{]} \,  \sem{ \Gamma } $.
    We now proceed by induction on $\Gamma  \vdash  \tau  \leq  \tau'$.
    \begin{eqncase}{
        \tau  \ottsym{=}   \set{  \nu  \COL \TINT \mid  \varphi }  & \tau'  \ottsym{=}   \set{  \nu  \COL \TINT \mid  \varphi' }  \\
        \models   \sem{ \Gamma }   \wedge  \varphi  \implies  \varphi'
      }
      From $ \ottkw{SATv} ( \ottnt{H} , \ottnt{R} , \ottnt{v} , \tau ) $ we have $\models  \ottsym{[}  \ottnt{R}  \ottsym{]} \, \ottsym{[}  \ottnt{v}  \ottsym{/}  \nu  \ottsym{]}  \varphi$.
      We must show that $\models  \ottsym{[}  \ottnt{R}  \ottsym{]} \, \ottsym{[}  \ottnt{v}  \ottsym{/}  \nu  \ottsym{]}  \varphi'$.
      From $\models   \sem{ \Gamma }   \wedge  \varphi  \implies  \varphi'$ we must have
      $\models  \ottsym{[}  \ottnt{R}  \ottsym{]} \, \ottsym{[}  \ottnt{v}  \ottsym{/}  \nu  \ottsym{]}   \sem{ \Gamma }   \wedge  \ottsym{[}  \ottnt{R}  \ottsym{]} \, \ottsym{[}  \ottnt{v}  \ottsym{/}  \nu  \ottsym{]}  \varphi  \implies  \ottsym{[}  \ottnt{R}  \ottsym{]} \, \ottsym{[}  \ottnt{v}  \ottsym{/}  \nu  \ottsym{]}  \varphi'$ is valid.
      As $\nu$ does not appear free in $ \sem{ \Gamma } $, we have
      $\models  \ottsym{[}  \ottnt{R}  \ottsym{]} \,  \sem{ \Gamma }   \wedge  \ottsym{[}  \ottnt{R}  \ottsym{]} \, \ottsym{[}  \ottnt{v}  \ottsym{/}  \nu  \ottsym{]}  \varphi  \implies  \ottsym{[}  \ottnt{R}  \ottsym{]} \, \ottsym{[}  \ottnt{v}  \ottsym{/}  \nu  \ottsym{]}  \varphi'$ is
      valid whereby the result is immediate.
    \end{eqncase}
    
    \begin{eqncase}{
        \tau  \ottsym{=}   \tau_{\ottmv{p}}  \TREF^{ r_{{\mathrm{1}}} }  & \tau'  \ottsym{=}   \tau'_{\ottmv{p}}  \TREF^{ r_{{\mathrm{2}}} }  \\
        r_{{\mathrm{2}}}  \le  r_{{\mathrm{1}}}
      }
      Immediate from the induction hypothesis.
    \end{eqncase}
  \item Immediate consequence of 3 and that
    $\Gamma  \leq  \Gamma'$ implies that $\Gamma  \vdash  \Gamma  \ottsym{(}  \mathit{x}  \ottsym{)}  \leq  \Gamma'  \ottsym{(}  \mathit{x}  \ottsym{)}$ for any $ \mathit{x}  \in \DOM( \Gamma' ) $.
  \item Immediate from 2 and 4.
  \end{enumerate}
\end{proof}

To show consistency is preserved during evaluation, \Cref{lem:sattosat,lem:ownequiv-preserv}
show types equivalent according to $ \approx $ are equivalent for the purposes
of $\ottkw{own}$ and $\ottkw{SATv}$. Then \Cref{lem:ownadd,lem:satadd} show that the
type addition operator $+$ ``distributes'' over $\ottkw{SATv}$ and $\ottkw{own}$.

\begin{lemma}[Type Equivalence Preserves Satisfiability] %
  \label{lem:sattosat}
  If $\tau_{{\mathrm{1}}}  \approx  \tau_{{\mathrm{2}}}$, then $ \ottkw{SATv} ( \ottnt{H} , \ottnt{R} , \ottnt{v} , \tau_{{\mathrm{1}}} )   \iff   \ottkw{SATv} ( \ottnt{H} , \ottnt{R} , \ottnt{v} , \tau_{{\mathrm{2}}} ) $.
\end{lemma}
\begin{proof}
  We prove the forward case by induction on $ \bullet   \vdash  \tau_{{\mathrm{1}}}  \leq  \tau_{{\mathrm{2}}}$ as implied by
  $\tau_{{\mathrm{1}}}  \approx  \tau_{{\mathrm{2}}}$. The inductive case follows from the IH. In the the base case
  where $\tau_{{\mathrm{1}}}  \ottsym{=}   \set{  \nu  \COL \TINT \mid  \varphi_{{\mathrm{1}}} } $ and $\tau_{{\mathrm{2}}}  \ottsym{=}   \set{  \nu  \COL \TINT \mid  \varphi_{{\mathrm{2}}} } $, from $ \bullet   \vdash  \tau_{{\mathrm{1}}}  \leq  \tau_{{\mathrm{2}}}$
  we have that $\models  \varphi_{{\mathrm{1}}}  \implies  \varphi_{{\mathrm{2}}}$ is valid, from which we must have $\models  \ottsym{[}  \ottnt{R}  \ottsym{]} \, \ottsym{[}  \ottnt{v}  \ottsym{/}  \nu  \ottsym{]}  \varphi_{{\mathrm{1}}}  \implies  \ottsym{[}  \ottnt{R}  \ottsym{]} \, \ottsym{[}  \ottnt{v}  \ottsym{/}  \nu  \ottsym{]}  \varphi_{{\mathrm{2}}}$,
  where from the definition of $ \ottkw{SATv} ( \ottnt{H} , \ottnt{R} , \ottnt{v} , \tau_{{\mathrm{1}}} ) $ we must then have $ \ottkw{SATv} ( \ottnt{H} , \ottnt{R} , \ottnt{v} , \tau_{{\mathrm{2}}} ) $.
    
  The backwards case follows similar reasoning by induction on $ \bullet   \vdash  \tau_{{\mathrm{2}}}  \leq  \tau_{{\mathrm{1}}}$.
\end{proof}

\begin{lemma} %
  \label{lem:ownequiv-preserv}
  If $\tau_{{\mathrm{1}}}  \approx  \tau_{{\mathrm{2}}}$, then $\ottkw{own} \, \ottsym{(}  \ottnt{H}  \ottsym{,}  \ottnt{v}  \ottsym{,}  \tau_{{\mathrm{1}}}  \ottsym{)}  \ottsym{=}  \ottkw{own} \, \ottsym{(}  \ottnt{H}  \ottsym{,}  \ottnt{v}  \ottsym{,}  \tau_{{\mathrm{2}}}  \ottsym{)}$.
\end{lemma}
\begin{proof}
  By reasoning similar to that in \Cref{lem:sattosat}.
\end{proof}

\begin{lemma} %
  \label{lem:ownadd}
  If $\tau_{\ottmv{p}}  \ottsym{=}  \tau_{{\mathrm{1}}}  \ottsym{+}  \tau_{{\mathrm{2}}}$, then $\ottkw{own} \, \ottsym{(}  \ottnt{H}  \ottsym{,}  \ottnt{v}  \ottsym{,}  \tau_{\ottmv{p}}  \ottsym{)}  \ottsym{=}  \ottkw{own} \, \ottsym{(}  \ottnt{H}  \ottsym{,}  \ottnt{v}  \ottsym{,}  \tau_{{\mathrm{1}}}  \ottsym{)}  \ottsym{+}  \ottkw{own} \, \ottsym{(}  \ottnt{H}  \ottsym{,}  \ottnt{v}  \ottsym{,}  \tau_{{\mathrm{2}}}  \ottsym{)}$.
\end{lemma}
\begin{proof}
  By induction on the rules used to derive $\tau_{{\mathrm{1}}}  \ottsym{+}  \tau_{{\mathrm{2}}}  \ottsym{=}  \tau_{\ottmv{p}}$.
  \begin{rncase}{Tadd-Int}
    We have $\ottkw{own} \, \ottsym{(}  \ottnt{H}  \ottsym{,}  \ottnt{v}  \ottsym{,}  \tau_{\ottmv{p}}  \ottsym{)}  \ottsym{=}  \ottkw{own} \, \ottsym{(}  \ottnt{H}  \ottsym{,}  \ottnt{v}  \ottsym{,}  \tau_{{\mathrm{1}}}  \ottsym{+}  \tau_{{\mathrm{2}}}  \ottsym{)}$, where $\tau_{{\mathrm{1}}}  \ottsym{+}  \tau_{{\mathrm{2}}}  \ottsym{=}   \set{  \nu  \COL \TINT \mid   \varphi_{{\mathrm{1}}}  \wedge  \varphi_{{\mathrm{2}}}  } $,
    $\ottkw{own} \, \ottsym{(}  \ottnt{H}  \ottsym{,}  \ottnt{v}  \ottsym{,}  \tau_{{\mathrm{1}}}  \ottsym{)}$ and $\ottkw{own} \, \ottsym{(}  \ottnt{H}  \ottsym{,}  \ottnt{v}  \ottsym{,}  \tau_{{\mathrm{2}}}  \ottsym{)}$, where $\tau_{{\mathrm{1}}}  \ottsym{=}   \set{  \nu  \COL \TINT \mid  \varphi_{{\mathrm{1}}} } , \tau_{{\mathrm{2}}}  \ottsym{=}   \set{  \nu  \COL \TINT \mid  \varphi_{{\mathrm{2}}} } $.
  
    From the definition of ownership, we have $\ottkw{own} \, \ottsym{(}  \ottnt{H}  \ottsym{,}  \ottnt{v}  \ottsym{,}  \tau_{\ottmv{p}}  \ottsym{)}  \ottsym{=}  \ottkw{own} \, \ottsym{(}  \ottnt{H}  \ottsym{,}  \ottnt{v}  \ottsym{,}  \tau_{{\mathrm{1}}}  \ottsym{)}  \ottsym{=}  \ottkw{own} \, \ottsym{(}  \ottnt{H}  \ottsym{,}  \ottnt{v}  \ottsym{,}  \tau_{{\mathrm{2}}}  \ottsym{)}  \ottsym{=}   \emptyset $.
    It is thus trivial that $\ottkw{own} \, \ottsym{(}  \ottnt{H}  \ottsym{,}  \ottnt{v}  \ottsym{,}  \tau_{\ottmv{p}}  \ottsym{)}  \ottsym{=}  \ottkw{own} \, \ottsym{(}  \ottnt{H}  \ottsym{,}  \ottnt{v}  \ottsym{,}  \tau_{{\mathrm{1}}}  \ottsym{)}  \ottsym{+}  \ottkw{own} \, \ottsym{(}  \ottnt{H}  \ottsym{,}  \ottnt{v}  \ottsym{,}  \tau_{{\mathrm{2}}}  \ottsym{)}$.
  \end{rncase}

  \begin{rncase}{Tadd-Ref}
    We assume $\ottnt{v} \, \ottsym{=} \, \ottmv{a}$ and $ \ottmv{a}  \in   \DOM( \ottnt{H} )  $, otherwise the result trivially holds.
    
    We have $\ottkw{own} \, \ottsym{(}  \ottnt{H}  \ottsym{,}  \ottnt{v}  \ottsym{,}  \tau_{\ottmv{p}}  \ottsym{)}  \ottsym{=}  \ottkw{own} \, \ottsym{(}  \ottnt{H}  \ottsym{,}  \ottnt{v}  \ottsym{,}  \tau_{{\mathrm{1}}}  \ottsym{+}  \tau_{{\mathrm{2}}}  \ottsym{)}$, where $\tau_{{\mathrm{1}}}  \ottsym{+}  \tau_{{\mathrm{2}}}  \ottsym{=}   \ottsym{(}  \tau'_{{\mathrm{1}}}  \ottsym{+}  \tau'_{{\mathrm{2}}}  \ottsym{)}  \TREF^{ r_{{\mathrm{1}}}  \ottsym{+}  r_{{\mathrm{2}}} } $,
    and  $\tau_{{\mathrm{1}}}  \ottsym{=}   \tau'_{{\mathrm{1}}}  \TREF^{ r_{{\mathrm{1}}} } $, $\tau_{{\mathrm{2}}}  \ottsym{=}   \tau'_{{\mathrm{2}}}  \TREF^{ r_{{\mathrm{2}}} } $.
    
    From the definition of ownership, we have $\ottkw{own} \, \ottsym{(}  \ottnt{H}  \ottsym{,}  \ottnt{v}  \ottsym{,}  \tau_{\ottmv{p}}  \ottsym{)}  \ottsym{=}  \ottsym{\{}  \ottmv{a}  \mapsto  r_{{\mathrm{1}}}  \ottsym{+}  r_{{\mathrm{2}}}  \ottsym{\}}  \ottsym{+}  \ottkw{own} \, \ottsym{(}  \ottnt{H}  \ottsym{,}  \ottnt{H}  \ottsym{(}  \ottnt{v}  \ottsym{)}  \ottsym{,}  \tau'_{{\mathrm{1}}}  \ottsym{+}  \tau'_{{\mathrm{2}}}  \ottsym{)}$ and:
    \begin{align*}
      \ottkw{own} \, \ottsym{(}  \ottnt{H}  \ottsym{,}  \ottnt{v}  \ottsym{,}  \tau_{{\mathrm{1}}}  \ottsym{)}  \ottsym{+}  \ottkw{own} \, \ottsym{(}  \ottnt{H}  \ottsym{,}  \ottnt{v}  \ottsym{,}  \tau_{{\mathrm{2}}}  \ottsym{)} & = \ottsym{\{}  \ottmv{a}  \mapsto  r_{{\mathrm{1}}}  \ottsym{\}}  \ottsym{+}  \ottkw{own} \, \ottsym{(}  \ottnt{H}  \ottsym{,}  \ottnt{H}  \ottsym{(}  \ottnt{v}  \ottsym{)}  \ottsym{,}  \tau'_{{\mathrm{1}}}  \ottsym{)}  \ottsym{+}  \ottsym{\{}  \ottmv{a}  \mapsto  r_{{\mathrm{2}}}  \ottsym{\}}  \ottsym{+}  \ottkw{own} \, \ottsym{(}  \ottnt{H}  \ottsym{,}  \ottnt{H}  \ottsym{(}  \ottnt{v}  \ottsym{)}  \ottsym{,}  \tau'_{{\mathrm{2}}}  \ottsym{)} \\
                                    & = \ottsym{\{}  \ottmv{a}  \mapsto  r_{{\mathrm{1}}}  \ottsym{+}  r_{{\mathrm{2}}}  \ottsym{\}}  \ottsym{+}  \ottkw{own} \, \ottsym{(}  \ottnt{H}  \ottsym{,}  \ottnt{H}  \ottsym{(}  \ottnt{v}  \ottsym{)}  \ottsym{,}  \tau'_{{\mathrm{1}}}  \ottsym{)}  \ottsym{+}  \ottkw{own} \, \ottsym{(}  \ottnt{H}  \ottsym{,}  \ottnt{H}  \ottsym{(}  \ottnt{v}  \ottsym{)}  \ottsym{,}  \tau'_{{\mathrm{2}}}  \ottsym{)}
    \end{align*}
    By the induction hypothesis, have that $\ottkw{own} \, \ottsym{(}  \ottnt{H}  \ottsym{,}  \ottnt{H}  \ottsym{(}  \ottnt{v}  \ottsym{)}  \ottsym{,}  \tau'_{{\mathrm{1}}}  \ottsym{+}  \tau'_{{\mathrm{2}}}  \ottsym{)}  \ottsym{=}  \ottkw{own} \, \ottsym{(}  \ottnt{H}  \ottsym{,}  \ottnt{H}  \ottsym{(}  \ottnt{v}  \ottsym{)}  \ottsym{,}  \tau'_{{\mathrm{1}}}  \ottsym{)}  \ottsym{+}  \ottkw{own} \, \ottsym{(}  \ottnt{H}  \ottsym{,}  \ottnt{H}  \ottsym{(}  \ottnt{v}  \ottsym{)}  \ottsym{,}  \tau'_{{\mathrm{2}}}  \ottsym{)}$ and can conclude that $\ottkw{own} \, \ottsym{(}  \ottnt{H}  \ottsym{,}  \ottnt{v}  \ottsym{,}  \tau_{\ottmv{p}}  \ottsym{)}  \ottsym{=}  \ottkw{own} \, \ottsym{(}  \ottnt{H}  \ottsym{,}  \ottnt{v}  \ottsym{,}  \tau_{{\mathrm{1}}}  \ottsym{)}  \ottsym{+}  \ottkw{own} \, \ottsym{(}  \ottnt{H}  \ottsym{,}  \ottnt{v}  \ottsym{,}  \tau_{{\mathrm{2}}}  \ottsym{)}$.
  \end{rncase}
\end{proof}

\begin{lemma} %
  \label{lem:satadd}
  If $\tau_{\ottmv{p}}  \ottsym{=}  \tau_{{\mathrm{1}}}  \ottsym{+}  \tau_{{\mathrm{2}}}$, we have
    $ \ottkw{SATv} ( \ottnt{H} , \ottnt{R} , v , \tau_{\ottmv{p}} ) $ iff. $ \ottkw{SATv} ( \ottnt{H} , \ottnt{R} , v , \tau_{{\mathrm{1}}} ) $ and $ \ottkw{SATv} ( \ottnt{H} , \ottnt{R} , v , \tau_{{\mathrm{2}}} ) $
\end{lemma}
\begin{proof}
  By induction on the rules used to derive $\tau_{{\mathrm{1}}}  \ottsym{+}  \tau_{{\mathrm{2}}}$. In the following
  we only prove the forward direction of the implication; the backwards
  direction is symmetric.
  \begin{rncase}{Tadd-Int}
    We have $ \ottkw{SATv} ( \ottnt{H} , \ottnt{R} , v , \tau_{{\mathrm{1}}}  \ottsym{+}  \tau_{{\mathrm{2}}} ) $, where $\tau_{{\mathrm{1}}}  \ottsym{+}  \tau_{{\mathrm{2}}}  \ottsym{=}   \set{  \nu  \COL \TINT \mid   \varphi_{{\mathrm{1}}}  \wedge  \varphi_{{\mathrm{2}}}  } $,
    we must show $ \ottkw{SATv} ( \ottnt{H} , \ottnt{R} , v , \tau_{{\mathrm{1}}} ) $ and $ \ottkw{SATv} ( \ottnt{H} , \ottnt{R} , v , \tau_{{\mathrm{2}}} ) $,
    where $\tau_{{\mathrm{1}}}  \ottsym{=}   \set{  \nu  \COL \TINT \mid  \varphi_{{\mathrm{1}}} } , \tau_{{\mathrm{2}}}  \ottsym{=}   \set{  \nu  \COL \TINT \mid  \varphi_{{\mathrm{2}}} } $.

    From the definition of $\ottkw{SATv}$, we must show $\ottsym{[}  \ottnt{R}  \ottsym{]} \, \ottsym{[}  v  \ottsym{/}  \nu  \ottsym{]}  \varphi_{{\mathrm{1}}}$ and $\ottsym{[}  \ottnt{R}  \ottsym{]} \, \ottsym{[}  v  \ottsym{/}  \nu  \ottsym{]}  \varphi_{{\mathrm{2}}}$.
    From $ \ottkw{SATv} ( \ottnt{H} , \ottnt{R} , v , \tau_{{\mathrm{1}}}  \ottsym{+}  \tau_{{\mathrm{2}}} ) $ we have $\ottsym{[}  \ottnt{R}  \ottsym{]} \, \ottsym{[}  \ottnt{v}  \ottsym{/}  \nu  \ottsym{]}  \ottsym{(}   \varphi_{{\mathrm{1}}}  \wedge  \varphi_{{\mathrm{2}}}   \ottsym{)}$.
    It is immediate that for any value $v$ such that $\ottsym{[}  \ottnt{R}  \ottsym{]} \, \ottsym{[}  \ottnt{v}  \ottsym{/}  \nu  \ottsym{]}  \ottsym{(}   \varphi_{{\mathrm{1}}}  \wedge  \varphi_{{\mathrm{2}}}   \ottsym{)}$, we must have $\ottsym{[}  \ottnt{R}  \ottsym{]} \, \ottsym{[}  \ottnt{v}  \ottsym{/}  \nu  \ottsym{]}  \varphi_{{\mathrm{1}}}$ and $\ottsym{[}  \ottnt{R}  \ottsym{]} \, \ottsym{[}  \ottnt{v}  \ottsym{/}  \nu  \ottsym{]}  \varphi_{{\mathrm{2}}}$.
    We then conclude $ \ottkw{SATv} ( \ottnt{H} , \ottnt{R} , v , \tau_{{\mathrm{1}}}  \ottsym{+}  \tau_{{\mathrm{2}}} ) $ implies $ \ottkw{SATv} ( \ottnt{H} , \ottnt{R} , v , \tau_{{\mathrm{1}}} ) $ and $ \ottkw{SATv} ( \ottnt{H} , \ottnt{R} , v , \tau_{{\mathrm{2}}} ) $.
  \end{rncase}

  \begin{rncase}{Tadd-Ref}
    Immediate from the definition of $\ottkw{SATv}$ and the inductive hypothesis.
  \end{rncase}
\end{proof}

\begin{definition}
  The valid substitution relation, written $ \ottnt{R} \vdash _{vs}  \tau $ is the smallest relation closed
  under the following rules:
  \infrule[]{
    \forall \,  \mathit{x}  \in   \ottkw{FPV} \, \ottsym{(}  \varphi  \ottsym{)}  \setminus   \set{ \nu }     \ottsym{.}  \exists \, n  \ottsym{.}  \ottnt{R}  \ottsym{(}  \mathit{x}  \ottsym{)}  \ottsym{=}  n
  }{
     \ottnt{R} \vdash _{vs}   \set{  \nu  \COL \TINT \mid  \varphi }  
  }
  \infrule[]{
     \ottnt{R} \vdash _{vs}  \tau 
  }{
    \ottnt{R} \vdash _{vs}   \tau  \TREF^{ r }  
 }
\end{definition}

\begin{lemma} %
  \label{lem:r-valid-subst}
  If $ \oldvec{\ell}   \vdash _{\wf}  \Gamma $ and $\ottkw{Cons} \, \ottsym{(}  \ottnt{H}  \ottsym{,}  \ottnt{R}  \ottsym{,}  \Gamma  \ottsym{)}$, then $\forall \,  \mathit{x}  \in   \DOM( \Gamma )    \ottsym{.}   \ottnt{R} \vdash _{vs}  \Gamma  \ottsym{(}  \mathit{x}  \ottsym{)} $.
\end{lemma}
\begin{proof}
  By $\ottkw{Cons} \, \ottsym{(}  \ottnt{H}  \ottsym{,}  \ottnt{R}  \ottsym{,}  \Gamma  \ottsym{)}$, all integer variables in $\Gamma$ must be in
  the domain of $\ottnt{R}$ and must be an integer. From $ \oldvec{\ell}   \vdash _{\wf}  \Gamma $, any free variables
  in any refinement of any type in $\Gamma$ must be an integer valued variable in $\Gamma$,
  which gives the required result.
\end{proof}

\begin{definition}
  We will write $ \ottnt{R}  \sqsubseteq  \ottnt{R'} $ to denote two register files such that:
  \begin{enumerate}
  \item $ \DOM( \ottnt{R} )  \subseteq  \DOM( \ottnt{R'} ) $, and
  \item $\forall \,  \mathit{x}  \in \DOM( \ottnt{R} )   \ottsym{.}  \ottnt{R}  \ottsym{(}  \mathit{x}  \ottsym{)}  \ottsym{=}  \ottnt{R'}  \ottsym{(}  \mathit{x}  \ottsym{)}$
  \end{enumerate}
\end{definition}

\begin{definition}
  Two heaps $\ottnt{H}$ and $\ottnt{H'}$ are \emph{equivalent modulo $\ottmv{a}$}, written $ \ottnt{H}   \approx _ \ottmv{a}   \ottnt{H'} $ if:
  \begin{enumerate}
  \item $ \DOM( \ottnt{H} )   \ottsym{=}   \DOM( \ottnt{H'} ) $
  \item $\forall \,  \ottmv{a'}  \in   \DOM( \ottnt{H} )    \ottsym{.}  \ottmv{a'} \, \neq \, \ottmv{a}  \implies  \ottnt{H}  \ottsym{(}  \ottmv{a'}  \ottsym{)}  \ottsym{=}  \ottnt{H'}  \ottsym{(}  \ottmv{a'}  \ottsym{)}$
  \item For any $n$, $ \ottnt{H} \vdash   \ottmv{a}  \Downarrow  n $ iff $ \ottnt{H'} \vdash   \ottmv{a}  \Downarrow  n $
  \end{enumerate}
\end{definition}

\begin{lemma}[Same-shape heap update]
  \label{lem:same-shape-update}
  Suppose $\ottnt{H}  \ottsym{(}  \ottmv{a}  \ottsym{)} \, \ottsym{=} \, \ottnt{v}$, $ \ottnt{H} \vdash   \ottnt{v}  \Downarrow  n $, and
  $ \ottnt{H} \vdash   \ottnt{v'}  \Downarrow  n $.  If $\ottnt{H'}  \ottsym{=}  \ottnt{H}  \ottsym{\{}  \ottmv{a}  \hookleftarrow  \ottnt{v'}  \ottsym{\}}$, then $ \ottnt{H}   \approx _ \ottmv{a}   \ottnt{H'} $.
\end{lemma}
\begin{proof}
  Easy.
\end{proof}

\begin{lemma}
  \label{lem:top-type-empty-omap}
  For any type $\tau  \ottsym{=}  \top_{\ottmv{n}}$, $\ottnt{H}, \ottnt{v}$, $\ottkw{own} \, \ottsym{(}  \ottnt{H}  \ottsym{,}  \ottnt{v}  \ottsym{,}  \top_{\ottmv{n}}  \ottsym{)}  \ottsym{=}   \emptyset $.
\end{lemma}
\begin{proof}
  By induction on $\tau$. In the base case, the result is trivial.
  Then consider the case where $\tau  \ottsym{=}   \top_{{\ottmv{n}-1}}  \TREF^{ \ottsym{0} } $. If $ \ottnt{v}  \not\in   \textbf{Addr}  $,
  or if $\ottnt{v} \, \ottsym{=} \, \ottmv{a}$ and $ \ottmv{a}  \not\in   \DOM( \ottnt{H} )  $, then the result trivially holds.
  Otherwise the result holds from the inductive hypothesis, the definition of $+$ and $\ottsym{\{}  \ottmv{a}  \mapsto  \ottsym{0}  \ottsym{\}}$.
\end{proof}

\begin{lemma}[Heap Update Ownership Preservation] %
  \label{lem:heapop}
  If $ \mathcal{L}   \mid   \Gamma   \vdash _{\wf}  \tau $, $ \ottnt{H}   \approx _ \ottmv{a}   \ottnt{H'} $, and $\ottkw{own} \, \ottsym{(}  \ottnt{H}  \ottsym{,}  \ottnt{v}  \ottsym{,}  \tau  \ottsym{)}  \ottsym{(}  \ottmv{a}  \ottsym{)}  \ottsym{=}  \ottsym{0}$, then
  $\ottkw{own} \, \ottsym{(}  \ottnt{H}  \ottsym{,}  \ottnt{v}  \ottsym{,}  \tau  \ottsym{)}  \ottsym{=}  \ottkw{own} \, \ottsym{(}  \ottnt{H'}  \ottsym{,}  \ottnt{v}  \ottsym{,}  \tau  \ottsym{)}$.
\end{lemma}
\begin{proof}
  By induction on the shape of $\tau$. If $\tau  \ottsym{=}   \set{  \nu  \COL \TINT \mid  \varphi } $ then
  the result trivially holds. Otherwise, $\tau  \ottsym{=}   \tau'  \TREF^{ r } $. We assume
  that $\ottnt{v} \, \ottsym{=} \, \ottmv{a''}$ and $ \ottmv{a''}  \in \DOM( \ottnt{H} ) $ (otherwise the result
  trivially holds, as $ \DOM( \ottnt{H} )   \ottsym{=}   \DOM( \ottnt{H'} ) $ by $ \ottnt{H}   \approx _ \ottmv{a}   \ottnt{H'} $).
  Consider the case where $\ottmv{a''} \, \ottsym{=} \, \ottmv{a}$. By definition
  $\ottkw{own} \, \ottsym{(}  \ottnt{H}  \ottsym{,}  \ottmv{a}  \ottsym{,}  \tau  \ottsym{)}  \ottsym{=}  \ottsym{\{}  \ottmv{a}  \mapsto  r  \ottsym{\}}  \ottsym{+}  \ottkw{own} \, \ottsym{(}  \ottnt{H}  \ottsym{,}  \ottnt{H}  \ottsym{(}  \ottmv{a}  \ottsym{)}  \ottsym{,}  \tau'  \ottsym{)}$, and by the
  assumption that $\ottkw{own} \, \ottsym{(}  \ottnt{H}  \ottsym{,}  \ottmv{a}  \ottsym{,}  \tau  \ottsym{)}  \ottsym{(}  \ottmv{a}  \ottsym{)}  \ottsym{=}  \ottsym{0}$ we must have that
  $r  \ottsym{=}  \ottsym{0}$. Further, by the ownership well-formedness of types,
  we must have $\tau'  \ottsym{=}  \top_{\ottmv{n}}$ for some $n$, thus by \Cref{lem:top-type-empty-omap}
  we have $\ottkw{own} \, \ottsym{(}  \ottnt{H}  \ottsym{,}  \ottnt{v}  \ottsym{,}  \tau  \ottsym{)}  \ottsym{=}   \emptyset   \ottsym{=}  \ottsym{\{}  \ottmv{a}  \mapsto  \ottsym{0}  \ottsym{\}}  \ottsym{+}  \ottkw{own} \, \ottsym{(}  \ottnt{H'}  \ottsym{,}  \ottnt{H'}  \ottsym{(}  \ottmv{a}  \ottsym{)}  \ottsym{,}  \top_{\ottmv{n}}  \ottsym{)}  \ottsym{=}  \ottkw{own} \, \ottsym{(}  \ottnt{H'}  \ottsym{,}  \ottnt{v}  \ottsym{,}  \tau  \ottsym{)}$.

  Finally, consider the case where $\ottmv{a''} \, \neq \, \ottmv{a}$. Then from the
  definition of $\ottkw{own} \, \ottsym{(}  \ottnt{H}  \ottsym{,}  \ottmv{a''}  \ottsym{,}  \tau  \ottsym{)}$ and our assumption that
  $\ottkw{own} \, \ottsym{(}  \ottnt{H}  \ottsym{,}  \ottmv{a''}  \ottsym{,}  \tau  \ottsym{)}  \ottsym{(}  \ottmv{a}  \ottsym{)}  \ottsym{=}  \ottsym{0}$, we have $\ottkw{own} \, \ottsym{(}  \ottnt{H}  \ottsym{,}  \ottnt{H}  \ottsym{(}  \ottmv{a''}  \ottsym{)}  \ottsym{,}  \tau'  \ottsym{)}  \ottsym{(}  \ottmv{a}  \ottsym{)}  \ottsym{=}  \ottsym{0}$,
  and the result holds from the inductive hypothesis and that $\ottnt{H}  \ottsym{(}  \ottmv{a''}  \ottsym{)} \, \ottsym{=} \, \ottnt{H'}  \ottsym{(}  \ottmv{a''}  \ottsym{)}$.
\end{proof}

\begin{lemma}
  \label{lem:sat-implies-shape-cons}
  For any $\ottnt{H}$, $\ottnt{R}$, $\ottnt{v}$, and $\tau$, if $ \ottkw{SATv} ( \ottnt{H} , \ottnt{R} , \ottnt{v} , \tau ) $ then $ \ottnt{H} \vdash   \ottnt{v}  \Downarrow   | \tau |  $
\end{lemma}
\begin{proof}
  By induction on $\tau$ and the definition of $\ottkw{SATv}$.
\end{proof}

\begin{lemma} %
  \label{lem:top-type-sat-all}
  For any $n$, if $ \ottnt{H} \vdash   v  \Downarrow  n $ then for any $\ottnt{R}$, $ \ottkw{SATv} ( \ottnt{H} , \ottnt{R} , v , \top_{\ottmv{n}} ) $.
\end{lemma}
\begin{proof}
  By induction on $n$. In the base case, by inversion on $ \ottnt{H} \vdash   v  \Downarrow  \ottsym{0} $ we have
  $ v  \in  \mathbb{Z} $ and as $\ottsym{[}  \ottnt{R}  \ottsym{]} \, \ottsym{[}  v  \ottsym{/}  \nu  \ottsym{]}   \top   \implies   \top $, we conclude $ \ottkw{SATv} ( \ottnt{H} , \ottnt{R} , v , \top_{{\mathrm{0}}} ) $.

  For $n > 0$, by inversion on $ \ottnt{H} \vdash   v  \Downarrow  n $ we have that
  $v \, \ottsym{=} \, \ottmv{a}$, $ \ottmv{a}  \in   \DOM( \ottnt{H} )  $, and $ \ottnt{H} \vdash   \ottnt{H}  \ottsym{(}  \ottmv{a}  \ottsym{)}  \Downarrow  n  \ottsym{-}  \ottsym{1} $,
  whereby the result holds from the inductive hypothesis.
\end{proof}
  
\begin{lemma}[Heap Update Consistency Preservation] %
  \label{lem:heapfor0}
  If $ \mathcal{L}   \mid   \Gamma   \vdash _{\wf}  \tau $, $ \ottnt{H}   \approx _ \ottmv{a}   \ottnt{H'} $, $\ottkw{own} \, \ottsym{(}  \ottnt{H}  \ottsym{,}  \ottnt{v}  \ottsym{,}  \tau  \ottsym{)}  \ottsym{(}  \ottmv{a}  \ottsym{)}  \ottsym{=}  \ottsym{0}$, and $ \ottkw{SATv} ( \ottnt{H} , \ottnt{R} , \ottnt{v} , \tau ) $, then $ \ottkw{SATv} ( \ottnt{H'} , \ottnt{R} , \ottnt{v} , \tau ) $.
\end{lemma}
\begin{proof}
  By induction on the shape of $\tau$. The base case where
  $\tau  \ottsym{=}   \set{  \nu  \COL \TINT \mid  \varphi } $ is trivial. We therefore consider the case
  where $\ottnt{v} \, \ottsym{=} \, \ottmv{a'}$ and $\tau  \ottsym{=}   \tau'  \TREF^{ r } $.

  We first consider the case where $\ottmv{a'} \, \ottsym{=} \, \ottmv{a}$, then by
  our assumption that $\ottkw{own} \, \ottsym{(}  \ottnt{H}  \ottsym{,}  \ottmv{a}  \ottsym{,}  \tau  \ottsym{)}  \ottsym{(}  \ottmv{a}  \ottsym{)}  \ottsym{=}  \ottsym{0}$, we must have that
  $\tau  \ottsym{=}   \tau'  \TREF^{ \ottsym{0} } $, whereby $\tau  \ottsym{=}  \top_{\ottmv{n}}$ for some $\ottmv{n}$.
  From $ \ottkw{SATv} ( \ottnt{H} , \ottnt{R} , \ottmv{a} , \tau ) $ and \Cref{lem:sat-implies-shape-cons}, we must have that
  $ \ottnt{H} \vdash   \ottmv{a}  \Downarrow   | \tau |  $,
  and from $ \ottnt{H}   \approx _ \ottmv{a}   \ottnt{H'} $, we therefore have that $ \ottnt{H'} \vdash   \ottmv{a}  \Downarrow   | \tau |  $
  whereby the result holds from \Cref{lem:top-type-sat-all}.

  Otherwise, we have that $\ottmv{a'} \, \neq \, \ottmv{a}$, and by definition we must have that
  $\ottkw{own} \, \ottsym{(}  \ottnt{H}  \ottsym{,}  \ottnt{H}  \ottsym{(}  \ottmv{a'}  \ottsym{)}  \ottsym{,}  \tau'  \ottsym{)}  \ottsym{(}  \ottmv{a}  \ottsym{)}  \ottsym{=}  \ottsym{0}$ and $\ottnt{H'}  \ottsym{(}  \ottmv{a'}  \ottsym{)} \, \ottsym{=} \, \ottnt{H}  \ottsym{(}  \ottmv{a'}  \ottsym{)}$ hence the result follows
  from the inductive hypothesis.
\end{proof}

\begin{lemma}[Register Weakening] %
  \label{lem:register}
  If $ \ottkw{SATv} ( \ottnt{H} , \ottnt{R} , v , \tau ) $ and $ \ottnt{R} \vdash _{vs}  \tau $,
  then for any $\ottnt{R'}$ such that $ \ottnt{R}  \sqsubseteq  \ottnt{R'} $, $ \ottkw{SATv} ( \ottnt{H} , \ottnt{R'} , v , \tau ) $.
\end{lemma}
\begin{proof}
  By induction on the shape of $\tau$. If $\tau  \ottsym{=}   \tau'  \TREF^{ r } $, then
  the result follows from the inductive hypothesis.
  We therefore consider the case where
  $\tau  \ottsym{=}   \set{  \nu  \COL \TINT \mid  \varphi } $. Without loss of generality, we consider
  the case where $ \DOM( \ottnt{R'} )  \setminus  \DOM( \ottnt{R} )  = \set{x}$,
  and $\ottnt{R'}  \ottsym{(}  \mathit{x}  \ottsym{)} \, \ottsym{=} \, n$. (If $\ottnt{R'}  \ottsym{(}  \mathit{x}  \ottsym{)} \, \ottsym{=} \, \ottmv{a}$, the extra binding
  also has no effect, and the case where more than one binding is
  added follows from $n$ applications of the following argument.)

  From $ \ottkw{SATv} ( \ottnt{H} , \ottnt{R} , v , \tau ) $, we conclude that $ v  \in  \mathbb{Z} $ and that
  $\ottsym{[}  \ottnt{R}  \ottsym{]} \, \ottsym{[}  v  \ottsym{/}  \nu  \ottsym{]}  \varphi$. If $ \mathit{x}  \not\in  \ottkw{FPV} \, \ottsym{(}  \varphi  \ottsym{)} $ then $\ottsym{[}  \ottnt{R}  \ottsym{]} \, \ottsym{[}  v  \ottsym{/}  \nu  \ottsym{]}  \varphi  \iff  \ottsym{[}  \ottnt{R'}  \ottsym{]} \, \ottsym{[}  v  \ottsym{/}  \nu  \ottsym{]}  \varphi$
  and the result holds trivially. Otherwise, if $ \mathit{x}  \in  \ottkw{FPV} \, \ottsym{(}  \varphi  \ottsym{)} $ and $ \mathit{x}  \not\in \DOM( \ottnt{R} ) $
  then $\ottnt{R}$ is not a valid substitution, violating our assumption.
\end{proof}

\begin{lemma}[Heap Extension Consistency Preservation] %
  \label{lem:newaddheap}
  If we have heap $H$, such that $ \ottkw{SATv} ( \ottnt{H} , \ottnt{R} , \ottnt{v} , \tau ) $, for any heap 
  $\ottnt{H'}  \ottsym{=}  \ottnt{H}  \ottsym{\{}  \ottmv{a}  \mapsto  \ottnt{v'}  \ottsym{\}},  \ottmv{a}  \not\in \DOM( \ottnt{H} ) $, then we have $ \ottkw{SATv} ( \ottnt{H'} , \ottnt{R} , \ottnt{v} , \tau ) $.
\end{lemma}
\begin{proof}
  By induction on the shape of $\tau$.
  The base case where $\tau  \ottsym{=}   \set{  \nu  \COL \TINT \mid  \varphi } $ is trivial.
  Next, we consider the case where $\tau  \ottsym{=}   \tau'  \TREF^{ r } $.
  We must show that $ \ottnt{v}  \in   \DOM( \ottnt{H'} )  $ and $ \ottkw{SATv} ( \ottnt{H'} , \ottnt{R} , \ottnt{H'}  \ottsym{(}  \ottnt{v}  \ottsym{)} , \tau' ) $.
  The first condition is immediately satisfied by inversion on $ \ottkw{SATv} ( \ottnt{H} , \ottnt{R} , \ottnt{v} , \tau' ) $,
  and from $ \ottmv{a}  \not\in \DOM( \ottnt{H} ) $, we have $v\neq a$, which gives that $H'(v)=H(v)$.
  That is we must show $ \ottkw{SATv} ( \ottnt{H'} , \ottnt{R} , \ottnt{H}  \ottsym{(}  \ottnt{v}  \ottsym{)} , \tau' ) $, which is follows from
  the induction hypothesis.
\end{proof}

\begin{lemma}[Full-reference exclusion]
  \label{lem:full-reference-exclusion}
  Suppose $ \mathcal{L}   \vdash _{\wf}  \Gamma $, $\ottkw{Cons} \, \ottsym{(}  \ottnt{H}  \ottsym{,}  \ottnt{R}  \ottsym{,}  \Gamma  \ottsym{)}$, $\Gamma  \ottsym{(}  \mathit{y}  \ottsym{)}  \ottsym{=}   \tau  \TREF^{ \ottsym{1} } $, and
  $\ottnt{R}  \ottsym{(}  \mathit{y}  \ottsym{)} \, \ottsym{=} \, \ottmv{a}$.  Then $\ottnt{H}  \ottsym{(}  \ottmv{a}  \ottsym{)}$ is defined,
  $\ottkw{own} \, \ottsym{(}  \ottnt{H}  \ottsym{,}  \ottnt{H}  \ottsym{(}  \ottmv{a}  \ottsym{)}  \ottsym{,}  \tau  \ottsym{)}  \ottsym{(}  \ottmv{a}  \ottsym{)}  \ottsym{=}  \ottsym{0}$, and for every
  $ \mathit{x}  \in   \DOM( \Gamma )  $ with $\mathit{x} \, \neq \, \mathit{y}$,
  $\ottkw{own} \, \ottsym{(}  \ottnt{H}  \ottsym{,}  \ottnt{R}  \ottsym{(}  \mathit{x}  \ottsym{)}  \ottsym{,}  \Gamma  \ottsym{(}  \mathit{x}  \ottsym{)}  \ottsym{)}  \ottsym{(}  \ottmv{a}  \ottsym{)}  \ottsym{=}  \ottsym{0}$.
\end{lemma}
\begin{proof}
  Satisfaction of $\mathit{y}$ gives the heap lookup.  Its contribution at
  $\ottmv{a}$ is
  $1+\ottkw{own} \, \ottsym{(}  \ottnt{H}  \ottsym{,}  \ottnt{H}  \ottsym{(}  \ottmv{a}  \ottsym{)}  \ottsym{,}  \tau  \ottsym{)}  \ottsym{(}  \ottmv{a}  \ottsym{)}$.  All remaining contributions, including the
  recursive contribution of $\tau$, are nonnegative, while consistency
  bounds their total by
  one.  Hence every contribution other than the displayed $1$ is zero.
\end{proof}

\begin{lemma}[Heap Extension Ownership Preservation]
  \label{lem:ownaddheap}
  If $ \ottkw{SATv} ( \ottnt{H} , \ottnt{R} , \ottnt{v} , \tau ) $, then for any $ \ottmv{a}  \not\in   \DOM( \ottnt{H} )  $
  $\ottkw{own} \, \ottsym{(}  \ottnt{H}  \ottsym{,}  \ottnt{v}  \ottsym{,}  \tau  \ottsym{)}  \ottsym{=}  \ottkw{own} \, \ottsym{(}  \ottnt{H}  \ottsym{\{}  \ottmv{a}  \mapsto  \ottnt{v'}  \ottsym{\}}  \ottsym{,}  \ottnt{v}  \ottsym{,}  \tau  \ottsym{)}$ for any value $\ottnt{v'}$.
\end{lemma}
\begin{proof}
  By induction on $\tau$. The base case is trivial as $\ottkw{own} \, \ottsym{(}  \ottnt{H}  \ottsym{,}  \ottnt{v}  \ottsym{,}   \set{  \nu  \COL \TINT \mid  \varphi }   \ottsym{)}  \ottsym{=}   \emptyset   \ottsym{=}  \ottkw{own} \, \ottsym{(}  \ottnt{H}  \ottsym{\{}  \ottmv{a}  \mapsto  \ottnt{v}  \ottsym{\}}  \ottsym{,}  \ottnt{v}  \ottsym{,}   \set{  \nu  \COL \TINT \mid  \varphi }   \ottsym{)}$.
  We therefore consider the case where $\tau  \ottsym{=}   \tau'  \TREF^{ r } $.

  From $ \ottkw{SATv} ( \ottnt{H} , \ottnt{R} , \ottnt{v} , \tau ) $ we must have that $\ottnt{v} \, \ottsym{=} \, \ottmv{a'}$ and
  $ \ottmv{a'}  \in   \DOM( \ottnt{H} )  $ (and by extension $ \ottmv{a'}  \in   \DOM( \ottnt{H}  \ottsym{\{}  \ottmv{a}  \mapsto  \ottnt{v'}  \ottsym{\}} )  $).
  From the definition of the ownership function, we have that $\ottkw{own} \, \ottsym{(}  \ottnt{H}  \ottsym{,}  \ottnt{v}  \ottsym{,}  \tau  \ottsym{)}  \ottsym{=}  \ottkw{own} \, \ottsym{(}  \ottnt{H}  \ottsym{,}  \ottnt{H}  \ottsym{(}  \ottmv{a'}  \ottsym{)}  \ottsym{,}  \tau'  \ottsym{)}  \ottsym{+}  \ottsym{\{}  \ottmv{a'}  \mapsto  r  \ottsym{\}}$.
  and $\ottkw{own} \, \ottsym{(}  \ottnt{H}  \ottsym{\{}  \ottmv{a}  \mapsto  \ottnt{v'}  \ottsym{\}}  \ottsym{,}  \ottnt{v}  \ottsym{,}  \tau  \ottsym{)}  \ottsym{=}  \ottkw{own} \, \ottsym{(}  \ottnt{H}  \ottsym{\{}  \ottmv{a}  \mapsto  \ottnt{v'}  \ottsym{\}}  \ottsym{,}  \ottnt{H}  \ottsym{\{}  \ottmv{a}  \mapsto  \ottnt{v'}  \ottsym{\}}  \ottsym{(}  \ottmv{a'}  \ottsym{)}  \ottsym{,}  \tau'  \ottsym{)}  \ottsym{+}  \ottsym{\{}  \ottmv{a'}  \mapsto  r  \ottsym{\}}$
  Then from our requirement that  $ \ottmv{a}  \not\in   \DOM( \ottnt{H} )  $, we have $\ottmv{a} \, \neq \, \ottmv{a'}$ and therefore
  $\ottnt{H}  \ottsym{(}  \ottmv{a'}  \ottsym{)} \, \ottsym{=} \, \ottnt{H}  \ottsym{\{}  \ottmv{a}  \mapsto  \ottnt{v'}  \ottsym{\}}  \ottsym{(}  \ottmv{a'}  \ottsym{)}$, whereby the result holds from the inductive hypothesis.
\end{proof}

\begin{lemma}[Environment Weakening] %
  \label{lem:tyenv-weaken}
  Define the partial operation $ \Gamma_{{\mathrm{1}}}  \uplus  \Gamma_{{\mathrm{2}}} $ for two environments such that $  \DOM( \Gamma_{{\mathrm{1}}} )   \cap   \DOM( \Gamma_{{\mathrm{2}}} )    \ottsym{=}   \emptyset $:
  \[
    \ottsym{(}   \Gamma_{{\mathrm{1}}}  \uplus  \Gamma_{{\mathrm{2}}}   \ottsym{)} \quad \ottsym{(}  \mathit{x}  \ottsym{)} = \begin{cases}
      \Gamma_{{\mathrm{1}}} \quad \ottsym{(}  \mathit{x}  \ottsym{)} &  \mathit{x}  \in   \DOM( \Gamma_{{\mathrm{1}}} )   \\
      \Gamma_{{\mathrm{2}}} \quad \ottsym{(}  \mathit{x}  \ottsym{)} &  \mathit{x}  \in   \DOM( \Gamma_{{\mathrm{2}}} )   \\
      \textit{undef} & o.w.
    \end{cases}
  \]
  
  Then, for any $\Gamma$ and $\Gamma''$ where $  \DOM( \Gamma )   \cap   \DOM( \Gamma'' )    \ottsym{=}   \emptyset $:
  \begin{enumerate}
  \item $\Gamma  \vdash  \tau_{{\mathrm{1}}}  \leq  \tau_{{\mathrm{2}}}$ implies $ \Gamma  \uplus  \Gamma''   \vdash  \tau_{{\mathrm{1}}}  \leq  \tau_{{\mathrm{2}}}$
  \item $\Gamma  \leq  \Gamma'$ implies $ \Gamma  \uplus  \Gamma''   \leq   \Gamma'  \uplus  \Gamma'' $
  \item If $ \Theta   \mid   \mathcal{L}   \mid   \Gamma   \vdash   \ottnt{e}  :  \tau   \produces   \Gamma' $, $ \mathcal{L}   \vdash _{\wf}   \Gamma  \uplus  \Gamma''  $, and $ \mathcal{L}   \vdash _{\wf}   \Gamma'  \uplus  \Gamma''  $, then $ \Theta   \mid   \mathcal{L}   \mid    \Gamma  \uplus  \Gamma''    \vdash   \ottnt{e}  :  \tau   \produces    \Gamma'  \uplus  \Gamma''  $
  \end{enumerate}
\end{lemma}
\begin{proof}
  \leavevmode
  \begin{enumerate}
  \item As in the proof of \subref{lem:ctxt-substitution}{itm:ctxt-sub-subtype}, at the root of
    the subtyping derivation is a logical judgment of the form
    $\models   \sem{ \Gamma }   \wedge  \varphi_{{\mathrm{1}}}  \implies  \varphi_{{\mathrm{2}}}$ which can be shown to be valid. We
    must then show that $\models   \sem{  \Gamma  \uplus  \Gamma''  }   \wedge  \varphi_{{\mathrm{1}}}  \implies  \varphi_{{\mathrm{2}}}$ is valid. As
    $ \sem{  \Gamma''  \uplus  \Gamma  }   \wedge  \varphi_{{\mathrm{1}}}  \ottsym{=}   \sem{ \Gamma'' }   \wedge   \sem{ \Gamma }   \wedge  \varphi$ only strengthens the pre-condition
    $ \sem{ \Gamma }   \wedge  \varphi_{{\mathrm{1}}}$, $\models   \sem{  \Gamma''  \uplus  \Gamma  }   \wedge  \varphi_{{\mathrm{1}}}  \implies  \varphi_{{\mathrm{2}}}$ must
    also be valid.
  \item It suffices to show that $ \Gamma  \uplus  \Gamma''   \vdash  \ottsym{(}   \Gamma  \uplus  \Gamma''   \ottsym{)}  \ottsym{(}  \mathit{x}  \ottsym{)}  \leq  \ottsym{(}   \Gamma'  \uplus  \Gamma''   \ottsym{)}  \ottsym{(}  \mathit{x}  \ottsym{)}$ for
    any arbitrary $ \mathit{x}  \in   \DOM(  \Gamma'  \uplus  \Gamma''  )  $. If $ \mathit{x}  \in \DOM( \Gamma' ) $ then we must have
    $\Gamma  \vdash  \Gamma  \ottsym{(}  \mathit{x}  \ottsym{)}  \leq  \Gamma'  \ottsym{(}  \mathit{x}  \ottsym{)}$ by inversion on $\Gamma  \leq  \Gamma'$,
    whereby $ \Gamma  \uplus  \Gamma''   \vdash  \ottsym{(}   \Gamma  \uplus  \Gamma''   \ottsym{)}  \ottsym{(}  \mathit{x}  \ottsym{)}  \ottsym{=}  \Gamma  \ottsym{(}  \mathit{x}  \ottsym{)}  \leq   \Gamma'  \uplus  \Gamma''   \ottsym{(}  \mathit{x}  \ottsym{)}  \ottsym{=}  \Gamma'  \ottsym{(}  \mathit{x}  \ottsym{)}$ from part 1.

    If $ \mathit{x}  \not\in \DOM( \Gamma' ) $, then we must show $\ottsym{(}   \Gamma  \uplus  \Gamma''   \ottsym{)}  \vdash  \Gamma''  \ottsym{(}  \mathit{x}  \ottsym{)}  \leq  \Gamma''  \ottsym{(}  \mathit{x}  \ottsym{)}$, which trivially holds.
  \item By induction on the typing derivation of
    $ \Theta   \mid   \mathcal{L}   \mid   \Gamma   \vdash   \ottnt{e}  :  \tau   \produces   \Gamma' $. We assume that the variables bound in
    any let expressions that appear within $\ottnt{e}$ are not in the domain
    of $\Gamma''$; this requirement can be easily enforced with
    consistent renaming and is preserved during evaluation. The only interesting cases are
    \rn{T-Sub} and \rn{T-Assert} and the let bindings;
    the other cases follow from the induction hypothesis.
    
    We now prove the relevant cases.

    \begin{rneqncase}{T-Let}{
         \Theta   \mid   \mathcal{L}   \mid   \Gamma   \vdash    \LET  \mathit{x}  =  \mathit{y}  \IN  \ottnt{e}   :  \tau   \produces   \Gamma'  \\
         \Theta   \mid   \mathcal{L}   \mid   \Gamma  \ottsym{[}  \mathit{y}  \hookleftarrow   \tau_{{\mathrm{1}}}  \wedge_{ \mathit{y} }   \mathit{y}  =_{ \tau_{{\mathrm{1}}} }  \mathit{x}    \ottsym{]}  \ottsym{,}  \mathit{x}  \ottsym{:}   \tau_{{\mathrm{2}}}  \wedge_{ \mathit{x} }   \mathit{x}  =_{ \tau_{{\mathrm{2}}} }  \mathit{y}     \vdash   \ottnt{e}  :  \tau   \produces   \Gamma'  \\
        \Gamma  \ottsym{(}  \mathit{y}  \ottsym{)}  \ottsym{=}  \tau_{{\mathrm{1}}}  \ottsym{+}  \tau_{{\mathrm{2}}} \andalso  \mathit{x}  \not\in   \DOM( \Gamma' )  
      }
      Let $ \Gamma'''  \ottsym{=}  \Gamma''  \uplus  \Gamma   \ottsym{[}  \mathit{y}  \hookleftarrow   \tau_{{\mathrm{1}}}  \wedge_{ \mathit{y} }   \mathit{y}  =_{ \tau_{{\mathrm{1}}} }  \mathit{x}    \ottsym{]}  \ottsym{,}  \mathit{x}  \ottsym{:}   \tau_{{\mathrm{2}}}  \wedge_{ \mathit{x} }   \mathit{x}  =_{ \tau_{{\mathrm{2}}} }  \mathit{y}  $.
      To use the inductive hypothesis, we must show that
      $ \mathcal{L}   \vdash _{\wf}  \Gamma''' $ and $ \mathcal{L}   \vdash _{\wf}   \Gamma'  \uplus  \Gamma''  $. The latter holds by assumption.
      To show the former, it suffices
      to show $ \mathcal{L}   \mid   \Gamma'''   \vdash _{\wf}   \tau_{{\mathrm{1}}}  \wedge_{ \mathit{y} }   \mathit{y}  =_{ \tau_{{\mathrm{1}}} }  \mathit{x}   $ and $ \mathcal{L}   \mid   \Gamma'''   \vdash _{\wf}   \tau_{{\mathrm{2}}}  \wedge_{ \mathit{x} }   \mathit{x}  =_{ \tau_{{\mathrm{2}}} }  \mathit{y}   $.
      From the assumed well-formedness $ \mathcal{L}   \vdash _{\wf}   \Gamma  \uplus  \Gamma''  $, we must have $ \mathcal{L}   \mid    \Gamma  \uplus  \Gamma''    \vdash _{\wf}  \tau_{{\mathrm{1}}}  \ottsym{+}  \tau_{{\mathrm{2}}} $,
      and in particular $ \mathcal{L}   \mid    \Gamma  \uplus  \Gamma''    \vdash _{\wf}  \tau_{{\mathrm{1}}} $ and $ \mathcal{L}   \mid    \Gamma  \uplus  \Gamma''    \vdash _{\wf}  \tau_{{\mathrm{2}}} $. From this
      we conclude both conditions hold. To show the well-typing of the overall let expression,
      we must show $ \mathit{x}  \not\in   \DOM(  \Gamma'  \uplus  \Gamma''  )  $, which follows from our assumption and $ \mathit{x}  \not\in   \DOM( \Gamma' )  $.
      Finally, we must also show that $ \mathcal{L}   \mid    \Gamma'  \uplus  \Gamma''    \vdash _{\wf}  \tau $. From $ \mathcal{L}   \mid   \Gamma'   \vdash _{\wf}  \tau $
      and the fact that$\forall \,  \mathit{x}  \in   \DOM( \Gamma' )    \ottsym{.}  \Gamma'  \ottsym{(}  \mathit{x}  \ottsym{)}  \ottsym{=}   \set{  \nu  \COL \TINT \mid \_ } $ iff
      $\ottsym{(}   \Gamma'  \uplus  \Gamma''   \ottsym{)}  \ottsym{(}  \mathit{x}  \ottsym{)}  \ottsym{=}   \set{  \nu  \COL \TINT \mid \_ } $, we must have $ \mathcal{L}   \mid    \Gamma'  \uplus  \Gamma''    \vdash _{\wf}  \tau $. 
    \end{rneqncase}
    \begin{namedcase}{\casefont{Cases \rn{T-LetInt}, \rn{T-Mkref}, \rn{T-Mkref}, \rn{T-Deref}, \rn{T-Call}: }}
      Similar to the reasoning in \rn{T-Let}.
    \end{namedcase}

    \begin{rneqncase}{T-Sub}{
        \Gamma  \leq  \Gamma_{{\mathrm{1}}} &  \Theta   \mid   \mathcal{L}   \mid   \Gamma_{{\mathrm{1}}}   \vdash   \ottnt{e}  :  \tau_{{\mathrm{2}}}   \produces   \Gamma_{{\mathrm{2}}}  \\
        \Gamma_{{\mathrm{2}}}  \ottsym{,}  \tau_{{\mathrm{2}}}  \leq  \Gamma'  \ottsym{,}  \tau \\
      }
      From the rules for subtyping, we must have
      $  \DOM( \Gamma_{{\mathrm{1}}} )   \subseteq   \DOM( \Gamma )  $ and $  \DOM( \Gamma' )   \subseteq   \DOM( \Gamma_{{\mathrm{2}}} )  $. A simple inductive argument
      gives that $  \DOM( \Gamma_{{\mathrm{2}}} )   \subseteq   \DOM( \Gamma_{{\mathrm{1}}} )  $, therefore we have $  \DOM( \Gamma' )   \subseteq   \DOM( \Gamma_{{\mathrm{1}}} )  $.
      Let $\mathcal{LV}$ be the set of free variables in the refinements of $\Gamma''$
      that are not in the domain of $\Gamma''$. From the assumed well-formedness of
      $ \mathcal{L}   \vdash _{\wf}   \Gamma'  \uplus  \Gamma''  $, we must have that
      $\forall x \in \mathcal{LV}. \mathit{x}  \in   \DOM( \Gamma' )    \wedge  \Gamma'  \ottsym{(}  \mathit{x}  \ottsym{)}  \ottsym{=}   \set{  \nu  \COL \TINT \mid \_ } $. Thus,
      $\mathcal{LV} \subseteq \Gamma_{{\mathrm{1}}}$ and $\mathcal{LV} \subseteq \Gamma_{{\mathrm{2}}}$. Further, by definition,
      for any $\Gamma_{\ottmv{p}}  \leq  \Gamma_{\ottmv{q}}$, if $\Gamma_{\ottmv{q}}  \ottsym{(}  \mathit{x}  \ottsym{)}  \ottsym{=}   \set{  \nu  \COL \TINT \mid \_ } $ then $\Gamma_{\ottmv{p}}  \ottsym{(}  \mathit{x}  \ottsym{)}  \ottsym{=}   \set{  \nu  \COL \TINT \mid \_ } $,
      i.e. subtyping preserves simple types.
      We conclude that $ \mathcal{L}   \vdash _{\wf}   \Gamma_{{\mathrm{1}}}  \uplus  \Gamma''  $ and
      $ \mathcal{L}   \vdash _{\wf}   \Gamma_{{\mathrm{2}}}  \uplus  \Gamma''  $, whereby the inductive hypothesis gives
      $ \Theta   \mid   \mathcal{L}   \mid    \Gamma_{{\mathrm{1}}}  \uplus  \Gamma''    \vdash   \ottnt{e}  :  \tau_{{\mathrm{2}}}   \produces    \Gamma_{{\mathrm{2}}}  \uplus  \Gamma''  $. To prove the overall result, we must
      show that $ \Gamma  \uplus  \Gamma''   \leq   \Gamma_{{\mathrm{1}}}  \uplus  \Gamma'' $ and $ \Gamma_{{\mathrm{2}}}  \uplus  \Gamma''   \ottsym{,}  \tau_{{\mathrm{2}}}  \leq   \Gamma'  \uplus  \Gamma''   \ottsym{,}  \tau$
      which follow from parts 1 and 2 above.
      That $ \mathcal{L}   \mid    \Gamma_{{\mathrm{2}}}  \uplus  \Gamma''    \vdash _{\wf}  \tau_{{\mathrm{2}}} $ follows by reasoning to the case for \rn{T-Let} above.
    \end{rneqncase}
    \begin{rncase}{T-Assert}
      We must show that $\models   \sem{  \Gamma''  \uplus  \Gamma  }   \implies  \varphi$ which is equivalent to
      $\models   \sem{ \Gamma'' }   \wedge   \sem{ \Gamma }   \implies  \varphi$. As the source term was well typed,
      $\models   \sem{ \Gamma }   \implies  \varphi$ is valid, we must then have $\models   \sem{ \Gamma'' }   \wedge   \sem{ \Gamma }   \implies  \varphi$
      whereby the inductive hypothesis gives the required result.
    \end{rncase}
  \end{enumerate}
\end{proof}

\section{Proof of \Cref{lem:preservation}}
\label{sec:preservation-proof}

We first prove two additional lemmas. 
\Cref{lem:stack-well-typed,lem:callfunc} give key facts used in the return and
call cases respectively; we have separated them into separate lemmas for clarity.

\begin{lemma} %
  \label{lem:stack-well-typed}
  For any $\Gamma_{{\mathrm{0}}}$ such that $ \Theta   \mid   \ell  \ottsym{:}  \oldvec{\ell}   \mid   \Gamma_{{\mathrm{0}}}   \vdash   \mathit{x}  :  \tau_{{\mathrm{1}}}   \produces   \Gamma_{{\mathrm{1}}} $ and $\Theta  \mid  \HOLE  \ottsym{:}  \tau_{{\mathrm{1}}}  \produces  \Gamma_{{\mathrm{1}}}  \mid  \oldvec{\ell}  \vdash_{\mathit{ectx} }   \ottnt{E} [\LET  \mathit{y}  =   \HOLE^ \ell   \IN  \ottnt{e}  ]   \ottsym{:}  \tau_{{\mathrm{2}}}  \produces  \Gamma_{{\mathrm{2}}}$ then
  $ \Theta   \mid   \oldvec{\ell}   \mid   \Gamma_{{\mathrm{0}}}   \vdash    \ottnt{E} [\LET  \mathit{y}  =   \HOLE^ \ell   \IN  \ottnt{e}  ]   \ottsym{[}  \mathit{x}  \ottsym{]}  :  \tau_{{\mathrm{2}}}   \produces   \Gamma_{{\mathrm{2}}} $.
\end{lemma}
\begin{proof}
  It suffices to show that $ \Theta   \mid   \oldvec{\ell}   \mid   \Gamma_{{\mathrm{0}}}   \vdash    \LET  \mathit{y}  =  \mathit{x}  \IN  \ottnt{e}   :  \tau'_{{\mathrm{1}}}   \produces   \Gamma'_{{\mathrm{1}}} $
  and $\Theta  \mid  \HOLE  \ottsym{:}  \tau'_{{\mathrm{1}}}  \produces  \Gamma'_{{\mathrm{1}}}  \mid  \oldvec{\ell}  \vdash_{\mathit{ectx} }  \ottnt{E}  \ottsym{:}  \tau_{{\mathrm{2}}}  \produces  \Gamma_{{\mathrm{2}}}$ for some $\tau'_{{\mathrm{1}}}$ and $\Gamma'_{{\mathrm{1}}}$
  whereby the result will hold from \Cref{lem:ectxt-sub-well-typed}.

  By inversion on
  $\Theta  \mid  \HOLE  \ottsym{:}  \tau_{{\mathrm{1}}}  \produces  \Gamma_{{\mathrm{1}}}  \mid  \oldvec{\ell}  \vdash_{\mathit{ectx} }   \ottnt{E} [\LET  \mathit{y}  =   \HOLE^ \ell   \IN  \ottnt{e}  ]   \ottsym{:}  \tau_{{\mathrm{2}}}  \produces  \Gamma_{{\mathrm{2}}}$ we have
  \begin{align}
    &  \Theta   \mid   \oldvec{\ell}   \mid   \Gamma_{{\mathrm{1}}}  \ottsym{,}  \mathit{y}  \ottsym{:}  \tau_{{\mathrm{1}}}   \vdash   \ottnt{e}  :  \tau''_{{\mathrm{1}}}   \produces   \Gamma''_{{\mathrm{1}}}  \label{eqn:let-body-ty} \\
    & \Theta  \mid  \HOLE  \ottsym{:}  \tau''_{{\mathrm{1}}}  \produces  \Gamma''_{{\mathrm{1}}}  \mid  \oldvec{\ell}  \vdash_{\mathit{ectx} }  \ottnt{E}  \ottsym{:}  \tau_{{\mathrm{2}}}  \produces  \Gamma_{{\mathrm{2}}} \label{eqn:ctxt-typed} \\
    &  \mathit{y}  \not\in   \DOM( \Gamma''_{{\mathrm{1}}} )   \label{eqn:y-not-free-var}
  \end{align}
  We take $\Gamma'_{{\mathrm{1}}}  \ottsym{=}  \Gamma''_{{\mathrm{1}}}$, $\tau'_{{\mathrm{1}}}  \ottsym{=}  \tau''_{{\mathrm{1}}}$, and then \Cref{eqn:ctxt-typed}
  gives the necessary typing for $\ottnt{E}$.

  It remains to to show that
  \[
     \Theta   \mid   \oldvec{\ell}   \mid   \Gamma_{{\mathrm{0}}}   \vdash    \LET  \mathit{y}  =  \mathit{x}  \IN  \ottnt{e}   :  \tau'_{{\mathrm{1}}}   \produces   \Gamma'_{{\mathrm{1}}} 
  \]
  (That $ \oldvec{\ell}   \vdash _{\wf}  \tau'_{{\mathrm{1}}}   \produces   \Gamma'_{{\mathrm{1}}} $ follows from \Cref{eqn:let-body-ty})
  
  By \Cref{lem:inversion} and \Cref{lem:var-trace-transport},
  from $ \Theta   \mid   \ell  \ottsym{:}  \oldvec{\ell}   \mid   \Gamma_{{\mathrm{0}}}   \vdash   \mathit{x}  :  \tau_{{\mathrm{1}}}   \produces   \Gamma_{{\mathrm{1}}} $ we conclude there
  exists some $\Gamma_{\ottmv{p}}$, $\tau_{\ottmv{p}}$, and $\Gamma'_{\ottmv{p}}$ such:
  \begin{align}
    & \Gamma_{{\mathrm{0}}}  \leq  \Gamma_{\ottmv{p}} \label{eqn:rin-sub} \\
    & \Gamma'_{\ottmv{p}}  \ottsym{,}  \tau_{\ottmv{p}}  \leq  \Gamma_{{\mathrm{1}}}  \ottsym{,}  \tau_{{\mathrm{1}}} \label{eqn:rout-sub} \\
    & \Gamma'_{\ottmv{p}}  \ottsym{=}  \Gamma_{\ottmv{p}}  \ottsym{[}  \mathit{x}  \hookleftarrow  \tau'_{\ottmv{p}}  \ottsym{]} \label{eqn:gp-up-def} \\
    & \Gamma_{\ottmv{p}}  \ottsym{(}  \mathit{x}  \ottsym{)}  \ottsym{=}  \tau_{\ottmv{p}}  \ottsym{+}  \tau'_{\ottmv{p}} \label{eqn:gp-x-def} \\
    &  \oldvec{\ell}   \vdash _{\wf}  \Gamma_{\ottmv{p}}  \label{eqn:renv-sub-in-wf}
  \end{align} 

  We first apply \rn{T-Sub} with \Cref{eqn:rin-sub,eqn:renv-sub-in-wf}, so it remains to show
  \[
     \Theta   \mid   \oldvec{\ell}   \mid   \Gamma_{\ottmv{p}}  \ottsym{[}  \mathit{x}  \ottsym{:}  \tau_{\ottmv{p}}  \ottsym{+}  \tau'_{\ottmv{p}}  \ottsym{]}   \vdash    \LET  \mathit{y}  =  \mathit{x}  \IN  \ottnt{e}   :  \tau'_{{\mathrm{1}}}   \produces   \Gamma'_{{\mathrm{1}}} 
  \]
  which, by \rn{T-Let} holds if we show that:
  \[
     \Theta   \mid   \oldvec{\ell}   \mid   \Gamma_{\ottmv{p}}  \ottsym{[}  \mathit{x}  \hookleftarrow   \tau'_{\ottmv{p}}  \wedge_{ \mathit{x} }   \mathit{x}  =_{ \tau'_{\ottmv{p}} }  \mathit{y}    \ottsym{]}  \ottsym{,}  \mathit{y}  \ottsym{:}   \tau_{\ottmv{p}}  \wedge_{ \mathit{y} }   \mathit{y}  =_{ \tau_{\ottmv{p}} }  \mathit{x}     \vdash   \ottnt{e}  :  \tau'_{{\mathrm{1}}}   \produces   \Gamma'_{{\mathrm{1}}} 
  \]
  ($ \mathit{y}  \not\in   \DOM( \Gamma'_{{\mathrm{1}}} )  $ follows from \Cref{eqn:y-not-free-var}, and the well-formedness of
  $\Gamma_{\ottmv{p}}  \ottsym{[}  \mathit{x}  \hookleftarrow   \tau'_{\ottmv{p}}  \wedge_{ \mathit{x} }   \mathit{x}  =_{ \tau'_{\ottmv{p}} }  \mathit{y}    \ottsym{]}  \ottsym{,}  \mathit{y}  \ottsym{:}   \tau_{\ottmv{p}}  \wedge_{ \mathit{y} }   \mathit{y}  =_{ \tau_{\ottmv{p}} }  \mathit{x}  $ follows from the well-formedness of
  $\Gamma_{\ottmv{p}}$, $\tau_{\ottmv{p}}$ and $\tau'_{\ottmv{p}}$ and that $\mathit{x}$ and $\mathit{y}$ appear in the refinements
  iff they are mapped to integer types in the new type environment.)
    
  We can use \rn{T-Sub} to weaken the type environment to:
  \[
     \Theta   \mid   \oldvec{\ell}   \mid   \Gamma_{\ottmv{p}}  \ottsym{[}  \mathit{x}  \hookleftarrow  \tau'_{\ottmv{p}}  \ottsym{]}  \ottsym{,}  \mathit{y}  \ottsym{:}  \tau_{\ottmv{p}}   \vdash   \ottnt{e}  :  \tau'_{{\mathrm{1}}}   \produces   \Gamma'_{{\mathrm{1}}} 
  \]
  From \Cref{eqn:rout-sub} above, we have that $\Gamma_{\ottmv{p}}  \ottsym{[}  \mathit{x}  \hookleftarrow  \tau'_{\ottmv{p}}  \ottsym{]}  \ottsym{,}  \mathit{y}  \ottsym{:}  \tau_{\ottmv{p}}  \leq  \Gamma_{{\mathrm{1}}}  \ottsym{,}  \mathit{y}  \ottsym{:}  \tau_{{\mathrm{1}}}$,
  whereby one final application of \rn{T-Sub} allows us to use to
  \Cref{eqn:let-body-ty} above.
\end{proof}

\begin{lemma} %
  \label{lem:callfunc}
  Let $\ottnt{E}  \ottsym{[}   \LET  \mathit{x}  =   \mathit{f} ^ \ell (  \mathit{y_{{\mathrm{1}}}} ,\ldots, \mathit{y_{\ottmv{n}}}  )   \IN  \ottnt{e'}   \ottsym{]}$ be a term such that:
  
  \begin{bcpcasearray}
     \Theta   \mid   \oldvec{\ell}   \mid   \Gamma_{{\mathrm{0}}}   \vdash    \LET  \mathit{x}  =   \mathit{f} ^ \ell (  \mathit{y_{{\mathrm{1}}}} ,\ldots, \mathit{y_{\ottmv{n}}}  )   \IN  \ottnt{e'}   :  \tau_{{\mathrm{1}}}   \produces   \Gamma_{{\mathrm{1}}}  &  \sigma_{\alpha}  \ottsym{=}  \ottsym{[}  \ell  \ottsym{:}  \oldvec{\ell}  \ottsym{/}  \lambda  \ottsym{]} \\
    \Theta  \mid  \HOLE  \ottsym{:}  \tau_{{\mathrm{1}}}  \produces  \Gamma_{{\mathrm{1}}}  \mid  \oldvec{\ell}  \vdash_{\mathit{ectx} }  \ottnt{E}  \ottsym{:}  \tau_{{\mathrm{2}}}  \produces  \Gamma_{{\mathrm{2}}} & \sigma_{x}  \ottsym{=}    [  \mathit{y_{{\mathrm{1}}}}  /  \mathit{x_{{\mathrm{1}}}}  ]  \cdots  [  \mathit{y_{\ottmv{n}}}  /  \mathit{x_{\ottmv{n}}}  ]   \\
     \mathit{f}  \mapsto  \ottsym{(}  \mathit{x_{{\mathrm{1}}}}  \ottsym{,} \, .. \, \ottsym{,}  \mathit{x_{\ottmv{n}}}  \ottsym{)}  \ottnt{e}  \in  \ottnt{D}  & \Theta  \vdash  \mathit{f}  \mapsto  \ottsym{(}  \mathit{x_{{\mathrm{1}}}}  \ottsym{,} \, .. \, \ottsym{,}  \mathit{x_{\ottmv{n}}}  \ottsym{)}  \ottnt{e} \\
     \vdash _{\wf}  \Theta 
  \end{bcpcasearray}

  where $\Theta  \ottsym{(}  \mathit{f}  \ottsym{)}  \ottsym{=}   \forall  \lambda .\tuple{ \mathit{x_{{\mathrm{1}}}} \COL \tau_{{\mathrm{1}}} ,\dots, \mathit{x_{\ottmv{n}}} \COL \tau_{\ottmv{n}} }\ra\tuple{ \mathit{x_{{\mathrm{1}}}} \COL \tau'_{{\mathrm{1}}} ,\dots, \mathit{x_{\ottmv{n}}} \COL \tau'_{\ottmv{n}}  \mid  \tau_{\ottmv{q}} } $.

  Then there exist some $\tau_{{\mathrm{3}}}$ and $\Gamma_{{\mathrm{3}}}$:
  \begin{align*}
    &  \Theta   \mid   \ell  \ottsym{:}  \oldvec{\ell}   \mid   \Gamma_{{\mathrm{0}}}   \vdash    \sigma_{x}   \ottnt{e}   :  \tau_{{\mathrm{3}}}   \produces   \Gamma_{{\mathrm{3}}}  \\
    & \Theta  \mid  \HOLE  \ottsym{:}  \tau_{{\mathrm{3}}}  \produces  \Gamma_{{\mathrm{3}}}  \mid  \oldvec{\ell}  \vdash_{\mathit{ectx} }   \ottnt{E} [\LET  \mathit{x}  =   \HOLE^ \ell   \IN  \ottnt{e'}  ]   \ottsym{:}  \tau_{{\mathrm{2}}}  \produces  \Gamma_{{\mathrm{2}}}
  \end{align*}
\end{lemma}
\begin{proof}
  From \Cref{lem:inversion} on
  $ \Theta   \mid   \oldvec{\ell}   \mid   \Gamma_{{\mathrm{0}}}   \vdash    \LET  \mathit{x}  =   \mathit{f} ^ \ell (  \mathit{y_{{\mathrm{1}}}} ,\ldots, \mathit{y_{\ottmv{n}}}  )   \IN  \ottnt{e'}   :  \tau_{{\mathrm{1}}}   \produces   \Gamma_{{\mathrm{1}}} $ we have, for some
  $\Gamma_{\ottmv{p}}, \tau_{\ottmv{p}}, \Gamma'_{\ottmv{p}}$, that:
  \begin{align}
    & \Gamma_{{\mathrm{0}}}  \leq  \Gamma_{\ottmv{p}} \label{eqn:in-sub} \\
    & \Gamma'_{\ottmv{p}}  \ottsym{,}  \tau_{\ottmv{p}}  \leq  \Gamma_{{\mathrm{1}}}  \ottsym{,}  \tau_{{\mathrm{1}}} \label{eqn:out-sub} \\
    & \Gamma_{\ottmv{p}}  \ottsym{(}  \mathit{y_{\ottmv{i}}}  \ottsym{)}  \ottsym{=}  \sigma_{\alpha} \, \sigma_{x} \, \tau_{\ottmv{i}} \label{eqn:arg-typed} \\
    &  \oldvec{\ell}   \vdash _{\wf}  \Gamma_{\ottmv{p}}  \label{eqn:sub-in-env-wf} \\
    &  \Theta   \mid   \oldvec{\ell}   \mid   \Gamma_{\ottmv{p}}  \ottsym{[}  \mathit{y_{\ottmv{i}}}  \hookleftarrow  \sigma_{\alpha} \, \sigma_{x} \, \tau'_{\ottmv{i}}  \ottsym{]}  \ottsym{,}  \mathit{x}  \ottsym{:}  \sigma_{\alpha} \, \sigma_{x} \, \tau_{\ottmv{q}}   \vdash   \ottnt{e'}  :  \tau_{\ottmv{p}}   \produces   \Gamma'_{\ottmv{p}}  \label{eqn:let-body-well-typed} \\
    &  \mathit{x}  \not\in   \DOM( \Gamma'_{\ottmv{p}} )   \label{eqn:cbind-erased}
  \end{align}
  
  To prove the first part, from the well-typing of the function body, we have
  $ \Theta   \mid   \lambda   \mid    \mathit{x_{{\mathrm{1}}}} \COL \tau_{{\mathrm{1}}} ,\ldots, \mathit{x_{\ottmv{n}}} \COL \tau_{\ottmv{n}}    \vdash   \ottnt{e}  :  \tau_{\ottmv{q}}   \produces    \mathit{x_{{\mathrm{1}}}} \COL \tau'_{{\mathrm{1}}} ,\ldots, \mathit{x_{\ottmv{n}}} \COL \tau'_{\ottmv{n}}  $.
  From our assumption that all variable names are distinct,
  by $n$ applications of the substitution lemma (\Cref{lem:substitution}) we have:
  $ \Theta   \mid   \lambda   \mid    \mathit{y_{{\mathrm{1}}}} \COL \sigma_{x} \, \tau_{{\mathrm{1}}} ,\ldots, \mathit{y_{\ottmv{n}}} \COL \sigma_{x} \, \tau_{\ottmv{n}}    \vdash    \sigma_{x}   \ottnt{e}   :  \sigma_{x} \, \tau_{\ottmv{q}}   \produces    \mathit{y_{{\mathrm{1}}}} \COL \sigma_{x} \, \tau'_{{\mathrm{1}}} ,\ldots, \mathit{y_{\ottmv{n}}} \COL \sigma_{x} \, \tau'_{\ottmv{n}}  $.
  By \subref{lem:ctxt-substitution}{itm:ctxt-subst-well-typed} we then have
  $ \Theta   \mid   \ell  \ottsym{:}  \oldvec{\ell}   \mid    \mathit{y_{{\mathrm{1}}}} \COL \sigma_{\alpha} \, \sigma_{x} \, \tau_{{\mathrm{1}}} ,\ldots, \mathit{y_{\ottmv{n}}} \COL \sigma_{\alpha} \, \sigma_{x} \, \tau_{\ottmv{n}}    \vdash    \sigma_{x}   \ottnt{e}   :  \sigma_{\alpha} \, \sigma_{x} \, \tau_{\ottmv{q}}   \produces    \mathit{y_{{\mathrm{1}}}} \COL \sigma_{\alpha} \, \sigma_{x} \, \tau'_{{\mathrm{1}}} ,\ldots, \mathit{y_{\ottmv{n}}} \COL \sigma_{\alpha} \, \sigma_{x} \, \tau'_{\ottmv{n}}  $.
  We take $\tau_{{\mathrm{3}}}  \ottsym{=}  \sigma_{\alpha} \, \sigma_{x} \, \tau_{\ottmv{q}}$ and $\Gamma_{{\mathrm{3}}}  \ottsym{=}  \Gamma_{\ottmv{p}}  \ottsym{[}  \mathit{y_{\ottmv{i}}}  \hookleftarrow  \sigma_{\alpha} \, \sigma_{x} \, \tau'_{\ottmv{i}}  \ottsym{]}$.
  
  By the well-formedness of function types and well-formedness of $\Gamma_{\ottmv{p}}$,
  we must have that $ \ell  \ottsym{:}  \oldvec{\ell}   \vdash _{\wf}  \Gamma_{{\mathrm{3}}} $.
  Then by \Cref{eqn:arg-typed,eqn:sub-in-env-wf,lem:tyenv-weaken} we have
  $ \Theta   \mid   \ell  \ottsym{:}  \oldvec{\ell}   \mid   \Gamma_{\ottmv{p}}   \vdash    \sigma_{x}   \ottnt{e}   :  \tau_{{\mathrm{3}}}   \produces   \Gamma_{{\mathrm{3}}} $, whereby \Cref{eqn:in-sub}
  and an application of \rn{T-Sub}
  gives $ \Theta   \mid   \ell  \ottsym{:}  \oldvec{\ell}   \mid   \Gamma_{{\mathrm{0}}}   \vdash    \sigma_{x}   \ottnt{e}   :  \tau_{{\mathrm{3}}}   \produces   \Gamma_{{\mathrm{3}}} $, i.e., the first result.

  To prove the second part, from the typing rule for \rn{TE-Stack} we must show:
  \begin{align}
    & \Theta  \mid  \HOLE  \ottsym{:}  \tau_{{\mathrm{1}}}  \produces  \Gamma_{{\mathrm{1}}}  \mid  \oldvec{\ell}  \vdash_{\mathit{ectx} }  \ottnt{E}  \ottsym{:}  \tau_{{\mathrm{2}}}  \produces  \Gamma_{{\mathrm{2}}} \label{eqn:context-well-typed} \\
    &  \mathit{x}  \not\in   \DOM( \Gamma_{{\mathrm{1}}} )   \label{eqn:x-not-free} \\
    &  \Theta   \mid   \oldvec{\ell}   \mid   \Gamma_{{\mathrm{3}}}  \ottsym{,}  \mathit{x}  \ottsym{:}  \tau_{{\mathrm{3}}}   \vdash   \ottnt{e'}  :  \tau_{{\mathrm{1}}}   \produces   \Gamma_{{\mathrm{1}}}  \label{eqn:let-body-sub-typed}
  \end{align}
  \Cref{eqn:context-well-typed} holds by assumption, and \Cref{eqn:x-not-free} follows from
  \Cref{eqn:cbind-erased} and that $\Gamma'_{\ottmv{p}}  \leq  \Gamma_{{\mathrm{1}}}$ implies $  \DOM( \Gamma_{{\mathrm{1}}} )   \subseteq   \DOM( \Gamma'_{\ottmv{p}} )  $.
  $ \oldvec{\ell}   \vdash _{\wf}  \tau_{{\mathrm{1}}}   \produces   \Gamma_{{\mathrm{1}}} $ follows from the well-typing of the function call term,
  and $ \oldvec{\ell}   \vdash _{\wf}  \Gamma_{{\mathrm{3}}}  \ottsym{,}  \mathit{x}  \ottsym{:}  \tau_{{\mathrm{3}}} $ follows from \Cref{eqn:let-body-well-typed}.

  From \Cref{eqn:out-sub,eqn:let-body-well-typed} we then have
  \Cref{eqn:let-body-sub-typed} via an application of \rn{T-Sub}.
\end{proof}

\begin{proof}[Preservation; \Cref{lem:preservation}]
  The proof is organized by cases analysis on the transition rule used of $\ottnt{e}$, and showing that the output configuration is well typed by
  $ \vdash_{\mathit{conf} }^D $. We must therefore find a $\Gamma''$ that is consistent with $\ottnt{H'}$ and $\ottnt{R'}$ and also satisfies the other conditions imposed by
  the definition of $ \vdash_{\mathit{conf} }^D $. Here $\Gamma'',\ottnt{H'}, \ottnt{R'}$ represent the type environment, heap and register after the transition respectively.
  We identify the heap and register file before transition with $\ottnt{H}$ and $\ottnt{R}$ respectively.
  In order to show that the ownership invariant is preserved, we need to prove that $\forall \,  \ottmv{a}  \in \DOM( \ottnt{H'} )   \ottsym{.}  \ottkw{Own} \, \ottsym{(}  \ottnt{H'}  \ottsym{,}  \ottnt{R}  \ottsym{,}  \Gamma''  \ottsym{)}  \ottsym{(}  \ottmv{a}  \ottsym{)}  \le  \ottsym{1}$.
  In many cases, we will show that $\ottkw{Own} \, \ottsym{(}  \ottnt{H}  \ottsym{,}  \ottnt{R}  \ottsym{,}  \Gamma  \ottsym{)}  \ottsym{=}  \ottkw{Own} \, \ottsym{(}  \ottnt{H'}  \ottsym{,}  \ottnt{R'}  \ottsym{,}  \Gamma''  \ottsym{)}$, whereby
  from the assumption that $\ottkw{Cons} \, \ottsym{(}  \ottnt{H}  \ottsym{,}  \ottnt{R}  \ottsym{,}  \Gamma  \ottsym{)}$ as implied by $ \vdash_{\mathit{conf} }^D $
  we have $\forall \,  \ottmv{a}  \in \DOM( \ottnt{H} )   \ottsym{.}  \ottkw{Own} \, \ottsym{(}  \ottnt{H}  \ottsym{,}  \ottnt{R}  \ottsym{,}  \Gamma  \ottsym{)}  \ottsym{(}  \ottmv{a}  \ottsym{)}  \le  \ottsym{1}$, giving the desired result.
  \begin{rneqncase}{R-Var}{
     \vdash_{\mathit{conf} }^D   \tuple{ \ottnt{H} ,  \ottnt{R} ,  F_{{\ottmv{n}-1}}  \ottsym{:}  \oldvec{F} ,  \mathit{x} }  ,   \tuple{ \ottnt{H} ,  \ottnt{R} ,  F_{{\ottmv{n}-1}}  \ottsym{:}  \oldvec{F} ,  \mathit{x} }     \longrightarrow _{ \ottnt{D} }     \tuple{ \ottnt{H} ,  \ottnt{R} ,  \oldvec{F} ,  F_{{\ottmv{n}-1}}  \ottsym{[}  \mathit{x}  \ottsym{]} }   \\
  }
  By inversion on configuration typing $ \vdash_{\mathit{conf} }^D   \tuple{ \ottnt{H} ,  \ottnt{R} ,  F_{{\ottmv{n}-1}}  \ottsym{:}  \oldvec{F} ,  \mathit{x} }  $, we have:
    \begin{align*}
      &  \Theta   \mid   \oldvec{\ell}   \mid   \Gamma   \vdash   \mathit{x}  :  \tau_{\ottmv{n}}   \produces   \Gamma_{\ottmv{n}}  \\
      & \forall i\in\set{1..n}.\Theta  \mid  \HOLE  \ottsym{:}  \tau_{\ottmv{i}}  \produces  \Gamma_{\ottmv{i}}  \mid  \oldvec{\ell}_{{\ottmv{i}-1}}  \vdash_{\mathit{ectx} }  F_{{\ottmv{i}-1}}  \ottsym{:}  \tau_{{\ottmv{i}-1}}  \produces  \Gamma_{{\ottmv{i}-1}}
    \end{align*}
    
    Using \Cref{lem:stack-well-typed}, we can conclude that $ \Theta   \mid   \oldvec{\ell}_{{\ottmv{n}-1}}   \mid   \Gamma   \vdash   F_{{\ottmv{n}-1}}  \ottsym{[}  \mathit{x}  \ottsym{]}  :  \tau_{{\ottmv{n}-1}}   \produces   \Gamma_{{\ottmv{n}-1}} $. We therefore take $\Gamma''  \ottsym{=}  \Gamma$.

    It remains to show that $\ottkw{Cons} \, \ottsym{(}  \ottnt{H}  \ottsym{,}  \ottnt{R}  \ottsym{,}  \Gamma''  \ottsym{)}$ which follows immediately from
    $\ottkw{Cons} \, \ottsym{(}  \ottnt{H}  \ottsym{,}  \ottnt{R}  \ottsym{,}  \Gamma  \ottsym{)}$.
  \end{rneqncase}

  \begin{rneqncase}{R-Deref}{
       \vdash_{\mathit{conf} }^D   \tuple{ \ottnt{H} ,  \ottnt{R} ,  \oldvec{F} ,  \ottnt{E}  \ottsym{[}   \LET  \mathit{x}  =   *  \mathit{y}   \IN  \ottnt{e}   \ottsym{]} }   \\
        \tuple{ \ottnt{H} ,  \ottnt{R} ,  \oldvec{F} ,  \ottnt{E}  \ottsym{[}   \LET  \mathit{x}  =   *  \mathit{y}   \IN  \ottnt{e}   \ottsym{]} }     \longrightarrow _{ \ottnt{D} }     \tuple{ \ottnt{H} ,  \ottnt{R}  \ottsym{\{}  \mathit{x'}  \mapsto  \ottnt{v}  \ottsym{\}} ,  \oldvec{F} ,  \ottnt{E}  \ottsym{[}    [  \mathit{x'}  /  \mathit{x}  ]    \ottnt{e}   \ottsym{]} }   \\
      \ottnt{H}  \ottsym{(}  \ottmv{a}  \ottsym{)} \, \ottsym{=} \, \ottnt{v} \andalso \ottnt{R}  \ottsym{(}  \mathit{y}  \ottsym{)} \, \ottsym{=} \, \ottmv{a} \andalso \ottnt{R'}  \ottsym{=}  \ottnt{R}  \ottsym{\{}  \mathit{x'}  \mapsto  \ottnt{v}  \ottsym{\}}\\
    }
    By inversion on the configuration typing relationship, we have that:
    \begin{align*}
       \Theta   \mid   \oldvec{\ell}   \mid   \Gamma_{{\mathrm{0}}}   \vdash   \ottnt{E}  \ottsym{[}   \LET  \mathit{x}  =   *  \mathit{y}   \IN  \ottnt{e}   \ottsym{]}  :  \tau_{\ottmv{n}}   \produces   \Gamma_{\ottmv{n}}  && \ottkw{Cons} \, \ottsym{(}  \ottnt{H}  \ottsym{,}  \ottnt{R}  \ottsym{,}  \Gamma_{{\mathrm{0}}}  \ottsym{)}
    \end{align*}
    By \Cref{lem:stack_var}, we have some $\tau, \Gamma'_{{\mathrm{0}}}$ such that:
    \begin{align*}
      \Theta  \mid  \HOLE  \ottsym{:}  \tau  \produces  \Gamma'_{{\mathrm{0}}}  \mid  \oldvec{\ell}  \vdash_{\mathit{ectx} }  \ottnt{E}  \ottsym{:}  \tau_{\ottmv{n}}  \produces  \Gamma_{\ottmv{n}} &&  \Theta   \mid   \oldvec{\ell}   \mid   \Gamma_{{\mathrm{0}}}   \vdash    \LET  \mathit{x}  =   *  \mathit{y}   \IN  \ottnt{e}   :  \tau   \produces   \Gamma'_{{\mathrm{0}}} 
    \end{align*}
    Using \Cref{lem:inversion}, we have some $\Gamma_{\ottmv{p}}$, $\Gamma'_{\ottmv{p}}$ and $\tau_{\ottmv{p}}$ such
    that:
    
    \begin{gather*}
      \Gamma_{{\mathrm{0}}}  \leq  \Gamma_{\ottmv{p}} \quad\quad  \oldvec{\ell}   \vdash _{\wf}  \Gamma_{\ottmv{p}}  \quad\quad \Gamma'_{\ottmv{p}}  \ottsym{,}  \tau_{\ottmv{p}}  \leq  \Gamma'_{{\mathrm{0}}}  \ottsym{,}  \tau \\
      \Gamma_{\ottmv{p}}  \ottsym{(}  \mathit{y}  \ottsym{)}  \ottsym{=}   \ottsym{(}  \tau_{{\mathrm{1}}}  \ottsym{+}  \tau_{{\mathrm{2}}}  \ottsym{)}  \TREF^{ r }  \quad\quad  \mathit{x}  \not\in \DOM( \Gamma'_{\ottmv{p}} )  \\
      \tau'' = \begin{cases}
        \ottsym{(}   \tau_{{\mathrm{1}}}  \wedge_{ \mathit{y} }   \mathit{y}  =_{ \tau_{{\mathrm{1}}} }  \mathit{x}    \ottsym{)} &  r   \ottsym{>}   \ottsym{0}  \\
        \tau_{{\mathrm{1}}} & r  \ottsym{=}  \ottsym{0}
      \end{cases} \\
       \Theta   \mid   \oldvec{\ell}   \mid   \Gamma_{\ottmv{p}}  \ottsym{[}  \mathit{y}  \hookleftarrow   \tau''  \TREF^{ r }   \ottsym{]}  \ottsym{,}  \mathit{x}  \ottsym{:}  \tau_{{\mathrm{2}}}   \vdash   \ottnt{e}  :  \tau_{\ottmv{p}}   \produces   \Gamma'_{\ottmv{p}} 
    \end{gather*}
    
    From \Cref{lem:subtyp-preserves-cons}, we then have $\ottkw{Cons} \, \ottsym{(}  \ottnt{H}  \ottsym{,}  \ottnt{R}  \ottsym{,}  \Gamma_{\ottmv{p}}  \ottsym{)}$.
    We will now show that:
    \begin{align}
      \label{eqn:new-tyenv-cons} & \ottkw{Cons} \, \ottsym{(}  \ottnt{H}  \ottsym{,}  \ottnt{R}  \ottsym{\{}  \mathit{x'}  \mapsto  \ottnt{v}  \ottsym{\}}  \ottsym{,}  \Gamma''  \ottsym{)} \\
                                 &  \Theta   \mid   \oldvec{\ell}   \mid   \Gamma''   \vdash     [  \mathit{x'}  /  \mathit{x}  ]    \ottnt{e}   :   [  \mathit{x'}  /  \mathit{x}  ]  \, \tau_{\ottmv{p}}   \produces    [  \mathit{x'}  /  \mathit{x}  ]  \, \Gamma'_{\ottmv{p}} 
    \end{align}
    where $\Gamma''  \ottsym{=}  \Gamma_{\ottmv{p}}  \ottsym{[}  \mathit{y}  \hookleftarrow   \ottsym{(}   [  \mathit{x'}  /  \mathit{x}  ]  \, \tau''  \ottsym{)}  \TREF^{ r }   \ottsym{]}  \ottsym{,}  \mathit{x'}  \ottsym{:}  \tau_{{\mathrm{2}}}$.

    Together these give our desired result. To see how,
    from $ \Theta   \mid   \oldvec{\ell}   \mid   \Gamma_{\ottmv{p}}  \ottsym{[}  \mathit{y}  \hookleftarrow   \tau''  \TREF^{ r }   \ottsym{]}  \ottsym{,}  \mathit{x}  \ottsym{:}  \tau_{{\mathrm{2}}}   \vdash   \ottnt{e}  :  \tau_{\ottmv{p}}   \produces   \Gamma'_{\ottmv{p}} $ above,
    we must have that $ \oldvec{\ell}   \vdash _{\wf}  \tau_{\ottmv{p}}   \produces   \Gamma'_{\ottmv{p}} $.
    From $ \mathit{x}  \not\in \DOM( \Gamma'_{\ottmv{p}} ) $ we must therefore have that
    $ [  \mathit{x'}  /  \mathit{x}  ]  \, \tau_{\ottmv{p}}  \ottsym{=}  \tau_{\ottmv{p}}$ and $ [  \mathit{x'}  /  \mathit{x}  ]  \, \Gamma'_{\ottmv{p}}  \ottsym{=}  \Gamma'_{\ottmv{p}}$. As $\Gamma'_{\ottmv{p}}  \ottsym{,}  \tau_{\ottmv{p}}  \leq  \Gamma'_{{\mathrm{0}}}  \ottsym{,}  \tau$
    an application of \rn{T-Sub} gives
    $ \Theta   \mid   \oldvec{\ell}   \mid   \Gamma''   \vdash     [  \mathit{x'}  /  \mathit{x}  ]    \ottnt{e}   :  \tau   \produces   \Gamma'_{{\mathrm{0}}} $. Then \Cref{lem:ectxt-sub-well-typed}
    will give that $ \Theta   \mid   \oldvec{\ell}   \mid   \Gamma''   \vdash   \ottnt{E}  \ottsym{[}    [  \mathit{x'}  /  \mathit{x}  ]    \ottnt{e}   \ottsym{]}  :  \tau_{\ottmv{n}}   \produces   \Gamma_{\ottmv{n}} $.
    
    As $\ottnt{E}$ and the stack $\oldvec{F}$ remained unchanged, combined with
    \Cref{eqn:new-tyenv-cons} this gives
    $ \vdash_{\mathit{conf} }^D   \tuple{ \ottnt{H} ,  \ottnt{R}  \ottsym{\{}  \mathit{x'}  \mapsto  \ottnt{v}  \ottsym{\}} ,  \oldvec{F} ,  \ottnt{E}  \ottsym{[}    [  \mathit{x'}  /  \mathit{x}  ]    \ottnt{e}   \ottsym{]} }  $ as required.
    As the above argument is used almost completely unchanged
    in all of the following cases, we will invert the redex without regard
    for the \rn{T-Sub} rule, with the understanding that the subtyping rule
    is handled with an argument identical to the above.
    
    We now show that $ \Theta   \mid   \oldvec{\ell}   \mid   \Gamma''   \vdash     [  \mathit{x'}  /  \mathit{x}  ]    \ottnt{e}   :   [  \mathit{x'}  /  \mathit{x}  ]  \, \tau_{\ottmv{p}}   \produces    [  \mathit{x'}  /  \mathit{x}  ]  \, \Gamma'_{\ottmv{p}} $ and
    $\ottkw{Cons} \, \ottsym{(}  \ottnt{H}  \ottsym{,}  \ottnt{R}  \ottsym{\{}  \mathit{x'}  \mapsto  \ottnt{v}  \ottsym{\}}  \ottsym{,}  \Gamma''  \ottsym{)}$.
    The first is easy to obtain using \Cref{lem:substitution}; from $\ottkw{Cons} \, \ottsym{(}  \ottnt{H}  \ottsym{,}  \ottnt{R}  \ottsym{,}  \Gamma_{\ottmv{p}}  \ottsym{)}$,
    we have that $\forall \,  \mathit{x}  \in \DOM( \Gamma_{\ottmv{p}} )   \ottsym{.}   \mathit{x}  \in   \DOM( \ottnt{R} )  $, whereby from $ \mathit{x}  \not\in   \DOM( \ottnt{R} )  $ we have
    $ \mathit{x'}  \not\in   \DOM( \Gamma_{\ottmv{p}} )  $.
    It therefore remains to show $\ottkw{Cons} \, \ottsym{(}  \ottnt{H}  \ottsym{,}  \ottnt{R}  \ottsym{\{}  \mathit{x'}  \mapsto  \ottnt{v}  \ottsym{\}}  \ottsym{,}  \Gamma''  \ottsym{)}$.

    To show $\ottkw{SAT} \, \ottsym{(}  \ottnt{H}  \ottsym{,}  \ottnt{R'}  \ottsym{,}  \Gamma''  \ottsym{)}$, it suffices to
    show that $ \ottkw{SATv} ( \ottnt{H} , \ottnt{R'} , \ottnt{R'}  \ottsym{(}  \mathit{x'}  \ottsym{)} , \tau_{{\mathrm{2}}} ) $
    and $ \ottkw{SATv} ( \ottnt{H} , \ottnt{R'} , \ottnt{H}  \ottsym{(}  \ottnt{R'}  \ottsym{(}  \mathit{y}  \ottsym{)}  \ottsym{)} , \tau'' ) $ (that $\ottkw{SATv}$ holds for all other variables
    $\mathit{z}$ follows from $\Gamma_{\ottmv{p}}  \ottsym{(}  \mathit{z}  \ottsym{)}  \ottsym{=}  \Gamma''  \ottsym{(}  \mathit{z}  \ottsym{)}$ and \Cref{lem:r-valid-subst,lem:register}).
    If $\tau_{{\mathrm{1}}}$ is an integer type and $ r   \ottsym{>}   \ottsym{0} $, then
    by the definition of the strengthening operator, the latter is equivalent to show
    that $ \ottkw{SATv} ( \ottnt{H} , \ottnt{R'} , \ottnt{H}  \ottsym{(}  \ottnt{R'}  \ottsym{(}  \mathit{y}  \ottsym{)}  \ottsym{)} , \tau_{{\mathrm{1}}} ) $
    and that $\ottnt{R'}  \ottsym{(}  \mathit{x'}  \ottsym{)}  \ottsym{=}  \ottnt{H}  \ottsym{(}  \ottnt{R'}  \ottsym{(}  \mathit{y}  \ottsym{)}  \ottsym{)}  \ottsym{=}  \ottnt{H}  \ottsym{(}  \ottnt{R}  \ottsym{(}  \mathit{y}  \ottsym{)}  \ottsym{)}$, which is immediate from the definition of $\rn{R-Deref}$.
    If $\tau_{{\mathrm{1}}}$ is not an integer or if $r  \ottsym{=}  \ottsym{0}$,
    then we must only show that $ \ottkw{SATv} ( \ottnt{H} , \ottnt{R'} , \ottnt{H}  \ottsym{(}  \ottnt{R'}  \ottsym{(}  \mathit{y}  \ottsym{)}  \ottsym{)} , \tau_{{\mathrm{1}}} ) $.
    
    From $\ottkw{Cons} \, \ottsym{(}  \ottnt{H}  \ottsym{,}  \ottnt{R}  \ottsym{,}  \Gamma_{\ottmv{p}}  \ottsym{)}$, we know that $\ottkw{SAT} \, \ottsym{(}  \ottnt{H}  \ottsym{,}  \ottnt{R}  \ottsym{,}  \Gamma_{\ottmv{p}}  \ottsym{)}$, in particular, $ \ottkw{SATv} ( \ottnt{H} , \ottnt{R} , \ottnt{R}  \ottsym{(}  \mathit{y}  \ottsym{)} , \Gamma_{\ottmv{p}}  \ottsym{(}  \mathit{y}  \ottsym{)} )  $.
    Then by \Cref{lem:satadd,lem:r-valid-subst,lem:register}, from $ \ottnt{R}  \sqsubseteq  \ottnt{R'} $, $ \oldvec{\ell}   \vdash _{\wf}  \Gamma_{\ottmv{p}} $, and $ \ottkw{SATv} ( \ottnt{H} , \ottnt{R} , \ottnt{v} , \tau_{{\mathrm{1}}}  \ottsym{+}  \tau_{{\mathrm{2}}} ) $ we obtain
    that $ \ottkw{SATv} ( \ottnt{H} , \ottnt{R'} , \ottnt{v} , \tau_{{\mathrm{1}}} ) $ and $ \ottkw{SATv} ( \ottnt{H} , \ottnt{R'} , \ottnt{v} , \tau_{{\mathrm{2}}} ) $, where $\ottnt{v} \, \ottsym{=} \, \ottnt{H}  \ottsym{(}  \ottnt{R}  \ottsym{(}  \mathit{y}  \ottsym{)}  \ottsym{)}$.
    We thus have that $ \ottkw{SATv} ( \ottnt{H} , \ottnt{R'} , \ottnt{R'}  \ottsym{(}  \mathit{x'}  \ottsym{)} , \tau_{{\mathrm{2}}} ) $ and $ \ottkw{SATv} ( \ottnt{H} , \ottnt{R'} , \ottnt{H}  \ottsym{(}  \ottnt{R'}  \ottsym{(}  \mathit{y}  \ottsym{)}  \ottsym{)} , \tau_{{\mathrm{1}}} ) $ are satisfied.
    
    We must also show that the ownership invariant is preserved.
    Then, it's to show $\forall \,  \ottmv{a}  \in \DOM( \ottnt{H} )   \ottsym{.}  \ottkw{Own} \, \ottsym{(}  \ottnt{H}  \ottsym{,}  \ottnt{R'}  \ottsym{,}  \Gamma''  \ottsym{)}  \ottsym{(}  \ottmv{a}  \ottsym{)}  \le  \ottsym{1}$. Define $\ottnt{O'_{{\mathrm{0}}}}$ and $\ottnt{O'_{{\mathrm{1}}}}$ as follows:
    \begin{align*}
      \ottkw{Own} \, \ottsym{(}  \ottnt{H}  \ottsym{,}  \ottnt{R}  \ottsym{,}  \Gamma_{\ottmv{p}}  \ottsym{)} & = \ottnt{O'_{{\mathrm{0}}}}  \ottsym{+}  \ottkw{own} \, \ottsym{(}  \ottnt{H}  \ottsym{,}  \ottnt{R}  \ottsym{(}  \mathit{y}  \ottsym{)}  \ottsym{,}  \Gamma_{\ottmv{p}}  \ottsym{(}  \mathit{y}  \ottsym{)}  \ottsym{)} \\
      \ottkw{Own} \, \ottsym{(}  \ottnt{H}  \ottsym{,}  \ottnt{R'}  \ottsym{,}  \Gamma''  \ottsym{)} & =\ottnt{O'_{{\mathrm{1}}}}  \ottsym{+}  \ottkw{own} \, \ottsym{(}  \ottnt{H}  \ottsym{,}  \ottnt{R'}  \ottsym{(}  \mathit{y}  \ottsym{)}  \ottsym{,}  \Gamma''  \ottsym{(}  \mathit{y}  \ottsym{)}  \ottsym{)}  \ottsym{+}  \ottkw{own} \, \ottsym{(}  \ottnt{H}  \ottsym{,}  \ottnt{R'}  \ottsym{(}  \mathit{x'}  \ottsym{)}  \ottsym{,}  \Gamma''  \ottsym{(}  \mathit{x'}  \ottsym{)}  \ottsym{)} \\
      \ottnt{O'_{{\mathrm{0}}}} &  =  \Sigma _{ \mathit{z} \in   \DOM( \Gamma_{\ottmv{p}} )   \setminus   \set{ \mathit{y} }   }\, \ottkw{own} \, \ottsym{(}  \ottnt{H}  \ottsym{,}  \ottnt{R}  \ottsym{(}  \mathit{z}  \ottsym{)}  \ottsym{,}  \Gamma_{\ottmv{p}}  \ottsym{(}  \mathit{z}  \ottsym{)}  \ottsym{)}  \\
      \ottnt{O'_{{\mathrm{1}}}} & =  \Sigma _{ \mathit{z} \in   \DOM( \Gamma'' )   \setminus  \ottsym{\{}  \mathit{y}  \ottsym{,}  \mathit{x'}  \ottsym{\}}  }\, \ottkw{own} \, \ottsym{(}  \ottnt{H}  \ottsym{,}  \ottnt{R'}  \ottsym{(}  \mathit{z}  \ottsym{)}  \ottsym{,}  \Gamma''  \ottsym{(}  \mathit{z}  \ottsym{)}  \ottsym{)} 
    \end{align*}
    Since $\ottnt{R'}$ agrees with $\ottnt{R}$ on all old variables and the
    corresponding types are unchanged, $\ottnt{O'_{{\mathrm{0}}}}  \ottsym{=}  \ottnt{O'_{{\mathrm{1}}}}$.  Then it suffices to show that
    $\ottkw{own} \, \ottsym{(}  \ottnt{H}  \ottsym{,}  \ottnt{R}  \ottsym{(}  \mathit{y}  \ottsym{)}  \ottsym{,}  \Gamma_{\ottmv{p}}  \ottsym{(}  \mathit{y}  \ottsym{)}  \ottsym{)}  \ottsym{=}  \ottkw{own} \, \ottsym{(}  \ottnt{H}  \ottsym{,}  \ottnt{R'}  \ottsym{(}  \mathit{y}  \ottsym{)}  \ottsym{,}  \Gamma''  \ottsym{(}  \mathit{y}  \ottsym{)}  \ottsym{)}  \ottsym{+}  \ottkw{own} \, \ottsym{(}  \ottnt{H}  \ottsym{,}  \ottnt{R'}  \ottsym{(}  \mathit{x'}  \ottsym{)}  \ottsym{,}  \Gamma''  \ottsym{(}  \mathit{x'}  \ottsym{)}  \ottsym{)}$.
    
    As $\ottnt{R'}  \ottsym{(}  \mathit{x'}  \ottsym{)}  \ottsym{=}  \ottnt{H}  \ottsym{(}  \ottnt{R'}  \ottsym{(}  \mathit{y}  \ottsym{)}  \ottsym{)}  \ottsym{=}  \ottnt{H}  \ottsym{(}  \ottnt{R}  \ottsym{(}  \mathit{y}  \ottsym{)}  \ottsym{)}$ and from the definition of $\Gamma''$, we have:
    \begin{align*}
       \ottkw{own} \, \ottsym{(}  \ottnt{H}  \ottsym{,}  \ottnt{R'}  \ottsym{(}  \mathit{x'}  \ottsym{)}  \ottsym{,}  \Gamma''  \ottsym{(}  \mathit{x'}  \ottsym{)}  \ottsym{)}  & =  \ottkw{own} \, \ottsym{(}  \ottnt{H}  \ottsym{,}  \ottnt{H}  \ottsym{(}  \ottnt{R}  \ottsym{(}  \mathit{y}  \ottsym{)}  \ottsym{)}  \ottsym{,}  \tau_{{\mathrm{2}}}  \ottsym{)}  \\
       \ottkw{own} \, \ottsym{(}  \ottnt{H}  \ottsym{,}  \ottnt{R'}  \ottsym{(}  \mathit{y}  \ottsym{)}  \ottsym{,}  \Gamma''  \ottsym{(}  \mathit{y}  \ottsym{)}  \ottsym{)}  & =  \ottsym{\{}  \ottmv{a}  \mapsto  r  \ottsym{\}}  \ottsym{+}  \ottkw{own} \, \ottsym{(}  \ottnt{H}  \ottsym{,}  \ottnt{H}  \ottsym{(}  \ottnt{R}  \ottsym{(}  \mathit{y}  \ottsym{)}  \ottsym{)}  \ottsym{,}  \tau_{{\mathrm{1}}}  \ottsym{)} 
    \end{align*}
    From the definition of the ownership function, we have that
    \[
      \ottkw{own} \, \ottsym{(}  \ottnt{H}  \ottsym{,}  \ottnt{R}  \ottsym{(}  \mathit{y}  \ottsym{)}  \ottsym{,}  \Gamma_{\ottmv{p}}  \ottsym{(}  \mathit{y}  \ottsym{)}  \ottsym{)}  \ottsym{=}  \ottsym{\{}  \ottmv{a}  \mapsto  r  \ottsym{\}}  \ottsym{+}  \ottkw{own} \, \ottsym{(}  \ottnt{H}  \ottsym{,}  \ottnt{H}  \ottsym{(}  \ottnt{R}  \ottsym{(}  \mathit{y}  \ottsym{)}  \ottsym{)}  \ottsym{,}  \tau_{{\mathrm{1}}}  \ottsym{+}  \tau_{{\mathrm{2}}}  \ottsym{)}
    \]
    which, by \Cref{lem:ownadd}, is equivalent to:
    \[
      \ottsym{\{}  \ottmv{a}  \mapsto  r  \ottsym{\}}  \ottsym{+}  \ottkw{own} \, \ottsym{(}  \ottnt{H}  \ottsym{,}  \ottnt{H}  \ottsym{(}  \ottnt{R}  \ottsym{(}  \mathit{y}  \ottsym{)}  \ottsym{)}  \ottsym{,}  \tau_{{\mathrm{1}}}  \ottsym{)}  \ottsym{+}  \ottkw{own} \, \ottsym{(}  \ottnt{H}  \ottsym{,}  \ottnt{H}  \ottsym{(}  \ottnt{R}  \ottsym{(}  \mathit{y}  \ottsym{)}  \ottsym{)}  \ottsym{,}  \tau_{{\mathrm{2}}}  \ottsym{)}
    \]
    We therefore have $\ottkw{own} \, \ottsym{(}  \ottnt{H}  \ottsym{,}  \ottnt{R}  \ottsym{(}  \mathit{y}  \ottsym{)}  \ottsym{,}  \Gamma_{\ottmv{p}}  \ottsym{(}  \mathit{y}  \ottsym{)}  \ottsym{)}  \ottsym{=}  \ottkw{own} \, \ottsym{(}  \ottnt{H}  \ottsym{,}  \ottnt{R'}  \ottsym{(}  \mathit{y}  \ottsym{)}  \ottsym{,}  \Gamma''  \ottsym{(}  \mathit{y}  \ottsym{)}  \ottsym{)}  \ottsym{+}  \ottkw{own} \, \ottsym{(}  \ottnt{H}  \ottsym{,}  \ottnt{R'}  \ottsym{(}  \mathit{x'}  \ottsym{)}  \ottsym{,}  \Gamma''  \ottsym{(}  \mathit{x'}  \ottsym{)}  \ottsym{)}$, and conclude that $\ottkw{Own} \, \ottsym{(}  \ottnt{H}  \ottsym{,}  \ottnt{R}  \ottsym{,}  \Gamma_{\ottmv{p}}  \ottsym{)}  \ottsym{=}  \ottkw{Own} \, \ottsym{(}  \ottnt{H}  \ottsym{,}  \ottnt{R'}  \ottsym{,}  \Gamma''  \ottsym{)}$.
  \end{rneqncase} %

  \begin{rneqncase}{R-Seq}{
       \vdash_{\mathit{conf} }^D   \tuple{ \ottnt{H} ,  \ottnt{R} ,  \oldvec{F} ,  \ottnt{E}  \ottsym{[}   \mathit{x}  \SEQ  \ottnt{e}   \ottsym{]} }  \\
        \tuple{ \ottnt{H} ,  \ottnt{R} ,  \oldvec{F} ,  \ottnt{E}  \ottsym{[}   \mathit{x}  \SEQ  \ottnt{e}   \ottsym{]} }     \longrightarrow _{ \ottnt{D} }     \tuple{ \ottnt{H} ,  \ottnt{R} ,  \oldvec{F} ,  \ottnt{E}  \ottsym{[}  \ottnt{e}  \ottsym{]} }  \\
    }
    By inversion (see \rn{R-Deref}) we have for some $\Gamma$ that:
    \begin{align*}
      &  \Theta   \mid   \oldvec{\ell}   \mid   \Gamma  \ottsym{[}  \mathit{x}  \ottsym{:}  \tau_{{\mathrm{0}}}  \ottsym{+}  \tau_{{\mathrm{1}}}  \ottsym{]}   \vdash   \mathit{x}  :  \tau_{{\mathrm{0}}}   \produces   \Gamma  \ottsym{[}  \mathit{x}  \hookleftarrow  \tau_{{\mathrm{1}}}  \ottsym{]}  \\
      &  \Theta   \mid   \oldvec{\ell}   \mid   \Gamma  \ottsym{[}  \mathit{x}  \hookleftarrow  \tau_{{\mathrm{1}}}  \ottsym{]}   \vdash   \ottnt{e}  :  \tau'   \produces   \Gamma'  \\
      & \ottkw{Cons} \, \ottsym{(}  \ottnt{H}  \ottsym{,}  \ottnt{R}  \ottsym{,}  \Gamma  \ottsym{)}
    \end{align*}
    We take $\Gamma''  \ottsym{=}  \Gamma  \ottsym{[}  \mathit{x}  \hookleftarrow  \tau_{{\mathrm{1}}}  \ottsym{]}$.
    
    It suffices to show (see \rn{R-Deref})
    that $ \Theta   \mid   \oldvec{\ell}   \mid   \Gamma''   \vdash   \ottnt{e}  :  \tau'   \produces   \Gamma' $,
    and $\ottkw{Cons} \, \ottsym{(}  \ottnt{H}  \ottsym{,}  \ottnt{R}  \ottsym{,}  \Gamma''  \ottsym{)}$.
    The first is immediate from the inversion above, and
    $\ottkw{Cons} \, \ottsym{(}  \ottnt{H}  \ottsym{,}  \ottnt{R}  \ottsym{,}  \Gamma''  \ottsym{)}$ follows from \Cref{lem:ownadd,lem:satadd}.
  \end{rneqncase}

  \begin{rneqncase}{R-Let}{
       \vdash_{\mathit{conf} }^D   \tuple{ \ottnt{H} ,  \ottnt{R} ,  \oldvec{F} ,  \ottnt{E}  \ottsym{[}   \LET  \mathit{x}  =  \mathit{y}  \IN  \ottnt{e}   \ottsym{]} }  \\
        \tuple{ \ottnt{H} ,  \ottnt{R} ,  \oldvec{F} ,  \ottnt{E}  \ottsym{[}   \LET  \mathit{x}  =  \mathit{y}  \IN  \ottnt{e}   \ottsym{]} }     \longrightarrow _{ \ottnt{D} }     \tuple{ \ottnt{H} ,  \ottnt{R}  \ottsym{\{}  \mathit{x'}  \mapsto  \ottnt{R}  \ottsym{(}  \mathit{y}  \ottsym{)}  \ottsym{\}} ,  \oldvec{F} ,  \ottnt{E}  \ottsym{[}    [  \mathit{x'}  /  \mathit{x}  ]    \ottnt{e}   \ottsym{]} }  \\
       \mathit{x'}  \not\in \DOM( \ottnt{R} )  \andalso \ottnt{R'}  \ottsym{=}  \ottnt{R}  \ottsym{\{}  \mathit{x'}  \mapsto  \ottnt{R}  \ottsym{(}  \mathit{y}  \ottsym{)}  \ottsym{\}}
    }
    By inversion (see \rn{R-Deref}) we have that for some $\Gamma$ that:
    \begin{align*}
      &  \oldvec{\ell}   \vdash _{\wf}  \Gamma \\
      & \Gamma  \ottsym{(}  \mathit{y}  \ottsym{)}  \ottsym{=}  \tau_{{\mathrm{1}}}  \ottsym{+}  \tau_{{\mathrm{2}}} \\
      &  \Theta   \mid   \oldvec{\ell}   \mid   \Gamma  \ottsym{[}  \mathit{y}  \hookleftarrow   \tau_{{\mathrm{1}}}  \wedge_{ \mathit{y} }   \mathit{y}  =_{ \tau_{{\mathrm{1}}} }  \mathit{x}    \ottsym{]}  \ottsym{,}  \mathit{x}  \ottsym{:}  \ottsym{(}   \tau_{{\mathrm{2}}}  \wedge_{ \mathit{x} }   \mathit{x}  =_{ \tau_{{\mathrm{2}}} }  \mathit{y}    \ottsym{)}   \vdash   \ottnt{e}  :  \tau   \produces   \Gamma'  \\
      & \ottkw{Cons} \, \ottsym{(}  \ottnt{H}  \ottsym{,}  \ottnt{R}  \ottsym{,}  \Gamma  \ottsym{)} \andalso  \mathit{x}  \not\in \DOM( \Gamma' ) 
    \end{align*}
    We give $\Gamma''  \ottsym{=}  \Gamma  \ottsym{[}  \mathit{y}  \hookleftarrow   \tau_{{\mathrm{1}}}  \wedge_{ \mathit{y} }   \mathit{y}  =_{ \tau_{{\mathrm{1}}} }  \mathit{x'}    \ottsym{]}  \ottsym{,}  \mathit{x'}  \ottsym{:}  \ottsym{(}   \tau_{{\mathrm{2}}}  \wedge_{ \mathit{x'} }   \mathit{x'}  =_{ \tau_{{\mathrm{2}}} }  \mathit{y}    \ottsym{)}$.
    
    It suffices to show (see \rn{R-Deref})
    that $ \Theta   \mid   \oldvec{\ell}   \mid   \Gamma''   \vdash     [  \mathit{x'}  /  \mathit{x}  ]    \ottnt{e}   :  \tau   \produces   \Gamma' $
    and $\ottkw{Cons} \, \ottsym{(}  \ottnt{H}  \ottsym{,}  \ottnt{R}  \ottsym{\{}  \mathit{x'}  \mapsto  \ottnt{R}  \ottsym{(}  \mathit{y}  \ottsym{)}  \ottsym{\}}  \ottsym{,}  \Gamma''  \ottsym{)}$.
    The first is easy to obtain using reasoning as in the \rn{R-Deref} case.
    It therefore remains to show $\ottkw{Cons} \, \ottsym{(}  \ottnt{H}  \ottsym{,}  \ottnt{R'}  \ottsym{,}  \Gamma''  \ottsym{)}$.
    
    To show that the output environment is consistent, we must show that $ \ottkw{SATv} ( \ottnt{H} , \ottnt{R'} , \ottnt{R'}  \ottsym{(}  \mathit{x'}  \ottsym{)} ,  \tau_{{\mathrm{2}}}  \wedge_{ \mathit{x'} }   \mathit{x'}  =_{ \tau_{{\mathrm{2}}} }  \mathit{y}   ) $ and $ \ottkw{SATv} ( \ottnt{H} , \ottnt{R'} , \ottnt{R'}  \ottsym{(}  \mathit{y}  \ottsym{)} ,  \tau_{{\mathrm{1}}}  \wedge_{ \mathit{y} }   \mathit{y}  =_{ \tau_{{\mathrm{1}}} }  \mathit{x'}   ) $.
    By reasoning similar to that in $\rn{R-Deref}$, it suffices to show that $ \ottkw{SATv} ( \ottnt{H} , \ottnt{R'} , \ottnt{R'}  \ottsym{(}  \mathit{x'}  \ottsym{)} , \tau_{{\mathrm{2}}} ) $ and $ \ottkw{SATv} ( \ottnt{H} , \ottnt{R'} , \ottnt{R'}  \ottsym{(}  \mathit{y}  \ottsym{)} , \tau_{{\mathrm{1}}} ) $.
    We know that $\ottkw{Cons} \, \ottsym{(}  \ottnt{H}  \ottsym{,}  \ottnt{R}  \ottsym{,}  \Gamma  \ottsym{)}$, from which we have $\ottkw{SAT} \, \ottsym{(}  \ottnt{H}  \ottsym{,}  \ottnt{R}  \ottsym{,}  \Gamma  \ottsym{)}$, in particular $ \mathit{y}  \in \DOM( \ottnt{R} ) $ and $ \ottkw{SATv} ( \ottnt{H} , \ottnt{R} , \ottnt{R}  \ottsym{(}  \mathit{y}  \ottsym{)} , \Gamma  \ottsym{(}  \mathit{y}  \ottsym{)} ) $.
    As $ \ottnt{R}  \sqsubseteq  \ottnt{R'} $ and $ \oldvec{\ell}   \vdash _{\wf}  \Gamma $, from \Cref{lem:satadd,lem:r-valid-subst,lem:register}, we obtain from $ \ottkw{SATv} ( \ottnt{H} , \ottnt{R} , \ottnt{v} , \tau_{{\mathrm{1}}}  \ottsym{+}  \tau_{{\mathrm{2}}} ) $ that
    $ \ottkw{SATv} ( \ottnt{H} , \ottnt{R'} , \ottnt{v} , \tau_{{\mathrm{1}}} ) $ and $ \ottkw{SATv} ( \ottnt{H} , \ottnt{R'} , \ottnt{v} , \tau_{{\mathrm{2}}} ) $ where $v = R(y)$.
    We then have $ \ottkw{SATv} ( \ottnt{H} , \ottnt{R'} , \ottnt{R'}  \ottsym{(}  \mathit{x'}  \ottsym{)} , \tau_{{\mathrm{2}}} ) $ and $ \ottkw{SATv} ( \ottnt{H} , \ottnt{R'} , \ottnt{R'}  \ottsym{(}  \mathit{y}  \ottsym{)} , \tau_{{\mathrm{1}}} ) $ are satisfied.
    
    We must also show that the ownership invariant is preserved.
    Then, it's to show $\forall \,  \ottmv{a}  \in \DOM( \ottnt{H} )   \ottsym{.}  \ottkw{Own} \, \ottsym{(}  \ottnt{H}  \ottsym{,}  \ottnt{R'}  \ottsym{,}  \Gamma''  \ottsym{)}  \ottsym{(}  \ottmv{a}  \ottsym{)}  \le  \ottsym{1}$. Define $\ottnt{O'_{{\mathrm{0}}}}$ and $\ottnt{O'_{{\mathrm{1}}}}$ as follows:
    \begin{align*}
    \ottkw{Own} \, \ottsym{(}  \ottnt{H}  \ottsym{,}  \ottnt{R}  \ottsym{,}  \Gamma  \ottsym{)} & = \ottnt{O'_{{\mathrm{0}}}}  \ottsym{+}  \ottkw{own} \, \ottsym{(}  \ottnt{H}  \ottsym{,}  \ottnt{R}  \ottsym{(}  \mathit{y}  \ottsym{)}  \ottsym{,}  \Gamma  \ottsym{(}  \mathit{y}  \ottsym{)}  \ottsym{)} \\ 
    \ottkw{Own} \, \ottsym{(}  \ottnt{H}  \ottsym{,}  \ottnt{R'}  \ottsym{,}  \Gamma''  \ottsym{)} & = \ottnt{O'_{{\mathrm{1}}}}  \ottsym{+}  \ottkw{own} \, \ottsym{(}  \ottnt{H}  \ottsym{,}  \ottnt{R'}  \ottsym{(}  \mathit{y}  \ottsym{)}  \ottsym{,}  \Gamma''  \ottsym{(}  \mathit{y}  \ottsym{)}  \ottsym{)}  \ottsym{+}  \ottkw{own} \, \ottsym{(}  \ottnt{H}  \ottsym{,}  \ottnt{R'}  \ottsym{(}  \mathit{x'}  \ottsym{)}  \ottsym{,}  \Gamma''  \ottsym{(}  \mathit{x'}  \ottsym{)}  \ottsym{)} \\
    \ottnt{O'_{{\mathrm{0}}}} & =  \Sigma _{ \mathit{z} \in   \DOM( \Gamma )   \setminus   \set{ \mathit{y} }   }\, \ottkw{own} \, \ottsym{(}  \ottnt{H}  \ottsym{,}  \ottnt{R}  \ottsym{(}  \mathit{z}  \ottsym{)}  \ottsym{,}  \Gamma  \ottsym{(}  \mathit{z}  \ottsym{)}  \ottsym{)}  \\
    \ottnt{O'_{{\mathrm{1}}}} & =  \Sigma _{ \mathit{z} \in   \DOM( \Gamma'' )   \setminus  \ottsym{\{}  \mathit{y}  \ottsym{,}  \mathit{x'}  \ottsym{\}}  }\, \ottkw{own} \, \ottsym{(}  \ottnt{H}  \ottsym{,}  \ottnt{R'}  \ottsym{(}  \mathit{z}  \ottsym{)}  \ottsym{,}  \Gamma''  \ottsym{(}  \mathit{z}  \ottsym{)}  \ottsym{)} 
    \end{align*}
    Since $\ottnt{R'}$ agrees with $\ottnt{R}$ on old variables and their types are
    unchanged, $\ottnt{O'_{{\mathrm{0}}}}  \ottsym{=}  \ottnt{O'_{{\mathrm{1}}}}$ holds.
    That $\ottkw{own} \, \ottsym{(}  \ottnt{H}  \ottsym{,}  \ottnt{R'}  \ottsym{(}  \mathit{x'}  \ottsym{)}  \ottsym{,}  \tau_{{\mathrm{2}}}  \ottsym{)}  \ottsym{+}  \ottkw{own} \, \ottsym{(}  \ottnt{H}  \ottsym{,}  \ottnt{R'}  \ottsym{(}  \mathit{y}  \ottsym{)}  \ottsym{,}  \tau_{{\mathrm{1}}}  \ottsym{)}  \ottsym{=}  \ottkw{own} \, \ottsym{(}  \ottnt{H}  \ottsym{,}  \ottnt{R}  \ottsym{(}  \mathit{y}  \ottsym{)}  \ottsym{,}  \tau_{{\mathrm{1}}}  \ottsym{+}  \tau_{{\mathrm{2}}}  \ottsym{)}$ follows immediately
    from \Cref{lem:ownadd} and the condition $\ottnt{R}  \ottsym{(}  \mathit{y}  \ottsym{)}  \ottsym{=}  \ottnt{R'}  \ottsym{(}  \mathit{x'}  \ottsym{)}  \ottsym{=}  \ottnt{R'}  \ottsym{(}  \mathit{y}  \ottsym{)}$.
    We therefore conclude that $\ottkw{Own} \, \ottsym{(}  \ottnt{H}  \ottsym{,}  \ottnt{R}  \ottsym{,}  \Gamma  \ottsym{)}  \ottsym{=}  \ottkw{Own} \, \ottsym{(}  \ottnt{H}  \ottsym{,}  \ottnt{R'}  \ottsym{,}  \Gamma''  \ottsym{)}$.
  \end{rneqncase} %
 
  \begin{rneqncase}{R-LetInt}{ %
       \vdash_{\mathit{conf} }^D   \tuple{ \ottnt{H} ,  \ottnt{R} ,  \oldvec{F} ,  \ottnt{E}  \ottsym{[}   \LET  \mathit{x}  =  n  \IN  \ottnt{e}   \ottsym{]} }  \\
        \tuple{ \ottnt{H} ,  \ottnt{R} ,  \oldvec{F} ,  \ottnt{E}  \ottsym{[}   \LET  \mathit{x}  =  n  \IN  \ottnt{e}   \ottsym{]} }     \longrightarrow _{ \ottnt{D} }     \tuple{ \ottnt{H} ,  \ottnt{R}  \ottsym{\{}  \mathit{x'}  \mapsto  n  \ottsym{\}} ,  \oldvec{F} ,  \ottnt{E}  \ottsym{[}    [  \mathit{x'}  /  \mathit{x}  ]    \ottnt{e}   \ottsym{]} }  
    }
    By inversion (see \rn{R-Deref}) we have that, for some $\Gamma$:
    \begin{align*}
      &  \Theta   \mid   \oldvec{\ell}   \mid   \Gamma  \ottsym{,}  \mathit{x}  \ottsym{:}   \set{  \nu  \COL \TINT \mid  \nu \, \ottsym{=} \, n }    \vdash   \ottnt{e}  :  \tau   \produces   \Gamma'  \\
      & \ottkw{Cons} \, \ottsym{(}  \ottnt{H}  \ottsym{,}  \ottnt{R}  \ottsym{,}  \Gamma  \ottsym{)} \andalso  \mathit{x}  \not\in   \DOM( \Gamma' )  
    \end{align*}
    
    We give that $\Gamma''  \ottsym{=}  \Gamma  \ottsym{,}  \mathit{x'}  \ottsym{:}   \set{  \nu  \COL \TINT \mid  \nu \, \ottsym{=} \, n } $, and it thus suffices to show
    that $ \Theta   \mid   \oldvec{\ell}   \mid   \Gamma''   \vdash     [  \mathit{x'}  /  \mathit{x}  ]    \ottnt{e}   :  \tau   \produces   \Gamma' $
    and $\ottkw{Cons} \, \ottsym{(}  \ottnt{H}  \ottsym{,}  \ottnt{R}  \ottsym{\{}  \mathit{x'}  \mapsto  n  \ottsym{\}}  \ottsym{,}  \Gamma''  \ottsym{)}$.
    The first one is easy to obtain using the \Cref{lem:substitution} (see \rn{R-Deref})
    and the latter is trivial by reasoning similar to the \rn{T-Let} and \rn{T-Deref} cases.
  \end{rneqncase} %
  
  \begin{rneqncase}{R-IfTrue}{ %
       \vdash_{\mathit{conf} }^D   \tuple{ \ottnt{H} ,  \ottnt{R} ,  \oldvec{F} ,  \ottnt{E}  \ottsym{[}   \IFZERO  \mathit{y}  \THEN  \ottnt{e_{{\mathrm{1}}}}  \ELSE  \ottnt{e_{{\mathrm{2}}}}   \ottsym{]} }  \\
        \tuple{ \ottnt{H} ,  \ottnt{R} ,  \oldvec{F} ,  \ottnt{E}  \ottsym{[}   \IFZERO  \mathit{y}  \THEN  \ottnt{e_{{\mathrm{1}}}}  \ELSE  \ottnt{e_{{\mathrm{2}}}}   \ottsym{]} }     \longrightarrow _{ \ottnt{D} }     \tuple{ \ottnt{H} ,  \ottnt{R} ,  \oldvec{F} ,  \ottnt{E}  \ottsym{[}  \ottnt{e_{{\mathrm{1}}}}  \ottsym{]} }  
    }
    By inversion (see \rn{R-Deref}) we have that for some $\Gamma$:
    \begin{align*}
      & \Gamma  \ottsym{(}  \mathit{x}  \ottsym{)}  \ottsym{=}   \set{  \nu  \COL \TINT \mid  \varphi }  \\
      &  \Theta   \mid   \oldvec{\ell}   \mid   \Gamma  \ottsym{[}  \mathit{x}  \hookleftarrow   \set{  \nu  \COL \TINT \mid   \varphi  \wedge  \nu \, \ottsym{=} \, \ottsym{0}  }   \ottsym{]}   \vdash   \ottnt{e_{{\mathrm{1}}}}  :  \tau   \produces   \Gamma'  \\
      & \ottkw{Cons} \, \ottsym{(}  \ottnt{H}  \ottsym{,}  \ottnt{R}  \ottsym{,}  \Gamma  \ottsym{)}
    \end{align*}
    We take $\Gamma''  \ottsym{=}  \Gamma  \ottsym{[}  \mathit{x}  \hookleftarrow   \set{  \nu  \COL \TINT \mid   \varphi  \wedge  \nu \, \ottsym{=} \, \ottsym{0}  }   \ottsym{]}$, and want to show that $\ottkw{Cons} \, \ottsym{(}  \ottnt{H}  \ottsym{,}  \ottnt{R}  \ottsym{,}  \Gamma''  \ottsym{)}$ (that $ \Theta   \mid   \oldvec{\ell}   \mid   \Gamma''   \vdash   \ottnt{e_{{\mathrm{1}}}}  :  \tau   \produces   \Gamma' $ is immediate).
    
    By definition, from $\ottkw{Cons} \, \ottsym{(}  \ottnt{H}  \ottsym{,}  \ottnt{R}  \ottsym{,}  \Gamma  \ottsym{)}$ we have $\ottkw{SAT} \, \ottsym{(}  \ottnt{H}  \ottsym{,}  \ottnt{R}  \ottsym{,}  \Gamma  \ottsym{)}$, in particular $ \mathit{x}  \in \DOM( \ottnt{R} ) $, $ \ottnt{R}  \ottsym{(}  \mathit{x}  \ottsym{)}  \in  \mathbb{Z} $ and $\ottsym{[}  \ottnt{R}  \ottsym{]} \, \ottsym{[}  \ottnt{R}  \ottsym{(}  \mathit{x}  \ottsym{)}  \ottsym{/}  \nu  \ottsym{]}  \varphi$.
    The refinement predicates $\varphi$ still holds in the output environment, since nothing changes in the register after transition.
    Also from precondition of $\rn{R-IfTrue}$, we have $\ottnt{R}  \ottsym{(}  \mathit{x}  \ottsym{)} \, \ottsym{=} \, \ottsym{0}$, thus $\mathit{x}$ satisfies the refinement that $\nu \, \ottsym{=} \, \ottsym{0}$.
    Thus $\ottsym{[}  \ottnt{R}  \ottsym{]} \, \ottsym{[}  \ottnt{R}  \ottsym{(}  \mathit{x}  \ottsym{)}  \ottsym{/}  \nu  \ottsym{]}  \ottsym{(}   \varphi  \wedge  \nu \, \ottsym{=} \, \ottsym{0}   \ottsym{)}$ is trivially satisfied.
  \end{rneqncase} %
  
  \begin{rncase}{R-IfFalse} %
    Similar to the case for \rn{R-IfTrue}.
  \end{rncase}
  
  \begin{rneqncase}{R-MkRef}{ 
       \vdash_{\mathit{conf} }^D   \tuple{ \ottnt{H} ,  \ottnt{R} ,  \oldvec{F} ,  \ottnt{E}  \ottsym{[}   \LET  \mathit{x}  =   \MKREF  \mathit{y}   \IN  \ottnt{e}   \ottsym{]} }  \\
        \tuple{ \ottnt{H} ,  \ottnt{R} ,  \oldvec{F} ,  \ottnt{E}  \ottsym{[}   \LET  \mathit{x}  =   \MKREF  \mathit{y}   \IN  \ottnt{e}   \ottsym{]} }     \longrightarrow _{ \ottnt{D} }     \tuple{ \ottnt{H'} ,  \ottnt{R'} ,  \oldvec{F} ,  \ottnt{E}  \ottsym{[}    [  \mathit{x'}  /  \mathit{x}  ]    \ottnt{e}   \ottsym{]} }  \\
       \ottmv{a}  \not\in \DOM( \ottnt{H} )  \andalso  \mathit{x'}  \not\in \DOM( \ottnt{R} )  \\
      \ottnt{H'}  \ottsym{=}  \ottnt{H}  \ottsym{\{}  \ottmv{a}  \mapsto  \ottnt{R}  \ottsym{(}  \mathit{y}  \ottsym{)}  \ottsym{\}} \andalso \ottnt{R'}  \ottsym{=}  \ottnt{R}  \ottsym{\{}  \mathit{x'}  \mapsto  \ottmv{a}  \ottsym{\}}
    }
    By inversion (see \rn{R-Deref}) we have that for some $\Gamma$:
    \begin{align*}
      &  \oldvec{\ell}   \vdash _{\wf}  \Gamma  \\
      & \Gamma  \ottsym{(}  \mathit{y}  \ottsym{)}  \ottsym{=}  \tau_{{\mathrm{1}}}  \ottsym{+}  \tau_{{\mathrm{2}}} \\
      &  \Theta   \mid   \oldvec{\ell}   \mid   \Gamma  \ottsym{[}  \mathit{y}  \hookleftarrow  \tau_{{\mathrm{1}}}  \ottsym{]}  \ottsym{,}  \mathit{x}  \ottsym{:}   \ottsym{(}   \tau_{{\mathrm{2}}}  \wedge_{ \mathit{x} }   \mathit{x}  =_{ \tau_{{\mathrm{2}}} }  \mathit{y}    \ottsym{)}  \TREF^{ \ottsym{1} }    \vdash   \ottnt{e}  :  \tau   \produces   \Gamma'  \\
      & \ottkw{Cons} \, \ottsym{(}  \ottnt{H}  \ottsym{,}  \ottnt{R}  \ottsym{,}  \Gamma  \ottsym{)} \andalso  \mathit{x}  \not\in   \DOM( \Gamma' )  
    \end{align*}
    We give $\Gamma''  \ottsym{=}  \Gamma  \ottsym{[}  \mathit{y}  \hookleftarrow  \tau_{{\mathrm{1}}}  \ottsym{]}  \ottsym{,}  \mathit{x'}  \ottsym{:}   \ottsym{(}   \tau_{{\mathrm{2}}}  \wedge_{ \mathit{x'} }   \mathit{x'}  =_{ \tau_{{\mathrm{2}}} }  \mathit{y}    \ottsym{)}  \TREF^{ \ottsym{1} } $,
    and must show that $ \Theta   \mid   \oldvec{\ell}   \mid   \Gamma''   \vdash     [  \mathit{x'}  /  \mathit{x}  ]    \ottnt{e}   :  \tau   \produces   \Gamma' $ and
    $\ottkw{Cons} \, \ottsym{(}  \ottnt{H'}  \ottsym{,}  \ottnt{R'}  \ottsym{,}  \Gamma''  \ottsym{)}$. The first follows
    from \Cref{lem:substitution} and the reasoning found in \rn{R-Deref},
    so it remains to show $\ottkw{Cons} \, \ottsym{(}  \ottnt{H'}  \ottsym{,}  \ottnt{R'}  \ottsym{,}  \Gamma''  \ottsym{)}$.
    To show that the output environment is consistent, we must show that $ \ottkw{SATv} ( \ottnt{H'} , \ottnt{R'} , \ottnt{R'}  \ottsym{(}  \mathit{x'}  \ottsym{)} ,  \ottsym{(}   \tau_{{\mathrm{2}}}  \wedge_{ \mathit{x'} }   \mathit{x'}  =_{ \tau_{{\mathrm{2}}} }  \mathit{y}    \ottsym{)}  \TREF^{ \ottsym{1} }  ) $ and $ \ottkw{SATv} ( \ottnt{H'} , \ottnt{R'} , \ottnt{R'}  \ottsym{(}  \mathit{y}  \ottsym{)} , \tau_{{\mathrm{1}}} ) $.
    By reasoning similar to that in $\rn{R-Deref}$, it suffices to show that $ \ottkw{SATv} ( \ottnt{H'} , \ottnt{R'} , \ottnt{R'}  \ottsym{(}  \mathit{x'}  \ottsym{)} ,  \tau_{{\mathrm{2}}}  \TREF^{ \ottsym{1} }  ) $ and $ \ottkw{SATv} ( \ottnt{H'} , \ottnt{R'} , \ottnt{R'}  \ottsym{(}  \mathit{y}  \ottsym{)} , \tau_{{\mathrm{1}}} ) $.
    We know that $\ottkw{Cons} \, \ottsym{(}  \ottnt{H}  \ottsym{,}  \ottnt{R}  \ottsym{,}  \Gamma  \ottsym{)}$, from which we have $\ottkw{SAT} \, \ottsym{(}  \ottnt{H}  \ottsym{,}  \ottnt{R}  \ottsym{,}  \Gamma  \ottsym{)}$, in particular $ \mathit{y}  \in \DOM( \ottnt{R} ) $ and $ \ottkw{SATv} ( \ottnt{H} , \ottnt{R} , \ottnt{R}  \ottsym{(}  \mathit{y}  \ottsym{)} , \Gamma  \ottsym{(}  \mathit{y}  \ottsym{)} )  $.
    As $ \ottnt{R}  \sqsubseteq  \ottnt{R'} $ and $ \oldvec{\ell}   \vdash _{\wf}  \Gamma $, from \Cref{lem:register,lem:r-valid-subst}, we have $ \ottkw{SATv} ( \ottnt{H} , \ottnt{R} , \ottnt{R}  \ottsym{(}  \mathit{y}  \ottsym{)} , \tau_{{\mathrm{1}}}  \ottsym{+}  \tau_{{\mathrm{2}}} ) $ implies $ \ottkw{SATv} ( \ottnt{H} , \ottnt{R'} , \ottnt{R'}  \ottsym{(}  \mathit{y}  \ottsym{)} , \tau_{{\mathrm{1}}}  \ottsym{+}  \tau_{{\mathrm{2}}} ) $.
    By \Cref{lem:newaddheap}, we then have $ \ottkw{SATv} ( \ottnt{H'} , \ottnt{R'} , \ottnt{R}  \ottsym{(}  \mathit{y}  \ottsym{)} , \tau_{{\mathrm{1}}}  \ottsym{+}  \tau_{{\mathrm{2}}} ) $.
    Then by \Cref{lem:satadd}, we have $ \ottkw{SATv} ( \ottnt{H'} , \ottnt{R'} , \ottnt{v} , \tau_{{\mathrm{1}}} ) $ and $ \ottkw{SATv} ( \ottnt{H'} , \ottnt{R'} , \ottnt{v} , \tau_{{\mathrm{2}}} ) $ where $\ottnt{v} \, \ottsym{=} \, \ottnt{R}  \ottsym{(}  \mathit{y}  \ottsym{)}$.
    We then have $ \ottkw{SATv} ( \ottnt{H'} , \ottnt{R'} , \ottnt{R'}  \ottsym{(}  \mathit{x'}  \ottsym{)} ,  \tau_{{\mathrm{2}}}  \TREF^{ \ottsym{1} }  ) $ and $ \ottkw{SATv} ( \ottnt{H'} , \ottnt{R'} , \ottnt{R'}  \ottsym{(}  \mathit{y}  \ottsym{)} , \tau_{{\mathrm{1}}} ) $ are satisfied.
    
    We must also show that the ownership invariant is preserved.
    Then, it's to show $\forall \,  \ottmv{a'}  \in \DOM( \ottnt{H'} )   \ottsym{.}  \ottkw{Own} \, \ottsym{(}  \ottnt{H'}  \ottsym{,}  \ottnt{R'}  \ottsym{,}  \Gamma''  \ottsym{)}  \ottsym{(}  \ottmv{a'}  \ottsym{)}  \le  \ottsym{1}$. From $\ottkw{Cons} \, \ottsym{(}  \ottnt{H}  \ottsym{,}  \ottnt{R}  \ottsym{,}  \Gamma  \ottsym{)}$ and \Cref{lem:ownaddheap,lem:ownadd} we have:
    \begin{align*}
      \ottkw{Own} \, \ottsym{(}  \ottnt{H'}  \ottsym{,}  \ottnt{R'}  \ottsym{,}  \Gamma''  \ottsym{)} & =  \Sigma _{ \mathit{z} \in  \DOM( \Gamma'' )  }\, \ottkw{own} \, \ottsym{(}  \ottnt{H'}  \ottsym{,}  \ottnt{R'}  \ottsym{(}  \mathit{z}  \ottsym{)}  \ottsym{,}  \Gamma''  \ottsym{(}  \mathit{z}  \ottsym{)}  \ottsym{)}  \\
             & =  \Sigma _{ \mathit{z} \in  \DOM( \Gamma )  }\, \ottkw{own} \, \ottsym{(}  \ottnt{H}  \ottsym{,}  \ottnt{R}  \ottsym{(}  \mathit{z}  \ottsym{)}  \ottsym{,}  \Gamma  \ottsym{(}  \mathit{z}  \ottsym{)}  \ottsym{)}   \ottsym{+}  \ottsym{\{}  \ottmv{a}  \mapsto  \ottsym{1}  \ottsym{\}} \\
             & = \ottkw{Own} \, \ottsym{(}  \ottnt{H}  \ottsym{,}  \ottnt{R}  \ottsym{,}  \Gamma  \ottsym{)}  \ottsym{+}  \ottsym{\{}  \ottmv{a}  \mapsto  \ottsym{1}  \ottsym{\}}
    \end{align*}

    Since $ \ottmv{a}  \not\in \DOM( \ottnt{H} ) $ and $\forall \,  \ottmv{a'}  \in \DOM( \ottnt{H} )   \ottsym{.}  \ottkw{Own} \, \ottsym{(}  \ottnt{H'}  \ottsym{,}  \ottnt{R}  \ottsym{,}  \Gamma  \ottsym{)}  \ottsym{(}  \ottmv{a'}  \ottsym{)}  \le  \ottsym{1}$,
    we have $\forall \,  \ottmv{a'}  \in \DOM( \ottnt{H'} )   \ottsym{.}  \ottkw{Own} \, \ottsym{(}  \ottnt{H'}  \ottsym{,}  \ottnt{R'}  \ottsym{,}  \Gamma''  \ottsym{)}  \ottsym{(}  \ottmv{a'}  \ottsym{)}  \le  \ottsym{1}$.

  \end{rneqncase} %

  \begin{rneqncase}{R-Assign}{ %
       \vdash_{\mathit{conf} }^D   \tuple{ \ottnt{H} ,  \ottnt{R} ,  \oldvec{F} ,  \ottnt{E}  \ottsym{[}   \mathit{y}  \WRITE  \mathit{x}  \SEQ  \ottnt{e}   \ottsym{]} }  \\
        \tuple{ \ottnt{H} ,  \ottnt{R} ,  \oldvec{F} ,  \ottnt{E}  \ottsym{[}   \mathit{y}  \WRITE  \mathit{x}  \SEQ  \ottnt{e}   \ottsym{]} }     \longrightarrow _{ \ottnt{D} }     \tuple{ \ottnt{H'} ,  \ottnt{R'} ,  \oldvec{F} ,  \ottnt{E}  \ottsym{[}  \ottnt{e}  \ottsym{]} }   \\
      \ottmv{a} \, \ottsym{=} \, \ottnt{R}  \ottsym{(}  \mathit{y}  \ottsym{)} \andalso \ottnt{H'}  \ottsym{=}  \ottnt{H}  \ottsym{\{}  \ottmv{a}  \hookleftarrow  \ottnt{R}  \ottsym{(}  \mathit{x}  \ottsym{)}  \ottsym{\}} \andalso \ottnt{R'}  \ottsym{=}  \ottnt{R}
    }
    By inversion (see the \rn{R-Deref} case) we have that
    \begin{align*}
      &  \Theta   \mid   \oldvec{\ell}   \mid   \Gamma  \ottsym{[}  \mathit{x}  \ottsym{:}  \tau_{{\mathrm{1}}}  \ottsym{+}  \tau_{{\mathrm{2}}}  \ottsym{]}  \ottsym{[}  \mathit{y}  \ottsym{:}   \tau'  \TREF^{ \ottsym{1} }   \ottsym{]}   \vdash    \mathit{y}  \WRITE  \mathit{x}  \SEQ  \ottnt{e}   :  \tau   \produces   \Gamma'  \\
      &  \Theta   \mid   \oldvec{\ell}   \mid   \Gamma  \ottsym{[}  \mathit{x}  \hookleftarrow  \tau_{{\mathrm{1}}}  \ottsym{]}  \ottsym{[}  \mathit{y}  \hookleftarrow   \ottsym{(}   \tau_{{\mathrm{2}}}  \wedge_{ \mathit{y} }   \mathit{y}  =_{ \tau_{{\mathrm{2}}} }  \mathit{x}    \ottsym{)}  \TREF^{ \ottsym{1} }   \ottsym{]}   \vdash   \ottnt{e}  :  \tau   \produces   \Gamma'  \\
      & \mathit{x} \, \neq \, \mathit{y} \andalso  \llparenthesis  \tau' \rrparenthesis   \ottsym{=}   \llparenthesis  \tau_{{\mathrm{2}}} \rrparenthesis  \\
      & \ottkw{Cons} \, \ottsym{(}  \ottnt{H}  \ottsym{,}  \ottnt{R}  \ottsym{,}  \Gamma  \ottsym{)}.
    \end{align*}
    We give $\Gamma''  \ottsym{=}  \Gamma  \ottsym{[}  \mathit{x}  \hookleftarrow  \tau_{{\mathrm{1}}}  \ottsym{]}  \ottsym{[}  \mathit{y}  \hookleftarrow   \ottsym{(}   \tau_{{\mathrm{2}}}  \wedge_{ \mathit{y} }   \mathit{y}  =_{ \tau_{{\mathrm{2}}} }  \mathit{x}    \ottsym{)}  \TREF^{ \ottsym{1} }   \ottsym{]}$. That
    $ \Theta   \mid   \oldvec{\ell}   \mid   \Gamma''   \vdash   \ottnt{e}  :  \tau   \produces   \Gamma' $ is immediate.
    
    We must therefore show that $\ottkw{Cons} \, \ottsym{(}  \ottnt{H'}  \ottsym{,}  \ottnt{R'}  \ottsym{,}  \Gamma''  \ottsym{)}$.
    To show that the output environment is consistent, we must show that $ \ottkw{SATv} ( \ottnt{H'} , \ottnt{R} , \ottnt{R}  \ottsym{(}  \mathit{y}  \ottsym{)} ,  \ottsym{(}   \tau_{{\mathrm{2}}}  \wedge_{ \mathit{y} }   \mathit{y}  =_{ \tau_{{\mathrm{2}}} }  \mathit{x}    \ottsym{)}  \TREF^{ \ottsym{1} }  ) $ and $ \ottkw{SATv} ( \ottnt{H'} , \ottnt{R} , \ottnt{R}  \ottsym{(}  \mathit{x}  \ottsym{)} , \tau_{{\mathrm{1}}} ) $.
    By reasoning similar to that in $\rn{R-Deref}$, it suffices to show that $ \ottkw{SATv} ( \ottnt{H'} , \ottnt{R} , \ottnt{R}  \ottsym{(}  \mathit{y}  \ottsym{)} ,  \tau_{{\mathrm{2}}}  \TREF^{ \ottsym{1} }  ) $ and $ \ottkw{SATv} ( \ottnt{H'} , \ottnt{R} , \ottnt{R}  \ottsym{(}  \mathit{x}  \ottsym{)} , \tau_{{\mathrm{1}}} ) $.

    From $\ottkw{Cons} \, \ottsym{(}  \ottnt{H}  \ottsym{,}  \ottnt{R}  \ottsym{,}  \Gamma  \ottsym{)}$, we know that $\ottkw{SAT} \, \ottsym{(}  \ottnt{H}  \ottsym{,}  \ottnt{R}  \ottsym{,}  \Gamma  \ottsym{)}$, in
    particular $ \ottkw{SATv} ( \ottnt{H} , \ottnt{R} , \ottnt{R}  \ottsym{(}  \mathit{x}  \ottsym{)} , \tau_{{\mathrm{1}}}  \ottsym{+}  \tau_{{\mathrm{2}}} ) $ and
    $ \ottkw{SATv} ( \ottnt{H} , \ottnt{R} , \ottnt{R}  \ottsym{(}  \mathit{y}  \ottsym{)} ,  \tau'  \TREF^{ \ottsym{1} }  ) $.  Since $\mathit{x} \, \neq \, \mathit{y}$,
    \Cref{lem:full-reference-exclusion} gives
    \[
      \ottkw{own} \, \ottsym{(}  \ottnt{H}  \ottsym{,}  \ottnt{R}  \ottsym{(}  \mathit{x}  \ottsym{)}  \ottsym{,}  \tau_{{\mathrm{1}}}  \ottsym{+}  \tau_{{\mathrm{2}}}  \ottsym{)}  \ottsym{(}  \ottmv{a}  \ottsym{)}  \ottsym{=}  \ottsym{0}
      \quad\hbox{and}\quad
      \ottkw{own} \, \ottsym{(}  \ottnt{H}  \ottsym{,}  \ottnt{H}  \ottsym{(}  \ottmv{a}  \ottsym{)}  \ottsym{,}  \tau'  \ottsym{)}  \ottsym{(}  \ottmv{a}  \ottsym{)}  \ottsym{=}  \ottsym{0},
    \]
    and also gives $\ottkw{own} \, \ottsym{(}  \ottnt{H}  \ottsym{,}  \ottnt{R}  \ottsym{(}  \mathit{z}  \ottsym{)}  \ottsym{,}  \Gamma  \ottsym{(}  \mathit{z}  \ottsym{)}  \ottsym{)}  \ottsym{(}  \ottmv{a}  \ottsym{)}  \ottsym{=}  \ottsym{0}$ for every other
    $\mathit{z} \, \neq \, \mathit{y}$.

    By \Cref{lem:satadd,lem:sat-implies-shape-cons}, the old contents
    $\ottnt{H}  \ottsym{(}  \ottmv{a}  \ottsym{)}$ have shape $ \llparenthesis  \tau' \rrparenthesis $ and $\ottnt{R}  \ottsym{(}  \mathit{x}  \ottsym{)}$ has shape $ \llparenthesis  \tau_{{\mathrm{2}}} \rrparenthesis $.
    These shapes are equal by the assignment premise, so
    \Cref{lem:same-shape-update} gives $ \ottnt{H}   \approx _ \ottmv{a}   \ottnt{H'} $.
    We may now apply \Cref{lem:heapfor0} to the binding of $\mathit{x}$, and then
    \Cref{lem:satadd}, to obtain
    $ \ottkw{SATv} ( \ottnt{H'} , \ottnt{R} , \ottnt{R}  \ottsym{(}  \mathit{x}  \ottsym{)} , \tau_{{\mathrm{1}}} ) $ and $ \ottkw{SATv} ( \ottnt{H'} , \ottnt{R} , \ottnt{R}  \ottsym{(}  \mathit{x}  \ottsym{)} , \tau_{{\mathrm{2}}} ) $.
    Since $\ottnt{H'}  \ottsym{(}  \ottnt{R}  \ottsym{(}  \mathit{y}  \ottsym{)}  \ottsym{)} \, \ottsym{=} \, \ottnt{R}  \ottsym{(}  \mathit{x}  \ottsym{)}$, the binding of $\mathit{y}$, including its equality
    refinement, is satisfied.  Every other binding is transported by
    \Cref{lem:heapfor0} using the zero-ownership facts above.  Thus
    $\ottkw{SAT} \, \ottsym{(}  \ottnt{H'}  \ottsym{,}  \ottnt{R}  \ottsym{,}  \Gamma''  \ottsym{)}$ holds.

    It remains to show that the ownership bound is preserved.  We retain the
    decomposition used for the input and output ownership maps, and
    we define $\ottnt{O'_{{\mathrm{0}}}}, \ottnt{O''_{{\mathrm{0}}}}, \ottnt{O'_{{\mathrm{1}}}}$ and $\ottnt{O''_{{\mathrm{1}}}}$ as below:
    \begin{align*}
      \ottkw{Own} \, \ottsym{(}  \ottnt{H}  \ottsym{,}  \ottnt{R}  \ottsym{,}  \Gamma  \ottsym{)} & = \ottnt{O'_{{\mathrm{0}}}}  \ottsym{+}  \ottnt{O''_{{\mathrm{0}}}} \\
      \ottkw{Own} \, \ottsym{(}  \ottnt{H'}  \ottsym{,}  \ottnt{R}  \ottsym{,}  \Gamma''  \ottsym{)} & = \ottnt{O'_{{\mathrm{1}}}}  \ottsym{+}  \ottnt{O''_{{\mathrm{1}}}} \\
      \ottnt{O'_{{\mathrm{0}}}} & =  \Sigma _{ \mathit{z} \in   \DOM( \Gamma )   \setminus  \ottsym{\{}  \mathit{y}  \ottsym{,}  \mathit{x}  \ottsym{\}}  }\, \ottkw{own} \, \ottsym{(}  \ottnt{H}  \ottsym{,}  \ottnt{R}  \ottsym{(}  \mathit{z}  \ottsym{)}  \ottsym{,}  \Gamma  \ottsym{(}  \mathit{z}  \ottsym{)}  \ottsym{)}  \\
      \ottnt{O''_{{\mathrm{0}}}} &= \ottkw{own} \, \ottsym{(}  \ottnt{H}  \ottsym{,}  \ottnt{R}  \ottsym{(}  \mathit{y}  \ottsym{)}  \ottsym{,}  \Gamma  \ottsym{(}  \mathit{y}  \ottsym{)}  \ottsym{)}  \ottsym{+}  \ottkw{own} \, \ottsym{(}  \ottnt{H}  \ottsym{,}  \ottnt{R}  \ottsym{(}  \mathit{x}  \ottsym{)}  \ottsym{,}  \Gamma  \ottsym{(}  \mathit{x}  \ottsym{)}  \ottsym{)} \\
      \ottnt{O'_{{\mathrm{1}}}} & =  \Sigma _{ \mathit{z} \in   \DOM( \Gamma'' )   \setminus  \ottsym{\{}  \mathit{y}  \ottsym{,}  \mathit{x}  \ottsym{\}}  }\, \ottkw{own} \, \ottsym{(}  \ottnt{H'}  \ottsym{,}  \ottnt{R}  \ottsym{(}  \mathit{z}  \ottsym{)}  \ottsym{,}  \Gamma''  \ottsym{(}  \mathit{z}  \ottsym{)}  \ottsym{)}  \\
      \ottnt{O''_{{\mathrm{1}}}} &= \ottkw{own} \, \ottsym{(}  \ottnt{H'}  \ottsym{,}  \ottnt{R}  \ottsym{(}  \mathit{y}  \ottsym{)}  \ottsym{,}  \Gamma''  \ottsym{(}  \mathit{y}  \ottsym{)}  \ottsym{)}  \ottsym{+}  \ottkw{own} \, \ottsym{(}  \ottnt{H'}  \ottsym{,}  \ottnt{R}  \ottsym{(}  \mathit{x}  \ottsym{)}  \ottsym{,}  \Gamma''  \ottsym{(}  \mathit{x}  \ottsym{)}  \ottsym{)}
    \end{align*}
    By the definition of the ownership function, $\Gamma  \ottsym{(}  \mathit{y}  \ottsym{)}  \ottsym{=}   \tau'  \TREF^{ \ottsym{1} } $ and $\Gamma  \ottsym{(}  \mathit{x}  \ottsym{)}  \ottsym{=}  \tau_{{\mathrm{1}}}  \ottsym{+}  \tau_{{\mathrm{2}}}$, we have:
    \begin{align*}
      \ottnt{O''_{{\mathrm{0}}}} & = \ottkw{own} \, \ottsym{(}  \ottnt{H}  \ottsym{,}  \ottnt{H}  \ottsym{(}  \ottnt{R}  \ottsym{(}  \mathit{y}  \ottsym{)}  \ottsym{)}  \ottsym{,}  \tau'  \ottsym{)}  \ottsym{+}  \ottsym{\{}  \ottmv{a}  \mapsto  \ottsym{1}  \ottsym{\}}  \ottsym{+}  \ottkw{own} \, \ottsym{(}  \ottnt{H}  \ottsym{,}  \ottnt{R}  \ottsym{(}  \mathit{x}  \ottsym{)}  \ottsym{,}  \tau_{{\mathrm{1}}}  \ottsym{+}  \tau_{{\mathrm{2}}}  \ottsym{)} \\
      \ottnt{O''_{{\mathrm{1}}}} & = \ottkw{own} \, \ottsym{(}  \ottnt{H'}  \ottsym{,}  \ottnt{H'}  \ottsym{(}  \ottnt{R}  \ottsym{(}  \mathit{y}  \ottsym{)}  \ottsym{)}  \ottsym{,}  \tau_{{\mathrm{2}}}  \ottsym{)}  \ottsym{+}  \ottsym{\{}  \ottmv{a}  \mapsto  \ottsym{1}  \ottsym{\}}  \ottsym{+}  \ottkw{own} \, \ottsym{(}  \ottnt{H'}  \ottsym{,}  \ottnt{R}  \ottsym{(}  \mathit{x}  \ottsym{)}  \ottsym{,}  \tau_{{\mathrm{1}}}  \ottsym{)}
    \end{align*}
    For every $ \mathit{z}  \in    \DOM( \Gamma )   \setminus  \ottsym{\{}  \mathit{y}  \ottsym{,}  \mathit{x}  \ottsym{\}}  $, full-reference exclusion gives zero
    ownership at $\ottmv{a}$, so \Cref{lem:heapop} gives $\ottnt{O'_{{\mathrm{0}}}}  \ottsym{=}  \ottnt{O'_{{\mathrm{1}}}}$.
    Then, from
    $\forall \,  \ottmv{a'}  \in \DOM( \ottnt{H} )   \ottsym{.}  \ottkw{Own} \, \ottsym{(}  \ottnt{H}  \ottsym{,}  \ottnt{R}  \ottsym{,}  \Gamma  \ottsym{)}  \ottsym{(}  \ottmv{a'}  \ottsym{)}  \ottsym{=}  \ottnt{O'_{{\mathrm{0}}}}  \ottsym{(}  \ottmv{a'}  \ottsym{)}  \ottsym{+}  \ottnt{O''_{{\mathrm{0}}}}  \ottsym{(}  \ottmv{a'}  \ottsym{)} \, \le \, \ottsym{1}$,
    it suffices to show that $O_1''(a')\leq O_0''(a')$ for every
    $ \ottmv{a'}  \in \DOM( \ottnt{H'} ) $.
       
    We first consider the case for $\ottmv{a'} \, \ottsym{=} \, \ottmv{a}$:
    \begin{align*}
      \ottnt{O''_{{\mathrm{1}}}}  \ottsym{(}  \ottmv{a}  \ottsym{)} & = \ottkw{own} \, \ottsym{(}  \ottnt{H'}  \ottsym{,}  \ottnt{H'}  \ottsym{(}  \ottnt{R}  \ottsym{(}  \mathit{y}  \ottsym{)}  \ottsym{)}  \ottsym{,}  \tau_{{\mathrm{2}}}  \ottsym{)}  \ottsym{(}  \ottmv{a}  \ottsym{)}  \ottsym{+}  \ottkw{own} \, \ottsym{(}  \ottnt{H'}  \ottsym{,}  \ottnt{R}  \ottsym{(}  \mathit{x}  \ottsym{)}  \ottsym{,}  \tau_{{\mathrm{1}}}  \ottsym{)}  \ottsym{(}  \ottmv{a}  \ottsym{)}  \ottsym{+}  \ottsym{1} \\
      \ottnt{O''_{{\mathrm{0}}}}  \ottsym{(}  \ottmv{a}  \ottsym{)} & = \ottkw{own} \, \ottsym{(}  \ottnt{H}  \ottsym{,}  \ottnt{R}  \ottsym{(}  \mathit{x}  \ottsym{)}  \ottsym{,}  \tau_{{\mathrm{1}}}  \ottsym{+}  \tau_{{\mathrm{2}}}  \ottsym{)}  \ottsym{(}  \ottmv{a}  \ottsym{)}  \ottsym{+}  \ottkw{own} \, \ottsym{(}  \ottnt{H}  \ottsym{,}  \ottnt{H}  \ottsym{(}  \ottnt{R}  \ottsym{(}  \mathit{y}  \ottsym{)}  \ottsym{)}  \ottsym{,}  \tau'  \ottsym{)}  \ottsym{(}  \ottmv{a}  \ottsym{)}  \ottsym{+}  \ottsym{1}.
    \end{align*}
    The two nonconstant terms in $\ottnt{O''_{{\mathrm{0}}}}  \ottsym{(}  \ottmv{a}  \ottsym{)}$ are zero by
    \Cref{lem:full-reference-exclusion}.  By \Cref{lem:heapop,lem:ownadd},
    \[
      \ottkw{own} \, \ottsym{(}  \ottnt{H}  \ottsym{,}  \ottnt{R}  \ottsym{(}  \mathit{x}  \ottsym{)}  \ottsym{,}  \tau_{{\mathrm{1}}}  \ottsym{+}  \tau_{{\mathrm{2}}}  \ottsym{)}  \ottsym{=}  \ottkw{own} \, \ottsym{(}  \ottnt{H'}  \ottsym{,}  \ottnt{R}  \ottsym{(}  \mathit{x}  \ottsym{)}  \ottsym{,}  \tau_{{\mathrm{1}}}  \ottsym{)}  \ottsym{+}  \ottkw{own} \, \ottsym{(}  \ottnt{H'}  \ottsym{,}  \ottnt{R}  \ottsym{(}  \mathit{x}  \ottsym{)}  \ottsym{,}  \tau_{{\mathrm{2}}}  \ottsym{)}.
    \]
    Since this map is zero at $\ottmv{a}$, nonnegativity implies that both
    summands are zero there.  Using $\ottnt{H'}  \ottsym{(}  \ottnt{R}  \ottsym{(}  \mathit{y}  \ottsym{)}  \ottsym{)} \, \ottsym{=} \, \ottnt{R}  \ottsym{(}  \mathit{x}  \ottsym{)}$, we therefore have
    $\ottnt{O''_{{\mathrm{1}}}}  \ottsym{(}  \ottmv{a}  \ottsym{)}  \ottsym{=}  \ottsym{1}  \ottsym{=}  \ottnt{O''_{{\mathrm{0}}}}  \ottsym{(}  \ottmv{a}  \ottsym{)}$.

    Next, consider $\ottmv{a'} \, \neq \, \ottmv{a}$:
    \begin{align*}
      \ottnt{O''_{{\mathrm{1}}}}  \ottsym{(}  \ottmv{a'}  \ottsym{)} & = \ottkw{own} \, \ottsym{(}  \ottnt{H'}  \ottsym{,}  \ottnt{H'}  \ottsym{(}  \ottnt{R}  \ottsym{(}  \mathit{y}  \ottsym{)}  \ottsym{)}  \ottsym{,}  \tau_{{\mathrm{2}}}  \ottsym{)}  \ottsym{(}  \ottmv{a'}  \ottsym{)}  \ottsym{+}  \ottkw{own} \, \ottsym{(}  \ottnt{H'}  \ottsym{,}  \ottnt{R}  \ottsym{(}  \mathit{x}  \ottsym{)}  \ottsym{,}  \tau_{{\mathrm{1}}}  \ottsym{)}  \ottsym{(}  \ottmv{a'}  \ottsym{)} \\
      \ottnt{O''_{{\mathrm{0}}}}  \ottsym{(}  \ottmv{a'}  \ottsym{)} & = \ottkw{own} \, \ottsym{(}  \ottnt{H}  \ottsym{,}  \ottnt{R}  \ottsym{(}  \mathit{x}  \ottsym{)}  \ottsym{,}  \tau_{{\mathrm{1}}}  \ottsym{+}  \tau_{{\mathrm{2}}}  \ottsym{)}  \ottsym{(}  \ottmv{a'}  \ottsym{)}  \ottsym{+}  \ottkw{own} \, \ottsym{(}  \ottnt{H}  \ottsym{,}  \ottnt{H}  \ottsym{(}  \ottnt{R}  \ottsym{(}  \mathit{y}  \ottsym{)}  \ottsym{)}  \ottsym{,}  \tau'  \ottsym{)}  \ottsym{(}  \ottmv{a'}  \ottsym{)}.
    \end{align*}
    Again using \Cref{lem:heapop,lem:ownadd} and
    $\ottnt{H'}  \ottsym{(}  \ottnt{R}  \ottsym{(}  \mathit{y}  \ottsym{)}  \ottsym{)} \, \ottsym{=} \, \ottnt{R}  \ottsym{(}  \mathit{x}  \ottsym{)}$, $\ottnt{O''_{{\mathrm{1}}}}  \ottsym{(}  \ottmv{a'}  \ottsym{)}$ equals
    $\ottkw{own} \, \ottsym{(}  \ottnt{H}  \ottsym{,}  \ottnt{R}  \ottsym{(}  \mathit{x}  \ottsym{)}  \ottsym{,}  \tau_{{\mathrm{1}}}  \ottsym{+}  \tau_{{\mathrm{2}}}  \ottsym{)}  \ottsym{(}  \ottmv{a'}  \ottsym{)} \leq \ottkw{own} \, \ottsym{(}  \ottnt{H}  \ottsym{,}  \ottnt{R}  \ottsym{(}  \mathit{x}  \ottsym{)}  \ottsym{,}  \tau_{{\mathrm{1}}}  \ottsym{+}  \tau_{{\mathrm{2}}}  \ottsym{)}  \ottsym{(}  \ottmv{a'}  \ottsym{)}  \ottsym{+}  \ottkw{own} \, \ottsym{(}  \ottnt{H}  \ottsym{,}  \ottnt{H}  \ottsym{(}  \ottnt{R}  \ottsym{(}  \mathit{y}  \ottsym{)}  \ottsym{)}  \ottsym{,}  \tau'  \ottsym{)}  \ottsym{(}  \ottmv{a'}  \ottsym{)}$.  We therefore conclude that
    $\forall \,  \ottmv{a'}  \in \DOM( \ottnt{H'} )   \ottsym{.}  \ottkw{Own} \, \ottsym{(}  \ottnt{H'}  \ottsym{,}  \ottnt{R}  \ottsym{,}  \Gamma''  \ottsym{)}  \ottsym{(}  \ottmv{a'}  \ottsym{)}  \le  \ottsym{1}$.
  \end{rneqncase} %
  
  \begin{rneqncase}{R-Alias}{
     \vdash_{\mathit{conf} }^D   \tuple{ \ottnt{H} ,  \ottnt{R} ,  \oldvec{F} ,  \ottnt{E}  \ottsym{[}   \ALIAS( \mathit{x}  =  \mathit{y} ) \SEQ  \ottnt{e}   \ottsym{]} }  \\
      \tuple{ \ottnt{H} ,  \ottnt{R} ,  \oldvec{F} ,  \ottnt{E}  \ottsym{[}   \ALIAS( \mathit{x}  =  \mathit{y} ) \SEQ  \ottnt{e}   \ottsym{]} }     \longrightarrow _{ \ottnt{D} }     \tuple{ \ottnt{H} ,  \ottnt{R} ,  \oldvec{F} ,  \ottnt{E}  \ottsym{[}  \ottnt{e}  \ottsym{]} }  \\
    \ottnt{R}  \ottsym{(}  \mathit{x}  \ottsym{)} \, \ottsym{=} \, \ottnt{R}  \ottsym{(}  \mathit{y}  \ottsym{)}
    }
    By inversion (see \rn{R-Deref}) we have for some $\Gamma$ that:
    \begin{align*}
      &  \Theta   \mid   \oldvec{\ell}   \mid   \Gamma  \ottsym{[}  \mathit{x}  \ottsym{:}   \tau_{{\mathrm{1}}}  \TREF^{ r_{{\mathrm{1}}} }   \ottsym{]}  \ottsym{[}  \mathit{y}  \ottsym{:}   \tau_{{\mathrm{2}}}  \TREF^{ r_{{\mathrm{2}}} }   \ottsym{]}   \vdash    \ALIAS( \mathit{x}  =  \mathit{y} ) \SEQ  \ottnt{e}   :  \tau   \produces   \Gamma'  \\
      &  \Theta   \mid   \oldvec{\ell}   \mid   \Gamma  \ottsym{[}  \mathit{x}  \hookleftarrow   \tau'_{{\mathrm{1}}}  \TREF^{ r'_{{\mathrm{1}}} }   \ottsym{]}  \ottsym{[}  \mathit{y}  \hookleftarrow   \tau'_{{\mathrm{2}}}  \TREF^{ r'_{{\mathrm{2}}} }   \ottsym{]}   \vdash   \ottnt{e}  :  \tau   \produces   \Gamma'  \\
      &   \tau_{{\mathrm{1}}}  \TREF^{ r_{{\mathrm{1}}} }   \ottsym{+}  \tau_{{\mathrm{2}}}  \TREF^{ r_{{\mathrm{2}}} }   \approx    \tau'_{{\mathrm{1}}}  \TREF^{ r'_{{\mathrm{1}}} }   \ottsym{+}  \tau'_{{\mathrm{2}}}  \TREF^{ r'_{{\mathrm{2}}} }  \\
      & \mathit{x} \, \neq \, \mathit{y} \\
      & \ottkw{Cons} \, \ottsym{(}  \ottnt{H}  \ottsym{,}  \ottnt{R}  \ottsym{,}  \Gamma  \ottsym{)}
    \end{align*}
    We give $\Gamma''  \ottsym{=}  \Gamma  \ottsym{[}  \mathit{x}  \hookleftarrow   \tau'_{{\mathrm{1}}}  \TREF^{ r'_{{\mathrm{1}}} }   \ottsym{]}  \ottsym{[}  \mathit{y}  \hookleftarrow   \tau'_{{\mathrm{2}}}  \TREF^{ r'_{{\mathrm{2}}} }   \ottsym{]}$, and must show
    $ \Theta   \mid   \oldvec{\ell}   \mid   \Gamma''   \vdash   \ottnt{e}  :  \tau   \produces   \Gamma' $ and $\ottkw{Cons} \, \ottsym{(}  \ottnt{H}  \ottsym{,}  \ottnt{R}  \ottsym{,}  \Gamma''  \ottsym{)}$.
    The first is immediate.
    
    To show $\ottkw{Cons} \, \ottsym{(}  \ottnt{H}  \ottsym{,}  \ottnt{R}  \ottsym{,}  \Gamma''  \ottsym{)}$ we first define:
    \begin{align*}
      \tau_{\ottmv{p}\,{\mathrm{1}}} = &  \tau_{{\mathrm{1}}}  \TREF^{ r_{{\mathrm{1}}} }  \\
      \tau_{\ottmv{p}\,{\mathrm{2}}} = &  \tau_{{\mathrm{2}}}  \TREF^{ r_{{\mathrm{2}}} }  \\
      \tau_{\ottmv{q}\,{\mathrm{1}}} = &  \tau'_{{\mathrm{1}}}  \TREF^{ r'_{{\mathrm{1}}} }  \\
      \tau_{\ottmv{q}\,{\mathrm{2}}} = &  \tau'_{{\mathrm{2}}}  \TREF^{ r'_{{\mathrm{2}}} }  \\
      \tau_{\ottmv{q}} = & \tau_{\ottmv{q}\,{\mathrm{1}}}  \ottsym{+}  \tau_{\ottmv{q}\,{\mathrm{2}}} \\
      \tau_{\ottmv{p}} = & \tau_{\ottmv{p}\,{\mathrm{1}}}  \ottsym{+}  \tau_{\ottmv{p}\,{\mathrm{2}}}
    \end{align*}
    We thus have $\tau_{\ottmv{q}}  \approx  \tau_{\ottmv{p}}$.
    
    We know that $\ottkw{Cons} \, \ottsym{(}  \ottnt{H}  \ottsym{,}  \ottnt{R}  \ottsym{,}  \Gamma  \ottsym{)}$, from which we have $\ottkw{SAT} \, \ottsym{(}  \ottnt{H}  \ottsym{,}  \ottnt{R}  \ottsym{,}  \Gamma  \ottsym{)}$, in particular
    $ \ottkw{SATv} ( \ottnt{H} , \ottnt{R} , \ottnt{R}  \ottsym{(}  \mathit{y}  \ottsym{)} ,  \tau_{{\mathrm{2}}}  \TREF^{ r_{{\mathrm{2}}} }  ) $ and $ \ottkw{SATv} ( \ottnt{H} , \ottnt{R} , \ottnt{R}  \ottsym{(}  \mathit{x}  \ottsym{)} ,  \tau_{{\mathrm{1}}}  \TREF^{ r_{{\mathrm{1}}} }  ) $.
    From $\tau_{\ottmv{p}\,{\mathrm{1}}}  \ottsym{+}  \tau_{\ottmv{p}\,{\mathrm{2}}}  \ottsym{=}  \tau_{\ottmv{p}}$ and \Cref{lem:satadd}, we have
    $ \ottkw{SATv} ( \ottnt{H} , \ottnt{R} , \ottnt{v} , \tau_{\ottmv{p}\,{\mathrm{1}}} ) $ and $ \ottkw{SATv} ( \ottnt{H} , \ottnt{R} , \ottnt{v} , \tau_{\ottmv{p}\,{\mathrm{2}}} ) $ imply $ \ottkw{SATv} ( \ottnt{H} , \ottnt{R} , \ottnt{v} , \tau_{\ottmv{p}} ) $, where $\ottnt{v}  \ottsym{=}  \ottnt{R}  \ottsym{(}  \mathit{y}  \ottsym{)}  \ottsym{=}  \ottnt{R}  \ottsym{(}  \mathit{x}  \ottsym{)}$.
    From $\tau_{\ottmv{q}}  \approx  \tau_{\ottmv{p}}$ and \Cref{lem:sattosat}, we have that $ \ottkw{SATv} ( \ottnt{H} , \ottnt{R} , \ottnt{v} , \tau_{\ottmv{p}} ) $ implies $ \ottkw{SATv} ( \ottnt{H} , \ottnt{R} , \ottnt{v} , \tau_{\ottmv{q}} ) $.
    From \Cref{lem:satadd} we also have that $ \ottkw{SATv} ( \ottnt{H} , \ottnt{R} , \ottnt{v} , \tau_{\ottmv{q}} ) $ implies $ \ottkw{SATv} ( \ottnt{H} , \ottnt{R} , \ottnt{v} , \tau_{\ottmv{q}\,{\mathrm{1}}} ) $ and $ \ottkw{SATv} ( \ottnt{H} , \ottnt{R} , \ottnt{v} , \tau_{\ottmv{q}\,{\mathrm{2}}} ) $, where
    again $\ottnt{v}  \ottsym{=}  \ottnt{R}  \ottsym{(}  \mathit{y}  \ottsym{)}  \ottsym{=}  \ottnt{R}  \ottsym{(}  \mathit{x}  \ottsym{)}$.

    Then from the reasoning above, the refinements of $\tau_{\ottmv{q}\,{\mathrm{1}}}$ and $\tau_{\ottmv{q}\,{\mathrm{2}}}$ are valid and $\ottkw{SAT} \, \ottsym{(}  \ottnt{H}  \ottsym{,}  \ottnt{R}  \ottsym{,}  \Gamma''  \ottsym{)}$ holds.

    We now show $\forall \,  \ottmv{a}  \in \DOM( \ottnt{H} )   \ottsym{.}  \ottkw{Own} \, \ottsym{(}  \ottnt{H}  \ottsym{,}  \ottnt{R}  \ottsym{,}  \Gamma''  \ottsym{)}  \ottsym{(}  \ottmv{a}  \ottsym{)}  \le  \ottsym{1}$.
    To prove that $\ottkw{Own} \, \ottsym{(}  \ottnt{H}  \ottsym{,}  \ottnt{R}  \ottsym{,}  \Gamma  \ottsym{)}  \ottsym{=}  \ottkw{Own} \, \ottsym{(}  \ottnt{H}  \ottsym{,}  \ottnt{R}  \ottsym{,}  \Gamma''  \ottsym{)}$ follows from:
    \begin{align*}
      &\ottkw{own} \, \ottsym{(}  \ottnt{H}  \ottsym{,}  \ottnt{R}  \ottsym{(}  \mathit{x}  \ottsym{)}  \ottsym{,}   \tau_{{\mathrm{1}}}  \TREF^{ r_{{\mathrm{1}}} }   \ottsym{)}  \ottsym{+}  \ottkw{own} \, \ottsym{(}  \ottnt{H}  \ottsym{,}  \ottnt{R}  \ottsym{(}  \mathit{y}  \ottsym{)}  \ottsym{,}   \tau_{{\mathrm{2}}}  \TREF^{ r_{{\mathrm{2}}} }   \ottsym{)} =  \\
      &\,\,\,\,\ottkw{own} \, \ottsym{(}  \ottnt{H}  \ottsym{,}  \ottnt{R}  \ottsym{(}  \mathit{x}  \ottsym{)}  \ottsym{,}   \tau'_{{\mathrm{1}}}  \TREF^{ r'_{{\mathrm{1}}} }   \ottsym{)}  \ottsym{+}  \ottkw{own} \, \ottsym{(}  \ottnt{H}  \ottsym{,}  \ottnt{R}  \ottsym{(}  \mathit{y}  \ottsym{)}  \ottsym{,}   \tau'_{{\mathrm{2}}}  \TREF^{ r'_{{\mathrm{2}}} }   \ottsym{)}
    \end{align*}
    which follows immediately from the conditions $  \tau_{{\mathrm{1}}}  \TREF^{ r_{{\mathrm{1}}} }   \ottsym{+}  \tau_{{\mathrm{2}}}  \TREF^{ r_{{\mathrm{2}}} }   \approx    \tau'_{{\mathrm{1}}}  \TREF^{ r'_{{\mathrm{1}}} }   \ottsym{+}  \tau'_{{\mathrm{2}}}  \TREF^{ r'_{{\mathrm{2}}} } $, $\ottnt{R}  \ottsym{(}  \mathit{x}  \ottsym{)} \, \ottsym{=} \, \ottnt{R}  \ottsym{(}  \mathit{y}  \ottsym{)}$, and \Cref{lem:ownadd,lem:ownequiv-preserv}.
  \end{rneqncase} %

  \begin{rncase}{R-AliasPtr}
    By reasoning similar to the \rn{R-Alias} case.
  \end{rncase}

  \begin{rncase}{R-AliasFail,R-AliasPtrFail} %
    The result configuration $ \mathbf{AliasFail} $ is trivially well-typed.
  \end{rncase}
  
  \begin{rneqncase}{R-Assert}{ %
     \vdash_{\mathit{conf} }^D   \tuple{ \ottnt{H} ,  \ottnt{R} ,  \oldvec{F} ,  \ottnt{E}  \ottsym{[}   \ASSERT( \varphi ) \SEQ  \ottnt{e}   \ottsym{]} }  \andalso \Gamma  \models  \ottsym{[}  \ottnt{R}  \ottsym{]} \, \varphi\\
      \tuple{ \ottnt{H} ,  \ottnt{R} ,  \oldvec{F} ,  \ottnt{E}  \ottsym{[}   \ASSERT( \varphi ) \SEQ  \ottnt{e}   \ottsym{]} }     \longrightarrow _{ \ottnt{D} }     \tuple{ \ottnt{H} ,  \ottnt{R} ,  \oldvec{F} ,  \ottnt{E}  \ottsym{[}  \ottnt{e}  \ottsym{]} }  \\
    }
    By inversion (see \rn{R-Deref}) we can obtain $ \Theta   \mid   \oldvec{\ell}   \mid   \Gamma   \vdash    \ASSERT( \varphi ) \SEQ  \ottnt{e}   :  \tau   \produces   \Gamma' 
$ and $ \Theta   \mid   \oldvec{\ell}   \mid   \Gamma   \vdash   \ottnt{e}  :  \tau   \produces   \Gamma' $, and
    the result follows immediately by taking $\Gamma''  \ottsym{=}  \Gamma$.
  \end{rneqncase} %
  
  \begin{rneqncase}{R-AssertFail}{ %
       \vdash_{\mathit{conf} }^D   \tuple{ \ottnt{H} ,  \ottnt{R} ,  \oldvec{F} ,  \ottnt{E}  \ottsym{[}   \ASSERT( \varphi ) \SEQ  \ottnt{e}   \ottsym{]} }   \\
        \tuple{ \ottnt{H} ,  \ottnt{R} ,  \oldvec{F} ,  \ottnt{E}  \ottsym{[}   \ASSERT( \varphi ) \SEQ  \ottnt{e}   \ottsym{]} }     \longrightarrow _{ \ottnt{D} }     \mathbf{AssertFail}  \\
       \Theta   \mid   \oldvec{\ell}   \mid   \Gamma   \vdash    \ASSERT( \varphi ) \SEQ  \ottnt{e}   :  \tau   \produces   \Gamma' 
    }
    By inversion (see the \rn{R-Deref} case) we have that
    $\Gamma  \models  \varphi$, i.e., $\models   \sem{ \Gamma }   \implies  \varphi$,
    for some $\Gamma$ such that $\ottkw{Cons} \, \ottsym{(}  \ottnt{H}  \ottsym{,}  \ottnt{R}  \ottsym{,}  \Gamma  \ottsym{)}$.
    From \Cref{lem:sat-implies-gamma} we therefore have $\models  \ottsym{[}  \ottnt{R}  \ottsym{]} \,  \sem{ \Gamma } $.
    From the precondition of \rn{R-AssertFail} we have that
    $\not\models  \ottsym{[}  \ottnt{R}  \ottsym{]} \, \varphi$. But from $\models   \sem{ \Gamma }   \implies  \varphi$ and $\models  \ottsym{[}  \ottnt{R}  \ottsym{]} \,  \sem{ \Gamma } $
    we can conclude that $\models  \ottsym{[}  \ottnt{R}  \ottsym{]} \, \varphi$, yielding a contradiction.
    We therefore conclude that this case is impossible.
  \end{rneqncase} %
  
  \begin{rneqncase}{R-Call}{
       \vdash_{\mathit{conf} }^D   \tuple{ \ottnt{H} ,  \ottnt{R} ,  \oldvec{F} ,  \ottnt{E}  \ottsym{[}   \LET  \mathit{x}  =   \mathit{f} ^ \ell (  \mathit{y_{{\mathrm{1}}}} ,\ldots, \mathit{y_{\ottmv{n}}}  )   \IN  \ottnt{e'}   \ottsym{]} }  \\
       \mathit{f}  \mapsto  \ottsym{(}  \mathit{x_{{\mathrm{1}}}}  \ottsym{,} \, .. \, \ottsym{,}  \mathit{x_{\ottmv{n}}}  \ottsym{)}  \ottnt{e}  \in  \ottnt{D}  \\
       \begin{array}{l}  \tuple{ \ottnt{H} ,  \ottnt{R} ,  \oldvec{F} ,  \ottnt{E}  \ottsym{[}   \LET  \mathit{x}  =   \mathit{f} ^ \ell (  \mathit{y_{{\mathrm{1}}}} ,\ldots, \mathit{y_{\ottmv{n}}}  )   \IN  \ottnt{e'}   \ottsym{]} }   \\ \quad   \longrightarrow _{ \ottnt{D} }     \tuple{ \ottnt{H} ,  \ottnt{R} ,   \ottnt{E} [\LET  \mathit{x}  =   \HOLE^ \ell   \IN  \ottnt{e'}  ]   \ottsym{:}  \oldvec{F} ,     [  \mathit{y_{{\mathrm{1}}}}  /  \mathit{x_{{\mathrm{1}}}}  ]  \cdots  [  \mathit{y_{\ottmv{n}}}  /  \mathit{x_{\ottmv{n}}}  ]     \ottnt{e}  }  \end{array}  \\
    }
    We must show that
    $ \vdash_{\mathit{conf} }^D   \tuple{ \ottnt{H} ,  \ottnt{R} ,   \ottnt{E} [\LET  \mathit{x}  =   \HOLE^ \ell   \IN  \ottnt{e'}  ]   \ottsym{:}  \oldvec{F} ,     [  \mathit{y_{{\mathrm{1}}}}  /  \mathit{x_{{\mathrm{1}}}}  ]  \cdots  [  \mathit{y_{\ottmv{n}}}  /  \mathit{x_{\ottmv{n}}}  ]     \ottnt{e}  }  $
    for some $\Gamma''$.
    
    By inversion on the configuration typing, we have that, for some $\Gamma$:
    \[
       \Theta   \mid   \oldvec{\ell}   \mid   \Gamma   \vdash   \ottnt{E}  \ottsym{[}   \LET  \mathit{x}  =   \mathit{f} ^ \ell (  \mathit{y_{{\mathrm{1}}}} ,\ldots, \mathit{y_{\ottmv{n}}}  )   \IN  \ottnt{e'}   \ottsym{]}  :  \tau_{\ottmv{n}}   \produces   \Gamma_{\ottmv{n}}  .
    \]
    By \Cref{lem:stack_var}, we then have for some $\tau$, and $\Gamma'$ that:
    \begin{align*}
      &  \Theta   \mid   \oldvec{\ell}   \mid   \Gamma   \vdash    \LET  \mathit{x}  =   \mathit{f} ^ \ell (  \mathit{y_{{\mathrm{1}}}} ,\ldots, \mathit{y_{\ottmv{n}}}  )   \IN  \ottnt{e'}   :  \tau   \produces   \Gamma'  \\
      & \Theta  \mid  \HOLE  \ottsym{:}  \tau  \produces  \Gamma'  \mid  \oldvec{\ell}  \vdash_{\mathit{ectx} }  \ottnt{E}  \ottsym{:}  \tau_{\ottmv{n}}  \produces  \Gamma_{\ottmv{n}}
    \end{align*}
    Taking $\tau_{{\mathrm{1}}}  \ottsym{=}  \tau, \Gamma_{{\mathrm{1}}}  \ottsym{=}  \Gamma', \Gamma_{{\mathrm{0}}}  \ottsym{=}  \Gamma, \Gamma_{{\mathrm{2}}}  \ottsym{=}  \Gamma_{\ottmv{n}}, \tau_{{\mathrm{2}}}  \ottsym{=}  \tau_{\ottmv{n}}$,
    by \Cref{lem:callfunc} we have, for some $\tau''', \Gamma'''$:
    \begin{align*}
      &  \Theta   \mid   \ell  \ottsym{:}  \oldvec{\ell}   \mid   \Gamma   \vdash    \sigma_{x}   \ottnt{e}   :  \tau'''   \produces   \Gamma'''  \\
      & \Theta  \mid  \HOLE  \ottsym{:}  \tau'''  \produces  \Gamma'''  \mid  \oldvec{\ell}  \vdash_{\mathit{ectx} }   \ottnt{E} [\LET  \mathit{x}  =   \HOLE^ \ell   \IN  \ottnt{e'}  ]   \ottsym{:}  \tau_{\ottmv{n}}  \produces  \Gamma_{\ottmv{n}}
    \end{align*}
    where:
    \begin{align*}
      \sigma_{x} & =   [  \mathit{y_{{\mathrm{1}}}}  /  \mathit{x_{{\mathrm{1}}}}  ]  \cdots  [  \mathit{y_{\ottmv{n}}}  /  \mathit{x_{\ottmv{n}}}  ]   \\
      \Theta  \ottsym{(}  \mathit{f}  \ottsym{)} & =  \forall  \lambda .\tuple{ \mathit{x_{{\mathrm{1}}}} \COL \tau_{\ottmv{i}} ,\dots, \mathit{x_{\ottmv{n}}} \COL \tau_{\ottmv{n}} }\ra\tuple{ \mathit{x_{{\mathrm{1}}}} \COL \tau'_{{\mathrm{1}}} ,\dots, \mathit{x_{\ottmv{n}}} \COL \tau'_{\ottmv{n}}  \mid  \tau_{\ottmv{p}} } 
    \end{align*}
    We therefore take $\Gamma''  \ottsym{=}  \Gamma$.
     
    We must also prove that $\forall i\in\set{1..n+1}.\Theta  \mid  \HOLE  \ottsym{:}  \tau_{\ottmv{i}}  \produces  \Gamma_{\ottmv{i}}  \mid  \oldvec{\ell}_{{\ottmv{i}-1}}  \vdash_{\mathit{ectx} }  \ottnt{E'_{{\ottmv{i}-1}}}  \ottsym{:}  \tau_{{\ottmv{i}-1}}  \produces  \Gamma_{{\ottmv{i}-1}}$  where $\ottnt{E'_{\ottmv{n}}}  \ottsym{=}  \ottnt{E}  \ottsym{[}   \LET  \mathit{x}  =   \HOLE^ \ell   \IN  \ottnt{e'}   \ottsym{]}$ and $\ottnt{E'_{\ottmv{i}}}  \ottsym{=}  \ottnt{E_{\ottmv{i}}} (0 \leq i < n)$,
    which can be divided into proving
    $\forall i\in\set{1..n}.\Theta  \mid  \HOLE  \ottsym{:}  \tau_{\ottmv{i}}  \produces  \Gamma_{\ottmv{i}}  \mid  \oldvec{\ell}_{{\ottmv{i}-1}}  \vdash_{\mathit{ectx} }  \ottnt{E'_{{\ottmv{i}-1}}}  \ottsym{:}  \tau_{{\ottmv{i}-1}}  \produces  \Gamma_{{\ottmv{i}-1}}$ and $\Theta  \mid  \HOLE  \ottsym{:}   \tau_{n+1}   \produces   \tenv_{n+1}   \mid  \oldvec{\ell}_{\ottmv{n}}  \vdash_{\mathit{ectx} }  \ottnt{E'_{\ottmv{n}}}  \ottsym{:}  \tau_{\ottmv{n}}  \produces  \Gamma_{\ottmv{n}}$. The first follows by inversion on $ \vdash_{\mathit{conf} }^D   \tuple{ \ottnt{H} ,  \ottnt{R} ,  \oldvec{F} ,  \ottnt{E}  \ottsym{[}   \LET  \mathit{x}  =   \mathit{f} ^ \ell (  \mathit{y_{{\mathrm{1}}}} ,\ldots, \mathit{y_{\ottmv{n}}}  )   \IN  \ottnt{e'}   \ottsym{]} }  $.
    To show the latter, we define $ \tenv_{n+1}   \ottsym{=}  \Gamma'''$ and $ \tau_{n+1}   \ottsym{=}  \tau'''$, whereby the
    well-typing holds from the result of applying \Cref{lem:callfunc} above.
    
    Finally, $\ottkw{Cons} \, \ottsym{(}  \ottnt{H}  \ottsym{,}  \ottnt{R}  \ottsym{,}  \Gamma''  \ottsym{)}$ follows immediately from $\ottkw{Cons} \, \ottsym{(}  \ottnt{H}  \ottsym{,}  \ottnt{R}  \ottsym{,}  \Gamma  \ottsym{)}$ and $\Gamma''  \ottsym{=}  \Gamma$.
  \end{rneqncase}
\end{proof}

\section{Proof of Progress}
\label{sec:progress-proof}

We first state the standard decomposition lemma.
\begin{lemma}[Decomposition]
  \label{lem:decomposition}
  For any term $\ottnt{e}$, either $\ottnt{e}  \ottsym{=}  \mathit{x}$ for some $\mathit{x}$ or there exists some $\ottnt{E}$ and $\ottnt{e'}$ where $\ottnt{E}  \ottsym{[}  \ottnt{e'}  \ottsym{]}  \ottsym{=}  \ottnt{e}$ and one of the following
  cases hold:
  \begin{enumerate}
  \item $\ottnt{e'}  \ottsym{=}   \LET  \mathit{x}  =   \MKREF  \mathit{y}   \IN  \ottnt{e''} $
  \item $\ottnt{e'}  \ottsym{=}   \LET  \mathit{x}  =  \mathit{y}  \IN  \ottnt{e''} $
  \item $\ottnt{e'}  \ottsym{=}   \LET  \mathit{x}  =  n  \IN  \ottnt{e''} $
  \item $\ottnt{e'}  \ottsym{=}   \LET  \mathit{x}  =   *  \mathit{y}   \IN  \ottnt{e''} $
  \item $\ottnt{e'}  \ottsym{=}   \LET  \mathit{x}  =   \mathit{f} ^ \ell (  \mathit{y_{{\mathrm{1}}}} ,\ldots, \mathit{y_{\ottmv{n}}}  )   \IN  \ottnt{e''} $
  \item $ \ottnt{e'}  \ottsym{=}  \mathit{x}  \SEQ  \ottnt{e''} $
  \item $\ottnt{e'}  \ottsym{=}   \ALIAS( \mathit{x}  =  \mathit{y} ) \SEQ  \ottnt{e''} $
  \item $\ottnt{e'}  \ottsym{=}   \ALIAS( \mathit{x}  = *  \mathit{y} ) \SEQ  \ottnt{e''} $
  \item $\ottnt{e'}  \ottsym{=}   \IFZERO  \mathit{x}  \THEN  \ottnt{e_{{\mathrm{1}}}}  \ELSE  \ottnt{e_{{\mathrm{2}}}} $
  \item $\ottnt{e'}  \ottsym{=}   \ASSERT( \varphi ) \SEQ  \ottnt{e''} $
  \item $\ottnt{e'}  \ottsym{=}   \mathit{x}  \WRITE  \mathit{y}  \SEQ  \ottnt{e} $
  \end{enumerate}
\end{lemma}
\begin{proof}
  Straightforward induction on $\ottnt{e}$.
\end{proof}

\begin{proof}[Progress; \Cref{lem:progress}]
  By inversion on $ \vdash_{\mathit{conf} }^D  \mathbf{C} $, either $\mathbf{C}  \ottsym{=}   \mathbf{AliasFail} $ or $\mathbf{C}  \ottsym{=}   \tuple{ \ottnt{H} ,  \ottnt{R} ,  \oldvec{F} ,  \ottnt{e} } $.
  In the former case the result is immediate. In the latter case we have
  that $ \Theta   \mid   \oldvec{\ell}   \mid   \Gamma   \vdash   \ottnt{e}  :  \tau   \produces   \Gamma' $ for some $\tau, \Gamma$ and $\Gamma'$,
  and further from   \Cref{lem:decomposition}, we have that either $\ottnt{e}  \ottsym{=}  \mathit{x}$ for some $\mathit{x}$ or there exists
  some $\ottnt{E}$ or $\ottnt{e'}$ where $\ottnt{e}  \ottsym{=}  \ottnt{E}  \ottsym{[}  \ottnt{e'}  \ottsym{]}$ and $\ottnt{e'}$ meets one of the cases in
  \Cref{lem:decomposition}.

  In the case $\ottnt{e}  \ottsym{=}  \mathit{x}$, we further make case analysis on the form of $\oldvec{F}$.
  The case where $\oldvec{F}  \ottsym{=}   \epsilon $ is immediate;
  In the other case where $\oldvec{F}  \ottsym{=}  F  \ottsym{:}  \oldvec{F}'$, the configuration can step
  to $ \tuple{ \ottnt{H} ,  \ottnt{R} ,  \oldvec{F}' ,  F  \ottsym{[}  \mathit{x}  \ottsym{]} } $ according to \rn{R-Var}.
  
  For the remaining cases where $\ottnt{e}  \ottsym{=}  \ottnt{E}  \ottsym{[}  \ottnt{e'}  \ottsym{]}$,
  by the well-typing of $\ottnt{e}$ with respect to $\Gamma$ and \Cref{lem:stack_var},
  we have that $ \Theta   \mid   \oldvec{\ell}   \mid   \Gamma   \vdash   \ottnt{e'}  :  \tau_{{\mathrm{0}}}   \produces   \Gamma_{{\mathrm{0}}} $ some $\tau_{{\mathrm{0}}}$ and $\Gamma_{{\mathrm{0}}}$.

  We now treat the remaining forms of $\ottnt{e'}$
  \begin{eqncase}{
      \ottnt{e'}  \ottsym{=}   \LET  \mathit{x}  =   *  \mathit{y}   \IN  \ottnt{e''} 
    }
    By inversion (\Cref{lem:inversion}) and \Cref{lem:subtyp-preserves-cons} we must have
    that for some $\Gamma_{\ottmv{p}}$ where $\ottkw{Cons} \, \ottsym{(}  \ottnt{H}  \ottsym{,}  \ottnt{R}  \ottsym{,}  \Gamma_{\ottmv{p}}  \ottsym{)}$ that $ \mathit{y}  \in   \DOM( \Gamma_{\ottmv{p}} )  $ and
    $\Gamma_{\ottmv{p}}  \ottsym{(}  \mathit{y}  \ottsym{)}  \ottsym{=}   \tau'  \TREF^{ r } $. From $\ottkw{Cons} \, \ottsym{(}  \ottnt{H}  \ottsym{,}  \ottnt{R}  \ottsym{,}  \Gamma_{\ottmv{p}}  \ottsym{)}$ we must have $ \mathit{y}  \in   \DOM( \ottnt{R} )  $
    and further $ \ottkw{SATv} ( \ottnt{H} , \ottnt{R} , \ottnt{R}  \ottsym{(}  \mathit{y}  \ottsym{)} ,  \tau'  \TREF^{ r' }  ) $, from which we must have
    $\ottnt{R}  \ottsym{(}  \mathit{y}  \ottsym{)} \, \ottsym{=} \, \ottmv{a}$ and  $ \ottmv{a}  \in   \DOM( \ottnt{H} )  $. Then $\mathbf{C}$ can step according to \rn{R-Deref}.
  \end{eqncase}
  
  \begin{eqncase}{\ottnt{e'}  \ottsym{=}   \LET  \mathit{x}  =  \mathit{y}  \IN  \ottnt{e''} }
    Again, by \Cref{lem:inversion,lem:subtyp-preserves-cons} and the definition
    of $\ottkw{Cons}$, we must have that $ \mathit{y}  \in   \DOM( \ottnt{R} )  $, and the system can step
    according to \rn{R-LetVar}.
  \end{eqncase}

  \begin{eqncase}{\ottnt{e'}  \ottsym{=}   \LET  \mathit{x}  =   \MKREF  \mathit{y}   \IN  \ottnt{e''} }
    Similar to the \rn{R-LetVar} case above.
  \end{eqncase}

  \begin{eqncase}{
      \ottnt{e'}  \ottsym{=}   \LET  \mathit{x}  =  n  \IN  \ottnt{e''}  \\
       \ottnt{e'}  \ottsym{=}  \mathit{x}  \SEQ  \ottnt{e''}  \\
      \ottnt{e'}  \ottsym{=}   \ASSERT( \varphi ) \SEQ  \ottnt{e''} 
    }
    The first two can trivially step according to \rn{R-LetInt} and \rn{R-Seq}, respectively.
    The last can step according to \rn{R-Assert} or \rn{R-AssertFalse} (although
    by \Cref{lem:preservation,lem:assertfail} the latter is impossible).
  \end{eqncase}
  
  \begin{eqncase}{
      \ottnt{e'}  \ottsym{=}   \ALIAS( \mathit{x}  =  \mathit{y} ) \SEQ  \ottnt{e''} 
    }
    Again by \Cref{lem:inversion,lem:subtyp-preserves-cons} and that
    $\ottkw{Cons} \, \ottsym{(}  \ottnt{H}  \ottsym{,}  \ottnt{R}  \ottsym{,}  \Gamma_{\ottmv{p}}  \ottsym{)}$ implies $\mathit{x}$ and $\mathit{y}$ are bound to addresses
    in the register file, we have that the configuration can step according to \rn{R-Alias}
    or \rn{R-AliasFail}.
  \end{eqncase}

  \begin{eqncase}{
      \ottnt{e'}  \ottsym{=}   \ALIAS( \mathit{x}  = *  \mathit{y} ) \SEQ  \ottnt{e''}  
    }
    Similar to the case above, we must have that $\mathit{x}$ is bound to an address
    in the register file, and that $\mathit{y}$ is bound to an address that
    is itself mapped to an address in the heap $\ottnt{H}$. Then the configuration may
    step according to \rn{R-AliasPtr} or \rn{R-AliasPtrFail}
  \end{eqncase}

  \begin{eqncase}{
      \ottnt{e}  \ottsym{=}   \IFZERO  \mathit{x}  \THEN  \ottnt{e_{{\mathrm{1}}}}  \ELSE  \ottnt{e_{{\mathrm{2}}}} 
    }
    As above, from the well-typing we must have that $\mathit{x}$ is bound in $\ottnt{R}$
    to some integer $n$. Then the configuration may step according to \rn{R-IfTrue}
    or \rn{R-IfFalse} depending on whether
    $n \, \ottsym{=} \, \ottsym{0}$ or $n \, \neq \, \ottsym{0}$.
  \end{eqncase}
  
  \begin{eqncase}{
      \ottnt{e'}  \ottsym{=}   \mathit{x}  \WRITE  \mathit{y}  \SEQ  \ottnt{e''} 
    }
    From the well-typing of $\ottnt{e'}$,
    \Cref{lem:inversion,lem:subtyp-preserves-cons} and the definition
    of $\ottkw{Cons}$, we must have that $ \mathit{y}  \in   \DOM( \ottnt{R} )  $,
    $ \mathit{x}  \in   \DOM( \ottnt{R} )  $, $\ottnt{R}  \ottsym{(}  \mathit{x}  \ottsym{)} \, \ottsym{=} \, \ottmv{a}$, and $ \ottmv{a}  \in   \DOM( \ottnt{H} )  $. Then
    the configuration can step according \rn{R-Assign}.
  \end{eqncase}

  \begin{eqncase}{
      \ottnt{e'}  \ottsym{=}   \LET  \mathit{x}  =   \mathit{f} ^ \ell (  \mathit{y_{{\mathrm{1}}}} ,\ldots, \mathit{y_{\ottmv{n}}}  )   \IN  \ottnt{e''} 
    }
    From the well-typing of the function call we must have that $ \mathit{f}  \in   \DOM( \Theta )  $.
    From $\Theta  \vdash  \ottnt{D}$ in the precondition of $ \vdash_{\mathit{conf} }^D  \mathbf{C} $, we must have that
    $ \mathit{f}  \mapsto  \ottsym{(}  \mathit{x_{{\mathrm{1}}}}  \ottsym{,} \, .. \, \ottsym{,}  \mathit{x_{\ottmv{j}}}  \ottsym{)}  \ottnt{e'''}  \in  \ottnt{D} $. Then from \rn{T-FunDef} we must have that
    $ j = n $ whereby the configuration can step according to \rn{R-Call}.
  \end{eqncase}
\end{proof}

\fi

\end{document}